\documentclass[runningheads]{llncs}
\usepackage{booktabs}
\usepackage{amssymb}
\usepackage{stmaryrd}
\usepackage{bbm}
\usepackage{wrapfig}
\usepackage[greek,english]{babel}
\usepackage{ucs}
\usepackage[utf8]{inputenc}
\usepackage{autofe}
\usepackage[references]{agda}
\usepackage{newunicodechar}
\usepackage[nocenter]{qtree}
\usepackage{multicol}
\usepackage{bbold}
\usepackage{tikz}
\usetikzlibrary{cd}
\usetikzlibrary{quotes}
\usepackage{tikzit}
\usepackage{microtype}

\newunicodechar{⊸}{$\multimap$}
\newunicodechar{Π}{$\Pi$}
\newunicodechar{≟}{$\mathbin{\stackrel{?}{=}}$}
\newunicodechar{𝕌}{$\mathbb{U}$}
\newunicodechar{𝔹}{$\mathbb{B}$}
\newunicodechar{𝔽}{$\mathbb{F}$}
\newunicodechar{𝕋}{$\mathbb{T}$}
\newunicodechar{ℕ}{$\mathbb{N}$}
\newunicodechar{𝟘}{$\mathbb{0}$}
\newunicodechar{𝟙}{$\mathbb{1}$}
\newunicodechar{𝟚}{$\mathbb{2}$}
\newunicodechar{⋆}{${}_*$}
\newunicodechar{○}{$\gcv$}
\newunicodechar{ᵤ}{${}_u$}
\newunicodechar{⊚}{$\fatsemi$}
\newunicodechar{⟷}{$\leftrightarrow$}
\newunicodechar{⟶}{$\longrightarrow$}
\newunicodechar{⟪}{$\ll$}
\newunicodechar{⟫}{$\gg$}
\newunicodechar{⊎}{$\uplus$}
\newunicodechar{ₗ}{$_l$}
\newunicodechar{ᵣ}{$_r$}
\newunicodechar{ₙ}{$_n$}
\newunicodechar{❰}{$\llparenthesis$}
\newunicodechar{❱}{$\rrparenthesis$}
\newunicodechar{η}{$\eta$}
\newunicodechar{ε}{$\epsilon$}
\newunicodechar{θ}{$\theta$}
\newunicodechar{λ}{$\lambda$}
\newunicodechar{∀}{$\forall$}
\newunicodechar{≠}{$\neq$}
\newunicodechar{₁}{$_1$}
\newunicodechar{₂}{$_2$}
\newunicodechar{↻}{$\gcv$}
\newunicodechar{⊗}{$\otimes$}
\newunicodechar{⊕}{$\oplus$}
\newunicodechar{₊}{$_+$}
\newunicodechar{₋}{$_-$}
\newunicodechar{≡}{$\equiv$}
\newunicodechar{⁴}{$^4$}
\newunicodechar{₃}{$_3$}
\newunicodechar{₄}{$_4$}
\newunicodechar{²}{$^2$}
\newunicodechar{³}{$^3$}
\newunicodechar{⟨}{$\langle$}
\newunicodechar{⟩}{$\rangle$}
\newunicodechar{□}{$\square$}
\newunicodechar{′}{$'$}
\newunicodechar{∣}{$|$}
\newunicodechar{∥}{$\|$}
\newunicodechar{↔}{$\leftrightarrow$}
\newunicodechar{⋮}{$\vdots$}
\newunicodechar{∈}{$\in$}
\newunicodechar{∙}{$\bullet$}
\newunicodechar{₀}{$_0$}
\newunicodechar{∃}{$\exists$}
\newunicodechar{●}{$\bullet$}
\newunicodechar{⊤}{$\top$}
\newunicodechar{⊥}{$\bot$}
\newunicodechar{∘}{$\circ$}
\newunicodechar{╱}{$/$}
\newcommand{\dmultimap}{\circ\!  \-- \! \circ}
\newunicodechar{⧟}{$\dmultimap$}
\newunicodechar{×}{$\times$}
\newunicodechar{→}{$\rightarrow$}
\newunicodechar{⟦}{$\llbracket$}
\newunicodechar{⟧}{$\rrbracket$}
\newunicodechar{Σ}{$\Sigma$}
\newunicodechar{¬}{$\neg$}
\newunicodechar{≫}{$\gg$}

\newcommand{\alt}{~|~}

\newcommand{\gcv}{\circlearrowright}

\newcommand{\inlv}[1]{\ensuremath{\mathit{inj}_1(#1)}}
\newcommand{\inrv}[1]{\ensuremath{\mathit{inj}_2(#1)}}
\newcommand{\ket}[1]{|#1\rangle}
\newcommand{\oneover}[1]{1/#1}

\newcommand{\identlp}{\mathit{unite}_+\mathit{l}}
\newcommand{\identrp}{\mathit{uniti}_+\mathit{l}}

\newcommand{\swapp}{\mathit{swap}_+}
\newcommand{\assoclp}{\mathit{assocl}_+}
\newcommand{\assocrp}{\mathit{assocr}_+}
\newcommand{\identlt}{\mathit{unite}_*\mathit{l}}
\newcommand{\identrt}{\mathit{uniti}_*\mathit{l}}

\newcommand{\swapt}{\mathit{swap}_*}
\newcommand{\assoclt}{\mathit{assocl}_*}
\newcommand{\assocrt}{\mathit{assocr}_*}
\newcommand{\absorbr}{\mathit{absorbr}}

\newcommand{\factorzl}{\mathit{factorzl}}
\newcommand{\dist}{\mathit{dist}}
\newcommand{\factor}{\mathit{factor}}

\newcommand{\distz}{\mathit{absorbr}}
\newcommand{\iso}{\leftrightarrow}
\newcommand{\proves}{\vdash}
\newcommand{\Afun}[1]{\AgdaFunction{#1}}
\newcommand{\Acon}[1]{\AgdaInductiveConstructor{#1}}
\newcommand{\Avar}[1]{\AgdaBound{#1}}
\newcommand{\idc}{\mathit{id}\!\!\leftrightarrow}
\newcommand{\Rule}[4]{
\makebox{{\rm #1}
$\displaystyle
\frac{\begin{array}{l}#2 \\\end{array}}
{\begin{array}{l}#3      \\\end{array}}$
 #4}}
\newcommand{\jdg}[3]{#2 \proves_{#1} #3}
\renewcommand{\AgdaCodeStyle}{\footnotesize} 
\AgdaNoSpaceAroundCode{}

\tikzstyle{control}=[fill=black, draw=black, shape=circle]
\tikzstyle{not}=[fill=white, draw=black, shape=circle]

\begin{code}[hide]%
\>[0]\AgdaSymbol{\{-\#}\AgdaSpace{}%
\AgdaKeyword{OPTIONS}\AgdaSpace{}%
\AgdaPragma{--without-K}\AgdaSpace{}%
\AgdaPragma{--safe}\AgdaSpace{}%
\AgdaSymbol{\#-\}}\<%
\\
\>[0]\AgdaKeyword{module}\AgdaSpace{}%
\AgdaModule{Pi}\AgdaSpace{}%
\AgdaKeyword{where}\<%
\\
\>[0]\AgdaKeyword{open}\AgdaSpace{}%
\AgdaKeyword{import}\AgdaSpace{}%
\AgdaModule{Data.Empty}\<%
\\
\>[0]\AgdaKeyword{open}\AgdaSpace{}%
\AgdaKeyword{import}\AgdaSpace{}%
\AgdaModule{Data.Unit}\<%
\\
\>[0]\AgdaKeyword{open}\AgdaSpace{}%
\AgdaKeyword{import}\AgdaSpace{}%
\AgdaModule{Data.Sum}\<%
\\
\>[0]\AgdaKeyword{open}\AgdaSpace{}%
\AgdaKeyword{import}\AgdaSpace{}%
\AgdaModule{Data.Product}\<%
\\
\>[0]\AgdaKeyword{open}\AgdaSpace{}%
\AgdaKeyword{import}\AgdaSpace{}%
\AgdaModule{Relation.Binary.PropositionalEquality}\<%
\\
\\[\AgdaEmptyExtraSkip]%
\>[0]\AgdaKeyword{infixr}\AgdaSpace{}%
\AgdaNumber{70}\AgdaSpace{}%
\AgdaOperator{\AgdaInductiveConstructor{\AgdaUnderscore{}×ᵤ\AgdaUnderscore{}}}\<%
\\
\>[0]\AgdaKeyword{infixr}\AgdaSpace{}%
\AgdaNumber{60}\AgdaSpace{}%
\AgdaOperator{\AgdaInductiveConstructor{\AgdaUnderscore{}+ᵤ\AgdaUnderscore{}}}\<%
\\
\>[0]\AgdaKeyword{infixr}\AgdaSpace{}%
\AgdaNumber{50}\AgdaSpace{}%
\AgdaOperator{\AgdaInductiveConstructor{\AgdaUnderscore{}⊚\AgdaUnderscore{}}}\<%
\\
\>[0]\AgdaKeyword{infix}\AgdaSpace{}%
\AgdaNumber{100}\AgdaSpace{}%
\AgdaOperator{\AgdaFunction{!\AgdaUnderscore{}}}\<%
\\
\\[\AgdaEmptyExtraSkip]%
\>[0]\AgdaKeyword{data}\AgdaSpace{}%
\AgdaDatatype{𝕌}\AgdaSpace{}%
\AgdaSymbol{:}\AgdaSpace{}%
\AgdaPrimitiveType{Set}\AgdaSpace{}%
\AgdaKeyword{where}\<%
\\
\>[0][@{}l@{\AgdaIndent{0}}]%
\>[2]\AgdaInductiveConstructor{𝟘}%
\>[10]\AgdaSymbol{:}\AgdaSpace{}%
\AgdaDatatype{𝕌}\<%
\\
\>[2]\AgdaInductiveConstructor{𝟙}%
\>[10]\AgdaSymbol{:}\AgdaSpace{}%
\AgdaDatatype{𝕌}\<%
\\
\>[2]\AgdaOperator{\AgdaInductiveConstructor{\AgdaUnderscore{}+ᵤ\AgdaUnderscore{}}}%
\>[10]\AgdaSymbol{:}\AgdaSpace{}%
\AgdaDatatype{𝕌}\AgdaSpace{}%
\AgdaSymbol{→}\AgdaSpace{}%
\AgdaDatatype{𝕌}\AgdaSpace{}%
\AgdaSymbol{→}\AgdaSpace{}%
\AgdaDatatype{𝕌}\<%
\\
\>[2]\AgdaOperator{\AgdaInductiveConstructor{\AgdaUnderscore{}×ᵤ\AgdaUnderscore{}}}%
\>[10]\AgdaSymbol{:}\AgdaSpace{}%
\AgdaDatatype{𝕌}\AgdaSpace{}%
\AgdaSymbol{→}\AgdaSpace{}%
\AgdaDatatype{𝕌}\AgdaSpace{}%
\AgdaSymbol{→}\AgdaSpace{}%
\AgdaDatatype{𝕌}\<%
\\
\\[\AgdaEmptyExtraSkip]%
\>[0]\AgdaOperator{\AgdaFunction{⟦\AgdaUnderscore{}⟧}}\AgdaSpace{}%
\AgdaSymbol{:}\AgdaSpace{}%
\AgdaSymbol{(}\AgdaBound{A}\AgdaSpace{}%
\AgdaSymbol{:}\AgdaSpace{}%
\AgdaDatatype{𝕌}\AgdaSymbol{)}\AgdaSpace{}%
\AgdaSymbol{→}\AgdaSpace{}%
\AgdaPrimitiveType{Set}\<%
\\
\>[0]\AgdaOperator{\AgdaFunction{⟦}}\AgdaSpace{}%
\AgdaInductiveConstructor{𝟘}\AgdaSpace{}%
\AgdaOperator{\AgdaFunction{⟧}}\AgdaSpace{}%
\AgdaSymbol{=}\AgdaSpace{}%
\AgdaDatatype{⊥}\<%
\\
\>[0]\AgdaOperator{\AgdaFunction{⟦}}\AgdaSpace{}%
\AgdaInductiveConstructor{𝟙}\AgdaSpace{}%
\AgdaOperator{\AgdaFunction{⟧}}\AgdaSpace{}%
\AgdaSymbol{=}\AgdaSpace{}%
\AgdaRecord{⊤}\<%
\\
\>[0]\AgdaOperator{\AgdaFunction{⟦}}\AgdaSpace{}%
\AgdaBound{t₁}\AgdaSpace{}%
\AgdaOperator{\AgdaInductiveConstructor{+ᵤ}}\AgdaSpace{}%
\AgdaBound{t₂}\AgdaSpace{}%
\AgdaOperator{\AgdaFunction{⟧}}\AgdaSpace{}%
\AgdaSymbol{=}\AgdaSpace{}%
\AgdaOperator{\AgdaFunction{⟦}}\AgdaSpace{}%
\AgdaBound{t₁}\AgdaSpace{}%
\AgdaOperator{\AgdaFunction{⟧}}\AgdaSpace{}%
\AgdaOperator{\AgdaDatatype{⊎}}\AgdaSpace{}%
\AgdaOperator{\AgdaFunction{⟦}}\AgdaSpace{}%
\AgdaBound{t₂}\AgdaSpace{}%
\AgdaOperator{\AgdaFunction{⟧}}\<%
\\
\>[0]\AgdaOperator{\AgdaFunction{⟦}}\AgdaSpace{}%
\AgdaBound{t₁}\AgdaSpace{}%
\AgdaOperator{\AgdaInductiveConstructor{×ᵤ}}\AgdaSpace{}%
\AgdaBound{t₂}\AgdaSpace{}%
\AgdaOperator{\AgdaFunction{⟧}}\AgdaSpace{}%
\AgdaSymbol{=}\AgdaSpace{}%
\AgdaOperator{\AgdaFunction{⟦}}\AgdaSpace{}%
\AgdaBound{t₁}\AgdaSpace{}%
\AgdaOperator{\AgdaFunction{⟧}}\AgdaSpace{}%
\AgdaOperator{\AgdaFunction{×}}\AgdaSpace{}%
\AgdaOperator{\AgdaFunction{⟦}}\AgdaSpace{}%
\AgdaBound{t₂}\AgdaSpace{}%
\AgdaOperator{\AgdaFunction{⟧}}\<%
\\
\\[\AgdaEmptyExtraSkip]%
\>[0]\AgdaKeyword{data}\AgdaSpace{}%
\AgdaOperator{\AgdaDatatype{\AgdaUnderscore{}⟷\AgdaUnderscore{}}}\AgdaSpace{}%
\AgdaSymbol{:}\AgdaSpace{}%
\AgdaDatatype{𝕌}\AgdaSpace{}%
\AgdaSymbol{→}\AgdaSpace{}%
\AgdaDatatype{𝕌}\AgdaSpace{}%
\AgdaSymbol{→}\AgdaSpace{}%
\AgdaPrimitiveType{Set}\AgdaSpace{}%
\AgdaKeyword{where}\<%
\\
\>[0][@{}l@{\AgdaIndent{0}}]%
\>[2]\AgdaInductiveConstructor{unite₊l}\AgdaSpace{}%
\AgdaSymbol{:}\AgdaSpace{}%
\AgdaSymbol{\{}\AgdaBound{t}\AgdaSpace{}%
\AgdaSymbol{:}\AgdaSpace{}%
\AgdaDatatype{𝕌}\AgdaSymbol{\}}\AgdaSpace{}%
\AgdaSymbol{→}\AgdaSpace{}%
\AgdaInductiveConstructor{𝟘}\AgdaSpace{}%
\AgdaOperator{\AgdaInductiveConstructor{+ᵤ}}\AgdaSpace{}%
\AgdaBound{t}\AgdaSpace{}%
\AgdaOperator{\AgdaDatatype{⟷}}\AgdaSpace{}%
\AgdaBound{t}\<%
\\
\>[2]\AgdaInductiveConstructor{uniti₊l}\AgdaSpace{}%
\AgdaSymbol{:}\AgdaSpace{}%
\AgdaSymbol{\{}\AgdaBound{t}\AgdaSpace{}%
\AgdaSymbol{:}\AgdaSpace{}%
\AgdaDatatype{𝕌}\AgdaSymbol{\}}\AgdaSpace{}%
\AgdaSymbol{→}\AgdaSpace{}%
\AgdaBound{t}\AgdaSpace{}%
\AgdaOperator{\AgdaDatatype{⟷}}\AgdaSpace{}%
\AgdaInductiveConstructor{𝟘}\AgdaSpace{}%
\AgdaOperator{\AgdaInductiveConstructor{+ᵤ}}\AgdaSpace{}%
\AgdaBound{t}\<%
\\
\>[2]\AgdaInductiveConstructor{unite₊r}\AgdaSpace{}%
\AgdaSymbol{:}\AgdaSpace{}%
\AgdaSymbol{\{}\AgdaBound{t}\AgdaSpace{}%
\AgdaSymbol{:}\AgdaSpace{}%
\AgdaDatatype{𝕌}\AgdaSymbol{\}}\AgdaSpace{}%
\AgdaSymbol{→}\AgdaSpace{}%
\AgdaBound{t}\AgdaSpace{}%
\AgdaOperator{\AgdaInductiveConstructor{+ᵤ}}\AgdaSpace{}%
\AgdaInductiveConstructor{𝟘}\AgdaSpace{}%
\AgdaOperator{\AgdaDatatype{⟷}}\AgdaSpace{}%
\AgdaBound{t}\<%
\\
\>[2]\AgdaInductiveConstructor{uniti₊r}\AgdaSpace{}%
\AgdaSymbol{:}\AgdaSpace{}%
\AgdaSymbol{\{}\AgdaBound{t}\AgdaSpace{}%
\AgdaSymbol{:}\AgdaSpace{}%
\AgdaDatatype{𝕌}\AgdaSymbol{\}}\AgdaSpace{}%
\AgdaSymbol{→}\AgdaSpace{}%
\AgdaBound{t}\AgdaSpace{}%
\AgdaOperator{\AgdaDatatype{⟷}}\AgdaSpace{}%
\AgdaBound{t}\AgdaSpace{}%
\AgdaOperator{\AgdaInductiveConstructor{+ᵤ}}\AgdaSpace{}%
\AgdaInductiveConstructor{𝟘}\<%
\\
\>[2]\AgdaInductiveConstructor{swap₊}%
\>[10]\AgdaSymbol{:}\AgdaSpace{}%
\AgdaSymbol{\{}\AgdaBound{t₁}\AgdaSpace{}%
\AgdaBound{t₂}\AgdaSpace{}%
\AgdaSymbol{:}\AgdaSpace{}%
\AgdaDatatype{𝕌}\AgdaSymbol{\}}\AgdaSpace{}%
\AgdaSymbol{→}\AgdaSpace{}%
\AgdaBound{t₁}\AgdaSpace{}%
\AgdaOperator{\AgdaInductiveConstructor{+ᵤ}}\AgdaSpace{}%
\AgdaBound{t₂}\AgdaSpace{}%
\AgdaOperator{\AgdaDatatype{⟷}}\AgdaSpace{}%
\AgdaBound{t₂}\AgdaSpace{}%
\AgdaOperator{\AgdaInductiveConstructor{+ᵤ}}\AgdaSpace{}%
\AgdaBound{t₁}\<%
\\
\>[2]\AgdaInductiveConstructor{assocl₊}\AgdaSpace{}%
\AgdaSymbol{:}\AgdaSpace{}%
\AgdaSymbol{\{}\AgdaBound{t₁}\AgdaSpace{}%
\AgdaBound{t₂}\AgdaSpace{}%
\AgdaBound{t₃}\AgdaSpace{}%
\AgdaSymbol{:}\AgdaSpace{}%
\AgdaDatatype{𝕌}\AgdaSymbol{\}}\AgdaSpace{}%
\AgdaSymbol{→}\AgdaSpace{}%
\AgdaBound{t₁}\AgdaSpace{}%
\AgdaOperator{\AgdaInductiveConstructor{+ᵤ}}\AgdaSpace{}%
\AgdaSymbol{(}\AgdaBound{t₂}\AgdaSpace{}%
\AgdaOperator{\AgdaInductiveConstructor{+ᵤ}}\AgdaSpace{}%
\AgdaBound{t₃}\AgdaSymbol{)}\AgdaSpace{}%
\AgdaOperator{\AgdaDatatype{⟷}}\AgdaSpace{}%
\AgdaSymbol{(}\AgdaBound{t₁}\AgdaSpace{}%
\AgdaOperator{\AgdaInductiveConstructor{+ᵤ}}\AgdaSpace{}%
\AgdaBound{t₂}\AgdaSymbol{)}\AgdaSpace{}%
\AgdaOperator{\AgdaInductiveConstructor{+ᵤ}}\AgdaSpace{}%
\AgdaBound{t₃}\<%
\\
\>[2]\AgdaInductiveConstructor{assocr₊}\AgdaSpace{}%
\AgdaSymbol{:}\AgdaSpace{}%
\AgdaSymbol{\{}\AgdaBound{t₁}\AgdaSpace{}%
\AgdaBound{t₂}\AgdaSpace{}%
\AgdaBound{t₃}\AgdaSpace{}%
\AgdaSymbol{:}\AgdaSpace{}%
\AgdaDatatype{𝕌}\AgdaSymbol{\}}\AgdaSpace{}%
\AgdaSymbol{→}\AgdaSpace{}%
\AgdaSymbol{(}\AgdaBound{t₁}\AgdaSpace{}%
\AgdaOperator{\AgdaInductiveConstructor{+ᵤ}}\AgdaSpace{}%
\AgdaBound{t₂}\AgdaSymbol{)}\AgdaSpace{}%
\AgdaOperator{\AgdaInductiveConstructor{+ᵤ}}\AgdaSpace{}%
\AgdaBound{t₃}\AgdaSpace{}%
\AgdaOperator{\AgdaDatatype{⟷}}\AgdaSpace{}%
\AgdaBound{t₁}\AgdaSpace{}%
\AgdaOperator{\AgdaInductiveConstructor{+ᵤ}}\AgdaSpace{}%
\AgdaSymbol{(}\AgdaBound{t₂}\AgdaSpace{}%
\AgdaOperator{\AgdaInductiveConstructor{+ᵤ}}\AgdaSpace{}%
\AgdaBound{t₃}\AgdaSymbol{)}\<%
\\
\>[2]\AgdaInductiveConstructor{unite⋆l}\AgdaSpace{}%
\AgdaSymbol{:}\AgdaSpace{}%
\AgdaSymbol{\{}\AgdaBound{t}\AgdaSpace{}%
\AgdaSymbol{:}\AgdaSpace{}%
\AgdaDatatype{𝕌}\AgdaSymbol{\}}\AgdaSpace{}%
\AgdaSymbol{→}\AgdaSpace{}%
\AgdaInductiveConstructor{𝟙}\AgdaSpace{}%
\AgdaOperator{\AgdaInductiveConstructor{×ᵤ}}\AgdaSpace{}%
\AgdaBound{t}\AgdaSpace{}%
\AgdaOperator{\AgdaDatatype{⟷}}\AgdaSpace{}%
\AgdaBound{t}\<%
\\
\>[2]\AgdaInductiveConstructor{uniti⋆l}\AgdaSpace{}%
\AgdaSymbol{:}\AgdaSpace{}%
\AgdaSymbol{\{}\AgdaBound{t}\AgdaSpace{}%
\AgdaSymbol{:}\AgdaSpace{}%
\AgdaDatatype{𝕌}\AgdaSymbol{\}}\AgdaSpace{}%
\AgdaSymbol{→}\AgdaSpace{}%
\AgdaBound{t}\AgdaSpace{}%
\AgdaOperator{\AgdaDatatype{⟷}}\AgdaSpace{}%
\AgdaInductiveConstructor{𝟙}\AgdaSpace{}%
\AgdaOperator{\AgdaInductiveConstructor{×ᵤ}}\AgdaSpace{}%
\AgdaBound{t}\<%
\\
\>[2]\AgdaInductiveConstructor{unite⋆r}\AgdaSpace{}%
\AgdaSymbol{:}\AgdaSpace{}%
\AgdaSymbol{\{}\AgdaBound{t}\AgdaSpace{}%
\AgdaSymbol{:}\AgdaSpace{}%
\AgdaDatatype{𝕌}\AgdaSymbol{\}}\AgdaSpace{}%
\AgdaSymbol{→}\AgdaSpace{}%
\AgdaBound{t}\AgdaSpace{}%
\AgdaOperator{\AgdaInductiveConstructor{×ᵤ}}\AgdaSpace{}%
\AgdaInductiveConstructor{𝟙}\AgdaSpace{}%
\AgdaOperator{\AgdaDatatype{⟷}}\AgdaSpace{}%
\AgdaBound{t}\<%
\\
\>[2]\AgdaInductiveConstructor{uniti⋆r}\AgdaSpace{}%
\AgdaSymbol{:}\AgdaSpace{}%
\AgdaSymbol{\{}\AgdaBound{t}\AgdaSpace{}%
\AgdaSymbol{:}\AgdaSpace{}%
\AgdaDatatype{𝕌}\AgdaSymbol{\}}\AgdaSpace{}%
\AgdaSymbol{→}\AgdaSpace{}%
\AgdaBound{t}\AgdaSpace{}%
\AgdaOperator{\AgdaDatatype{⟷}}\AgdaSpace{}%
\AgdaBound{t}\AgdaSpace{}%
\AgdaOperator{\AgdaInductiveConstructor{×ᵤ}}\AgdaSpace{}%
\AgdaInductiveConstructor{𝟙}\<%
\\
\>[2]\AgdaInductiveConstructor{swap⋆}%
\>[10]\AgdaSymbol{:}\AgdaSpace{}%
\AgdaSymbol{\{}\AgdaBound{t₁}\AgdaSpace{}%
\AgdaBound{t₂}\AgdaSpace{}%
\AgdaSymbol{:}\AgdaSpace{}%
\AgdaDatatype{𝕌}\AgdaSymbol{\}}\AgdaSpace{}%
\AgdaSymbol{→}\AgdaSpace{}%
\AgdaBound{t₁}\AgdaSpace{}%
\AgdaOperator{\AgdaInductiveConstructor{×ᵤ}}\AgdaSpace{}%
\AgdaBound{t₂}\AgdaSpace{}%
\AgdaOperator{\AgdaDatatype{⟷}}\AgdaSpace{}%
\AgdaBound{t₂}\AgdaSpace{}%
\AgdaOperator{\AgdaInductiveConstructor{×ᵤ}}\AgdaSpace{}%
\AgdaBound{t₁}\<%
\\
\>[2]\AgdaInductiveConstructor{assocl⋆}\AgdaSpace{}%
\AgdaSymbol{:}\AgdaSpace{}%
\AgdaSymbol{\{}\AgdaBound{t₁}\AgdaSpace{}%
\AgdaBound{t₂}\AgdaSpace{}%
\AgdaBound{t₃}\AgdaSpace{}%
\AgdaSymbol{:}\AgdaSpace{}%
\AgdaDatatype{𝕌}\AgdaSymbol{\}}\AgdaSpace{}%
\AgdaSymbol{→}\AgdaSpace{}%
\AgdaBound{t₁}\AgdaSpace{}%
\AgdaOperator{\AgdaInductiveConstructor{×ᵤ}}\AgdaSpace{}%
\AgdaSymbol{(}\AgdaBound{t₂}\AgdaSpace{}%
\AgdaOperator{\AgdaInductiveConstructor{×ᵤ}}\AgdaSpace{}%
\AgdaBound{t₃}\AgdaSymbol{)}\AgdaSpace{}%
\AgdaOperator{\AgdaDatatype{⟷}}\AgdaSpace{}%
\AgdaSymbol{(}\AgdaBound{t₁}\AgdaSpace{}%
\AgdaOperator{\AgdaInductiveConstructor{×ᵤ}}\AgdaSpace{}%
\AgdaBound{t₂}\AgdaSymbol{)}\AgdaSpace{}%
\AgdaOperator{\AgdaInductiveConstructor{×ᵤ}}\AgdaSpace{}%
\AgdaBound{t₃}\<%
\\
\>[2]\AgdaInductiveConstructor{assocr⋆}\AgdaSpace{}%
\AgdaSymbol{:}\AgdaSpace{}%
\AgdaSymbol{\{}\AgdaBound{t₁}\AgdaSpace{}%
\AgdaBound{t₂}\AgdaSpace{}%
\AgdaBound{t₃}\AgdaSpace{}%
\AgdaSymbol{:}\AgdaSpace{}%
\AgdaDatatype{𝕌}\AgdaSymbol{\}}\AgdaSpace{}%
\AgdaSymbol{→}\AgdaSpace{}%
\AgdaSymbol{(}\AgdaBound{t₁}\AgdaSpace{}%
\AgdaOperator{\AgdaInductiveConstructor{×ᵤ}}\AgdaSpace{}%
\AgdaBound{t₂}\AgdaSymbol{)}\AgdaSpace{}%
\AgdaOperator{\AgdaInductiveConstructor{×ᵤ}}\AgdaSpace{}%
\AgdaBound{t₃}\AgdaSpace{}%
\AgdaOperator{\AgdaDatatype{⟷}}\AgdaSpace{}%
\AgdaBound{t₁}\AgdaSpace{}%
\AgdaOperator{\AgdaInductiveConstructor{×ᵤ}}\AgdaSpace{}%
\AgdaSymbol{(}\AgdaBound{t₂}\AgdaSpace{}%
\AgdaOperator{\AgdaInductiveConstructor{×ᵤ}}\AgdaSpace{}%
\AgdaBound{t₃}\AgdaSymbol{)}\<%
\\
\>[2]\AgdaInductiveConstructor{absorbr}\AgdaSpace{}%
\AgdaSymbol{:}\AgdaSpace{}%
\AgdaSymbol{\{}\AgdaBound{t}\AgdaSpace{}%
\AgdaSymbol{:}\AgdaSpace{}%
\AgdaDatatype{𝕌}\AgdaSymbol{\}}\AgdaSpace{}%
\AgdaSymbol{→}\AgdaSpace{}%
\AgdaInductiveConstructor{𝟘}\AgdaSpace{}%
\AgdaOperator{\AgdaInductiveConstructor{×ᵤ}}\AgdaSpace{}%
\AgdaBound{t}\AgdaSpace{}%
\AgdaOperator{\AgdaDatatype{⟷}}\AgdaSpace{}%
\AgdaInductiveConstructor{𝟘}\<%
\\
\>[2]\AgdaInductiveConstructor{absorbl}\AgdaSpace{}%
\AgdaSymbol{:}\AgdaSpace{}%
\AgdaSymbol{\{}\AgdaBound{t}\AgdaSpace{}%
\AgdaSymbol{:}\AgdaSpace{}%
\AgdaDatatype{𝕌}\AgdaSymbol{\}}\AgdaSpace{}%
\AgdaSymbol{→}\AgdaSpace{}%
\AgdaBound{t}\AgdaSpace{}%
\AgdaOperator{\AgdaInductiveConstructor{×ᵤ}}\AgdaSpace{}%
\AgdaInductiveConstructor{𝟘}\AgdaSpace{}%
\AgdaOperator{\AgdaDatatype{⟷}}\AgdaSpace{}%
\AgdaInductiveConstructor{𝟘}\<%
\\
\>[2]\AgdaInductiveConstructor{factorzr}\AgdaSpace{}%
\AgdaSymbol{:}\AgdaSpace{}%
\AgdaSymbol{\{}\AgdaBound{t}\AgdaSpace{}%
\AgdaSymbol{:}\AgdaSpace{}%
\AgdaDatatype{𝕌}\AgdaSymbol{\}}\AgdaSpace{}%
\AgdaSymbol{→}\AgdaSpace{}%
\AgdaInductiveConstructor{𝟘}\AgdaSpace{}%
\AgdaOperator{\AgdaDatatype{⟷}}\AgdaSpace{}%
\AgdaBound{t}\AgdaSpace{}%
\AgdaOperator{\AgdaInductiveConstructor{×ᵤ}}\AgdaSpace{}%
\AgdaInductiveConstructor{𝟘}\<%
\\
\>[2]\AgdaInductiveConstructor{factorzl}\AgdaSpace{}%
\AgdaSymbol{:}\AgdaSpace{}%
\AgdaSymbol{\{}\AgdaBound{t}\AgdaSpace{}%
\AgdaSymbol{:}\AgdaSpace{}%
\AgdaDatatype{𝕌}\AgdaSymbol{\}}\AgdaSpace{}%
\AgdaSymbol{→}\AgdaSpace{}%
\AgdaInductiveConstructor{𝟘}\AgdaSpace{}%
\AgdaOperator{\AgdaDatatype{⟷}}\AgdaSpace{}%
\AgdaInductiveConstructor{𝟘}\AgdaSpace{}%
\AgdaOperator{\AgdaInductiveConstructor{×ᵤ}}\AgdaSpace{}%
\AgdaBound{t}\<%
\\
\>[2]\AgdaInductiveConstructor{dist}%
\>[10]\AgdaSymbol{:}\AgdaSpace{}%
\AgdaSymbol{\{}\AgdaBound{t₁}\AgdaSpace{}%
\AgdaBound{t₂}\AgdaSpace{}%
\AgdaBound{t₃}\AgdaSpace{}%
\AgdaSymbol{:}\AgdaSpace{}%
\AgdaDatatype{𝕌}\AgdaSymbol{\}}\AgdaSpace{}%
\AgdaSymbol{→}\AgdaSpace{}%
\AgdaSymbol{(}\AgdaBound{t₁}\AgdaSpace{}%
\AgdaOperator{\AgdaInductiveConstructor{+ᵤ}}\AgdaSpace{}%
\AgdaBound{t₂}\AgdaSymbol{)}\AgdaSpace{}%
\AgdaOperator{\AgdaInductiveConstructor{×ᵤ}}\AgdaSpace{}%
\AgdaBound{t₃}\AgdaSpace{}%
\AgdaOperator{\AgdaDatatype{⟷}}\AgdaSpace{}%
\AgdaSymbol{(}\AgdaBound{t₁}\AgdaSpace{}%
\AgdaOperator{\AgdaInductiveConstructor{×ᵤ}}\AgdaSpace{}%
\AgdaBound{t₃}\AgdaSymbol{)}\AgdaSpace{}%
\AgdaOperator{\AgdaInductiveConstructor{+ᵤ}}\AgdaSpace{}%
\AgdaSymbol{(}\AgdaBound{t₂}\AgdaSpace{}%
\AgdaOperator{\AgdaInductiveConstructor{×ᵤ}}\AgdaSpace{}%
\AgdaBound{t₃}\AgdaSymbol{)}\<%
\\
\>[2]\AgdaInductiveConstructor{factor}%
\>[10]\AgdaSymbol{:}\AgdaSpace{}%
\AgdaSymbol{\{}\AgdaBound{t₁}\AgdaSpace{}%
\AgdaBound{t₂}\AgdaSpace{}%
\AgdaBound{t₃}\AgdaSpace{}%
\AgdaSymbol{:}\AgdaSpace{}%
\AgdaDatatype{𝕌}\AgdaSymbol{\}}\AgdaSpace{}%
\AgdaSymbol{→}\AgdaSpace{}%
\AgdaSymbol{(}\AgdaBound{t₁}\AgdaSpace{}%
\AgdaOperator{\AgdaInductiveConstructor{×ᵤ}}\AgdaSpace{}%
\AgdaBound{t₃}\AgdaSymbol{)}\AgdaSpace{}%
\AgdaOperator{\AgdaInductiveConstructor{+ᵤ}}\AgdaSpace{}%
\AgdaSymbol{(}\AgdaBound{t₂}\AgdaSpace{}%
\AgdaOperator{\AgdaInductiveConstructor{×ᵤ}}\AgdaSpace{}%
\AgdaBound{t₃}\AgdaSymbol{)}\AgdaSpace{}%
\AgdaOperator{\AgdaDatatype{⟷}}\AgdaSpace{}%
\AgdaSymbol{(}\AgdaBound{t₁}\AgdaSpace{}%
\AgdaOperator{\AgdaInductiveConstructor{+ᵤ}}\AgdaSpace{}%
\AgdaBound{t₂}\AgdaSymbol{)}\AgdaSpace{}%
\AgdaOperator{\AgdaInductiveConstructor{×ᵤ}}\AgdaSpace{}%
\AgdaBound{t₃}\<%
\\
\>[2]\AgdaInductiveConstructor{distl}%
\>[10]\AgdaSymbol{:}\AgdaSpace{}%
\AgdaSymbol{\{}\AgdaBound{t₁}\AgdaSpace{}%
\AgdaBound{t₂}\AgdaSpace{}%
\AgdaBound{t₃}\AgdaSpace{}%
\AgdaSymbol{:}\AgdaSpace{}%
\AgdaDatatype{𝕌}\AgdaSymbol{\}}\AgdaSpace{}%
\AgdaSymbol{→}\AgdaSpace{}%
\AgdaBound{t₁}\AgdaSpace{}%
\AgdaOperator{\AgdaInductiveConstructor{×ᵤ}}\AgdaSpace{}%
\AgdaSymbol{(}\AgdaBound{t₂}\AgdaSpace{}%
\AgdaOperator{\AgdaInductiveConstructor{+ᵤ}}\AgdaSpace{}%
\AgdaBound{t₃}\AgdaSymbol{)}\AgdaSpace{}%
\AgdaOperator{\AgdaDatatype{⟷}}\AgdaSpace{}%
\AgdaSymbol{(}\AgdaBound{t₁}\AgdaSpace{}%
\AgdaOperator{\AgdaInductiveConstructor{×ᵤ}}\AgdaSpace{}%
\AgdaBound{t₂}\AgdaSymbol{)}\AgdaSpace{}%
\AgdaOperator{\AgdaInductiveConstructor{+ᵤ}}\AgdaSpace{}%
\AgdaSymbol{(}\AgdaBound{t₁}\AgdaSpace{}%
\AgdaOperator{\AgdaInductiveConstructor{×ᵤ}}\AgdaSpace{}%
\AgdaBound{t₃}\AgdaSymbol{)}\<%
\\
\>[2]\AgdaInductiveConstructor{factorl}\AgdaSpace{}%
\AgdaSymbol{:}\AgdaSpace{}%
\AgdaSymbol{\{}\AgdaBound{t₁}\AgdaSpace{}%
\AgdaBound{t₂}\AgdaSpace{}%
\AgdaBound{t₃}\AgdaSpace{}%
\AgdaSymbol{:}\AgdaSpace{}%
\AgdaDatatype{𝕌}\AgdaSpace{}%
\AgdaSymbol{\}}\AgdaSpace{}%
\AgdaSymbol{→}\AgdaSpace{}%
\AgdaSymbol{(}\AgdaBound{t₁}\AgdaSpace{}%
\AgdaOperator{\AgdaInductiveConstructor{×ᵤ}}\AgdaSpace{}%
\AgdaBound{t₂}\AgdaSymbol{)}\AgdaSpace{}%
\AgdaOperator{\AgdaInductiveConstructor{+ᵤ}}\AgdaSpace{}%
\AgdaSymbol{(}\AgdaBound{t₁}\AgdaSpace{}%
\AgdaOperator{\AgdaInductiveConstructor{×ᵤ}}\AgdaSpace{}%
\AgdaBound{t₃}\AgdaSymbol{)}\AgdaSpace{}%
\AgdaOperator{\AgdaDatatype{⟷}}\AgdaSpace{}%
\AgdaBound{t₁}\AgdaSpace{}%
\AgdaOperator{\AgdaInductiveConstructor{×ᵤ}}\AgdaSpace{}%
\AgdaSymbol{(}\AgdaBound{t₂}\AgdaSpace{}%
\AgdaOperator{\AgdaInductiveConstructor{+ᵤ}}\AgdaSpace{}%
\AgdaBound{t₃}\AgdaSymbol{)}\<%
\\
\>[2]\AgdaInductiveConstructor{id⟷}%
\>[10]\AgdaSymbol{:}\AgdaSpace{}%
\AgdaSymbol{\{}\AgdaBound{t}\AgdaSpace{}%
\AgdaSymbol{:}\AgdaSpace{}%
\AgdaDatatype{𝕌}\AgdaSymbol{\}}\AgdaSpace{}%
\AgdaSymbol{→}\AgdaSpace{}%
\AgdaBound{t}\AgdaSpace{}%
\AgdaOperator{\AgdaDatatype{⟷}}\AgdaSpace{}%
\AgdaBound{t}\<%
\\
\>[2]\AgdaOperator{\AgdaInductiveConstructor{\AgdaUnderscore{}⊚\AgdaUnderscore{}}}%
\>[10]\AgdaSymbol{:}\AgdaSpace{}%
\AgdaSymbol{\{}\AgdaBound{t₁}\AgdaSpace{}%
\AgdaBound{t₂}\AgdaSpace{}%
\AgdaBound{t₃}\AgdaSpace{}%
\AgdaSymbol{:}\AgdaSpace{}%
\AgdaDatatype{𝕌}\AgdaSymbol{\}}\AgdaSpace{}%
\AgdaSymbol{→}\AgdaSpace{}%
\AgdaSymbol{(}\AgdaBound{t₁}\AgdaSpace{}%
\AgdaOperator{\AgdaDatatype{⟷}}\AgdaSpace{}%
\AgdaBound{t₂}\AgdaSymbol{)}\AgdaSpace{}%
\AgdaSymbol{→}\AgdaSpace{}%
\AgdaSymbol{(}\AgdaBound{t₂}\AgdaSpace{}%
\AgdaOperator{\AgdaDatatype{⟷}}\AgdaSpace{}%
\AgdaBound{t₃}\AgdaSymbol{)}\AgdaSpace{}%
\AgdaSymbol{→}\AgdaSpace{}%
\AgdaSymbol{(}\AgdaBound{t₁}\AgdaSpace{}%
\AgdaOperator{\AgdaDatatype{⟷}}\AgdaSpace{}%
\AgdaBound{t₃}\AgdaSymbol{)}\<%
\\
\>[2]\AgdaOperator{\AgdaInductiveConstructor{\AgdaUnderscore{}⊕\AgdaUnderscore{}}}%
\>[10]\AgdaSymbol{:}\AgdaSpace{}%
\AgdaSymbol{\{}\AgdaBound{t₁}\AgdaSpace{}%
\AgdaBound{t₂}\AgdaSpace{}%
\AgdaBound{t₃}\AgdaSpace{}%
\AgdaBound{t₄}\AgdaSpace{}%
\AgdaSymbol{:}\AgdaSpace{}%
\AgdaDatatype{𝕌}\AgdaSymbol{\}}\AgdaSpace{}%
\AgdaSymbol{→}\AgdaSpace{}%
\AgdaSymbol{(}\AgdaBound{t₁}\AgdaSpace{}%
\AgdaOperator{\AgdaDatatype{⟷}}\AgdaSpace{}%
\AgdaBound{t₃}\AgdaSymbol{)}\AgdaSpace{}%
\AgdaSymbol{→}\AgdaSpace{}%
\AgdaSymbol{(}\AgdaBound{t₂}\AgdaSpace{}%
\AgdaOperator{\AgdaDatatype{⟷}}\AgdaSpace{}%
\AgdaBound{t₄}\AgdaSymbol{)}\AgdaSpace{}%
\AgdaSymbol{→}\AgdaSpace{}%
\AgdaSymbol{(}\AgdaBound{t₁}\AgdaSpace{}%
\AgdaOperator{\AgdaInductiveConstructor{+ᵤ}}\AgdaSpace{}%
\AgdaBound{t₂}\AgdaSpace{}%
\AgdaOperator{\AgdaDatatype{⟷}}\AgdaSpace{}%
\AgdaBound{t₃}\AgdaSpace{}%
\AgdaOperator{\AgdaInductiveConstructor{+ᵤ}}\AgdaSpace{}%
\AgdaBound{t₄}\AgdaSymbol{)}\<%
\\
\>[2]\AgdaOperator{\AgdaInductiveConstructor{\AgdaUnderscore{}⊗\AgdaUnderscore{}}}%
\>[10]\AgdaSymbol{:}\AgdaSpace{}%
\AgdaSymbol{\{}\AgdaBound{t₁}\AgdaSpace{}%
\AgdaBound{t₂}\AgdaSpace{}%
\AgdaBound{t₃}\AgdaSpace{}%
\AgdaBound{t₄}\AgdaSpace{}%
\AgdaSymbol{:}\AgdaSpace{}%
\AgdaDatatype{𝕌}\AgdaSymbol{\}}\AgdaSpace{}%
\AgdaSymbol{→}\AgdaSpace{}%
\AgdaSymbol{(}\AgdaBound{t₁}\AgdaSpace{}%
\AgdaOperator{\AgdaDatatype{⟷}}\AgdaSpace{}%
\AgdaBound{t₃}\AgdaSymbol{)}\AgdaSpace{}%
\AgdaSymbol{→}\AgdaSpace{}%
\AgdaSymbol{(}\AgdaBound{t₂}\AgdaSpace{}%
\AgdaOperator{\AgdaDatatype{⟷}}\AgdaSpace{}%
\AgdaBound{t₄}\AgdaSymbol{)}\AgdaSpace{}%
\AgdaSymbol{→}\AgdaSpace{}%
\AgdaSymbol{(}\AgdaBound{t₁}\AgdaSpace{}%
\AgdaOperator{\AgdaInductiveConstructor{×ᵤ}}\AgdaSpace{}%
\AgdaBound{t₂}\AgdaSpace{}%
\AgdaOperator{\AgdaDatatype{⟷}}\AgdaSpace{}%
\AgdaBound{t₃}\AgdaSpace{}%
\AgdaOperator{\AgdaInductiveConstructor{×ᵤ}}\AgdaSpace{}%
\AgdaBound{t₄}\AgdaSymbol{)}\<%
\\
\\[\AgdaEmptyExtraSkip]%
\>[0]\AgdaFunction{eval}\AgdaSpace{}%
\AgdaSymbol{:}\AgdaSpace{}%
\AgdaSymbol{\{}\AgdaBound{A}\AgdaSpace{}%
\AgdaBound{B}\AgdaSpace{}%
\AgdaSymbol{:}\AgdaSpace{}%
\AgdaDatatype{𝕌}\AgdaSymbol{\}}\AgdaSpace{}%
\AgdaSymbol{→}\AgdaSpace{}%
\AgdaSymbol{(}\AgdaBound{A}\AgdaSpace{}%
\AgdaOperator{\AgdaDatatype{⟷}}\AgdaSpace{}%
\AgdaBound{B}\AgdaSymbol{)}\AgdaSpace{}%
\AgdaSymbol{→}\AgdaSpace{}%
\AgdaOperator{\AgdaFunction{⟦}}\AgdaSpace{}%
\AgdaBound{A}\AgdaSpace{}%
\AgdaOperator{\AgdaFunction{⟧}}\AgdaSpace{}%
\AgdaSymbol{→}\AgdaSpace{}%
\AgdaOperator{\AgdaFunction{⟦}}\AgdaSpace{}%
\AgdaBound{B}\AgdaSpace{}%
\AgdaOperator{\AgdaFunction{⟧}}\<%
\\
\>[0]\AgdaFunction{eval}\AgdaSpace{}%
\AgdaInductiveConstructor{unite₊l}\AgdaSpace{}%
\AgdaSymbol{(}\AgdaInductiveConstructor{inj₂}\AgdaSpace{}%
\AgdaBound{v}\AgdaSymbol{)}\AgdaSpace{}%
\AgdaSymbol{=}\AgdaSpace{}%
\AgdaBound{v}\<%
\\
\>[0]\AgdaFunction{eval}\AgdaSpace{}%
\AgdaInductiveConstructor{uniti₊l}\AgdaSpace{}%
\AgdaBound{v}%
\>[16]\AgdaSymbol{=}\AgdaSpace{}%
\AgdaInductiveConstructor{inj₂}\AgdaSpace{}%
\AgdaBound{v}\<%
\\
\>[0]\AgdaFunction{eval}\AgdaSpace{}%
\AgdaInductiveConstructor{unite₊r}\AgdaSpace{}%
\AgdaSymbol{(}\AgdaInductiveConstructor{inj₁}\AgdaSpace{}%
\AgdaBound{v}\AgdaSymbol{)}\AgdaSpace{}%
\AgdaSymbol{=}\AgdaSpace{}%
\AgdaBound{v}\<%
\\
\>[0]\AgdaFunction{eval}\AgdaSpace{}%
\AgdaInductiveConstructor{uniti₊r}\AgdaSpace{}%
\AgdaBound{v}%
\>[16]\AgdaSymbol{=}\AgdaSpace{}%
\AgdaInductiveConstructor{inj₁}\AgdaSpace{}%
\AgdaBound{v}\<%
\\
\>[0]\AgdaFunction{eval}\AgdaSpace{}%
\AgdaInductiveConstructor{swap₊}\AgdaSpace{}%
\AgdaSymbol{(}\AgdaInductiveConstructor{inj₁}\AgdaSpace{}%
\AgdaBound{v}\AgdaSymbol{)}\AgdaSpace{}%
\AgdaSymbol{=}\AgdaSpace{}%
\AgdaInductiveConstructor{inj₂}\AgdaSpace{}%
\AgdaBound{v}\<%
\\
\>[0]\AgdaFunction{eval}\AgdaSpace{}%
\AgdaInductiveConstructor{swap₊}\AgdaSpace{}%
\AgdaSymbol{(}\AgdaInductiveConstructor{inj₂}\AgdaSpace{}%
\AgdaBound{v}\AgdaSymbol{)}\AgdaSpace{}%
\AgdaSymbol{=}\AgdaSpace{}%
\AgdaInductiveConstructor{inj₁}\AgdaSpace{}%
\AgdaBound{v}\<%
\\
\>[0]\AgdaFunction{eval}\AgdaSpace{}%
\AgdaInductiveConstructor{assocl₊}\AgdaSpace{}%
\AgdaSymbol{(}\AgdaInductiveConstructor{inj₁}\AgdaSpace{}%
\AgdaBound{v}\AgdaSymbol{)}\AgdaSpace{}%
\AgdaSymbol{=}\AgdaSpace{}%
\AgdaInductiveConstructor{inj₁}\AgdaSpace{}%
\AgdaSymbol{(}\AgdaInductiveConstructor{inj₁}\AgdaSpace{}%
\AgdaBound{v}\AgdaSymbol{)}\<%
\\
\>[0]\AgdaFunction{eval}\AgdaSpace{}%
\AgdaInductiveConstructor{assocl₊}\AgdaSpace{}%
\AgdaSymbol{(}\AgdaInductiveConstructor{inj₂}\AgdaSpace{}%
\AgdaSymbol{(}\AgdaInductiveConstructor{inj₁}\AgdaSpace{}%
\AgdaBound{v}\AgdaSymbol{))}\AgdaSpace{}%
\AgdaSymbol{=}\AgdaSpace{}%
\AgdaInductiveConstructor{inj₁}\AgdaSpace{}%
\AgdaSymbol{(}\AgdaInductiveConstructor{inj₂}\AgdaSpace{}%
\AgdaBound{v}\AgdaSymbol{)}\<%
\\
\>[0]\AgdaFunction{eval}\AgdaSpace{}%
\AgdaInductiveConstructor{assocl₊}\AgdaSpace{}%
\AgdaSymbol{(}\AgdaInductiveConstructor{inj₂}\AgdaSpace{}%
\AgdaSymbol{(}\AgdaInductiveConstructor{inj₂}\AgdaSpace{}%
\AgdaBound{v}\AgdaSymbol{))}\AgdaSpace{}%
\AgdaSymbol{=}\AgdaSpace{}%
\AgdaInductiveConstructor{inj₂}\AgdaSpace{}%
\AgdaBound{v}\<%
\\
\>[0]\AgdaFunction{eval}\AgdaSpace{}%
\AgdaInductiveConstructor{assocr₊}\AgdaSpace{}%
\AgdaSymbol{(}\AgdaInductiveConstructor{inj₁}\AgdaSpace{}%
\AgdaSymbol{(}\AgdaInductiveConstructor{inj₁}\AgdaSpace{}%
\AgdaBound{v}\AgdaSymbol{))}\AgdaSpace{}%
\AgdaSymbol{=}\AgdaSpace{}%
\AgdaInductiveConstructor{inj₁}\AgdaSpace{}%
\AgdaBound{v}\<%
\\
\>[0]\AgdaFunction{eval}\AgdaSpace{}%
\AgdaInductiveConstructor{assocr₊}\AgdaSpace{}%
\AgdaSymbol{(}\AgdaInductiveConstructor{inj₁}\AgdaSpace{}%
\AgdaSymbol{(}\AgdaInductiveConstructor{inj₂}\AgdaSpace{}%
\AgdaBound{v}\AgdaSymbol{))}\AgdaSpace{}%
\AgdaSymbol{=}\AgdaSpace{}%
\AgdaInductiveConstructor{inj₂}\AgdaSpace{}%
\AgdaSymbol{(}\AgdaInductiveConstructor{inj₁}\AgdaSpace{}%
\AgdaBound{v}\AgdaSymbol{)}\<%
\\
\>[0]\AgdaFunction{eval}\AgdaSpace{}%
\AgdaInductiveConstructor{assocr₊}\AgdaSpace{}%
\AgdaSymbol{(}\AgdaInductiveConstructor{inj₂}\AgdaSpace{}%
\AgdaBound{v}\AgdaSymbol{)}\AgdaSpace{}%
\AgdaSymbol{=}\AgdaSpace{}%
\AgdaInductiveConstructor{inj₂}\AgdaSpace{}%
\AgdaSymbol{(}\AgdaInductiveConstructor{inj₂}\AgdaSpace{}%
\AgdaBound{v}\AgdaSymbol{)}\<%
\\
\>[0]\AgdaFunction{eval}\AgdaSpace{}%
\AgdaInductiveConstructor{unite⋆l}\AgdaSpace{}%
\AgdaSymbol{(}\AgdaInductiveConstructor{tt}\AgdaSpace{}%
\AgdaOperator{\AgdaInductiveConstructor{,}}\AgdaSpace{}%
\AgdaBound{v}\AgdaSymbol{)}\AgdaSpace{}%
\AgdaSymbol{=}\AgdaSpace{}%
\AgdaBound{v}\<%
\\
\>[0]\AgdaFunction{eval}\AgdaSpace{}%
\AgdaInductiveConstructor{uniti⋆l}\AgdaSpace{}%
\AgdaBound{v}\AgdaSpace{}%
\AgdaSymbol{=}\AgdaSpace{}%
\AgdaSymbol{(}\AgdaInductiveConstructor{tt}\AgdaSpace{}%
\AgdaOperator{\AgdaInductiveConstructor{,}}\AgdaSpace{}%
\AgdaBound{v}\AgdaSymbol{)}\<%
\\
\>[0]\AgdaFunction{eval}\AgdaSpace{}%
\AgdaInductiveConstructor{unite⋆r}\AgdaSpace{}%
\AgdaSymbol{(}\AgdaBound{v}\AgdaSpace{}%
\AgdaOperator{\AgdaInductiveConstructor{,}}\AgdaSpace{}%
\AgdaInductiveConstructor{tt}\AgdaSymbol{)}\AgdaSpace{}%
\AgdaSymbol{=}\AgdaSpace{}%
\AgdaBound{v}\<%
\\
\>[0]\AgdaFunction{eval}\AgdaSpace{}%
\AgdaInductiveConstructor{uniti⋆r}\AgdaSpace{}%
\AgdaBound{v}\AgdaSpace{}%
\AgdaSymbol{=}\AgdaSpace{}%
\AgdaSymbol{(}\AgdaBound{v}\AgdaSpace{}%
\AgdaOperator{\AgdaInductiveConstructor{,}}\AgdaSpace{}%
\AgdaInductiveConstructor{tt}\AgdaSymbol{)}\<%
\\
\>[0]\AgdaFunction{eval}\AgdaSpace{}%
\AgdaInductiveConstructor{swap⋆}\AgdaSpace{}%
\AgdaSymbol{(}\AgdaBound{v₁}\AgdaSpace{}%
\AgdaOperator{\AgdaInductiveConstructor{,}}\AgdaSpace{}%
\AgdaBound{v₂}\AgdaSymbol{)}%
\>[30]\AgdaSymbol{=}\AgdaSpace{}%
\AgdaSymbol{(}\AgdaBound{v₂}\AgdaSpace{}%
\AgdaOperator{\AgdaInductiveConstructor{,}}\AgdaSpace{}%
\AgdaBound{v₁}\AgdaSymbol{)}\<%
\\
\>[0]\AgdaFunction{eval}\AgdaSpace{}%
\AgdaInductiveConstructor{assocl⋆}\AgdaSpace{}%
\AgdaSymbol{(}\AgdaBound{v₁}\AgdaSpace{}%
\AgdaOperator{\AgdaInductiveConstructor{,}}\AgdaSpace{}%
\AgdaSymbol{(}\AgdaBound{v₂}\AgdaSpace{}%
\AgdaOperator{\AgdaInductiveConstructor{,}}\AgdaSpace{}%
\AgdaBound{v₃}\AgdaSymbol{))}\AgdaSpace{}%
\AgdaSymbol{=}\AgdaSpace{}%
\AgdaSymbol{((}\AgdaBound{v₁}\AgdaSpace{}%
\AgdaOperator{\AgdaInductiveConstructor{,}}\AgdaSpace{}%
\AgdaBound{v₂}\AgdaSymbol{)}\AgdaSpace{}%
\AgdaOperator{\AgdaInductiveConstructor{,}}\AgdaSpace{}%
\AgdaBound{v₃}\AgdaSymbol{)}\<%
\\
\>[0]\AgdaFunction{eval}\AgdaSpace{}%
\AgdaInductiveConstructor{assocr⋆}\AgdaSpace{}%
\AgdaSymbol{((}\AgdaBound{v₁}\AgdaSpace{}%
\AgdaOperator{\AgdaInductiveConstructor{,}}\AgdaSpace{}%
\AgdaBound{v₂}\AgdaSymbol{)}\AgdaSpace{}%
\AgdaOperator{\AgdaInductiveConstructor{,}}\AgdaSpace{}%
\AgdaBound{v₃}\AgdaSymbol{)}\AgdaSpace{}%
\AgdaSymbol{=}\AgdaSpace{}%
\AgdaSymbol{(}\AgdaBound{v₁}\AgdaSpace{}%
\AgdaOperator{\AgdaInductiveConstructor{,}}\AgdaSpace{}%
\AgdaSymbol{(}\AgdaBound{v₂}\AgdaSpace{}%
\AgdaOperator{\AgdaInductiveConstructor{,}}\AgdaSpace{}%
\AgdaBound{v₃}\AgdaSymbol{))}\<%
\\
\>[0]\AgdaFunction{eval}\AgdaSpace{}%
\AgdaInductiveConstructor{absorbl}\AgdaSpace{}%
\AgdaSymbol{()}\<%
\\
\>[0]\AgdaFunction{eval}\AgdaSpace{}%
\AgdaInductiveConstructor{absorbr}\AgdaSpace{}%
\AgdaSymbol{()}\<%
\\
\>[0]\AgdaFunction{eval}\AgdaSpace{}%
\AgdaInductiveConstructor{factorzl}\AgdaSpace{}%
\AgdaSymbol{()}\<%
\\
\>[0]\AgdaFunction{eval}\AgdaSpace{}%
\AgdaInductiveConstructor{factorzr}\AgdaSpace{}%
\AgdaSymbol{()}\<%
\\
\>[0]\AgdaFunction{eval}\AgdaSpace{}%
\AgdaInductiveConstructor{dist}\AgdaSpace{}%
\AgdaSymbol{(}\AgdaInductiveConstructor{inj₁}\AgdaSpace{}%
\AgdaBound{v₁}\AgdaSpace{}%
\AgdaOperator{\AgdaInductiveConstructor{,}}\AgdaSpace{}%
\AgdaBound{v₃}\AgdaSymbol{)}\AgdaSpace{}%
\AgdaSymbol{=}\AgdaSpace{}%
\AgdaInductiveConstructor{inj₁}\AgdaSpace{}%
\AgdaSymbol{(}\AgdaBound{v₁}\AgdaSpace{}%
\AgdaOperator{\AgdaInductiveConstructor{,}}\AgdaSpace{}%
\AgdaBound{v₃}\AgdaSymbol{)}\<%
\\
\>[0]\AgdaFunction{eval}\AgdaSpace{}%
\AgdaInductiveConstructor{dist}\AgdaSpace{}%
\AgdaSymbol{(}\AgdaInductiveConstructor{inj₂}\AgdaSpace{}%
\AgdaBound{v₂}\AgdaSpace{}%
\AgdaOperator{\AgdaInductiveConstructor{,}}\AgdaSpace{}%
\AgdaBound{v₃}\AgdaSymbol{)}\AgdaSpace{}%
\AgdaSymbol{=}\AgdaSpace{}%
\AgdaInductiveConstructor{inj₂}\AgdaSpace{}%
\AgdaSymbol{(}\AgdaBound{v₂}\AgdaSpace{}%
\AgdaOperator{\AgdaInductiveConstructor{,}}\AgdaSpace{}%
\AgdaBound{v₃}\AgdaSymbol{)}\<%
\\
\>[0]\AgdaFunction{eval}\AgdaSpace{}%
\AgdaInductiveConstructor{factor}\AgdaSpace{}%
\AgdaSymbol{(}\AgdaInductiveConstructor{inj₁}\AgdaSpace{}%
\AgdaSymbol{(}\AgdaBound{v₁}\AgdaSpace{}%
\AgdaOperator{\AgdaInductiveConstructor{,}}\AgdaSpace{}%
\AgdaBound{v₃}\AgdaSymbol{))}\AgdaSpace{}%
\AgdaSymbol{=}\AgdaSpace{}%
\AgdaSymbol{(}\AgdaInductiveConstructor{inj₁}\AgdaSpace{}%
\AgdaBound{v₁}\AgdaSpace{}%
\AgdaOperator{\AgdaInductiveConstructor{,}}\AgdaSpace{}%
\AgdaBound{v₃}\AgdaSymbol{)}\<%
\\
\>[0]\AgdaFunction{eval}\AgdaSpace{}%
\AgdaInductiveConstructor{factor}\AgdaSpace{}%
\AgdaSymbol{(}\AgdaInductiveConstructor{inj₂}\AgdaSpace{}%
\AgdaSymbol{(}\AgdaBound{v₂}\AgdaSpace{}%
\AgdaOperator{\AgdaInductiveConstructor{,}}\AgdaSpace{}%
\AgdaBound{v₃}\AgdaSymbol{))}\AgdaSpace{}%
\AgdaSymbol{=}\AgdaSpace{}%
\AgdaSymbol{(}\AgdaInductiveConstructor{inj₂}\AgdaSpace{}%
\AgdaBound{v₂}\AgdaSpace{}%
\AgdaOperator{\AgdaInductiveConstructor{,}}\AgdaSpace{}%
\AgdaBound{v₃}\AgdaSymbol{)}\<%
\\
\>[0]\AgdaFunction{eval}\AgdaSpace{}%
\AgdaInductiveConstructor{distl}\AgdaSpace{}%
\AgdaSymbol{(}\AgdaBound{v}\AgdaSpace{}%
\AgdaOperator{\AgdaInductiveConstructor{,}}\AgdaSpace{}%
\AgdaInductiveConstructor{inj₁}\AgdaSpace{}%
\AgdaBound{v₁}\AgdaSymbol{)}\AgdaSpace{}%
\AgdaSymbol{=}\AgdaSpace{}%
\AgdaInductiveConstructor{inj₁}\AgdaSpace{}%
\AgdaSymbol{(}\AgdaBound{v}\AgdaSpace{}%
\AgdaOperator{\AgdaInductiveConstructor{,}}\AgdaSpace{}%
\AgdaBound{v₁}\AgdaSymbol{)}\<%
\\
\>[0]\AgdaFunction{eval}\AgdaSpace{}%
\AgdaInductiveConstructor{distl}\AgdaSpace{}%
\AgdaSymbol{(}\AgdaBound{v}\AgdaSpace{}%
\AgdaOperator{\AgdaInductiveConstructor{,}}\AgdaSpace{}%
\AgdaInductiveConstructor{inj₂}\AgdaSpace{}%
\AgdaBound{v₂}\AgdaSymbol{)}\AgdaSpace{}%
\AgdaSymbol{=}\AgdaSpace{}%
\AgdaInductiveConstructor{inj₂}\AgdaSpace{}%
\AgdaSymbol{(}\AgdaBound{v}\AgdaSpace{}%
\AgdaOperator{\AgdaInductiveConstructor{,}}\AgdaSpace{}%
\AgdaBound{v₂}\AgdaSymbol{)}\<%
\\
\>[0]\AgdaFunction{eval}\AgdaSpace{}%
\AgdaInductiveConstructor{factorl}\AgdaSpace{}%
\AgdaSymbol{(}\AgdaInductiveConstructor{inj₁}\AgdaSpace{}%
\AgdaSymbol{(}\AgdaBound{v}\AgdaSpace{}%
\AgdaOperator{\AgdaInductiveConstructor{,}}\AgdaSpace{}%
\AgdaBound{v₁}\AgdaSymbol{))}\AgdaSpace{}%
\AgdaSymbol{=}\AgdaSpace{}%
\AgdaSymbol{(}\AgdaBound{v}\AgdaSpace{}%
\AgdaOperator{\AgdaInductiveConstructor{,}}\AgdaSpace{}%
\AgdaInductiveConstructor{inj₁}\AgdaSpace{}%
\AgdaBound{v₁}\AgdaSymbol{)}\<%
\\
\>[0]\AgdaFunction{eval}\AgdaSpace{}%
\AgdaInductiveConstructor{factorl}\AgdaSpace{}%
\AgdaSymbol{(}\AgdaInductiveConstructor{inj₂}\AgdaSpace{}%
\AgdaSymbol{(}\AgdaBound{v}\AgdaSpace{}%
\AgdaOperator{\AgdaInductiveConstructor{,}}\AgdaSpace{}%
\AgdaBound{v₂}\AgdaSymbol{))}\AgdaSpace{}%
\AgdaSymbol{=}\AgdaSpace{}%
\AgdaSymbol{(}\AgdaBound{v}\AgdaSpace{}%
\AgdaOperator{\AgdaInductiveConstructor{,}}\AgdaSpace{}%
\AgdaInductiveConstructor{inj₂}\AgdaSpace{}%
\AgdaBound{v₂}\AgdaSymbol{)}\<%
\\
\>[0]\AgdaFunction{eval}\AgdaSpace{}%
\AgdaInductiveConstructor{id⟷}\AgdaSpace{}%
\AgdaBound{v}\AgdaSpace{}%
\AgdaSymbol{=}\AgdaSpace{}%
\AgdaBound{v}\<%
\\
\>[0]\AgdaFunction{eval}\AgdaSpace{}%
\AgdaSymbol{(}\AgdaBound{c₁}\AgdaSpace{}%
\AgdaOperator{\AgdaInductiveConstructor{⊚}}\AgdaSpace{}%
\AgdaBound{c₂}\AgdaSymbol{)}\AgdaSpace{}%
\AgdaBound{v}\AgdaSpace{}%
\AgdaSymbol{=}\AgdaSpace{}%
\AgdaFunction{eval}\AgdaSpace{}%
\AgdaBound{c₂}\AgdaSpace{}%
\AgdaSymbol{(}\AgdaFunction{eval}\AgdaSpace{}%
\AgdaBound{c₁}\AgdaSpace{}%
\AgdaBound{v}\AgdaSymbol{)}\<%
\\
\>[0]\AgdaFunction{eval}\AgdaSpace{}%
\AgdaSymbol{(}\AgdaBound{c₁}\AgdaSpace{}%
\AgdaOperator{\AgdaInductiveConstructor{⊕}}\AgdaSpace{}%
\AgdaBound{c₂}\AgdaSymbol{)}\AgdaSpace{}%
\AgdaSymbol{(}\AgdaInductiveConstructor{inj₁}\AgdaSpace{}%
\AgdaBound{v}\AgdaSymbol{)}\AgdaSpace{}%
\AgdaSymbol{=}\AgdaSpace{}%
\AgdaInductiveConstructor{inj₁}\AgdaSpace{}%
\AgdaSymbol{(}\AgdaFunction{eval}\AgdaSpace{}%
\AgdaBound{c₁}\AgdaSpace{}%
\AgdaBound{v}\AgdaSymbol{)}\<%
\\
\>[0]\AgdaFunction{eval}\AgdaSpace{}%
\AgdaSymbol{(}\AgdaBound{c₁}\AgdaSpace{}%
\AgdaOperator{\AgdaInductiveConstructor{⊕}}\AgdaSpace{}%
\AgdaBound{c₂}\AgdaSymbol{)}\AgdaSpace{}%
\AgdaSymbol{(}\AgdaInductiveConstructor{inj₂}\AgdaSpace{}%
\AgdaBound{v}\AgdaSymbol{)}\AgdaSpace{}%
\AgdaSymbol{=}\AgdaSpace{}%
\AgdaInductiveConstructor{inj₂}\AgdaSpace{}%
\AgdaSymbol{(}\AgdaFunction{eval}\AgdaSpace{}%
\AgdaBound{c₂}\AgdaSpace{}%
\AgdaBound{v}\AgdaSymbol{)}\<%
\\
\>[0]\AgdaFunction{eval}\AgdaSpace{}%
\AgdaSymbol{(}\AgdaBound{c₁}\AgdaSpace{}%
\AgdaOperator{\AgdaInductiveConstructor{⊗}}\AgdaSpace{}%
\AgdaBound{c₂}\AgdaSymbol{)}\AgdaSpace{}%
\AgdaSymbol{(}\AgdaBound{v₁}\AgdaSpace{}%
\AgdaOperator{\AgdaInductiveConstructor{,}}\AgdaSpace{}%
\AgdaBound{v₂}\AgdaSymbol{)}\AgdaSpace{}%
\AgdaSymbol{=}\AgdaSpace{}%
\AgdaSymbol{(}\AgdaFunction{eval}\AgdaSpace{}%
\AgdaBound{c₁}\AgdaSpace{}%
\AgdaBound{v₁}\AgdaSpace{}%
\AgdaOperator{\AgdaInductiveConstructor{,}}\AgdaSpace{}%
\AgdaFunction{eval}\AgdaSpace{}%
\AgdaBound{c₂}\AgdaSpace{}%
\AgdaBound{v₂}\AgdaSymbol{)}\<%
\\
\\[\AgdaEmptyExtraSkip]%
\>[0]\AgdaOperator{\AgdaFunction{!\AgdaUnderscore{}}}\AgdaSpace{}%
\AgdaSymbol{:}\AgdaSpace{}%
\AgdaSymbol{\{}\AgdaBound{A}\AgdaSpace{}%
\AgdaBound{B}\AgdaSpace{}%
\AgdaSymbol{:}\AgdaSpace{}%
\AgdaDatatype{𝕌}\AgdaSymbol{\}}\AgdaSpace{}%
\AgdaSymbol{→}\AgdaSpace{}%
\AgdaBound{A}\AgdaSpace{}%
\AgdaOperator{\AgdaDatatype{⟷}}\AgdaSpace{}%
\AgdaBound{B}\AgdaSpace{}%
\AgdaSymbol{→}\AgdaSpace{}%
\AgdaBound{B}\AgdaSpace{}%
\AgdaOperator{\AgdaDatatype{⟷}}\AgdaSpace{}%
\AgdaBound{A}\<%
\\
\>[0]\AgdaOperator{\AgdaFunction{!}}\AgdaSpace{}%
\AgdaInductiveConstructor{unite₊l}\AgdaSpace{}%
\AgdaSymbol{=}\AgdaSpace{}%
\AgdaInductiveConstructor{uniti₊l}\<%
\\
\>[0]\AgdaOperator{\AgdaFunction{!}}\AgdaSpace{}%
\AgdaInductiveConstructor{uniti₊l}\AgdaSpace{}%
\AgdaSymbol{=}\AgdaSpace{}%
\AgdaInductiveConstructor{unite₊l}\<%
\\
\>[0]\AgdaOperator{\AgdaFunction{!}}\AgdaSpace{}%
\AgdaInductiveConstructor{unite₊r}\AgdaSpace{}%
\AgdaSymbol{=}\AgdaSpace{}%
\AgdaInductiveConstructor{uniti₊r}\<%
\\
\>[0]\AgdaOperator{\AgdaFunction{!}}\AgdaSpace{}%
\AgdaInductiveConstructor{uniti₊r}\AgdaSpace{}%
\AgdaSymbol{=}\AgdaSpace{}%
\AgdaInductiveConstructor{unite₊r}\<%
\\
\>[0]\AgdaOperator{\AgdaFunction{!}}\AgdaSpace{}%
\AgdaInductiveConstructor{swap₊}\AgdaSpace{}%
\AgdaSymbol{=}\AgdaSpace{}%
\AgdaInductiveConstructor{swap₊}\<%
\\
\>[0]\AgdaOperator{\AgdaFunction{!}}\AgdaSpace{}%
\AgdaInductiveConstructor{assocl₊}\AgdaSpace{}%
\AgdaSymbol{=}\AgdaSpace{}%
\AgdaInductiveConstructor{assocr₊}\<%
\\
\>[0]\AgdaOperator{\AgdaFunction{!}}\AgdaSpace{}%
\AgdaInductiveConstructor{assocr₊}\AgdaSpace{}%
\AgdaSymbol{=}\AgdaSpace{}%
\AgdaInductiveConstructor{assocl₊}\<%
\\
\>[0]\AgdaOperator{\AgdaFunction{!}}\AgdaSpace{}%
\AgdaInductiveConstructor{unite⋆l}\AgdaSpace{}%
\AgdaSymbol{=}\AgdaSpace{}%
\AgdaInductiveConstructor{uniti⋆l}\<%
\\
\>[0]\AgdaOperator{\AgdaFunction{!}}\AgdaSpace{}%
\AgdaInductiveConstructor{uniti⋆l}\AgdaSpace{}%
\AgdaSymbol{=}\AgdaSpace{}%
\AgdaInductiveConstructor{unite⋆l}\<%
\\
\>[0]\AgdaOperator{\AgdaFunction{!}}\AgdaSpace{}%
\AgdaInductiveConstructor{unite⋆r}\AgdaSpace{}%
\AgdaSymbol{=}\AgdaSpace{}%
\AgdaInductiveConstructor{uniti⋆r}\<%
\\
\>[0]\AgdaOperator{\AgdaFunction{!}}\AgdaSpace{}%
\AgdaInductiveConstructor{uniti⋆r}\AgdaSpace{}%
\AgdaSymbol{=}\AgdaSpace{}%
\AgdaInductiveConstructor{unite⋆r}\<%
\\
\>[0]\AgdaOperator{\AgdaFunction{!}}\AgdaSpace{}%
\AgdaInductiveConstructor{swap⋆}\AgdaSpace{}%
\AgdaSymbol{=}\AgdaSpace{}%
\AgdaInductiveConstructor{swap⋆}\<%
\\
\>[0]\AgdaOperator{\AgdaFunction{!}}\AgdaSpace{}%
\AgdaInductiveConstructor{assocl⋆}\AgdaSpace{}%
\AgdaSymbol{=}\AgdaSpace{}%
\AgdaInductiveConstructor{assocr⋆}\<%
\\
\>[0]\AgdaOperator{\AgdaFunction{!}}\AgdaSpace{}%
\AgdaInductiveConstructor{assocr⋆}\AgdaSpace{}%
\AgdaSymbol{=}\AgdaSpace{}%
\AgdaInductiveConstructor{assocl⋆}\<%
\\
\>[0]\AgdaOperator{\AgdaFunction{!}}\AgdaSpace{}%
\AgdaInductiveConstructor{absorbr}\AgdaSpace{}%
\AgdaSymbol{=}\AgdaSpace{}%
\AgdaInductiveConstructor{factorzl}\<%
\\
\>[0]\AgdaOperator{\AgdaFunction{!}}\AgdaSpace{}%
\AgdaInductiveConstructor{absorbl}\AgdaSpace{}%
\AgdaSymbol{=}\AgdaSpace{}%
\AgdaInductiveConstructor{factorzr}\<%
\\
\>[0]\AgdaOperator{\AgdaFunction{!}}\AgdaSpace{}%
\AgdaInductiveConstructor{factorzr}\AgdaSpace{}%
\AgdaSymbol{=}\AgdaSpace{}%
\AgdaInductiveConstructor{absorbl}\<%
\\
\>[0]\AgdaOperator{\AgdaFunction{!}}\AgdaSpace{}%
\AgdaInductiveConstructor{factorzl}\AgdaSpace{}%
\AgdaSymbol{=}\AgdaSpace{}%
\AgdaInductiveConstructor{absorbr}\<%
\\
\>[0]\AgdaOperator{\AgdaFunction{!}}\AgdaSpace{}%
\AgdaInductiveConstructor{dist}\AgdaSpace{}%
\AgdaSymbol{=}\AgdaSpace{}%
\AgdaInductiveConstructor{factor}\<%
\\
\>[0]\AgdaOperator{\AgdaFunction{!}}\AgdaSpace{}%
\AgdaInductiveConstructor{factor}\AgdaSpace{}%
\AgdaSymbol{=}\AgdaSpace{}%
\AgdaInductiveConstructor{dist}\<%
\\
\>[0]\AgdaOperator{\AgdaFunction{!}}\AgdaSpace{}%
\AgdaInductiveConstructor{distl}\AgdaSpace{}%
\AgdaSymbol{=}\AgdaSpace{}%
\AgdaInductiveConstructor{factorl}\<%
\\
\>[0]\AgdaOperator{\AgdaFunction{!}}\AgdaSpace{}%
\AgdaInductiveConstructor{factorl}\AgdaSpace{}%
\AgdaSymbol{=}\AgdaSpace{}%
\AgdaInductiveConstructor{distl}\<%
\\
\>[0]\AgdaOperator{\AgdaFunction{!}}\AgdaSpace{}%
\AgdaInductiveConstructor{id⟷}\AgdaSpace{}%
\AgdaSymbol{=}\AgdaSpace{}%
\AgdaInductiveConstructor{id⟷}\<%
\\
\>[0]\AgdaOperator{\AgdaFunction{!}}\AgdaSpace{}%
\AgdaSymbol{(}\AgdaBound{c₁}\AgdaSpace{}%
\AgdaOperator{\AgdaInductiveConstructor{⊚}}\AgdaSpace{}%
\AgdaBound{c₂}\AgdaSymbol{)}\AgdaSpace{}%
\AgdaSymbol{=}\AgdaSpace{}%
\AgdaSymbol{(}\AgdaOperator{\AgdaFunction{!}}\AgdaSpace{}%
\AgdaBound{c₂}\AgdaSymbol{)}\AgdaSpace{}%
\AgdaOperator{\AgdaInductiveConstructor{⊚}}\AgdaSpace{}%
\AgdaSymbol{(}\AgdaOperator{\AgdaFunction{!}}\AgdaSpace{}%
\AgdaBound{c₁}\AgdaSymbol{)}\<%
\\
\>[0]\AgdaOperator{\AgdaFunction{!}}\AgdaSpace{}%
\AgdaSymbol{(}\AgdaBound{c₁}\AgdaSpace{}%
\AgdaOperator{\AgdaInductiveConstructor{⊕}}\AgdaSpace{}%
\AgdaBound{c₂}\AgdaSymbol{)}\AgdaSpace{}%
\AgdaSymbol{=}\AgdaSpace{}%
\AgdaSymbol{(}\AgdaOperator{\AgdaFunction{!}}\AgdaSpace{}%
\AgdaBound{c₁}\AgdaSymbol{)}\AgdaSpace{}%
\AgdaOperator{\AgdaInductiveConstructor{⊕}}\AgdaSpace{}%
\AgdaSymbol{(}\AgdaOperator{\AgdaFunction{!}}\AgdaSpace{}%
\AgdaBound{c₂}\AgdaSymbol{)}\<%
\\
\>[0]\AgdaOperator{\AgdaFunction{!}}\AgdaSpace{}%
\AgdaSymbol{(}\AgdaBound{c₁}\AgdaSpace{}%
\AgdaOperator{\AgdaInductiveConstructor{⊗}}\AgdaSpace{}%
\AgdaBound{c₂}\AgdaSymbol{)}\AgdaSpace{}%
\AgdaSymbol{=}\AgdaSpace{}%
\AgdaSymbol{(}\AgdaOperator{\AgdaFunction{!}}\AgdaSpace{}%
\AgdaBound{c₁}\AgdaSymbol{)}\AgdaSpace{}%
\AgdaOperator{\AgdaInductiveConstructor{⊗}}\AgdaSpace{}%
\AgdaSymbol{(}\AgdaOperator{\AgdaFunction{!}}\AgdaSpace{}%
\AgdaBound{c₂}\AgdaSymbol{)}\<%
\\
\\[\AgdaEmptyExtraSkip]%
\>[0]\AgdaFunction{ΠisRev}\AgdaSpace{}%
\AgdaSymbol{:}\AgdaSpace{}%
\AgdaSymbol{∀}\AgdaSpace{}%
\AgdaSymbol{\{}\AgdaBound{A}\AgdaSpace{}%
\AgdaBound{B}\AgdaSymbol{\}}\AgdaSpace{}%
\AgdaSymbol{→}\AgdaSpace{}%
\AgdaSymbol{(}\AgdaBound{c}\AgdaSpace{}%
\AgdaSymbol{:}\AgdaSpace{}%
\AgdaBound{A}\AgdaSpace{}%
\AgdaOperator{\AgdaDatatype{⟷}}\AgdaSpace{}%
\AgdaBound{B}\AgdaSymbol{)}\AgdaSpace{}%
\AgdaSymbol{(}\AgdaBound{a}\AgdaSpace{}%
\AgdaSymbol{:}\AgdaSpace{}%
\AgdaOperator{\AgdaFunction{⟦}}\AgdaSpace{}%
\AgdaBound{A}\AgdaSpace{}%
\AgdaOperator{\AgdaFunction{⟧}}\AgdaSymbol{)}\AgdaSpace{}%
\AgdaSymbol{→}\AgdaSpace{}%
\AgdaFunction{eval}\AgdaSpace{}%
\AgdaSymbol{(}\AgdaBound{c}\AgdaSpace{}%
\AgdaOperator{\AgdaInductiveConstructor{⊚}}\AgdaSpace{}%
\AgdaOperator{\AgdaFunction{!}}\AgdaSpace{}%
\AgdaBound{c}\AgdaSymbol{)}\AgdaSpace{}%
\AgdaBound{a}\AgdaSpace{}%
\AgdaOperator{\AgdaDatatype{≡}}\AgdaSpace{}%
\AgdaBound{a}\<%
\\
\>[0]\AgdaFunction{ΠisRev}\AgdaSpace{}%
\AgdaInductiveConstructor{unite₊l}\AgdaSpace{}%
\AgdaSymbol{(}\AgdaInductiveConstructor{inj₂}\AgdaSpace{}%
\AgdaBound{y}\AgdaSymbol{)}\AgdaSpace{}%
\AgdaSymbol{=}\AgdaSpace{}%
\AgdaInductiveConstructor{refl}\<%
\\
\>[0]\AgdaFunction{ΠisRev}\AgdaSpace{}%
\AgdaInductiveConstructor{uniti₊l}\AgdaSpace{}%
\AgdaBound{a}\AgdaSpace{}%
\AgdaSymbol{=}\AgdaSpace{}%
\AgdaInductiveConstructor{refl}\<%
\\
\>[0]\AgdaFunction{ΠisRev}\AgdaSpace{}%
\AgdaInductiveConstructor{unite₊r}\AgdaSpace{}%
\AgdaSymbol{(}\AgdaInductiveConstructor{inj₁}\AgdaSpace{}%
\AgdaBound{x}\AgdaSymbol{)}\AgdaSpace{}%
\AgdaSymbol{=}\AgdaSpace{}%
\AgdaInductiveConstructor{refl}\<%
\\
\>[0]\AgdaFunction{ΠisRev}\AgdaSpace{}%
\AgdaInductiveConstructor{uniti₊r}\AgdaSpace{}%
\AgdaBound{a}\AgdaSpace{}%
\AgdaSymbol{=}\AgdaSpace{}%
\AgdaInductiveConstructor{refl}\<%
\\
\>[0]\AgdaFunction{ΠisRev}\AgdaSpace{}%
\AgdaInductiveConstructor{swap₊}\AgdaSpace{}%
\AgdaSymbol{(}\AgdaInductiveConstructor{inj₁}\AgdaSpace{}%
\AgdaBound{x}\AgdaSymbol{)}\AgdaSpace{}%
\AgdaSymbol{=}\AgdaSpace{}%
\AgdaInductiveConstructor{refl}\<%
\\
\>[0]\AgdaFunction{ΠisRev}\AgdaSpace{}%
\AgdaInductiveConstructor{swap₊}\AgdaSpace{}%
\AgdaSymbol{(}\AgdaInductiveConstructor{inj₂}\AgdaSpace{}%
\AgdaBound{y}\AgdaSymbol{)}\AgdaSpace{}%
\AgdaSymbol{=}\AgdaSpace{}%
\AgdaInductiveConstructor{refl}\<%
\\
\>[0]\AgdaFunction{ΠisRev}\AgdaSpace{}%
\AgdaInductiveConstructor{assocl₊}\AgdaSpace{}%
\AgdaSymbol{(}\AgdaInductiveConstructor{inj₁}\AgdaSpace{}%
\AgdaBound{x}\AgdaSymbol{)}\AgdaSpace{}%
\AgdaSymbol{=}\AgdaSpace{}%
\AgdaInductiveConstructor{refl}\<%
\\
\>[0]\AgdaFunction{ΠisRev}\AgdaSpace{}%
\AgdaInductiveConstructor{assocl₊}\AgdaSpace{}%
\AgdaSymbol{(}\AgdaInductiveConstructor{inj₂}\AgdaSpace{}%
\AgdaSymbol{(}\AgdaInductiveConstructor{inj₁}\AgdaSpace{}%
\AgdaBound{x}\AgdaSymbol{))}\AgdaSpace{}%
\AgdaSymbol{=}\AgdaSpace{}%
\AgdaInductiveConstructor{refl}\<%
\\
\>[0]\AgdaFunction{ΠisRev}\AgdaSpace{}%
\AgdaInductiveConstructor{assocl₊}\AgdaSpace{}%
\AgdaSymbol{(}\AgdaInductiveConstructor{inj₂}\AgdaSpace{}%
\AgdaSymbol{(}\AgdaInductiveConstructor{inj₂}\AgdaSpace{}%
\AgdaBound{y}\AgdaSymbol{))}\AgdaSpace{}%
\AgdaSymbol{=}\AgdaSpace{}%
\AgdaInductiveConstructor{refl}\<%
\\
\>[0]\AgdaFunction{ΠisRev}\AgdaSpace{}%
\AgdaInductiveConstructor{assocr₊}\AgdaSpace{}%
\AgdaSymbol{(}\AgdaInductiveConstructor{inj₁}\AgdaSpace{}%
\AgdaSymbol{(}\AgdaInductiveConstructor{inj₁}\AgdaSpace{}%
\AgdaBound{x}\AgdaSymbol{))}\AgdaSpace{}%
\AgdaSymbol{=}\AgdaSpace{}%
\AgdaInductiveConstructor{refl}\<%
\\
\>[0]\AgdaFunction{ΠisRev}\AgdaSpace{}%
\AgdaInductiveConstructor{assocr₊}\AgdaSpace{}%
\AgdaSymbol{(}\AgdaInductiveConstructor{inj₁}\AgdaSpace{}%
\AgdaSymbol{(}\AgdaInductiveConstructor{inj₂}\AgdaSpace{}%
\AgdaBound{y}\AgdaSymbol{))}\AgdaSpace{}%
\AgdaSymbol{=}\AgdaSpace{}%
\AgdaInductiveConstructor{refl}\<%
\\
\>[0]\AgdaFunction{ΠisRev}\AgdaSpace{}%
\AgdaInductiveConstructor{assocr₊}\AgdaSpace{}%
\AgdaSymbol{(}\AgdaInductiveConstructor{inj₂}\AgdaSpace{}%
\AgdaBound{y}\AgdaSymbol{)}\AgdaSpace{}%
\AgdaSymbol{=}\AgdaSpace{}%
\AgdaInductiveConstructor{refl}\<%
\\
\>[0]\AgdaFunction{ΠisRev}\AgdaSpace{}%
\AgdaInductiveConstructor{unite⋆l}\AgdaSpace{}%
\AgdaSymbol{(}\AgdaInductiveConstructor{tt}\AgdaSpace{}%
\AgdaOperator{\AgdaInductiveConstructor{,}}\AgdaSpace{}%
\AgdaBound{y}\AgdaSymbol{)}\AgdaSpace{}%
\AgdaSymbol{=}\AgdaSpace{}%
\AgdaInductiveConstructor{refl}\<%
\\
\>[0]\AgdaFunction{ΠisRev}\AgdaSpace{}%
\AgdaInductiveConstructor{uniti⋆l}\AgdaSpace{}%
\AgdaBound{a}\AgdaSpace{}%
\AgdaSymbol{=}\AgdaSpace{}%
\AgdaInductiveConstructor{refl}\<%
\\
\>[0]\AgdaFunction{ΠisRev}\AgdaSpace{}%
\AgdaInductiveConstructor{unite⋆r}\AgdaSpace{}%
\AgdaSymbol{(}\AgdaBound{x}\AgdaSpace{}%
\AgdaOperator{\AgdaInductiveConstructor{,}}\AgdaSpace{}%
\AgdaInductiveConstructor{tt}\AgdaSymbol{)}\AgdaSpace{}%
\AgdaSymbol{=}\AgdaSpace{}%
\AgdaInductiveConstructor{refl}\<%
\\
\>[0]\AgdaFunction{ΠisRev}\AgdaSpace{}%
\AgdaInductiveConstructor{uniti⋆r}\AgdaSpace{}%
\AgdaBound{a}\AgdaSpace{}%
\AgdaSymbol{=}\AgdaSpace{}%
\AgdaInductiveConstructor{refl}\<%
\\
\>[0]\AgdaFunction{ΠisRev}\AgdaSpace{}%
\AgdaInductiveConstructor{swap⋆}\AgdaSpace{}%
\AgdaSymbol{(}\AgdaBound{x}\AgdaSpace{}%
\AgdaOperator{\AgdaInductiveConstructor{,}}\AgdaSpace{}%
\AgdaBound{y}\AgdaSymbol{)}\AgdaSpace{}%
\AgdaSymbol{=}\AgdaSpace{}%
\AgdaInductiveConstructor{refl}\<%
\\
\>[0]\AgdaFunction{ΠisRev}\AgdaSpace{}%
\AgdaInductiveConstructor{assocl⋆}\AgdaSpace{}%
\AgdaSymbol{(}\AgdaBound{x}\AgdaSpace{}%
\AgdaOperator{\AgdaInductiveConstructor{,}}\AgdaSpace{}%
\AgdaSymbol{(}\AgdaBound{y}\AgdaSpace{}%
\AgdaOperator{\AgdaInductiveConstructor{,}}\AgdaSpace{}%
\AgdaBound{z}\AgdaSymbol{))}\AgdaSpace{}%
\AgdaSymbol{=}\AgdaSpace{}%
\AgdaInductiveConstructor{refl}\<%
\\
\>[0]\AgdaFunction{ΠisRev}\AgdaSpace{}%
\AgdaInductiveConstructor{assocr⋆}\AgdaSpace{}%
\AgdaSymbol{((}\AgdaBound{x}\AgdaSpace{}%
\AgdaOperator{\AgdaInductiveConstructor{,}}\AgdaSpace{}%
\AgdaBound{y}\AgdaSymbol{)}\AgdaSpace{}%
\AgdaOperator{\AgdaInductiveConstructor{,}}\AgdaSpace{}%
\AgdaBound{z}\AgdaSymbol{)}\AgdaSpace{}%
\AgdaSymbol{=}\AgdaSpace{}%
\AgdaInductiveConstructor{refl}\<%
\\
\>[0]\AgdaFunction{ΠisRev}\AgdaSpace{}%
\AgdaInductiveConstructor{dist}\AgdaSpace{}%
\AgdaSymbol{(}\AgdaInductiveConstructor{inj₁}\AgdaSpace{}%
\AgdaBound{x}\AgdaSpace{}%
\AgdaOperator{\AgdaInductiveConstructor{,}}\AgdaSpace{}%
\AgdaBound{z}\AgdaSymbol{)}\AgdaSpace{}%
\AgdaSymbol{=}\AgdaSpace{}%
\AgdaInductiveConstructor{refl}\<%
\\
\>[0]\AgdaFunction{ΠisRev}\AgdaSpace{}%
\AgdaInductiveConstructor{dist}\AgdaSpace{}%
\AgdaSymbol{(}\AgdaInductiveConstructor{inj₂}\AgdaSpace{}%
\AgdaBound{y}\AgdaSpace{}%
\AgdaOperator{\AgdaInductiveConstructor{,}}\AgdaSpace{}%
\AgdaBound{z}\AgdaSymbol{)}\AgdaSpace{}%
\AgdaSymbol{=}\AgdaSpace{}%
\AgdaInductiveConstructor{refl}\<%
\\
\>[0]\AgdaFunction{ΠisRev}\AgdaSpace{}%
\AgdaInductiveConstructor{factor}\AgdaSpace{}%
\AgdaSymbol{(}\AgdaInductiveConstructor{inj₁}\AgdaSpace{}%
\AgdaSymbol{(}\AgdaBound{x}\AgdaSpace{}%
\AgdaOperator{\AgdaInductiveConstructor{,}}\AgdaSpace{}%
\AgdaBound{z}\AgdaSymbol{))}\AgdaSpace{}%
\AgdaSymbol{=}\AgdaSpace{}%
\AgdaInductiveConstructor{refl}\<%
\\
\>[0]\AgdaFunction{ΠisRev}\AgdaSpace{}%
\AgdaInductiveConstructor{factor}\AgdaSpace{}%
\AgdaSymbol{(}\AgdaInductiveConstructor{inj₂}\AgdaSpace{}%
\AgdaSymbol{(}\AgdaBound{y}\AgdaSpace{}%
\AgdaOperator{\AgdaInductiveConstructor{,}}\AgdaSpace{}%
\AgdaBound{z}\AgdaSymbol{))}\AgdaSpace{}%
\AgdaSymbol{=}\AgdaSpace{}%
\AgdaInductiveConstructor{refl}\<%
\\
\>[0]\AgdaFunction{ΠisRev}\AgdaSpace{}%
\AgdaInductiveConstructor{distl}\AgdaSpace{}%
\AgdaSymbol{(}\AgdaBound{x}\AgdaSpace{}%
\AgdaOperator{\AgdaInductiveConstructor{,}}\AgdaSpace{}%
\AgdaInductiveConstructor{inj₁}\AgdaSpace{}%
\AgdaBound{y}\AgdaSymbol{)}\AgdaSpace{}%
\AgdaSymbol{=}\AgdaSpace{}%
\AgdaInductiveConstructor{refl}\<%
\\
\>[0]\AgdaFunction{ΠisRev}\AgdaSpace{}%
\AgdaInductiveConstructor{distl}\AgdaSpace{}%
\AgdaSymbol{(}\AgdaBound{x}\AgdaSpace{}%
\AgdaOperator{\AgdaInductiveConstructor{,}}\AgdaSpace{}%
\AgdaInductiveConstructor{inj₂}\AgdaSpace{}%
\AgdaBound{z}\AgdaSymbol{)}\AgdaSpace{}%
\AgdaSymbol{=}\AgdaSpace{}%
\AgdaInductiveConstructor{refl}\<%
\\
\>[0]\AgdaFunction{ΠisRev}\AgdaSpace{}%
\AgdaInductiveConstructor{factorl}\AgdaSpace{}%
\AgdaSymbol{(}\AgdaInductiveConstructor{inj₁}\AgdaSpace{}%
\AgdaSymbol{(}\AgdaBound{x}\AgdaSpace{}%
\AgdaOperator{\AgdaInductiveConstructor{,}}\AgdaSpace{}%
\AgdaBound{y}\AgdaSymbol{))}\AgdaSpace{}%
\AgdaSymbol{=}\AgdaSpace{}%
\AgdaInductiveConstructor{refl}\<%
\\
\>[0]\AgdaFunction{ΠisRev}\AgdaSpace{}%
\AgdaInductiveConstructor{factorl}\AgdaSpace{}%
\AgdaSymbol{(}\AgdaInductiveConstructor{inj₂}\AgdaSpace{}%
\AgdaSymbol{(}\AgdaBound{x}\AgdaSpace{}%
\AgdaOperator{\AgdaInductiveConstructor{,}}\AgdaSpace{}%
\AgdaBound{z}\AgdaSymbol{))}\AgdaSpace{}%
\AgdaSymbol{=}\AgdaSpace{}%
\AgdaInductiveConstructor{refl}\<%
\\
\>[0]\AgdaFunction{ΠisRev}\AgdaSpace{}%
\AgdaInductiveConstructor{id⟷}\AgdaSpace{}%
\AgdaBound{a}\AgdaSpace{}%
\AgdaSymbol{=}\AgdaSpace{}%
\AgdaInductiveConstructor{refl}\<%
\\
\>[0]\AgdaFunction{ΠisRev}\AgdaSpace{}%
\AgdaSymbol{(}\AgdaBound{c₁}\AgdaSpace{}%
\AgdaOperator{\AgdaInductiveConstructor{⊚}}\AgdaSpace{}%
\AgdaBound{c₂}\AgdaSymbol{)}\AgdaSpace{}%
\AgdaBound{a}\AgdaSpace{}%
\AgdaKeyword{rewrite}\AgdaSpace{}%
\AgdaFunction{ΠisRev}\AgdaSpace{}%
\AgdaBound{c₂}\AgdaSpace{}%
\AgdaSymbol{(}\AgdaFunction{eval}\AgdaSpace{}%
\AgdaBound{c₁}\AgdaSpace{}%
\AgdaBound{a}\AgdaSymbol{)}\AgdaSpace{}%
\AgdaSymbol{=}\AgdaSpace{}%
\AgdaFunction{ΠisRev}\AgdaSpace{}%
\AgdaBound{c₁}\AgdaSpace{}%
\AgdaBound{a}\<%
\\
\>[0]\AgdaFunction{ΠisRev}\AgdaSpace{}%
\AgdaSymbol{(}\AgdaBound{c₁}\AgdaSpace{}%
\AgdaOperator{\AgdaInductiveConstructor{⊕}}\AgdaSpace{}%
\AgdaBound{c₂}\AgdaSymbol{)}\AgdaSpace{}%
\AgdaSymbol{(}\AgdaInductiveConstructor{inj₁}\AgdaSpace{}%
\AgdaBound{x}\AgdaSymbol{)}\AgdaSpace{}%
\AgdaKeyword{rewrite}\AgdaSpace{}%
\AgdaFunction{ΠisRev}\AgdaSpace{}%
\AgdaBound{c₁}\AgdaSpace{}%
\AgdaBound{x}\AgdaSpace{}%
\AgdaSymbol{=}\AgdaSpace{}%
\AgdaInductiveConstructor{refl}\<%
\\
\>[0]\AgdaFunction{ΠisRev}\AgdaSpace{}%
\AgdaSymbol{(}\AgdaBound{c₁}\AgdaSpace{}%
\AgdaOperator{\AgdaInductiveConstructor{⊕}}\AgdaSpace{}%
\AgdaBound{c₂}\AgdaSymbol{)}\AgdaSpace{}%
\AgdaSymbol{(}\AgdaInductiveConstructor{inj₂}\AgdaSpace{}%
\AgdaBound{y}\AgdaSymbol{)}\AgdaSpace{}%
\AgdaKeyword{rewrite}\AgdaSpace{}%
\AgdaFunction{ΠisRev}\AgdaSpace{}%
\AgdaBound{c₂}\AgdaSpace{}%
\AgdaBound{y}\AgdaSpace{}%
\AgdaSymbol{=}\AgdaSpace{}%
\AgdaInductiveConstructor{refl}\<%
\\
\>[0]\AgdaFunction{ΠisRev}\AgdaSpace{}%
\AgdaSymbol{(}\AgdaBound{c₁}\AgdaSpace{}%
\AgdaOperator{\AgdaInductiveConstructor{⊗}}\AgdaSpace{}%
\AgdaBound{c₂}\AgdaSymbol{)}\AgdaSpace{}%
\AgdaSymbol{(}\AgdaBound{x}\AgdaSpace{}%
\AgdaOperator{\AgdaInductiveConstructor{,}}\AgdaSpace{}%
\AgdaBound{y}\AgdaSymbol{)}\AgdaSpace{}%
\AgdaKeyword{rewrite}\AgdaSpace{}%
\AgdaFunction{ΠisRev}\AgdaSpace{}%
\AgdaBound{c₁}\AgdaSpace{}%
\AgdaBound{x}\AgdaSpace{}%
\AgdaSymbol{|}\AgdaSpace{}%
\AgdaFunction{ΠisRev}\AgdaSpace{}%
\AgdaBound{c₂}\AgdaSpace{}%
\AgdaBound{y}\AgdaSpace{}%
\AgdaSymbol{=}\AgdaSpace{}%
\AgdaInductiveConstructor{refl}\<%
\end{code}

\newcommand{\Bexamples}{%
\begin{code}%
\>[0]\AgdaKeyword{pattern}\AgdaSpace{}%
\AgdaInductiveConstructor{𝔽}\AgdaSpace{}%
\AgdaSymbol{=}\AgdaSpace{}%
\AgdaInductiveConstructor{inj₁}\AgdaSpace{}%
\AgdaInductiveConstructor{tt}\<%
\\
\>[0]\AgdaKeyword{pattern}\AgdaSpace{}%
\AgdaInductiveConstructor{𝕋}\AgdaSpace{}%
\AgdaSymbol{=}\AgdaSpace{}%
\AgdaInductiveConstructor{inj₂}\AgdaSpace{}%
\AgdaInductiveConstructor{tt}\<%
\\
\\[\AgdaEmptyExtraSkip]%
\>[0]\AgdaFunction{𝔹}\AgdaSpace{}%
\AgdaFunction{𝔹²}\AgdaSpace{}%
\AgdaFunction{𝔹³}\AgdaSpace{}%
\AgdaSymbol{:}\AgdaSpace{}%
\AgdaDatatype{𝕌}\<%
\\
\>[0]\AgdaFunction{𝔹}%
\>[4]\AgdaSymbol{=}\AgdaSpace{}%
\AgdaInductiveConstructor{𝟙}\AgdaSpace{}%
\AgdaOperator{\AgdaInductiveConstructor{+ᵤ}}\AgdaSpace{}%
\AgdaInductiveConstructor{𝟙}\<%
\\
\>[0]\AgdaFunction{𝔹²}%
\>[4]\AgdaSymbol{=}\AgdaSpace{}%
\AgdaFunction{𝔹}\AgdaSpace{}%
\AgdaOperator{\AgdaInductiveConstructor{×ᵤ}}\AgdaSpace{}%
\AgdaFunction{𝔹}\<%
\\
\>[0]\AgdaFunction{𝔹³}%
\>[4]\AgdaSymbol{=}\AgdaSpace{}%
\AgdaFunction{𝔹}\AgdaSpace{}%
\AgdaOperator{\AgdaInductiveConstructor{×ᵤ}}\AgdaSpace{}%
\AgdaFunction{𝔹²}\<%
\\
\\[\AgdaEmptyExtraSkip]%
\>[0]\AgdaFunction{ctrl}\AgdaSpace{}%
\AgdaSymbol{:}\AgdaSpace{}%
\AgdaSymbol{\{}\AgdaBound{A}\AgdaSpace{}%
\AgdaSymbol{:}\AgdaSpace{}%
\AgdaDatatype{𝕌}\AgdaSymbol{\}}\AgdaSpace{}%
\AgdaSymbol{→}\AgdaSpace{}%
\AgdaSymbol{(}\AgdaBound{A}\AgdaSpace{}%
\AgdaOperator{\AgdaDatatype{⟷}}\AgdaSpace{}%
\AgdaBound{A}\AgdaSymbol{)}\AgdaSpace{}%
\AgdaSymbol{→}\AgdaSpace{}%
\AgdaFunction{𝔹}\AgdaSpace{}%
\AgdaOperator{\AgdaInductiveConstructor{×ᵤ}}\AgdaSpace{}%
\AgdaBound{A}\AgdaSpace{}%
\AgdaOperator{\AgdaDatatype{⟷}}\AgdaSpace{}%
\AgdaFunction{𝔹}\AgdaSpace{}%
\AgdaOperator{\AgdaInductiveConstructor{×ᵤ}}\AgdaSpace{}%
\AgdaBound{A}\<%
\\
\>[0]\AgdaFunction{ctrl}\AgdaSpace{}%
\AgdaBound{c}\AgdaSpace{}%
\AgdaSymbol{=}\AgdaSpace{}%
\AgdaInductiveConstructor{dist}\AgdaSpace{}%
\AgdaOperator{\AgdaInductiveConstructor{⊚}}\AgdaSpace{}%
\AgdaSymbol{(}\AgdaInductiveConstructor{id⟷}\AgdaSpace{}%
\AgdaOperator{\AgdaInductiveConstructor{⊕}}\AgdaSpace{}%
\AgdaSymbol{(}\AgdaInductiveConstructor{id⟷}\AgdaSpace{}%
\AgdaOperator{\AgdaInductiveConstructor{⊗}}\AgdaSpace{}%
\AgdaBound{c}\AgdaSymbol{))}\AgdaSpace{}%
\AgdaOperator{\AgdaInductiveConstructor{⊚}}\AgdaSpace{}%
\AgdaInductiveConstructor{factor}\<%
\\
\\[\AgdaEmptyExtraSkip]%
\>[0]\AgdaFunction{NOT}\AgdaSpace{}%
\AgdaSymbol{:}\AgdaSpace{}%
\AgdaFunction{𝔹}\AgdaSpace{}%
\AgdaOperator{\AgdaDatatype{⟷}}\AgdaSpace{}%
\AgdaFunction{𝔹}\<%
\\
\>[0]\AgdaFunction{NOT}\AgdaSpace{}%
\AgdaSymbol{=}\AgdaSpace{}%
\AgdaInductiveConstructor{swap₊}\<%
\\
\\[\AgdaEmptyExtraSkip]%
\>[0]\AgdaFunction{CNOT}\AgdaSpace{}%
\AgdaSymbol{:}\AgdaSpace{}%
\AgdaFunction{𝔹²}\AgdaSpace{}%
\AgdaOperator{\AgdaDatatype{⟷}}\AgdaSpace{}%
\AgdaFunction{𝔹²}\<%
\\
\>[0]\AgdaFunction{CNOT}\AgdaSpace{}%
\AgdaSymbol{=}\AgdaSpace{}%
\AgdaFunction{ctrl}\AgdaSpace{}%
\AgdaFunction{NOT}\<%
\\
\\[\AgdaEmptyExtraSkip]%
\>[0]\AgdaFunction{TOFFOLI}\AgdaSpace{}%
\AgdaSymbol{:}\AgdaSpace{}%
\AgdaFunction{𝔹³}\AgdaSpace{}%
\AgdaOperator{\AgdaDatatype{⟷}}\AgdaSpace{}%
\AgdaFunction{𝔹³}\<%
\\
\>[0]\AgdaFunction{TOFFOLI}\AgdaSpace{}%
\AgdaSymbol{=}\AgdaSpace{}%
\AgdaFunction{ctrl}\AgdaSpace{}%
\AgdaSymbol{(}\AgdaFunction{ctrl}\AgdaSpace{}%
\AgdaFunction{NOT}\AgdaSymbol{)}\<%
\end{code}}

\begin{code}[hide]%
\>[0]\AgdaSymbol{\{-\#}\AgdaSpace{}%
\AgdaKeyword{OPTIONS}\AgdaSpace{}%
\AgdaPragma{--without-K}\AgdaSpace{}%
\AgdaPragma{--safe}\AgdaSpace{}%
\AgdaSymbol{\#-\}}\<%
\\
\>[0]\AgdaKeyword{module}\AgdaSpace{}%
\AgdaModule{Pi.}\AgdaOperator{\AgdaModule{D\AgdaUnderscore{}default}}\AgdaSpace{}%
\AgdaKeyword{where}\<%
\\
\>[0]\AgdaKeyword{open}\AgdaSpace{}%
\AgdaKeyword{import}\AgdaSpace{}%
\AgdaModule{Data.Empty}\<%
\\
\>[0]\AgdaKeyword{open}\AgdaSpace{}%
\AgdaKeyword{import}\AgdaSpace{}%
\AgdaModule{Data.Unit}\<%
\\
\>[0]\AgdaKeyword{open}\AgdaSpace{}%
\AgdaKeyword{import}\AgdaSpace{}%
\AgdaModule{Data.Nat}\<%
\\
\>[0]\AgdaKeyword{open}\AgdaSpace{}%
\AgdaKeyword{import}\AgdaSpace{}%
\AgdaModule{Data.Nat.Properties}\<%
\\
\>[0]\AgdaKeyword{open}\AgdaSpace{}%
\AgdaKeyword{import}\AgdaSpace{}%
\AgdaModule{Data.Sum}\<%
\\
\>[0]\AgdaKeyword{open}\AgdaSpace{}%
\AgdaKeyword{import}\AgdaSpace{}%
\AgdaModule{Data.Product}\<%
\\
\>[0]\AgdaKeyword{open}\AgdaSpace{}%
\AgdaKeyword{import}\AgdaSpace{}%
\AgdaModule{Data.Maybe}\<%
\\
\>[0]\AgdaKeyword{open}\AgdaSpace{}%
\AgdaKeyword{import}\AgdaSpace{}%
\AgdaModule{Function}\AgdaSpace{}%
\AgdaKeyword{using}\AgdaSpace{}%
\AgdaSymbol{(}\AgdaOperator{\AgdaFunction{\AgdaUnderscore{}∘\AgdaUnderscore{}}}\AgdaSymbol{)}\<%
\\
\>[0]\AgdaKeyword{open}\AgdaSpace{}%
\AgdaKeyword{import}\AgdaSpace{}%
\AgdaModule{Relation.Binary.PropositionalEquality}\<%
\\
\>[0]\AgdaKeyword{open}\AgdaSpace{}%
\AgdaKeyword{import}\AgdaSpace{}%
\AgdaModule{Relation.Binary.Core}\<%
\\
\>[0]\AgdaKeyword{open}\AgdaSpace{}%
\AgdaKeyword{import}\AgdaSpace{}%
\AgdaModule{Relation.Binary}\<%
\\
\>[0]\AgdaKeyword{open}\AgdaSpace{}%
\AgdaKeyword{import}\AgdaSpace{}%
\AgdaModule{Relation.Nullary}\<%
\\
\\[\AgdaEmptyExtraSkip]%
\>[0]\AgdaKeyword{infix}%
\>[7]\AgdaNumber{99}\AgdaSpace{}%
\AgdaOperator{\AgdaInductiveConstructor{𝟙/\AgdaUnderscore{}}}\<%
\\
\>[0]\AgdaKeyword{infixr}\AgdaSpace{}%
\AgdaNumber{70}\AgdaSpace{}%
\AgdaOperator{\AgdaInductiveConstructor{\AgdaUnderscore{}×ᵤ\AgdaUnderscore{}}}\<%
\\
\>[0]\AgdaKeyword{infixr}\AgdaSpace{}%
\AgdaNumber{60}\AgdaSpace{}%
\AgdaOperator{\AgdaInductiveConstructor{\AgdaUnderscore{}+ᵤ\AgdaUnderscore{}}}\<%
\\
\>[0]\AgdaKeyword{infix}%
\>[7]\AgdaNumber{40}\AgdaSpace{}%
\AgdaOperator{\AgdaDatatype{\AgdaUnderscore{}⟷\AgdaUnderscore{}}}\<%
\\
\>[0]\AgdaKeyword{infixr}\AgdaSpace{}%
\AgdaNumber{50}\AgdaSpace{}%
\AgdaOperator{\AgdaInductiveConstructor{\AgdaUnderscore{}⊚\AgdaUnderscore{}}}\<%
\\
\\[\AgdaEmptyExtraSkip]%
\>[0]\AgdaKeyword{data}\AgdaSpace{}%
\AgdaDatatype{◯}\AgdaSpace{}%
\AgdaSymbol{:}\AgdaSpace{}%
\AgdaPrimitiveType{Set}\AgdaSpace{}%
\AgdaKeyword{where}\<%
\\
\>[0][@{}l@{\AgdaIndent{0}}]%
\>[2]\AgdaInductiveConstructor{↻}\AgdaSpace{}%
\AgdaSymbol{:}\AgdaSpace{}%
\AgdaDatatype{◯}\<%
\\
\\[\AgdaEmptyExtraSkip]%
\>[0]\AgdaKeyword{data}\AgdaSpace{}%
\AgdaDatatype{𝕌}\AgdaSpace{}%
\AgdaSymbol{:}\AgdaSpace{}%
\AgdaPrimitiveType{Set}\AgdaSpace{}%
\AgdaKeyword{where}\<%
\\
\>[0][@{}l@{\AgdaIndent{0}}]%
\>[2]\AgdaInductiveConstructor{𝟘}%
\>[8]\AgdaSymbol{:}\AgdaSpace{}%
\AgdaDatatype{𝕌}\<%
\\
\>[2]\AgdaInductiveConstructor{𝟙}%
\>[8]\AgdaSymbol{:}\AgdaSpace{}%
\AgdaDatatype{𝕌}\<%
\\
\>[2]\AgdaOperator{\AgdaInductiveConstructor{\AgdaUnderscore{}+ᵤ\AgdaUnderscore{}}}%
\>[8]\AgdaSymbol{:}\AgdaSpace{}%
\AgdaDatatype{𝕌}\AgdaSpace{}%
\AgdaSymbol{→}\AgdaSpace{}%
\AgdaDatatype{𝕌}\AgdaSpace{}%
\AgdaSymbol{→}\AgdaSpace{}%
\AgdaDatatype{𝕌}\<%
\\
\>[2]\AgdaOperator{\AgdaInductiveConstructor{\AgdaUnderscore{}×ᵤ\AgdaUnderscore{}}}%
\>[8]\AgdaSymbol{:}\AgdaSpace{}%
\AgdaDatatype{𝕌}\AgdaSpace{}%
\AgdaSymbol{→}\AgdaSpace{}%
\AgdaDatatype{𝕌}\AgdaSpace{}%
\AgdaSymbol{→}\AgdaSpace{}%
\AgdaDatatype{𝕌}\<%
\\
\>[2]\AgdaOperator{\AgdaInductiveConstructor{𝟙/\AgdaUnderscore{}}}%
\>[8]\AgdaSymbol{:}\AgdaSpace{}%
\AgdaDatatype{𝕌}\AgdaSpace{}%
\AgdaSymbol{→}\AgdaSpace{}%
\AgdaDatatype{𝕌}\<%
\\
\\[\AgdaEmptyExtraSkip]%
\>[0]\AgdaOperator{\AgdaFunction{⟦\AgdaUnderscore{}⟧}}\AgdaSpace{}%
\AgdaSymbol{:}\AgdaSpace{}%
\AgdaDatatype{𝕌}\AgdaSpace{}%
\AgdaSymbol{→}\AgdaSpace{}%
\AgdaPrimitiveType{Set}\<%
\\
\>[0]\AgdaOperator{\AgdaFunction{⟦}}\AgdaSpace{}%
\AgdaInductiveConstructor{𝟘}\AgdaSpace{}%
\AgdaOperator{\AgdaFunction{⟧}}\AgdaSpace{}%
\AgdaSymbol{=}\AgdaSpace{}%
\AgdaDatatype{⊥}\<%
\\
\>[0]\AgdaOperator{\AgdaFunction{⟦}}\AgdaSpace{}%
\AgdaInductiveConstructor{𝟙}\AgdaSpace{}%
\AgdaOperator{\AgdaFunction{⟧}}\AgdaSpace{}%
\AgdaSymbol{=}\AgdaSpace{}%
\AgdaRecord{⊤}\<%
\\
\>[0]\AgdaOperator{\AgdaFunction{⟦}}\AgdaSpace{}%
\AgdaBound{t₁}\AgdaSpace{}%
\AgdaOperator{\AgdaInductiveConstructor{+ᵤ}}\AgdaSpace{}%
\AgdaBound{t₂}\AgdaSpace{}%
\AgdaOperator{\AgdaFunction{⟧}}\AgdaSpace{}%
\AgdaSymbol{=}\AgdaSpace{}%
\AgdaOperator{\AgdaFunction{⟦}}\AgdaSpace{}%
\AgdaBound{t₁}\AgdaSpace{}%
\AgdaOperator{\AgdaFunction{⟧}}\AgdaSpace{}%
\AgdaOperator{\AgdaDatatype{⊎}}\AgdaSpace{}%
\AgdaOperator{\AgdaFunction{⟦}}\AgdaSpace{}%
\AgdaBound{t₂}\AgdaSpace{}%
\AgdaOperator{\AgdaFunction{⟧}}\<%
\\
\>[0]\AgdaOperator{\AgdaFunction{⟦}}\AgdaSpace{}%
\AgdaBound{t₁}\AgdaSpace{}%
\AgdaOperator{\AgdaInductiveConstructor{×ᵤ}}\AgdaSpace{}%
\AgdaBound{t₂}\AgdaSpace{}%
\AgdaOperator{\AgdaFunction{⟧}}\AgdaSpace{}%
\AgdaSymbol{=}\AgdaSpace{}%
\AgdaOperator{\AgdaFunction{⟦}}\AgdaSpace{}%
\AgdaBound{t₁}\AgdaSpace{}%
\AgdaOperator{\AgdaFunction{⟧}}\AgdaSpace{}%
\AgdaOperator{\AgdaFunction{×}}\AgdaSpace{}%
\AgdaOperator{\AgdaFunction{⟦}}\AgdaSpace{}%
\AgdaBound{t₂}\AgdaSpace{}%
\AgdaOperator{\AgdaFunction{⟧}}\<%
\\
\>[0]\AgdaOperator{\AgdaFunction{⟦}}\AgdaSpace{}%
\AgdaOperator{\AgdaInductiveConstructor{𝟙/}}\AgdaSpace{}%
\AgdaBound{t}\AgdaSpace{}%
\AgdaOperator{\AgdaFunction{⟧}}\AgdaSpace{}%
\AgdaSymbol{=}\AgdaSpace{}%
\AgdaDatatype{◯}\<%
\\
\\[\AgdaEmptyExtraSkip]%
\>[0]\AgdaOperator{\AgdaFunction{\AgdaUnderscore{}≠𝟘}}\AgdaSpace{}%
\AgdaSymbol{:}\AgdaSpace{}%
\AgdaDatatype{𝕌}\AgdaSpace{}%
\AgdaSymbol{→}\AgdaSpace{}%
\AgdaPrimitiveType{Set}\<%
\\
\>[0]\AgdaInductiveConstructor{𝟘}\AgdaSpace{}%
\AgdaOperator{\AgdaFunction{≠𝟘}}\AgdaSpace{}%
\AgdaSymbol{=}\AgdaSpace{}%
\AgdaDatatype{⊥}\<%
\\
\>[0]\AgdaInductiveConstructor{𝟙}\AgdaSpace{}%
\AgdaOperator{\AgdaFunction{≠𝟘}}\AgdaSpace{}%
\AgdaSymbol{=}\AgdaSpace{}%
\AgdaRecord{⊤}\<%
\\
\>[0]\AgdaSymbol{(}\AgdaBound{t₁}\AgdaSpace{}%
\AgdaOperator{\AgdaInductiveConstructor{+ᵤ}}\AgdaSpace{}%
\AgdaBound{t₂}\AgdaSymbol{)}\AgdaSpace{}%
\AgdaOperator{\AgdaFunction{≠𝟘}}\AgdaSpace{}%
\AgdaSymbol{=}\AgdaSpace{}%
\AgdaSymbol{(}\AgdaBound{t₁}\AgdaSpace{}%
\AgdaOperator{\AgdaFunction{≠𝟘}}\AgdaSymbol{)}\AgdaSpace{}%
\AgdaOperator{\AgdaDatatype{⊎}}\AgdaSpace{}%
\AgdaSymbol{(}\AgdaBound{t₂}\AgdaSpace{}%
\AgdaOperator{\AgdaFunction{≠𝟘}}\AgdaSymbol{)}\<%
\\
\>[0]\AgdaSymbol{(}\AgdaBound{t₁}\AgdaSpace{}%
\AgdaOperator{\AgdaInductiveConstructor{×ᵤ}}\AgdaSpace{}%
\AgdaBound{t₂}\AgdaSymbol{)}\AgdaSpace{}%
\AgdaOperator{\AgdaFunction{≠𝟘}}\AgdaSpace{}%
\AgdaSymbol{=}\AgdaSpace{}%
\AgdaSymbol{(}\AgdaBound{t₁}\AgdaSpace{}%
\AgdaOperator{\AgdaFunction{≠𝟘}}\AgdaSymbol{)}\AgdaSpace{}%
\AgdaOperator{\AgdaFunction{×}}\AgdaSpace{}%
\AgdaSymbol{(}\AgdaBound{t₂}\AgdaSpace{}%
\AgdaOperator{\AgdaFunction{≠𝟘}}\AgdaSymbol{)}\<%
\\
\>[0]\AgdaSymbol{(}\AgdaOperator{\AgdaInductiveConstructor{𝟙/}}\AgdaSpace{}%
\AgdaBound{t}\AgdaSymbol{)}\AgdaSpace{}%
\AgdaOperator{\AgdaFunction{≠𝟘}}\AgdaSpace{}%
\AgdaSymbol{=}\AgdaSpace{}%
\AgdaBound{t}\AgdaSpace{}%
\AgdaOperator{\AgdaFunction{≠𝟘}}\<%
\\
\\[\AgdaEmptyExtraSkip]%
\>[0]\AgdaFunction{default}\AgdaSpace{}%
\AgdaSymbol{:}\AgdaSpace{}%
\AgdaSymbol{\{}\AgdaBound{t}\AgdaSpace{}%
\AgdaSymbol{:}\AgdaSpace{}%
\AgdaDatatype{𝕌}\AgdaSymbol{\}}\AgdaSpace{}%
\AgdaSymbol{→}\AgdaSpace{}%
\AgdaSymbol{(}\AgdaBound{t≠𝟘}\AgdaSpace{}%
\AgdaSymbol{:}\AgdaSpace{}%
\AgdaBound{t}\AgdaSpace{}%
\AgdaOperator{\AgdaFunction{≠𝟘}}\AgdaSymbol{)}\AgdaSpace{}%
\AgdaSymbol{→}\AgdaSpace{}%
\AgdaOperator{\AgdaFunction{⟦}}\AgdaSpace{}%
\AgdaBound{t}\AgdaSpace{}%
\AgdaOperator{\AgdaFunction{⟧}}\<%
\\
\>[0]\AgdaFunction{default}\AgdaSpace{}%
\AgdaSymbol{\{}\AgdaInductiveConstructor{𝟙}\AgdaSymbol{\}}\AgdaSpace{}%
\AgdaBound{t≠𝟘}\AgdaSpace{}%
\AgdaSymbol{=}\AgdaSpace{}%
\AgdaInductiveConstructor{tt}\<%
\\
\>[0]\AgdaFunction{default}\AgdaSpace{}%
\AgdaSymbol{\{}\AgdaBound{t₁}\AgdaSpace{}%
\AgdaOperator{\AgdaInductiveConstructor{+ᵤ}}\AgdaSpace{}%
\AgdaBound{t₂}\AgdaSymbol{\}}\AgdaSpace{}%
\AgdaSymbol{(}\AgdaInductiveConstructor{inj₁}\AgdaSpace{}%
\AgdaBound{t₁≠𝟘}\AgdaSymbol{)}\AgdaSpace{}%
\AgdaSymbol{=}\AgdaSpace{}%
\AgdaInductiveConstructor{inj₁}\AgdaSpace{}%
\AgdaSymbol{(}\AgdaFunction{default}\AgdaSpace{}%
\AgdaBound{t₁≠𝟘}\AgdaSymbol{)}\<%
\\
\>[0]\AgdaFunction{default}\AgdaSpace{}%
\AgdaSymbol{\{}\AgdaBound{t₁}\AgdaSpace{}%
\AgdaOperator{\AgdaInductiveConstructor{+ᵤ}}\AgdaSpace{}%
\AgdaBound{t₂}\AgdaSymbol{\}}\AgdaSpace{}%
\AgdaSymbol{(}\AgdaInductiveConstructor{inj₂}\AgdaSpace{}%
\AgdaBound{t₂≠𝟘}\AgdaSymbol{)}\AgdaSpace{}%
\AgdaSymbol{=}\AgdaSpace{}%
\AgdaInductiveConstructor{inj₂}\AgdaSpace{}%
\AgdaSymbol{(}\AgdaFunction{default}\AgdaSpace{}%
\AgdaBound{t₂≠𝟘}\AgdaSymbol{)}\<%
\\
\>[0]\AgdaFunction{default}\AgdaSpace{}%
\AgdaSymbol{\{}\AgdaBound{t₁}\AgdaSpace{}%
\AgdaOperator{\AgdaInductiveConstructor{×ᵤ}}\AgdaSpace{}%
\AgdaBound{t₂}\AgdaSymbol{\}}\AgdaSpace{}%
\AgdaSymbol{(}\AgdaBound{t₁≠𝟘}\AgdaSpace{}%
\AgdaOperator{\AgdaInductiveConstructor{,}}\AgdaSpace{}%
\AgdaBound{t₂≠𝟘}\AgdaSymbol{)}\AgdaSpace{}%
\AgdaSymbol{=}\AgdaSpace{}%
\AgdaFunction{default}\AgdaSpace{}%
\AgdaBound{t₁≠𝟘}\AgdaSpace{}%
\AgdaOperator{\AgdaInductiveConstructor{,}}\AgdaSpace{}%
\AgdaFunction{default}\AgdaSpace{}%
\AgdaBound{t₂≠𝟘}\<%
\\
\>[0]\AgdaFunction{default}\AgdaSpace{}%
\AgdaSymbol{\{}\AgdaOperator{\AgdaInductiveConstructor{𝟙/}}\AgdaSpace{}%
\AgdaBound{t}\AgdaSymbol{\}}\AgdaSpace{}%
\AgdaBound{t≠𝟘}\AgdaSpace{}%
\AgdaSymbol{=}\AgdaSpace{}%
\AgdaInductiveConstructor{↻}\<%
\\
\\[\AgdaEmptyExtraSkip]%
\>[0]\AgdaKeyword{data}\AgdaSpace{}%
\AgdaOperator{\AgdaDatatype{\AgdaUnderscore{}⟷\AgdaUnderscore{}}}\AgdaSpace{}%
\AgdaSymbol{:}\AgdaSpace{}%
\AgdaDatatype{𝕌}\AgdaSpace{}%
\AgdaSymbol{→}\AgdaSpace{}%
\AgdaDatatype{𝕌}\AgdaSpace{}%
\AgdaSymbol{→}\AgdaSpace{}%
\AgdaPrimitiveType{Set}\AgdaSpace{}%
\AgdaKeyword{where}\<%
\\
\>[0][@{}l@{\AgdaIndent{0}}]%
\>[2]\AgdaInductiveConstructor{unite₊l}\AgdaSpace{}%
\AgdaSymbol{:}\AgdaSpace{}%
\AgdaSymbol{\{}\AgdaBound{t}\AgdaSpace{}%
\AgdaSymbol{:}\AgdaSpace{}%
\AgdaDatatype{𝕌}\AgdaSymbol{\}}\AgdaSpace{}%
\AgdaSymbol{→}\AgdaSpace{}%
\AgdaInductiveConstructor{𝟘}\AgdaSpace{}%
\AgdaOperator{\AgdaInductiveConstructor{+ᵤ}}\AgdaSpace{}%
\AgdaBound{t}\AgdaSpace{}%
\AgdaOperator{\AgdaDatatype{⟷}}\AgdaSpace{}%
\AgdaBound{t}\<%
\\
\>[2]\AgdaInductiveConstructor{uniti₊l}\AgdaSpace{}%
\AgdaSymbol{:}\AgdaSpace{}%
\AgdaSymbol{\{}\AgdaBound{t}\AgdaSpace{}%
\AgdaSymbol{:}\AgdaSpace{}%
\AgdaDatatype{𝕌}\AgdaSymbol{\}}\AgdaSpace{}%
\AgdaSymbol{→}\AgdaSpace{}%
\AgdaBound{t}\AgdaSpace{}%
\AgdaOperator{\AgdaDatatype{⟷}}\AgdaSpace{}%
\AgdaInductiveConstructor{𝟘}\AgdaSpace{}%
\AgdaOperator{\AgdaInductiveConstructor{+ᵤ}}\AgdaSpace{}%
\AgdaBound{t}\<%
\\
\>[2]\AgdaInductiveConstructor{unite₊r}\AgdaSpace{}%
\AgdaSymbol{:}\AgdaSpace{}%
\AgdaSymbol{\{}\AgdaBound{t}\AgdaSpace{}%
\AgdaSymbol{:}\AgdaSpace{}%
\AgdaDatatype{𝕌}\AgdaSymbol{\}}\AgdaSpace{}%
\AgdaSymbol{→}\AgdaSpace{}%
\AgdaBound{t}\AgdaSpace{}%
\AgdaOperator{\AgdaInductiveConstructor{+ᵤ}}\AgdaSpace{}%
\AgdaInductiveConstructor{𝟘}\AgdaSpace{}%
\AgdaOperator{\AgdaDatatype{⟷}}\AgdaSpace{}%
\AgdaBound{t}\<%
\\
\>[2]\AgdaInductiveConstructor{uniti₊r}\AgdaSpace{}%
\AgdaSymbol{:}\AgdaSpace{}%
\AgdaSymbol{\{}\AgdaBound{t}\AgdaSpace{}%
\AgdaSymbol{:}\AgdaSpace{}%
\AgdaDatatype{𝕌}\AgdaSymbol{\}}\AgdaSpace{}%
\AgdaSymbol{→}\AgdaSpace{}%
\AgdaBound{t}\AgdaSpace{}%
\AgdaOperator{\AgdaDatatype{⟷}}\AgdaSpace{}%
\AgdaBound{t}\AgdaSpace{}%
\AgdaOperator{\AgdaInductiveConstructor{+ᵤ}}\AgdaSpace{}%
\AgdaInductiveConstructor{𝟘}\<%
\\
\>[2]\AgdaInductiveConstructor{swap₊}%
\>[10]\AgdaSymbol{:}\AgdaSpace{}%
\AgdaSymbol{\{}\AgdaBound{t₁}\AgdaSpace{}%
\AgdaBound{t₂}\AgdaSpace{}%
\AgdaSymbol{:}\AgdaSpace{}%
\AgdaDatatype{𝕌}\AgdaSymbol{\}}\AgdaSpace{}%
\AgdaSymbol{→}\AgdaSpace{}%
\AgdaBound{t₁}\AgdaSpace{}%
\AgdaOperator{\AgdaInductiveConstructor{+ᵤ}}\AgdaSpace{}%
\AgdaBound{t₂}\AgdaSpace{}%
\AgdaOperator{\AgdaDatatype{⟷}}\AgdaSpace{}%
\AgdaBound{t₂}\AgdaSpace{}%
\AgdaOperator{\AgdaInductiveConstructor{+ᵤ}}\AgdaSpace{}%
\AgdaBound{t₁}\<%
\\
\>[2]\AgdaInductiveConstructor{assocl₊}\AgdaSpace{}%
\AgdaSymbol{:}\AgdaSpace{}%
\AgdaSymbol{\{}\AgdaBound{t₁}\AgdaSpace{}%
\AgdaBound{t₂}\AgdaSpace{}%
\AgdaBound{t₃}\AgdaSpace{}%
\AgdaSymbol{:}\AgdaSpace{}%
\AgdaDatatype{𝕌}\AgdaSymbol{\}}\AgdaSpace{}%
\AgdaSymbol{→}\AgdaSpace{}%
\AgdaBound{t₁}\AgdaSpace{}%
\AgdaOperator{\AgdaInductiveConstructor{+ᵤ}}\AgdaSpace{}%
\AgdaSymbol{(}\AgdaBound{t₂}\AgdaSpace{}%
\AgdaOperator{\AgdaInductiveConstructor{+ᵤ}}\AgdaSpace{}%
\AgdaBound{t₃}\AgdaSymbol{)}\AgdaSpace{}%
\AgdaOperator{\AgdaDatatype{⟷}}\AgdaSpace{}%
\AgdaSymbol{(}\AgdaBound{t₁}\AgdaSpace{}%
\AgdaOperator{\AgdaInductiveConstructor{+ᵤ}}\AgdaSpace{}%
\AgdaBound{t₂}\AgdaSymbol{)}\AgdaSpace{}%
\AgdaOperator{\AgdaInductiveConstructor{+ᵤ}}\AgdaSpace{}%
\AgdaBound{t₃}\<%
\\
\>[2]\AgdaInductiveConstructor{assocr₊}\AgdaSpace{}%
\AgdaSymbol{:}\AgdaSpace{}%
\AgdaSymbol{\{}\AgdaBound{t₁}\AgdaSpace{}%
\AgdaBound{t₂}\AgdaSpace{}%
\AgdaBound{t₃}\AgdaSpace{}%
\AgdaSymbol{:}\AgdaSpace{}%
\AgdaDatatype{𝕌}\AgdaSymbol{\}}\AgdaSpace{}%
\AgdaSymbol{→}\AgdaSpace{}%
\AgdaSymbol{(}\AgdaBound{t₁}\AgdaSpace{}%
\AgdaOperator{\AgdaInductiveConstructor{+ᵤ}}\AgdaSpace{}%
\AgdaBound{t₂}\AgdaSymbol{)}\AgdaSpace{}%
\AgdaOperator{\AgdaInductiveConstructor{+ᵤ}}\AgdaSpace{}%
\AgdaBound{t₃}\AgdaSpace{}%
\AgdaOperator{\AgdaDatatype{⟷}}\AgdaSpace{}%
\AgdaBound{t₁}\AgdaSpace{}%
\AgdaOperator{\AgdaInductiveConstructor{+ᵤ}}\AgdaSpace{}%
\AgdaSymbol{(}\AgdaBound{t₂}\AgdaSpace{}%
\AgdaOperator{\AgdaInductiveConstructor{+ᵤ}}\AgdaSpace{}%
\AgdaBound{t₃}\AgdaSymbol{)}\<%
\\
\>[2]\AgdaInductiveConstructor{unite⋆l}\AgdaSpace{}%
\AgdaSymbol{:}\AgdaSpace{}%
\AgdaSymbol{\{}\AgdaBound{t}\AgdaSpace{}%
\AgdaSymbol{:}\AgdaSpace{}%
\AgdaDatatype{𝕌}\AgdaSymbol{\}}\AgdaSpace{}%
\AgdaSymbol{→}\AgdaSpace{}%
\AgdaInductiveConstructor{𝟙}\AgdaSpace{}%
\AgdaOperator{\AgdaInductiveConstructor{×ᵤ}}\AgdaSpace{}%
\AgdaBound{t}\AgdaSpace{}%
\AgdaOperator{\AgdaDatatype{⟷}}\AgdaSpace{}%
\AgdaBound{t}\<%
\\
\>[2]\AgdaInductiveConstructor{uniti⋆l}\AgdaSpace{}%
\AgdaSymbol{:}\AgdaSpace{}%
\AgdaSymbol{\{}\AgdaBound{t}\AgdaSpace{}%
\AgdaSymbol{:}\AgdaSpace{}%
\AgdaDatatype{𝕌}\AgdaSymbol{\}}\AgdaSpace{}%
\AgdaSymbol{→}\AgdaSpace{}%
\AgdaBound{t}\AgdaSpace{}%
\AgdaOperator{\AgdaDatatype{⟷}}\AgdaSpace{}%
\AgdaInductiveConstructor{𝟙}\AgdaSpace{}%
\AgdaOperator{\AgdaInductiveConstructor{×ᵤ}}\AgdaSpace{}%
\AgdaBound{t}\<%
\\
\>[2]\AgdaInductiveConstructor{unite⋆r}\AgdaSpace{}%
\AgdaSymbol{:}\AgdaSpace{}%
\AgdaSymbol{\{}\AgdaBound{t}\AgdaSpace{}%
\AgdaSymbol{:}\AgdaSpace{}%
\AgdaDatatype{𝕌}\AgdaSymbol{\}}\AgdaSpace{}%
\AgdaSymbol{→}\AgdaSpace{}%
\AgdaBound{t}\AgdaSpace{}%
\AgdaOperator{\AgdaInductiveConstructor{×ᵤ}}\AgdaSpace{}%
\AgdaInductiveConstructor{𝟙}\AgdaSpace{}%
\AgdaOperator{\AgdaDatatype{⟷}}\AgdaSpace{}%
\AgdaBound{t}\<%
\\
\>[2]\AgdaInductiveConstructor{uniti⋆r}\AgdaSpace{}%
\AgdaSymbol{:}\AgdaSpace{}%
\AgdaSymbol{\{}\AgdaBound{t}\AgdaSpace{}%
\AgdaSymbol{:}\AgdaSpace{}%
\AgdaDatatype{𝕌}\AgdaSymbol{\}}\AgdaSpace{}%
\AgdaSymbol{→}\AgdaSpace{}%
\AgdaBound{t}\AgdaSpace{}%
\AgdaOperator{\AgdaDatatype{⟷}}\AgdaSpace{}%
\AgdaBound{t}\AgdaSpace{}%
\AgdaOperator{\AgdaInductiveConstructor{×ᵤ}}\AgdaSpace{}%
\AgdaInductiveConstructor{𝟙}\<%
\\
\>[2]\AgdaInductiveConstructor{swap⋆}%
\>[10]\AgdaSymbol{:}\AgdaSpace{}%
\AgdaSymbol{\{}\AgdaBound{t₁}\AgdaSpace{}%
\AgdaBound{t₂}\AgdaSpace{}%
\AgdaSymbol{:}\AgdaSpace{}%
\AgdaDatatype{𝕌}\AgdaSymbol{\}}\AgdaSpace{}%
\AgdaSymbol{→}\AgdaSpace{}%
\AgdaBound{t₁}\AgdaSpace{}%
\AgdaOperator{\AgdaInductiveConstructor{×ᵤ}}\AgdaSpace{}%
\AgdaBound{t₂}\AgdaSpace{}%
\AgdaOperator{\AgdaDatatype{⟷}}\AgdaSpace{}%
\AgdaBound{t₂}\AgdaSpace{}%
\AgdaOperator{\AgdaInductiveConstructor{×ᵤ}}\AgdaSpace{}%
\AgdaBound{t₁}\<%
\\
\>[2]\AgdaInductiveConstructor{assocl⋆}\AgdaSpace{}%
\AgdaSymbol{:}\AgdaSpace{}%
\AgdaSymbol{\{}\AgdaBound{t₁}\AgdaSpace{}%
\AgdaBound{t₂}\AgdaSpace{}%
\AgdaBound{t₃}\AgdaSpace{}%
\AgdaSymbol{:}\AgdaSpace{}%
\AgdaDatatype{𝕌}\AgdaSymbol{\}}\AgdaSpace{}%
\AgdaSymbol{→}\AgdaSpace{}%
\AgdaBound{t₁}\AgdaSpace{}%
\AgdaOperator{\AgdaInductiveConstructor{×ᵤ}}\AgdaSpace{}%
\AgdaSymbol{(}\AgdaBound{t₂}\AgdaSpace{}%
\AgdaOperator{\AgdaInductiveConstructor{×ᵤ}}\AgdaSpace{}%
\AgdaBound{t₃}\AgdaSymbol{)}\AgdaSpace{}%
\AgdaOperator{\AgdaDatatype{⟷}}\AgdaSpace{}%
\AgdaSymbol{(}\AgdaBound{t₁}\AgdaSpace{}%
\AgdaOperator{\AgdaInductiveConstructor{×ᵤ}}\AgdaSpace{}%
\AgdaBound{t₂}\AgdaSymbol{)}\AgdaSpace{}%
\AgdaOperator{\AgdaInductiveConstructor{×ᵤ}}\AgdaSpace{}%
\AgdaBound{t₃}\<%
\\
\>[2]\AgdaInductiveConstructor{assocr⋆}\AgdaSpace{}%
\AgdaSymbol{:}\AgdaSpace{}%
\AgdaSymbol{\{}\AgdaBound{t₁}\AgdaSpace{}%
\AgdaBound{t₂}\AgdaSpace{}%
\AgdaBound{t₃}\AgdaSpace{}%
\AgdaSymbol{:}\AgdaSpace{}%
\AgdaDatatype{𝕌}\AgdaSymbol{\}}\AgdaSpace{}%
\AgdaSymbol{→}\AgdaSpace{}%
\AgdaSymbol{(}\AgdaBound{t₁}\AgdaSpace{}%
\AgdaOperator{\AgdaInductiveConstructor{×ᵤ}}\AgdaSpace{}%
\AgdaBound{t₂}\AgdaSymbol{)}\AgdaSpace{}%
\AgdaOperator{\AgdaInductiveConstructor{×ᵤ}}\AgdaSpace{}%
\AgdaBound{t₃}\AgdaSpace{}%
\AgdaOperator{\AgdaDatatype{⟷}}\AgdaSpace{}%
\AgdaBound{t₁}\AgdaSpace{}%
\AgdaOperator{\AgdaInductiveConstructor{×ᵤ}}\AgdaSpace{}%
\AgdaSymbol{(}\AgdaBound{t₂}\AgdaSpace{}%
\AgdaOperator{\AgdaInductiveConstructor{×ᵤ}}\AgdaSpace{}%
\AgdaBound{t₃}\AgdaSymbol{)}\<%
\\
\>[2]\AgdaInductiveConstructor{absorbr}\AgdaSpace{}%
\AgdaSymbol{:}\AgdaSpace{}%
\AgdaSymbol{\{}\AgdaBound{t}\AgdaSpace{}%
\AgdaSymbol{:}\AgdaSpace{}%
\AgdaDatatype{𝕌}\AgdaSymbol{\}}\AgdaSpace{}%
\AgdaSymbol{→}\AgdaSpace{}%
\AgdaInductiveConstructor{𝟘}\AgdaSpace{}%
\AgdaOperator{\AgdaInductiveConstructor{×ᵤ}}\AgdaSpace{}%
\AgdaBound{t}\AgdaSpace{}%
\AgdaOperator{\AgdaDatatype{⟷}}\AgdaSpace{}%
\AgdaInductiveConstructor{𝟘}\<%
\\
\>[2]\AgdaInductiveConstructor{absorbl}\AgdaSpace{}%
\AgdaSymbol{:}\AgdaSpace{}%
\AgdaSymbol{\{}\AgdaBound{t}\AgdaSpace{}%
\AgdaSymbol{:}\AgdaSpace{}%
\AgdaDatatype{𝕌}\AgdaSymbol{\}}\AgdaSpace{}%
\AgdaSymbol{→}\AgdaSpace{}%
\AgdaBound{t}\AgdaSpace{}%
\AgdaOperator{\AgdaInductiveConstructor{×ᵤ}}\AgdaSpace{}%
\AgdaInductiveConstructor{𝟘}\AgdaSpace{}%
\AgdaOperator{\AgdaDatatype{⟷}}\AgdaSpace{}%
\AgdaInductiveConstructor{𝟘}\<%
\\
\>[2]\AgdaInductiveConstructor{factorzr}\AgdaSpace{}%
\AgdaSymbol{:}\AgdaSpace{}%
\AgdaSymbol{\{}\AgdaBound{t}\AgdaSpace{}%
\AgdaSymbol{:}\AgdaSpace{}%
\AgdaDatatype{𝕌}\AgdaSymbol{\}}\AgdaSpace{}%
\AgdaSymbol{→}\AgdaSpace{}%
\AgdaInductiveConstructor{𝟘}\AgdaSpace{}%
\AgdaOperator{\AgdaDatatype{⟷}}\AgdaSpace{}%
\AgdaBound{t}\AgdaSpace{}%
\AgdaOperator{\AgdaInductiveConstructor{×ᵤ}}\AgdaSpace{}%
\AgdaInductiveConstructor{𝟘}\<%
\\
\>[2]\AgdaInductiveConstructor{factorzl}\AgdaSpace{}%
\AgdaSymbol{:}\AgdaSpace{}%
\AgdaSymbol{\{}\AgdaBound{t}\AgdaSpace{}%
\AgdaSymbol{:}\AgdaSpace{}%
\AgdaDatatype{𝕌}\AgdaSymbol{\}}\AgdaSpace{}%
\AgdaSymbol{→}\AgdaSpace{}%
\AgdaInductiveConstructor{𝟘}\AgdaSpace{}%
\AgdaOperator{\AgdaDatatype{⟷}}\AgdaSpace{}%
\AgdaInductiveConstructor{𝟘}\AgdaSpace{}%
\AgdaOperator{\AgdaInductiveConstructor{×ᵤ}}\AgdaSpace{}%
\AgdaBound{t}\<%
\\
\>[2]\AgdaInductiveConstructor{dist}%
\>[10]\AgdaSymbol{:}\AgdaSpace{}%
\AgdaSymbol{\{}\AgdaBound{t₁}\AgdaSpace{}%
\AgdaBound{t₂}\AgdaSpace{}%
\AgdaBound{t₃}\AgdaSpace{}%
\AgdaSymbol{:}\AgdaSpace{}%
\AgdaDatatype{𝕌}\AgdaSymbol{\}}\AgdaSpace{}%
\AgdaSymbol{→}\AgdaSpace{}%
\AgdaSymbol{(}\AgdaBound{t₁}\AgdaSpace{}%
\AgdaOperator{\AgdaInductiveConstructor{+ᵤ}}\AgdaSpace{}%
\AgdaBound{t₂}\AgdaSymbol{)}\AgdaSpace{}%
\AgdaOperator{\AgdaInductiveConstructor{×ᵤ}}\AgdaSpace{}%
\AgdaBound{t₃}\AgdaSpace{}%
\AgdaOperator{\AgdaDatatype{⟷}}\AgdaSpace{}%
\AgdaSymbol{(}\AgdaBound{t₁}\AgdaSpace{}%
\AgdaOperator{\AgdaInductiveConstructor{×ᵤ}}\AgdaSpace{}%
\AgdaBound{t₃}\AgdaSymbol{)}\AgdaSpace{}%
\AgdaOperator{\AgdaInductiveConstructor{+ᵤ}}\AgdaSpace{}%
\AgdaSymbol{(}\AgdaBound{t₂}\AgdaSpace{}%
\AgdaOperator{\AgdaInductiveConstructor{×ᵤ}}\AgdaSpace{}%
\AgdaBound{t₃}\AgdaSymbol{)}\<%
\\
\>[2]\AgdaInductiveConstructor{factor}%
\>[10]\AgdaSymbol{:}\AgdaSpace{}%
\AgdaSymbol{\{}\AgdaBound{t₁}\AgdaSpace{}%
\AgdaBound{t₂}\AgdaSpace{}%
\AgdaBound{t₃}\AgdaSpace{}%
\AgdaSymbol{:}\AgdaSpace{}%
\AgdaDatatype{𝕌}\AgdaSymbol{\}}\AgdaSpace{}%
\AgdaSymbol{→}\AgdaSpace{}%
\AgdaSymbol{(}\AgdaBound{t₁}\AgdaSpace{}%
\AgdaOperator{\AgdaInductiveConstructor{×ᵤ}}\AgdaSpace{}%
\AgdaBound{t₃}\AgdaSymbol{)}\AgdaSpace{}%
\AgdaOperator{\AgdaInductiveConstructor{+ᵤ}}\AgdaSpace{}%
\AgdaSymbol{(}\AgdaBound{t₂}\AgdaSpace{}%
\AgdaOperator{\AgdaInductiveConstructor{×ᵤ}}\AgdaSpace{}%
\AgdaBound{t₃}\AgdaSymbol{)}\AgdaSpace{}%
\AgdaOperator{\AgdaDatatype{⟷}}\AgdaSpace{}%
\AgdaSymbol{(}\AgdaBound{t₁}\AgdaSpace{}%
\AgdaOperator{\AgdaInductiveConstructor{+ᵤ}}\AgdaSpace{}%
\AgdaBound{t₂}\AgdaSymbol{)}\AgdaSpace{}%
\AgdaOperator{\AgdaInductiveConstructor{×ᵤ}}\AgdaSpace{}%
\AgdaBound{t₃}\<%
\\
\>[2]\AgdaInductiveConstructor{distl}%
\>[10]\AgdaSymbol{:}\AgdaSpace{}%
\AgdaSymbol{\{}\AgdaBound{t₁}\AgdaSpace{}%
\AgdaBound{t₂}\AgdaSpace{}%
\AgdaBound{t₃}\AgdaSpace{}%
\AgdaSymbol{:}\AgdaSpace{}%
\AgdaDatatype{𝕌}\AgdaSymbol{\}}\AgdaSpace{}%
\AgdaSymbol{→}\AgdaSpace{}%
\AgdaBound{t₁}\AgdaSpace{}%
\AgdaOperator{\AgdaInductiveConstructor{×ᵤ}}\AgdaSpace{}%
\AgdaSymbol{(}\AgdaBound{t₂}\AgdaSpace{}%
\AgdaOperator{\AgdaInductiveConstructor{+ᵤ}}\AgdaSpace{}%
\AgdaBound{t₃}\AgdaSymbol{)}\AgdaSpace{}%
\AgdaOperator{\AgdaDatatype{⟷}}\AgdaSpace{}%
\AgdaSymbol{(}\AgdaBound{t₁}\AgdaSpace{}%
\AgdaOperator{\AgdaInductiveConstructor{×ᵤ}}\AgdaSpace{}%
\AgdaBound{t₂}\AgdaSymbol{)}\AgdaSpace{}%
\AgdaOperator{\AgdaInductiveConstructor{+ᵤ}}\AgdaSpace{}%
\AgdaSymbol{(}\AgdaBound{t₁}\AgdaSpace{}%
\AgdaOperator{\AgdaInductiveConstructor{×ᵤ}}\AgdaSpace{}%
\AgdaBound{t₃}\AgdaSymbol{)}\<%
\\
\>[2]\AgdaInductiveConstructor{factorl}\AgdaSpace{}%
\AgdaSymbol{:}\AgdaSpace{}%
\AgdaSymbol{\{}\AgdaBound{t₁}\AgdaSpace{}%
\AgdaBound{t₂}\AgdaSpace{}%
\AgdaBound{t₃}\AgdaSpace{}%
\AgdaSymbol{:}\AgdaSpace{}%
\AgdaDatatype{𝕌}\AgdaSpace{}%
\AgdaSymbol{\}}\AgdaSpace{}%
\AgdaSymbol{→}\AgdaSpace{}%
\AgdaSymbol{(}\AgdaBound{t₁}\AgdaSpace{}%
\AgdaOperator{\AgdaInductiveConstructor{×ᵤ}}\AgdaSpace{}%
\AgdaBound{t₂}\AgdaSymbol{)}\AgdaSpace{}%
\AgdaOperator{\AgdaInductiveConstructor{+ᵤ}}\AgdaSpace{}%
\AgdaSymbol{(}\AgdaBound{t₁}\AgdaSpace{}%
\AgdaOperator{\AgdaInductiveConstructor{×ᵤ}}\AgdaSpace{}%
\AgdaBound{t₃}\AgdaSymbol{)}\AgdaSpace{}%
\AgdaOperator{\AgdaDatatype{⟷}}\AgdaSpace{}%
\AgdaBound{t₁}\AgdaSpace{}%
\AgdaOperator{\AgdaInductiveConstructor{×ᵤ}}\AgdaSpace{}%
\AgdaSymbol{(}\AgdaBound{t₂}\AgdaSpace{}%
\AgdaOperator{\AgdaInductiveConstructor{+ᵤ}}\AgdaSpace{}%
\AgdaBound{t₃}\AgdaSymbol{)}\<%
\\
\>[2]\AgdaInductiveConstructor{id⟷}%
\>[10]\AgdaSymbol{:}\AgdaSpace{}%
\AgdaSymbol{\{}\AgdaBound{t}\AgdaSpace{}%
\AgdaSymbol{:}\AgdaSpace{}%
\AgdaDatatype{𝕌}\AgdaSymbol{\}}\AgdaSpace{}%
\AgdaSymbol{→}\AgdaSpace{}%
\AgdaBound{t}\AgdaSpace{}%
\AgdaOperator{\AgdaDatatype{⟷}}\AgdaSpace{}%
\AgdaBound{t}\<%
\\
\>[2]\AgdaOperator{\AgdaInductiveConstructor{\AgdaUnderscore{}⊚\AgdaUnderscore{}}}%
\>[10]\AgdaSymbol{:}\AgdaSpace{}%
\AgdaSymbol{\{}\AgdaBound{t₁}\AgdaSpace{}%
\AgdaBound{t₂}\AgdaSpace{}%
\AgdaBound{t₃}\AgdaSpace{}%
\AgdaSymbol{:}\AgdaSpace{}%
\AgdaDatatype{𝕌}\AgdaSymbol{\}}\AgdaSpace{}%
\AgdaSymbol{→}\AgdaSpace{}%
\AgdaSymbol{(}\AgdaBound{t₁}\AgdaSpace{}%
\AgdaOperator{\AgdaDatatype{⟷}}\AgdaSpace{}%
\AgdaBound{t₂}\AgdaSymbol{)}\AgdaSpace{}%
\AgdaSymbol{→}\AgdaSpace{}%
\AgdaSymbol{(}\AgdaBound{t₂}\AgdaSpace{}%
\AgdaOperator{\AgdaDatatype{⟷}}\AgdaSpace{}%
\AgdaBound{t₃}\AgdaSymbol{)}\AgdaSpace{}%
\AgdaSymbol{→}\AgdaSpace{}%
\AgdaSymbol{(}\AgdaBound{t₁}\AgdaSpace{}%
\AgdaOperator{\AgdaDatatype{⟷}}\AgdaSpace{}%
\AgdaBound{t₃}\AgdaSymbol{)}\<%
\\
\>[2]\AgdaOperator{\AgdaInductiveConstructor{\AgdaUnderscore{}⊕\AgdaUnderscore{}}}%
\>[10]\AgdaSymbol{:}\AgdaSpace{}%
\AgdaSymbol{\{}\AgdaBound{t₁}\AgdaSpace{}%
\AgdaBound{t₂}\AgdaSpace{}%
\AgdaBound{t₃}\AgdaSpace{}%
\AgdaBound{t₄}\AgdaSpace{}%
\AgdaSymbol{:}\AgdaSpace{}%
\AgdaDatatype{𝕌}\AgdaSymbol{\}}\AgdaSpace{}%
\AgdaSymbol{→}\AgdaSpace{}%
\AgdaSymbol{(}\AgdaBound{t₁}\AgdaSpace{}%
\AgdaOperator{\AgdaDatatype{⟷}}\AgdaSpace{}%
\AgdaBound{t₃}\AgdaSymbol{)}\AgdaSpace{}%
\AgdaSymbol{→}\AgdaSpace{}%
\AgdaSymbol{(}\AgdaBound{t₂}\AgdaSpace{}%
\AgdaOperator{\AgdaDatatype{⟷}}\AgdaSpace{}%
\AgdaBound{t₄}\AgdaSymbol{)}\AgdaSpace{}%
\AgdaSymbol{→}\AgdaSpace{}%
\AgdaSymbol{(}\AgdaBound{t₁}\AgdaSpace{}%
\AgdaOperator{\AgdaInductiveConstructor{+ᵤ}}\AgdaSpace{}%
\AgdaBound{t₂}\AgdaSpace{}%
\AgdaOperator{\AgdaDatatype{⟷}}\AgdaSpace{}%
\AgdaBound{t₃}\AgdaSpace{}%
\AgdaOperator{\AgdaInductiveConstructor{+ᵤ}}\AgdaSpace{}%
\AgdaBound{t₄}\AgdaSymbol{)}\<%
\\
\>[2]\AgdaOperator{\AgdaInductiveConstructor{\AgdaUnderscore{}⊗\AgdaUnderscore{}}}%
\>[10]\AgdaSymbol{:}\AgdaSpace{}%
\AgdaSymbol{\{}\AgdaBound{t₁}\AgdaSpace{}%
\AgdaBound{t₂}\AgdaSpace{}%
\AgdaBound{t₃}\AgdaSpace{}%
\AgdaBound{t₄}\AgdaSpace{}%
\AgdaSymbol{:}\AgdaSpace{}%
\AgdaDatatype{𝕌}\AgdaSymbol{\}}\AgdaSpace{}%
\AgdaSymbol{→}\AgdaSpace{}%
\AgdaSymbol{(}\AgdaBound{t₁}\AgdaSpace{}%
\AgdaOperator{\AgdaDatatype{⟷}}\AgdaSpace{}%
\AgdaBound{t₃}\AgdaSymbol{)}\AgdaSpace{}%
\AgdaSymbol{→}\AgdaSpace{}%
\AgdaSymbol{(}\AgdaBound{t₂}\AgdaSpace{}%
\AgdaOperator{\AgdaDatatype{⟷}}\AgdaSpace{}%
\AgdaBound{t₄}\AgdaSymbol{)}\AgdaSpace{}%
\AgdaSymbol{→}\AgdaSpace{}%
\AgdaSymbol{(}\AgdaBound{t₁}\AgdaSpace{}%
\AgdaOperator{\AgdaInductiveConstructor{×ᵤ}}\AgdaSpace{}%
\AgdaBound{t₂}\AgdaSpace{}%
\AgdaOperator{\AgdaDatatype{⟷}}\AgdaSpace{}%
\AgdaBound{t₃}\AgdaSpace{}%
\AgdaOperator{\AgdaInductiveConstructor{×ᵤ}}\AgdaSpace{}%
\AgdaBound{t₄}\AgdaSymbol{)}\<%
\end{code}

\begin{code}[hide]%
\>[0]\AgdaOperator{\AgdaFunction{\AgdaUnderscore{}≟ᵤ\AgdaUnderscore{}}}\AgdaSpace{}%
\AgdaSymbol{:}\AgdaSpace{}%
\AgdaSymbol{\{}\AgdaBound{t}\AgdaSpace{}%
\AgdaSymbol{:}\AgdaSpace{}%
\AgdaDatatype{𝕌}\AgdaSymbol{\}}\AgdaSpace{}%
\AgdaSymbol{→}\AgdaSpace{}%
\AgdaFunction{Decidable}\AgdaSpace{}%
\AgdaSymbol{(}\AgdaOperator{\AgdaDatatype{\AgdaUnderscore{}≡\AgdaUnderscore{}}}\AgdaSpace{}%
\AgdaSymbol{\{}\AgdaArgument{A}\AgdaSpace{}%
\AgdaSymbol{=}\AgdaSpace{}%
\AgdaOperator{\AgdaFunction{⟦}}\AgdaSpace{}%
\AgdaBound{t}\AgdaSpace{}%
\AgdaOperator{\AgdaFunction{⟧}}\AgdaSymbol{\})}\<%
\\
\>[0]\AgdaOperator{\AgdaFunction{\AgdaUnderscore{}≟ᵤ\AgdaUnderscore{}}}\AgdaSpace{}%
\AgdaSymbol{\{}\AgdaInductiveConstructor{𝟙}\AgdaSymbol{\}}\AgdaSpace{}%
\AgdaInductiveConstructor{tt}\AgdaSpace{}%
\AgdaInductiveConstructor{tt}\AgdaSpace{}%
\AgdaSymbol{=}\AgdaSpace{}%
\AgdaInductiveConstructor{yes}\AgdaSpace{}%
\AgdaInductiveConstructor{refl}\<%
\\
\>[0]\AgdaOperator{\AgdaFunction{\AgdaUnderscore{}≟ᵤ\AgdaUnderscore{}}}\AgdaSpace{}%
\AgdaSymbol{\{}\AgdaBound{t₁}\AgdaSpace{}%
\AgdaOperator{\AgdaInductiveConstructor{+ᵤ}}\AgdaSpace{}%
\AgdaBound{t₂}\AgdaSymbol{\}}\AgdaSpace{}%
\AgdaSymbol{(}\AgdaInductiveConstructor{inj₁}\AgdaSpace{}%
\AgdaBound{x}\AgdaSymbol{)}\AgdaSpace{}%
\AgdaSymbol{(}\AgdaInductiveConstructor{inj₁}\AgdaSpace{}%
\AgdaBound{y}\AgdaSymbol{)}\AgdaSpace{}%
\AgdaKeyword{with}\AgdaSpace{}%
\AgdaBound{x}\AgdaSpace{}%
\AgdaOperator{\AgdaFunction{≟ᵤ}}\AgdaSpace{}%
\AgdaBound{y}\<%
\\
\>[0]\AgdaSymbol{...}\AgdaSpace{}%
\AgdaSymbol{|}\AgdaSpace{}%
\AgdaInductiveConstructor{yes}\AgdaSpace{}%
\AgdaInductiveConstructor{refl}\AgdaSpace{}%
\AgdaSymbol{=}\AgdaSpace{}%
\AgdaInductiveConstructor{yes}\AgdaSpace{}%
\AgdaInductiveConstructor{refl}\<%
\\
\>[0]\AgdaSymbol{...}\AgdaSpace{}%
\AgdaSymbol{|}\AgdaSpace{}%
\AgdaInductiveConstructor{no}%
\>[10]\AgdaBound{x≠y}%
\>[15]\AgdaSymbol{=}\AgdaSpace{}%
\AgdaInductiveConstructor{no}\AgdaSpace{}%
\AgdaSymbol{(λ}\AgdaSpace{}%
\AgdaSymbol{\{}\AgdaInductiveConstructor{refl}\AgdaSpace{}%
\AgdaSymbol{→}\AgdaSpace{}%
\AgdaBound{x≠y}\AgdaSpace{}%
\AgdaInductiveConstructor{refl}\AgdaSymbol{\})}\<%
\\
\>[0]\AgdaOperator{\AgdaFunction{\AgdaUnderscore{}≟ᵤ\AgdaUnderscore{}}}\AgdaSpace{}%
\AgdaSymbol{\{}\AgdaBound{t₁}\AgdaSpace{}%
\AgdaOperator{\AgdaInductiveConstructor{+ᵤ}}\AgdaSpace{}%
\AgdaBound{t₂}\AgdaSymbol{\}}\AgdaSpace{}%
\AgdaSymbol{(}\AgdaInductiveConstructor{inj₁}\AgdaSpace{}%
\AgdaBound{x}\AgdaSymbol{)}\AgdaSpace{}%
\AgdaSymbol{(}\AgdaInductiveConstructor{inj₂}\AgdaSpace{}%
\AgdaBound{y}\AgdaSymbol{)}\AgdaSpace{}%
\AgdaSymbol{=}\AgdaSpace{}%
\AgdaInductiveConstructor{no}\AgdaSpace{}%
\AgdaSymbol{(λ}\AgdaSpace{}%
\AgdaSymbol{())}\<%
\\
\>[0]\AgdaOperator{\AgdaFunction{\AgdaUnderscore{}≟ᵤ\AgdaUnderscore{}}}\AgdaSpace{}%
\AgdaSymbol{\{}\AgdaBound{t₁}\AgdaSpace{}%
\AgdaOperator{\AgdaInductiveConstructor{+ᵤ}}\AgdaSpace{}%
\AgdaBound{t₂}\AgdaSymbol{\}}\AgdaSpace{}%
\AgdaSymbol{(}\AgdaInductiveConstructor{inj₂}\AgdaSpace{}%
\AgdaBound{x}\AgdaSymbol{)}\AgdaSpace{}%
\AgdaSymbol{(}\AgdaInductiveConstructor{inj₁}\AgdaSpace{}%
\AgdaBound{y}\AgdaSymbol{)}\AgdaSpace{}%
\AgdaSymbol{=}\AgdaSpace{}%
\AgdaInductiveConstructor{no}\AgdaSpace{}%
\AgdaSymbol{(λ}\AgdaSpace{}%
\AgdaSymbol{())}\<%
\\
\>[0]\AgdaOperator{\AgdaFunction{\AgdaUnderscore{}≟ᵤ\AgdaUnderscore{}}}\AgdaSpace{}%
\AgdaSymbol{\{}\AgdaBound{t₁}\AgdaSpace{}%
\AgdaOperator{\AgdaInductiveConstructor{+ᵤ}}\AgdaSpace{}%
\AgdaBound{t₂}\AgdaSymbol{\}}\AgdaSpace{}%
\AgdaSymbol{(}\AgdaInductiveConstructor{inj₂}\AgdaSpace{}%
\AgdaBound{x}\AgdaSymbol{)}\AgdaSpace{}%
\AgdaSymbol{(}\AgdaInductiveConstructor{inj₂}\AgdaSpace{}%
\AgdaBound{y}\AgdaSymbol{)}\AgdaSpace{}%
\AgdaKeyword{with}\AgdaSpace{}%
\AgdaBound{x}\AgdaSpace{}%
\AgdaOperator{\AgdaFunction{≟ᵤ}}\AgdaSpace{}%
\AgdaBound{y}\<%
\\
\>[0]\AgdaSymbol{...}\AgdaSpace{}%
\AgdaSymbol{|}\AgdaSpace{}%
\AgdaInductiveConstructor{yes}\AgdaSpace{}%
\AgdaInductiveConstructor{refl}\AgdaSpace{}%
\AgdaSymbol{=}\AgdaSpace{}%
\AgdaInductiveConstructor{yes}\AgdaSpace{}%
\AgdaInductiveConstructor{refl}\<%
\\
\>[0]\AgdaSymbol{...}\AgdaSpace{}%
\AgdaSymbol{|}\AgdaSpace{}%
\AgdaInductiveConstructor{no}%
\>[10]\AgdaBound{x≠y}%
\>[15]\AgdaSymbol{=}\AgdaSpace{}%
\AgdaInductiveConstructor{no}\AgdaSpace{}%
\AgdaSymbol{(λ}\AgdaSpace{}%
\AgdaSymbol{\{}\AgdaInductiveConstructor{refl}\AgdaSpace{}%
\AgdaSymbol{→}\AgdaSpace{}%
\AgdaBound{x≠y}\AgdaSpace{}%
\AgdaInductiveConstructor{refl}\AgdaSymbol{\})}\<%
\\
\>[0]\AgdaOperator{\AgdaFunction{\AgdaUnderscore{}≟ᵤ\AgdaUnderscore{}}}\AgdaSpace{}%
\AgdaSymbol{\{}\AgdaBound{t₁}\AgdaSpace{}%
\AgdaOperator{\AgdaInductiveConstructor{×ᵤ}}\AgdaSpace{}%
\AgdaBound{t₂}\AgdaSymbol{\}}\AgdaSpace{}%
\AgdaSymbol{(}\AgdaBound{x₁}\AgdaSpace{}%
\AgdaOperator{\AgdaInductiveConstructor{,}}\AgdaSpace{}%
\AgdaBound{x₂}\AgdaSymbol{)}\AgdaSpace{}%
\AgdaSymbol{(}\AgdaBound{y₁}\AgdaSpace{}%
\AgdaOperator{\AgdaInductiveConstructor{,}}\AgdaSpace{}%
\AgdaBound{y₂}\AgdaSymbol{)}\AgdaSpace{}%
\AgdaKeyword{with}\AgdaSpace{}%
\AgdaBound{x₁}\AgdaSpace{}%
\AgdaOperator{\AgdaFunction{≟ᵤ}}\AgdaSpace{}%
\AgdaBound{y₁}\AgdaSpace{}%
\AgdaSymbol{|}\AgdaSpace{}%
\AgdaBound{x₂}\AgdaSpace{}%
\AgdaOperator{\AgdaFunction{≟ᵤ}}\AgdaSpace{}%
\AgdaBound{y₂}\<%
\\
\>[0]\AgdaSymbol{...}\AgdaSpace{}%
\AgdaSymbol{|}\AgdaSpace{}%
\AgdaInductiveConstructor{yes}\AgdaSpace{}%
\AgdaInductiveConstructor{refl}\AgdaSpace{}%
\AgdaSymbol{|}\AgdaSpace{}%
\AgdaInductiveConstructor{yes}\AgdaSpace{}%
\AgdaInductiveConstructor{refl}\AgdaSpace{}%
\AgdaSymbol{=}\AgdaSpace{}%
\AgdaInductiveConstructor{yes}\AgdaSpace{}%
\AgdaInductiveConstructor{refl}\<%
\\
\>[0]\AgdaSymbol{...}\AgdaSpace{}%
\AgdaSymbol{|}\AgdaSpace{}%
\AgdaInductiveConstructor{yes}\AgdaSpace{}%
\AgdaInductiveConstructor{refl}\AgdaSpace{}%
\AgdaSymbol{|}\AgdaSpace{}%
\AgdaInductiveConstructor{no}\AgdaSpace{}%
\AgdaBound{x₂≠y₂}\AgdaSpace{}%
\AgdaSymbol{=}\AgdaSpace{}%
\AgdaInductiveConstructor{no}%
\>[32]\AgdaSymbol{(}\AgdaBound{x₂≠y₂}\AgdaSpace{}%
\AgdaOperator{\AgdaFunction{∘}}\AgdaSpace{}%
\AgdaFunction{cong}\AgdaSpace{}%
\AgdaField{proj₂}\AgdaSymbol{)}\<%
\\
\>[0]\AgdaSymbol{...}\AgdaSpace{}%
\AgdaSymbol{|}\AgdaSpace{}%
\AgdaInductiveConstructor{no}\AgdaSpace{}%
\AgdaBound{x₁≠y₁}\AgdaSpace{}%
\AgdaSymbol{|}\AgdaSpace{}%
\AgdaInductiveConstructor{yes}\AgdaSpace{}%
\AgdaInductiveConstructor{refl}\AgdaSpace{}%
\AgdaSymbol{=}\AgdaSpace{}%
\AgdaInductiveConstructor{no}\AgdaSpace{}%
\AgdaSymbol{(}\AgdaBound{x₁≠y₁}\AgdaSpace{}%
\AgdaOperator{\AgdaFunction{∘}}\AgdaSpace{}%
\AgdaFunction{cong}\AgdaSpace{}%
\AgdaField{proj₁}\AgdaSymbol{)}\<%
\\
\>[0]\AgdaSymbol{...}\AgdaSpace{}%
\AgdaSymbol{|}\AgdaSpace{}%
\AgdaInductiveConstructor{no}\AgdaSpace{}%
\AgdaBound{x₁≠y₁}\AgdaSpace{}%
\AgdaSymbol{|}\AgdaSpace{}%
\AgdaInductiveConstructor{no}\AgdaSpace{}%
\AgdaBound{x₂≠y₂}\AgdaSpace{}%
\AgdaSymbol{=}\AgdaSpace{}%
\AgdaInductiveConstructor{no}\AgdaSpace{}%
\AgdaSymbol{(}\AgdaBound{x₁≠y₁}\AgdaSpace{}%
\AgdaOperator{\AgdaFunction{∘}}\AgdaSpace{}%
\AgdaFunction{cong}\AgdaSpace{}%
\AgdaField{proj₁}\AgdaSymbol{)}\<%
\\
\>[0]\AgdaOperator{\AgdaFunction{\AgdaUnderscore{}≟ᵤ\AgdaUnderscore{}}}\AgdaSpace{}%
\AgdaSymbol{\{}\AgdaOperator{\AgdaInductiveConstructor{𝟙/}}\AgdaSpace{}%
\AgdaBound{t}\AgdaSymbol{\}}\AgdaSpace{}%
\AgdaInductiveConstructor{↻}\AgdaSpace{}%
\AgdaInductiveConstructor{↻}\AgdaSpace{}%
\AgdaSymbol{=}\AgdaSpace{}%
\AgdaInductiveConstructor{yes}\AgdaSpace{}%
\AgdaInductiveConstructor{refl}\<%
\end{code}

\begin{code}[hide]%
\>[0]\AgdaFunction{interp}\AgdaSpace{}%
\AgdaInductiveConstructor{swap⋆}\AgdaSpace{}%
\AgdaSymbol{(}\AgdaBound{v₁}\AgdaSpace{}%
\AgdaOperator{\AgdaInductiveConstructor{,}}\AgdaSpace{}%
\AgdaBound{v₂}\AgdaSymbol{)}\AgdaSpace{}%
\AgdaSymbol{=}\AgdaSpace{}%
\AgdaInductiveConstructor{just}\AgdaSpace{}%
\AgdaSymbol{(}\AgdaBound{v₂}\AgdaSpace{}%
\AgdaOperator{\AgdaInductiveConstructor{,}}\AgdaSpace{}%
\AgdaBound{v₁}\AgdaSymbol{)}\<%
\\
\>[0]\AgdaFunction{interp}\AgdaSpace{}%
\AgdaInductiveConstructor{unite₊l}\AgdaSpace{}%
\AgdaSymbol{(}\AgdaInductiveConstructor{inj₁}\AgdaSpace{}%
\AgdaSymbol{())}\<%
\\
\>[0]\AgdaFunction{interp}\AgdaSpace{}%
\AgdaInductiveConstructor{unite₊l}\AgdaSpace{}%
\AgdaSymbol{(}\AgdaInductiveConstructor{inj₂}\AgdaSpace{}%
\AgdaBound{v}\AgdaSymbol{)}\AgdaSpace{}%
\AgdaSymbol{=}\AgdaSpace{}%
\AgdaInductiveConstructor{just}\AgdaSpace{}%
\AgdaBound{v}\<%
\\
\>[0]\AgdaFunction{interp}\AgdaSpace{}%
\AgdaInductiveConstructor{uniti₊l}\AgdaSpace{}%
\AgdaBound{v}\AgdaSpace{}%
\AgdaSymbol{=}\AgdaSpace{}%
\AgdaInductiveConstructor{just}\AgdaSpace{}%
\AgdaSymbol{(}\AgdaInductiveConstructor{inj₂}\AgdaSpace{}%
\AgdaBound{v}\AgdaSymbol{)}\<%
\\
\>[0]\AgdaFunction{interp}\AgdaSpace{}%
\AgdaInductiveConstructor{unite₊r}\AgdaSpace{}%
\AgdaSymbol{(}\AgdaInductiveConstructor{inj₁}\AgdaSpace{}%
\AgdaBound{v}\AgdaSymbol{)}\AgdaSpace{}%
\AgdaSymbol{=}\AgdaSpace{}%
\AgdaInductiveConstructor{just}\AgdaSpace{}%
\AgdaBound{v}\<%
\\
\>[0]\AgdaFunction{interp}\AgdaSpace{}%
\AgdaInductiveConstructor{unite₊r}\AgdaSpace{}%
\AgdaSymbol{(}\AgdaInductiveConstructor{inj₂}\AgdaSpace{}%
\AgdaSymbol{())}\<%
\\
\>[0]\AgdaFunction{interp}\AgdaSpace{}%
\AgdaInductiveConstructor{uniti₊r}\AgdaSpace{}%
\AgdaBound{v}\AgdaSpace{}%
\AgdaSymbol{=}\AgdaSpace{}%
\AgdaInductiveConstructor{just}\AgdaSpace{}%
\AgdaSymbol{(}\AgdaInductiveConstructor{inj₁}\AgdaSpace{}%
\AgdaBound{v}\AgdaSymbol{)}\<%
\\
\>[0]\AgdaFunction{interp}\AgdaSpace{}%
\AgdaInductiveConstructor{swap₊}\AgdaSpace{}%
\AgdaSymbol{(}\AgdaInductiveConstructor{inj₁}\AgdaSpace{}%
\AgdaBound{v}\AgdaSymbol{)}\AgdaSpace{}%
\AgdaSymbol{=}\AgdaSpace{}%
\AgdaInductiveConstructor{just}\AgdaSpace{}%
\AgdaSymbol{(}\AgdaInductiveConstructor{inj₂}\AgdaSpace{}%
\AgdaBound{v}\AgdaSymbol{)}\<%
\\
\>[0]\AgdaFunction{interp}\AgdaSpace{}%
\AgdaInductiveConstructor{swap₊}\AgdaSpace{}%
\AgdaSymbol{(}\AgdaInductiveConstructor{inj₂}\AgdaSpace{}%
\AgdaBound{v}\AgdaSymbol{)}\AgdaSpace{}%
\AgdaSymbol{=}\AgdaSpace{}%
\AgdaInductiveConstructor{just}\AgdaSpace{}%
\AgdaSymbol{(}\AgdaInductiveConstructor{inj₁}\AgdaSpace{}%
\AgdaBound{v}\AgdaSymbol{)}\<%
\\
\>[0]\AgdaFunction{interp}\AgdaSpace{}%
\AgdaInductiveConstructor{assocl₊}\AgdaSpace{}%
\AgdaSymbol{(}\AgdaInductiveConstructor{inj₁}\AgdaSpace{}%
\AgdaBound{v}\AgdaSymbol{)}\AgdaSpace{}%
\AgdaSymbol{=}\AgdaSpace{}%
\AgdaInductiveConstructor{just}\AgdaSpace{}%
\AgdaSymbol{(}\AgdaInductiveConstructor{inj₁}\AgdaSpace{}%
\AgdaSymbol{(}\AgdaInductiveConstructor{inj₁}\AgdaSpace{}%
\AgdaBound{v}\AgdaSymbol{))}\<%
\\
\>[0]\AgdaFunction{interp}\AgdaSpace{}%
\AgdaInductiveConstructor{assocl₊}\AgdaSpace{}%
\AgdaSymbol{(}\AgdaInductiveConstructor{inj₂}\AgdaSpace{}%
\AgdaSymbol{(}\AgdaInductiveConstructor{inj₁}\AgdaSpace{}%
\AgdaBound{v}\AgdaSymbol{))}\AgdaSpace{}%
\AgdaSymbol{=}\AgdaSpace{}%
\AgdaInductiveConstructor{just}\AgdaSpace{}%
\AgdaSymbol{(}\AgdaInductiveConstructor{inj₁}\AgdaSpace{}%
\AgdaSymbol{(}\AgdaInductiveConstructor{inj₂}\AgdaSpace{}%
\AgdaBound{v}\AgdaSymbol{))}\<%
\\
\>[0]\AgdaFunction{interp}\AgdaSpace{}%
\AgdaInductiveConstructor{assocl₊}\AgdaSpace{}%
\AgdaSymbol{(}\AgdaInductiveConstructor{inj₂}\AgdaSpace{}%
\AgdaSymbol{(}\AgdaInductiveConstructor{inj₂}\AgdaSpace{}%
\AgdaBound{v}\AgdaSymbol{))}\AgdaSpace{}%
\AgdaSymbol{=}\AgdaSpace{}%
\AgdaInductiveConstructor{just}\AgdaSpace{}%
\AgdaSymbol{(}\AgdaInductiveConstructor{inj₂}\AgdaSpace{}%
\AgdaBound{v}\AgdaSymbol{)}\<%
\\
\>[0]\AgdaFunction{interp}\AgdaSpace{}%
\AgdaInductiveConstructor{assocr₊}\AgdaSpace{}%
\AgdaSymbol{(}\AgdaInductiveConstructor{inj₁}\AgdaSpace{}%
\AgdaSymbol{(}\AgdaInductiveConstructor{inj₁}\AgdaSpace{}%
\AgdaBound{v}\AgdaSymbol{))}\AgdaSpace{}%
\AgdaSymbol{=}\AgdaSpace{}%
\AgdaInductiveConstructor{just}\AgdaSpace{}%
\AgdaSymbol{(}\AgdaInductiveConstructor{inj₁}\AgdaSpace{}%
\AgdaBound{v}\AgdaSymbol{)}\<%
\\
\>[0]\AgdaFunction{interp}\AgdaSpace{}%
\AgdaInductiveConstructor{assocr₊}\AgdaSpace{}%
\AgdaSymbol{(}\AgdaInductiveConstructor{inj₁}\AgdaSpace{}%
\AgdaSymbol{(}\AgdaInductiveConstructor{inj₂}\AgdaSpace{}%
\AgdaBound{v}\AgdaSymbol{))}\AgdaSpace{}%
\AgdaSymbol{=}\AgdaSpace{}%
\AgdaInductiveConstructor{just}\AgdaSpace{}%
\AgdaSymbol{(}\AgdaInductiveConstructor{inj₂}\AgdaSpace{}%
\AgdaSymbol{(}\AgdaInductiveConstructor{inj₁}\AgdaSpace{}%
\AgdaBound{v}\AgdaSymbol{))}\<%
\\
\>[0]\AgdaFunction{interp}\AgdaSpace{}%
\AgdaInductiveConstructor{assocr₊}\AgdaSpace{}%
\AgdaSymbol{(}\AgdaInductiveConstructor{inj₂}\AgdaSpace{}%
\AgdaBound{v}\AgdaSymbol{)}\AgdaSpace{}%
\AgdaSymbol{=}\AgdaSpace{}%
\AgdaInductiveConstructor{just}\AgdaSpace{}%
\AgdaSymbol{(}\AgdaInductiveConstructor{inj₂}\AgdaSpace{}%
\AgdaSymbol{(}\AgdaInductiveConstructor{inj₂}\AgdaSpace{}%
\AgdaBound{v}\AgdaSymbol{))}\<%
\\
\>[0]\AgdaFunction{interp}\AgdaSpace{}%
\AgdaInductiveConstructor{unite⋆l}\AgdaSpace{}%
\AgdaBound{v}\AgdaSpace{}%
\AgdaSymbol{=}\AgdaSpace{}%
\AgdaInductiveConstructor{just}\AgdaSpace{}%
\AgdaSymbol{(}\AgdaField{proj₂}\AgdaSpace{}%
\AgdaBound{v}\AgdaSymbol{)}\<%
\\
\>[0]\AgdaFunction{interp}\AgdaSpace{}%
\AgdaInductiveConstructor{uniti⋆l}\AgdaSpace{}%
\AgdaBound{v}\AgdaSpace{}%
\AgdaSymbol{=}\AgdaSpace{}%
\AgdaInductiveConstructor{just}\AgdaSpace{}%
\AgdaSymbol{(}\AgdaInductiveConstructor{tt}\AgdaSpace{}%
\AgdaOperator{\AgdaInductiveConstructor{,}}\AgdaSpace{}%
\AgdaBound{v}\AgdaSymbol{)}\<%
\\
\>[0]\AgdaFunction{interp}\AgdaSpace{}%
\AgdaInductiveConstructor{unite⋆r}\AgdaSpace{}%
\AgdaBound{v}\AgdaSpace{}%
\AgdaSymbol{=}\AgdaSpace{}%
\AgdaInductiveConstructor{just}\AgdaSpace{}%
\AgdaSymbol{(}\AgdaField{proj₁}\AgdaSpace{}%
\AgdaBound{v}\AgdaSymbol{)}\<%
\\
\>[0]\AgdaFunction{interp}\AgdaSpace{}%
\AgdaInductiveConstructor{uniti⋆r}\AgdaSpace{}%
\AgdaBound{v}\AgdaSpace{}%
\AgdaSymbol{=}\AgdaSpace{}%
\AgdaInductiveConstructor{just}\AgdaSpace{}%
\AgdaSymbol{(}\AgdaBound{v}\AgdaSpace{}%
\AgdaOperator{\AgdaInductiveConstructor{,}}\AgdaSpace{}%
\AgdaInductiveConstructor{tt}\AgdaSymbol{)}\<%
\\
\>[0]\AgdaFunction{interp}\AgdaSpace{}%
\AgdaInductiveConstructor{assocl⋆}\AgdaSpace{}%
\AgdaSymbol{(}\AgdaBound{v₁}\AgdaSpace{}%
\AgdaOperator{\AgdaInductiveConstructor{,}}\AgdaSpace{}%
\AgdaBound{v₂}\AgdaSpace{}%
\AgdaOperator{\AgdaInductiveConstructor{,}}\AgdaSpace{}%
\AgdaBound{v₃}\AgdaSymbol{)}\AgdaSpace{}%
\AgdaSymbol{=}\AgdaSpace{}%
\AgdaInductiveConstructor{just}\AgdaSpace{}%
\AgdaSymbol{((}\AgdaBound{v₁}\AgdaSpace{}%
\AgdaOperator{\AgdaInductiveConstructor{,}}\AgdaSpace{}%
\AgdaBound{v₂}\AgdaSymbol{)}\AgdaSpace{}%
\AgdaOperator{\AgdaInductiveConstructor{,}}\AgdaSpace{}%
\AgdaBound{v₃}\AgdaSymbol{)}\<%
\\
\>[0]\AgdaFunction{interp}\AgdaSpace{}%
\AgdaInductiveConstructor{assocr⋆}\AgdaSpace{}%
\AgdaSymbol{((}\AgdaBound{v₁}\AgdaSpace{}%
\AgdaOperator{\AgdaInductiveConstructor{,}}\AgdaSpace{}%
\AgdaBound{v₂}\AgdaSymbol{)}\AgdaSpace{}%
\AgdaOperator{\AgdaInductiveConstructor{,}}\AgdaSpace{}%
\AgdaBound{v₃}\AgdaSymbol{)}\AgdaSpace{}%
\AgdaSymbol{=}\AgdaSpace{}%
\AgdaInductiveConstructor{just}\AgdaSpace{}%
\AgdaSymbol{(}\AgdaBound{v₁}\AgdaSpace{}%
\AgdaOperator{\AgdaInductiveConstructor{,}}\AgdaSpace{}%
\AgdaBound{v₂}\AgdaSpace{}%
\AgdaOperator{\AgdaInductiveConstructor{,}}\AgdaSpace{}%
\AgdaBound{v₃}\AgdaSymbol{)}\<%
\\
\>[0]\AgdaFunction{interp}\AgdaSpace{}%
\AgdaInductiveConstructor{absorbr}\AgdaSpace{}%
\AgdaSymbol{(()}\AgdaSpace{}%
\AgdaOperator{\AgdaInductiveConstructor{,}}\AgdaSpace{}%
\AgdaBound{v}\AgdaSymbol{)}\<%
\\
\>[0]\AgdaFunction{interp}\AgdaSpace{}%
\AgdaInductiveConstructor{absorbl}\AgdaSpace{}%
\AgdaSymbol{(}\AgdaBound{v}\AgdaSpace{}%
\AgdaOperator{\AgdaInductiveConstructor{,}}\AgdaSpace{}%
\AgdaSymbol{())}\<%
\\
\>[0]\AgdaFunction{interp}\AgdaSpace{}%
\AgdaInductiveConstructor{factorzr}\AgdaSpace{}%
\AgdaSymbol{()}\<%
\\
\>[0]\AgdaFunction{interp}\AgdaSpace{}%
\AgdaInductiveConstructor{factorzl}\AgdaSpace{}%
\AgdaSymbol{()}\<%
\\
\>[0]\AgdaFunction{interp}\AgdaSpace{}%
\AgdaInductiveConstructor{dist}\AgdaSpace{}%
\AgdaSymbol{(}\AgdaInductiveConstructor{inj₁}\AgdaSpace{}%
\AgdaBound{v₁}\AgdaSpace{}%
\AgdaOperator{\AgdaInductiveConstructor{,}}\AgdaSpace{}%
\AgdaBound{v₃}\AgdaSymbol{)}\AgdaSpace{}%
\AgdaSymbol{=}\AgdaSpace{}%
\AgdaInductiveConstructor{just}\AgdaSpace{}%
\AgdaSymbol{(}\AgdaInductiveConstructor{inj₁}\AgdaSpace{}%
\AgdaSymbol{(}\AgdaBound{v₁}\AgdaSpace{}%
\AgdaOperator{\AgdaInductiveConstructor{,}}\AgdaSpace{}%
\AgdaBound{v₃}\AgdaSymbol{))}\<%
\\
\>[0]\AgdaFunction{interp}\AgdaSpace{}%
\AgdaInductiveConstructor{dist}\AgdaSpace{}%
\AgdaSymbol{(}\AgdaInductiveConstructor{inj₂}\AgdaSpace{}%
\AgdaBound{v₂}\AgdaSpace{}%
\AgdaOperator{\AgdaInductiveConstructor{,}}\AgdaSpace{}%
\AgdaBound{v₃}\AgdaSymbol{)}\AgdaSpace{}%
\AgdaSymbol{=}\AgdaSpace{}%
\AgdaInductiveConstructor{just}\AgdaSpace{}%
\AgdaSymbol{(}\AgdaInductiveConstructor{inj₂}\AgdaSpace{}%
\AgdaSymbol{(}\AgdaBound{v₂}\AgdaSpace{}%
\AgdaOperator{\AgdaInductiveConstructor{,}}\AgdaSpace{}%
\AgdaBound{v₃}\AgdaSymbol{))}\<%
\\
\>[0]\AgdaFunction{interp}\AgdaSpace{}%
\AgdaInductiveConstructor{factor}\AgdaSpace{}%
\AgdaSymbol{(}\AgdaInductiveConstructor{inj₁}\AgdaSpace{}%
\AgdaSymbol{(}\AgdaBound{v₁}\AgdaSpace{}%
\AgdaOperator{\AgdaInductiveConstructor{,}}\AgdaSpace{}%
\AgdaBound{v₃}\AgdaSymbol{))}\AgdaSpace{}%
\AgdaSymbol{=}\AgdaSpace{}%
\AgdaInductiveConstructor{just}\AgdaSpace{}%
\AgdaSymbol{(}\AgdaInductiveConstructor{inj₁}\AgdaSpace{}%
\AgdaBound{v₁}\AgdaSpace{}%
\AgdaOperator{\AgdaInductiveConstructor{,}}\AgdaSpace{}%
\AgdaBound{v₃}\AgdaSymbol{)}\<%
\\
\>[0]\AgdaFunction{interp}\AgdaSpace{}%
\AgdaInductiveConstructor{factor}\AgdaSpace{}%
\AgdaSymbol{(}\AgdaInductiveConstructor{inj₂}\AgdaSpace{}%
\AgdaSymbol{(}\AgdaBound{v₂}\AgdaSpace{}%
\AgdaOperator{\AgdaInductiveConstructor{,}}\AgdaSpace{}%
\AgdaBound{v₃}\AgdaSymbol{))}\AgdaSpace{}%
\AgdaSymbol{=}\AgdaSpace{}%
\AgdaInductiveConstructor{just}\AgdaSpace{}%
\AgdaSymbol{(}\AgdaInductiveConstructor{inj₂}\AgdaSpace{}%
\AgdaBound{v₂}\AgdaSpace{}%
\AgdaOperator{\AgdaInductiveConstructor{,}}\AgdaSpace{}%
\AgdaBound{v₃}\AgdaSymbol{)}\<%
\\
\>[0]\AgdaFunction{interp}\AgdaSpace{}%
\AgdaInductiveConstructor{distl}\AgdaSpace{}%
\AgdaSymbol{(}\AgdaBound{v₁}\AgdaSpace{}%
\AgdaOperator{\AgdaInductiveConstructor{,}}\AgdaSpace{}%
\AgdaInductiveConstructor{inj₁}\AgdaSpace{}%
\AgdaBound{v₂}\AgdaSymbol{)}\AgdaSpace{}%
\AgdaSymbol{=}\AgdaSpace{}%
\AgdaInductiveConstructor{just}\AgdaSpace{}%
\AgdaSymbol{(}\AgdaInductiveConstructor{inj₁}\AgdaSpace{}%
\AgdaSymbol{(}\AgdaBound{v₁}\AgdaSpace{}%
\AgdaOperator{\AgdaInductiveConstructor{,}}\AgdaSpace{}%
\AgdaBound{v₂}\AgdaSymbol{))}\<%
\\
\>[0]\AgdaFunction{interp}\AgdaSpace{}%
\AgdaInductiveConstructor{distl}\AgdaSpace{}%
\AgdaSymbol{(}\AgdaBound{v₁}\AgdaSpace{}%
\AgdaOperator{\AgdaInductiveConstructor{,}}\AgdaSpace{}%
\AgdaInductiveConstructor{inj₂}\AgdaSpace{}%
\AgdaBound{v₃}\AgdaSymbol{)}\AgdaSpace{}%
\AgdaSymbol{=}\AgdaSpace{}%
\AgdaInductiveConstructor{just}\AgdaSpace{}%
\AgdaSymbol{(}\AgdaInductiveConstructor{inj₂}\AgdaSpace{}%
\AgdaSymbol{(}\AgdaBound{v₁}\AgdaSpace{}%
\AgdaOperator{\AgdaInductiveConstructor{,}}\AgdaSpace{}%
\AgdaBound{v₃}\AgdaSymbol{))}\<%
\\
\>[0]\AgdaFunction{interp}\AgdaSpace{}%
\AgdaInductiveConstructor{factorl}\AgdaSpace{}%
\AgdaSymbol{(}\AgdaInductiveConstructor{inj₁}\AgdaSpace{}%
\AgdaSymbol{(}\AgdaBound{v₁}\AgdaSpace{}%
\AgdaOperator{\AgdaInductiveConstructor{,}}\AgdaSpace{}%
\AgdaBound{v₂}\AgdaSymbol{))}\AgdaSpace{}%
\AgdaSymbol{=}\AgdaSpace{}%
\AgdaInductiveConstructor{just}\AgdaSpace{}%
\AgdaSymbol{(}\AgdaBound{v₁}\AgdaSpace{}%
\AgdaOperator{\AgdaInductiveConstructor{,}}\AgdaSpace{}%
\AgdaInductiveConstructor{inj₁}\AgdaSpace{}%
\AgdaBound{v₂}\AgdaSymbol{)}\<%
\\
\>[0]\AgdaFunction{interp}\AgdaSpace{}%
\AgdaInductiveConstructor{factorl}\AgdaSpace{}%
\AgdaSymbol{(}\AgdaInductiveConstructor{inj₂}\AgdaSpace{}%
\AgdaSymbol{(}\AgdaBound{v₁}\AgdaSpace{}%
\AgdaOperator{\AgdaInductiveConstructor{,}}\AgdaSpace{}%
\AgdaBound{v₃}\AgdaSymbol{))}\AgdaSpace{}%
\AgdaSymbol{=}\AgdaSpace{}%
\AgdaInductiveConstructor{just}\AgdaSpace{}%
\AgdaSymbol{(}\AgdaBound{v₁}\AgdaSpace{}%
\AgdaOperator{\AgdaInductiveConstructor{,}}\AgdaSpace{}%
\AgdaInductiveConstructor{inj₂}\AgdaSpace{}%
\AgdaBound{v₃}\AgdaSymbol{)}\<%
\\
\>[0]\AgdaFunction{interp}\AgdaSpace{}%
\AgdaInductiveConstructor{id⟷}\AgdaSpace{}%
\AgdaBound{v}\AgdaSpace{}%
\AgdaSymbol{=}\AgdaSpace{}%
\AgdaInductiveConstructor{just}\AgdaSpace{}%
\AgdaBound{v}\<%
\\
\>[0]\AgdaFunction{interp}\AgdaSpace{}%
\AgdaSymbol{(}\AgdaBound{c₁}\AgdaSpace{}%
\AgdaOperator{\AgdaInductiveConstructor{⊕}}\AgdaSpace{}%
\AgdaBound{c₂}\AgdaSymbol{)}\AgdaSpace{}%
\AgdaSymbol{(}\AgdaInductiveConstructor{inj₁}\AgdaSpace{}%
\AgdaBound{v}\AgdaSymbol{)}\AgdaSpace{}%
\AgdaSymbol{=}\AgdaSpace{}%
\AgdaFunction{interp}\AgdaSpace{}%
\AgdaBound{c₁}\AgdaSpace{}%
\AgdaBound{v}\AgdaSpace{}%
\AgdaOperator{\AgdaFunction{>>=}}\AgdaSpace{}%
\AgdaInductiveConstructor{just}\AgdaSpace{}%
\AgdaOperator{\AgdaFunction{∘}}\AgdaSpace{}%
\AgdaInductiveConstructor{inj₁}\<%
\\
\>[0]\AgdaFunction{interp}\AgdaSpace{}%
\AgdaSymbol{(}\AgdaBound{c₁}\AgdaSpace{}%
\AgdaOperator{\AgdaInductiveConstructor{⊕}}\AgdaSpace{}%
\AgdaBound{c₂}\AgdaSymbol{)}\AgdaSpace{}%
\AgdaSymbol{(}\AgdaInductiveConstructor{inj₂}\AgdaSpace{}%
\AgdaBound{v}\AgdaSymbol{)}\AgdaSpace{}%
\AgdaSymbol{=}\AgdaSpace{}%
\AgdaFunction{interp}\AgdaSpace{}%
\AgdaBound{c₂}\AgdaSpace{}%
\AgdaBound{v}\AgdaSpace{}%
\AgdaOperator{\AgdaFunction{>>=}}\AgdaSpace{}%
\AgdaInductiveConstructor{just}\AgdaSpace{}%
\AgdaOperator{\AgdaFunction{∘}}\AgdaSpace{}%
\AgdaInductiveConstructor{inj₂}\<%
\\
\>[0]\AgdaFunction{interp}\AgdaSpace{}%
\AgdaSymbol{(}\AgdaBound{c₁}\AgdaSpace{}%
\AgdaOperator{\AgdaInductiveConstructor{⊗}}\AgdaSpace{}%
\AgdaBound{c₂}\AgdaSymbol{)}\AgdaSpace{}%
\AgdaSymbol{(}\AgdaBound{v₁}\AgdaSpace{}%
\AgdaOperator{\AgdaInductiveConstructor{,}}\AgdaSpace{}%
\AgdaBound{v₂}\AgdaSymbol{)}\AgdaSpace{}%
\AgdaSymbol{=}%
\>[1101I]\AgdaFunction{interp}\AgdaSpace{}%
\AgdaBound{c₁}\AgdaSpace{}%
\AgdaBound{v₁}\AgdaSpace{}%
\AgdaOperator{\AgdaFunction{>>=}}\<%
\\
\>[1101I][@{}l@{\AgdaIndent{0}}]%
\>[31]\AgdaSymbol{(λ}%
\>[1105I]\AgdaBound{v₁'}\AgdaSpace{}%
\AgdaSymbol{→}\AgdaSpace{}%
\AgdaFunction{interp}\AgdaSpace{}%
\AgdaBound{c₂}\AgdaSpace{}%
\AgdaBound{v₂}\AgdaSpace{}%
\AgdaOperator{\AgdaFunction{>>=}}\<%
\\
\>[.][@{}l@{}]\<[1105I]%
\>[34]\AgdaSymbol{λ}\AgdaSpace{}%
\AgdaBound{v₂'}\AgdaSpace{}%
\AgdaSymbol{→}\AgdaSpace{}%
\AgdaInductiveConstructor{just}\AgdaSpace{}%
\AgdaSymbol{(}\AgdaBound{v₁'}\AgdaSpace{}%
\AgdaOperator{\AgdaInductiveConstructor{,}}\AgdaSpace{}%
\AgdaBound{v₂'}\AgdaSymbol{))}\<%
\\
\\[\AgdaEmptyExtraSkip]%
\>[0]\AgdaFunction{𝔹}\AgdaSpace{}%
\AgdaSymbol{:}\AgdaSpace{}%
\AgdaDatatype{𝕌}\<%
\\
\>[0]\AgdaFunction{𝔹}\AgdaSpace{}%
\AgdaSymbol{=}\AgdaSpace{}%
\AgdaInductiveConstructor{𝟙}\AgdaSpace{}%
\AgdaOperator{\AgdaInductiveConstructor{+ᵤ}}\AgdaSpace{}%
\AgdaInductiveConstructor{𝟙}\<%
\\
\\[\AgdaEmptyExtraSkip]%
\>[0]\AgdaKeyword{pattern}\AgdaSpace{}%
\AgdaInductiveConstructor{𝔽}\AgdaSpace{}%
\AgdaSymbol{=}\AgdaSpace{}%
\AgdaInductiveConstructor{inj₁}\AgdaSpace{}%
\AgdaInductiveConstructor{tt}\<%
\\
\>[0]\AgdaKeyword{pattern}\AgdaSpace{}%
\AgdaInductiveConstructor{𝕋}\AgdaSpace{}%
\AgdaSymbol{=}\AgdaSpace{}%
\AgdaInductiveConstructor{inj₂}\AgdaSpace{}%
\AgdaInductiveConstructor{tt}\<%
\\
\\[\AgdaEmptyExtraSkip]%
\>[0]\AgdaFunction{CNOT}\AgdaSpace{}%
\AgdaFunction{CNOT'}\AgdaSpace{}%
\AgdaSymbol{:}\AgdaSpace{}%
\AgdaFunction{𝔹}\AgdaSpace{}%
\AgdaOperator{\AgdaInductiveConstructor{×ᵤ}}\AgdaSpace{}%
\AgdaFunction{𝔹}\AgdaSpace{}%
\AgdaOperator{\AgdaDatatype{⟷}}\AgdaSpace{}%
\AgdaFunction{𝔹}\AgdaSpace{}%
\AgdaOperator{\AgdaInductiveConstructor{×ᵤ}}\AgdaSpace{}%
\AgdaFunction{𝔹}\<%
\\
\>[0]\AgdaFunction{CNOT}\AgdaSpace{}%
\AgdaSymbol{=}\AgdaSpace{}%
\AgdaInductiveConstructor{dist}\AgdaSpace{}%
\AgdaOperator{\AgdaInductiveConstructor{⊚}}\AgdaSpace{}%
\AgdaSymbol{(}\AgdaInductiveConstructor{id⟷}\AgdaSpace{}%
\AgdaOperator{\AgdaInductiveConstructor{⊕}}\AgdaSpace{}%
\AgdaSymbol{(}\AgdaInductiveConstructor{id⟷}\AgdaSpace{}%
\AgdaOperator{\AgdaInductiveConstructor{⊗}}\AgdaSpace{}%
\AgdaInductiveConstructor{swap₊}\AgdaSymbol{))}\AgdaSpace{}%
\AgdaOperator{\AgdaInductiveConstructor{⊚}}\AgdaSpace{}%
\AgdaInductiveConstructor{factor}\<%
\\
\>[0]\AgdaFunction{CNOT'}\AgdaSpace{}%
\AgdaSymbol{=}\AgdaSpace{}%
\AgdaInductiveConstructor{distl}\AgdaSpace{}%
\AgdaOperator{\AgdaInductiveConstructor{⊚}}\AgdaSpace{}%
\AgdaSymbol{(}\AgdaInductiveConstructor{id⟷}\AgdaSpace{}%
\AgdaOperator{\AgdaInductiveConstructor{⊕}}\AgdaSpace{}%
\AgdaSymbol{(}\AgdaInductiveConstructor{swap₊}\AgdaSpace{}%
\AgdaOperator{\AgdaInductiveConstructor{⊗}}\AgdaSpace{}%
\AgdaInductiveConstructor{id⟷}\AgdaSymbol{))}\AgdaSpace{}%
\AgdaOperator{\AgdaInductiveConstructor{⊚}}\AgdaSpace{}%
\AgdaInductiveConstructor{factorl}\<%
\\
\\[\AgdaEmptyExtraSkip]%
\>[0]\AgdaKeyword{infixr}\AgdaSpace{}%
\AgdaNumber{2}%
\>[10]\AgdaOperator{\AgdaFunction{\AgdaUnderscore{}⟷⟨\AgdaUnderscore{}⟩\AgdaUnderscore{}}}\<%
\\
\>[0]\AgdaKeyword{infix}%
\>[7]\AgdaNumber{3}%
\>[10]\AgdaOperator{\AgdaFunction{\AgdaUnderscore{}□}}\<%
\\
\\[\AgdaEmptyExtraSkip]%
\>[0]\AgdaOperator{\AgdaFunction{\AgdaUnderscore{}⟷⟨\AgdaUnderscore{}⟩\AgdaUnderscore{}}}\AgdaSpace{}%
\AgdaSymbol{:}%
\>[1162I]\AgdaSymbol{(}\AgdaBound{t₁}\AgdaSpace{}%
\AgdaSymbol{:}\AgdaSpace{}%
\AgdaDatatype{𝕌}\AgdaSymbol{)}\AgdaSpace{}%
\AgdaSymbol{\{}\AgdaBound{t₂}\AgdaSpace{}%
\AgdaSymbol{:}\AgdaSpace{}%
\AgdaDatatype{𝕌}\AgdaSymbol{\}}\AgdaSpace{}%
\AgdaSymbol{\{}\AgdaBound{t₃}\AgdaSpace{}%
\AgdaSymbol{:}\AgdaSpace{}%
\AgdaDatatype{𝕌}\AgdaSymbol{\}}\AgdaSpace{}%
\AgdaSymbol{→}\<%
\\
\>[1162I][@{}l@{\AgdaIndent{0}}]%
\>[10]\AgdaSymbol{(}\AgdaBound{t₁}\AgdaSpace{}%
\AgdaOperator{\AgdaDatatype{⟷}}\AgdaSpace{}%
\AgdaBound{t₂}\AgdaSymbol{)}\AgdaSpace{}%
\AgdaSymbol{→}\AgdaSpace{}%
\AgdaSymbol{(}\AgdaBound{t₂}\AgdaSpace{}%
\AgdaOperator{\AgdaDatatype{⟷}}\AgdaSpace{}%
\AgdaBound{t₃}\AgdaSymbol{)}\AgdaSpace{}%
\AgdaSymbol{→}\AgdaSpace{}%
\AgdaSymbol{(}\AgdaBound{t₁}\AgdaSpace{}%
\AgdaOperator{\AgdaDatatype{⟷}}\AgdaSpace{}%
\AgdaBound{t₃}\AgdaSymbol{)}\<%
\\
\>[0]\AgdaSymbol{\AgdaUnderscore{}}\AgdaSpace{}%
\AgdaOperator{\AgdaFunction{⟷⟨}}\AgdaSpace{}%
\AgdaBound{α}\AgdaSpace{}%
\AgdaOperator{\AgdaFunction{⟩}}\AgdaSpace{}%
\AgdaBound{β}\AgdaSpace{}%
\AgdaSymbol{=}\AgdaSpace{}%
\AgdaBound{α}\AgdaSpace{}%
\AgdaOperator{\AgdaInductiveConstructor{⊚}}\AgdaSpace{}%
\AgdaBound{β}\<%
\\
\\[\AgdaEmptyExtraSkip]%
\>[0]\AgdaOperator{\AgdaFunction{\AgdaUnderscore{}□}}\AgdaSpace{}%
\AgdaSymbol{:}\AgdaSpace{}%
\AgdaSymbol{(}\AgdaBound{t}\AgdaSpace{}%
\AgdaSymbol{:}\AgdaSpace{}%
\AgdaDatatype{𝕌}\AgdaSymbol{)}\AgdaSpace{}%
\AgdaSymbol{→}\AgdaSpace{}%
\AgdaSymbol{\{}\AgdaBound{t}\AgdaSpace{}%
\AgdaSymbol{:}\AgdaSpace{}%
\AgdaDatatype{𝕌}\AgdaSymbol{\}}\AgdaSpace{}%
\AgdaSymbol{→}\AgdaSpace{}%
\AgdaSymbol{(}\AgdaBound{t}\AgdaSpace{}%
\AgdaOperator{\AgdaDatatype{⟷}}\AgdaSpace{}%
\AgdaBound{t}\AgdaSymbol{)}\<%
\\
\>[0]\AgdaOperator{\AgdaFunction{\AgdaUnderscore{}□}}\AgdaSpace{}%
\AgdaBound{t}\AgdaSpace{}%
\AgdaSymbol{=}\AgdaSpace{}%
\AgdaInductiveConstructor{id⟷}\<%
\\
\\[\AgdaEmptyExtraSkip]%
\>[0]\AgdaFunction{𝔹≠𝟘}\AgdaSpace{}%
\AgdaSymbol{:}\AgdaSpace{}%
\AgdaFunction{𝔹}\AgdaSpace{}%
\AgdaOperator{\AgdaFunction{≠𝟘}}\<%
\\
\>[0]\AgdaFunction{𝔹≠𝟘}\AgdaSpace{}%
\AgdaSymbol{=}\AgdaSpace{}%
\AgdaInductiveConstructor{inj₁}\AgdaSpace{}%
\AgdaInductiveConstructor{tt}\<%
\end{code}

\begin{code}[hide]%
\>[0]\AgdaSymbol{\{-\#}\AgdaSpace{}%
\AgdaKeyword{OPTIONS}\AgdaSpace{}%
\AgdaPragma{--without-K}\AgdaSpace{}%
\AgdaPragma{--safe}\AgdaSpace{}%
\AgdaSymbol{\#-\}}\<%
\\
\>[0]\AgdaKeyword{module}\AgdaSpace{}%
\AgdaModule{Pi.D}\AgdaSpace{}%
\AgdaKeyword{where}\<%
\\
\>[0]\AgdaKeyword{open}\AgdaSpace{}%
\AgdaKeyword{import}\AgdaSpace{}%
\AgdaModule{Data.Empty}\<%
\\
\>[0]\AgdaKeyword{open}\AgdaSpace{}%
\AgdaKeyword{import}\AgdaSpace{}%
\AgdaModule{Data.Unit}\<%
\\
\>[0]\AgdaKeyword{open}\AgdaSpace{}%
\AgdaKeyword{import}\AgdaSpace{}%
\AgdaModule{Data.Nat}\<%
\\
\>[0]\AgdaKeyword{open}\AgdaSpace{}%
\AgdaKeyword{import}\AgdaSpace{}%
\AgdaModule{Data.Nat.Properties}\<%
\\
\>[0]\AgdaKeyword{open}\AgdaSpace{}%
\AgdaKeyword{import}\AgdaSpace{}%
\AgdaModule{Data.Sum}\<%
\\
\>[0]\AgdaKeyword{open}\AgdaSpace{}%
\AgdaKeyword{import}\AgdaSpace{}%
\AgdaModule{Data.Product}\<%
\\
\>[0]\AgdaKeyword{open}\AgdaSpace{}%
\AgdaKeyword{import}\AgdaSpace{}%
\AgdaModule{Data.Maybe}\AgdaSpace{}%
\AgdaKeyword{renaming}\AgdaSpace{}%
\AgdaSymbol{(}\AgdaOperator{\AgdaFunction{\AgdaUnderscore{}>>=\AgdaUnderscore{}}}\AgdaSpace{}%
\AgdaSymbol{to}\AgdaSpace{}%
\AgdaOperator{\AgdaFunction{\AgdaUnderscore{}≫=\AgdaUnderscore{}}}\AgdaSymbol{)}\<%
\\
\>[0]\AgdaKeyword{open}\AgdaSpace{}%
\AgdaKeyword{import}\AgdaSpace{}%
\AgdaModule{Function}\AgdaSpace{}%
\AgdaKeyword{using}\AgdaSpace{}%
\AgdaSymbol{(}\AgdaOperator{\AgdaFunction{\AgdaUnderscore{}∘\AgdaUnderscore{}}}\AgdaSymbol{)}\<%
\\
\>[0]\AgdaKeyword{open}\AgdaSpace{}%
\AgdaKeyword{import}\AgdaSpace{}%
\AgdaModule{Relation.Binary.PropositionalEquality}\<%
\\
\>[0]\AgdaKeyword{open}\AgdaSpace{}%
\AgdaKeyword{import}\AgdaSpace{}%
\AgdaModule{Relation.Binary.Core}\<%
\\
\>[0]\AgdaKeyword{open}\AgdaSpace{}%
\AgdaKeyword{import}\AgdaSpace{}%
\AgdaModule{Relation.Nullary}\<%
\\
\>[0]\AgdaKeyword{open}\AgdaSpace{}%
\AgdaKeyword{import}\AgdaSpace{}%
\AgdaModule{Induction.Nat}\<%
\\
\>[0]\AgdaKeyword{open}\AgdaSpace{}%
\AgdaKeyword{import}\AgdaSpace{}%
\AgdaModule{Relation.Binary}\<%
\\
\\[\AgdaEmptyExtraSkip]%
\>[0]\AgdaKeyword{infix}%
\>[7]\AgdaNumber{90}\AgdaSpace{}%
\AgdaOperator{\AgdaInductiveConstructor{𝟙/\AgdaUnderscore{}}}\<%
\\
\>[0]\AgdaKeyword{infixr}\AgdaSpace{}%
\AgdaNumber{70}\AgdaSpace{}%
\AgdaOperator{\AgdaInductiveConstructor{\AgdaUnderscore{}×ᵤ\AgdaUnderscore{}}}\<%
\\
\>[0]\AgdaKeyword{infixr}\AgdaSpace{}%
\AgdaNumber{60}\AgdaSpace{}%
\AgdaOperator{\AgdaInductiveConstructor{\AgdaUnderscore{}+ᵤ\AgdaUnderscore{}}}\<%
\\
\>[0]\AgdaKeyword{infix}%
\>[7]\AgdaNumber{40}\AgdaSpace{}%
\AgdaOperator{\AgdaDatatype{\AgdaUnderscore{}⟷\AgdaUnderscore{}}}\<%
\\
\>[0]\AgdaKeyword{infixr}\AgdaSpace{}%
\AgdaNumber{50}\AgdaSpace{}%
\AgdaOperator{\AgdaInductiveConstructor{\AgdaUnderscore{}⊚\AgdaUnderscore{}}}\<%
\\
\>[0]\AgdaKeyword{infix}\AgdaSpace{}%
\AgdaNumber{100}\AgdaSpace{}%
\AgdaOperator{\AgdaFunction{!\AgdaUnderscore{}}}\<%
\\
\\[\AgdaEmptyExtraSkip]%
\>[0]\AgdaKeyword{data}\AgdaSpace{}%
\AgdaDatatype{◯}\AgdaSpace{}%
\AgdaSymbol{:}\AgdaSpace{}%
\AgdaPrimitiveType{Set}\AgdaSpace{}%
\AgdaKeyword{where}\<%
\\
\>[0][@{}l@{\AgdaIndent{0}}]%
\>[2]\AgdaInductiveConstructor{↻}\AgdaSpace{}%
\AgdaSymbol{:}\AgdaSpace{}%
\AgdaDatatype{◯}\<%
\\
\\[\AgdaEmptyExtraSkip]%
\>[0]\AgdaKeyword{mutual}\<%
\end{code}
\newcommand{\PIFDUdef}{%
\begin{code}[hide]%
\>[0][@{}l@{\AgdaIndent{1}}]%
\>[2]\AgdaKeyword{data}\AgdaSpace{}%
\AgdaDatatype{𝕌}\AgdaSpace{}%
\AgdaSymbol{:}\AgdaSpace{}%
\AgdaPrimitiveType{Set}\AgdaSpace{}%
\AgdaKeyword{where}\<%
\\
\>[2][@{}l@{\AgdaIndent{0}}]%
\>[4]\AgdaInductiveConstructor{𝟘}\AgdaSpace{}%
\AgdaSymbol{:}\AgdaSpace{}%
\AgdaDatatype{𝕌}\<%
\\
\>[4]\AgdaInductiveConstructor{𝟙}\AgdaSpace{}%
\AgdaSymbol{:}\AgdaSpace{}%
\AgdaDatatype{𝕌}\<%
\\
\>[4]\AgdaOperator{\AgdaInductiveConstructor{\AgdaUnderscore{}+ᵤ\AgdaUnderscore{}}}\AgdaSpace{}%
\AgdaSymbol{:}\AgdaSpace{}%
\AgdaDatatype{𝕌}\AgdaSpace{}%
\AgdaSymbol{→}\AgdaSpace{}%
\AgdaDatatype{𝕌}\AgdaSpace{}%
\AgdaSymbol{→}\AgdaSpace{}%
\AgdaDatatype{𝕌}\<%
\\
\>[4]\AgdaOperator{\AgdaInductiveConstructor{\AgdaUnderscore{}×ᵤ\AgdaUnderscore{}}}\AgdaSpace{}%
\AgdaSymbol{:}\AgdaSpace{}%
\AgdaDatatype{𝕌}\AgdaSpace{}%
\AgdaSymbol{→}\AgdaSpace{}%
\AgdaDatatype{𝕌}\AgdaSpace{}%
\AgdaSymbol{→}\AgdaSpace{}%
\AgdaDatatype{𝕌}\<%
\end{code}
\begin{code}%
\>[4]\AgdaOperator{\AgdaInductiveConstructor{𝟙/\AgdaUnderscore{}}}\AgdaSpace{}%
\AgdaSymbol{:}\AgdaSpace{}%
\AgdaSymbol{\{}\AgdaBound{t}\AgdaSpace{}%
\AgdaSymbol{:}\AgdaSpace{}%
\AgdaDatatype{𝕌}\AgdaSymbol{\}}\AgdaSpace{}%
\AgdaSymbol{→}\AgdaSpace{}%
\AgdaOperator{\AgdaFunction{⟦}}\AgdaSpace{}%
\AgdaBound{t}\AgdaSpace{}%
\AgdaOperator{\AgdaFunction{⟧}}\AgdaSpace{}%
\AgdaSymbol{→}\AgdaSpace{}%
\AgdaDatatype{𝕌}\<%
\end{code}}
\begin{code}[hide]%
\>[2]\AgdaOperator{\AgdaFunction{⟦\AgdaUnderscore{}⟧}}\AgdaSpace{}%
\AgdaSymbol{:}\AgdaSpace{}%
\AgdaDatatype{𝕌}\AgdaSpace{}%
\AgdaSymbol{→}\AgdaSpace{}%
\AgdaPrimitiveType{Set}\<%
\\
\>[2]\AgdaOperator{\AgdaFunction{⟦}}\AgdaSpace{}%
\AgdaInductiveConstructor{𝟘}\AgdaSpace{}%
\AgdaOperator{\AgdaFunction{⟧}}\AgdaSpace{}%
\AgdaSymbol{=}\AgdaSpace{}%
\AgdaDatatype{⊥}\<%
\\
\>[2]\AgdaOperator{\AgdaFunction{⟦}}\AgdaSpace{}%
\AgdaInductiveConstructor{𝟙}\AgdaSpace{}%
\AgdaOperator{\AgdaFunction{⟧}}\AgdaSpace{}%
\AgdaSymbol{=}\AgdaSpace{}%
\AgdaRecord{⊤}\<%
\\
\>[2]\AgdaOperator{\AgdaFunction{⟦}}\AgdaSpace{}%
\AgdaBound{t₁}\AgdaSpace{}%
\AgdaOperator{\AgdaInductiveConstructor{+ᵤ}}\AgdaSpace{}%
\AgdaBound{t₂}\AgdaSpace{}%
\AgdaOperator{\AgdaFunction{⟧}}\AgdaSpace{}%
\AgdaSymbol{=}\AgdaSpace{}%
\AgdaOperator{\AgdaFunction{⟦}}\AgdaSpace{}%
\AgdaBound{t₁}\AgdaSpace{}%
\AgdaOperator{\AgdaFunction{⟧}}\AgdaSpace{}%
\AgdaOperator{\AgdaDatatype{⊎}}\AgdaSpace{}%
\AgdaOperator{\AgdaFunction{⟦}}\AgdaSpace{}%
\AgdaBound{t₂}\AgdaSpace{}%
\AgdaOperator{\AgdaFunction{⟧}}\<%
\\
\>[2]\AgdaOperator{\AgdaFunction{⟦}}\AgdaSpace{}%
\AgdaBound{t₁}\AgdaSpace{}%
\AgdaOperator{\AgdaInductiveConstructor{×ᵤ}}\AgdaSpace{}%
\AgdaBound{t₂}\AgdaSpace{}%
\AgdaOperator{\AgdaFunction{⟧}}\AgdaSpace{}%
\AgdaSymbol{=}\AgdaSpace{}%
\AgdaOperator{\AgdaFunction{⟦}}\AgdaSpace{}%
\AgdaBound{t₁}\AgdaSpace{}%
\AgdaOperator{\AgdaFunction{⟧}}\AgdaSpace{}%
\AgdaOperator{\AgdaFunction{×}}\AgdaSpace{}%
\AgdaOperator{\AgdaFunction{⟦}}\AgdaSpace{}%
\AgdaBound{t₂}\AgdaSpace{}%
\AgdaOperator{\AgdaFunction{⟧}}\<%
\\
\>[2]\AgdaOperator{\AgdaFunction{⟦}}\AgdaSpace{}%
\AgdaOperator{\AgdaInductiveConstructor{𝟙/}}\AgdaSpace{}%
\AgdaBound{v}\AgdaSpace{}%
\AgdaOperator{\AgdaFunction{⟧}}\AgdaSpace{}%
\AgdaSymbol{=}\AgdaSpace{}%
\AgdaDatatype{◯}\<%
\end{code}
\newcommand{\PIFDCombdef}{%
\begin{code}[hide]%
\>[2]\AgdaKeyword{data}\AgdaSpace{}%
\AgdaOperator{\AgdaDatatype{\AgdaUnderscore{}⟷\AgdaUnderscore{}}}\AgdaSpace{}%
\AgdaSymbol{:}\AgdaSpace{}%
\AgdaDatatype{𝕌}\AgdaSpace{}%
\AgdaSymbol{→}\AgdaSpace{}%
\AgdaDatatype{𝕌}\AgdaSpace{}%
\AgdaSymbol{→}\AgdaSpace{}%
\AgdaPrimitiveType{Set}\AgdaSpace{}%
\AgdaKeyword{where}\<%
\\
\>[2][@{}l@{\AgdaIndent{0}}]%
\>[4]\AgdaInductiveConstructor{unite₊l}\AgdaSpace{}%
\AgdaSymbol{:}\AgdaSpace{}%
\AgdaSymbol{\{}\AgdaBound{t}\AgdaSpace{}%
\AgdaSymbol{:}\AgdaSpace{}%
\AgdaDatatype{𝕌}\AgdaSymbol{\}}\AgdaSpace{}%
\AgdaSymbol{→}\AgdaSpace{}%
\AgdaInductiveConstructor{𝟘}\AgdaSpace{}%
\AgdaOperator{\AgdaInductiveConstructor{+ᵤ}}\AgdaSpace{}%
\AgdaBound{t}\AgdaSpace{}%
\AgdaOperator{\AgdaDatatype{⟷}}\AgdaSpace{}%
\AgdaBound{t}\<%
\\
\>[4]\AgdaInductiveConstructor{uniti₊l}\AgdaSpace{}%
\AgdaSymbol{:}\AgdaSpace{}%
\AgdaSymbol{\{}\AgdaBound{t}\AgdaSpace{}%
\AgdaSymbol{:}\AgdaSpace{}%
\AgdaDatatype{𝕌}\AgdaSymbol{\}}\AgdaSpace{}%
\AgdaSymbol{→}\AgdaSpace{}%
\AgdaBound{t}\AgdaSpace{}%
\AgdaOperator{\AgdaDatatype{⟷}}\AgdaSpace{}%
\AgdaInductiveConstructor{𝟘}\AgdaSpace{}%
\AgdaOperator{\AgdaInductiveConstructor{+ᵤ}}\AgdaSpace{}%
\AgdaBound{t}\<%
\\
\>[4]\AgdaInductiveConstructor{unite₊r}\AgdaSpace{}%
\AgdaSymbol{:}\AgdaSpace{}%
\AgdaSymbol{\{}\AgdaBound{t}\AgdaSpace{}%
\AgdaSymbol{:}\AgdaSpace{}%
\AgdaDatatype{𝕌}\AgdaSymbol{\}}\AgdaSpace{}%
\AgdaSymbol{→}\AgdaSpace{}%
\AgdaBound{t}\AgdaSpace{}%
\AgdaOperator{\AgdaInductiveConstructor{+ᵤ}}\AgdaSpace{}%
\AgdaInductiveConstructor{𝟘}\AgdaSpace{}%
\AgdaOperator{\AgdaDatatype{⟷}}\AgdaSpace{}%
\AgdaBound{t}\<%
\\
\>[4]\AgdaInductiveConstructor{uniti₊r}\AgdaSpace{}%
\AgdaSymbol{:}\AgdaSpace{}%
\AgdaSymbol{\{}\AgdaBound{t}\AgdaSpace{}%
\AgdaSymbol{:}\AgdaSpace{}%
\AgdaDatatype{𝕌}\AgdaSymbol{\}}\AgdaSpace{}%
\AgdaSymbol{→}\AgdaSpace{}%
\AgdaBound{t}\AgdaSpace{}%
\AgdaOperator{\AgdaDatatype{⟷}}\AgdaSpace{}%
\AgdaBound{t}\AgdaSpace{}%
\AgdaOperator{\AgdaInductiveConstructor{+ᵤ}}\AgdaSpace{}%
\AgdaInductiveConstructor{𝟘}\<%
\\
\>[4]\AgdaInductiveConstructor{swap₊}%
\>[12]\AgdaSymbol{:}\AgdaSpace{}%
\AgdaSymbol{\{}\AgdaBound{t₁}\AgdaSpace{}%
\AgdaBound{t₂}\AgdaSpace{}%
\AgdaSymbol{:}\AgdaSpace{}%
\AgdaDatatype{𝕌}\AgdaSymbol{\}}\AgdaSpace{}%
\AgdaSymbol{→}\AgdaSpace{}%
\AgdaBound{t₁}\AgdaSpace{}%
\AgdaOperator{\AgdaInductiveConstructor{+ᵤ}}\AgdaSpace{}%
\AgdaBound{t₂}\AgdaSpace{}%
\AgdaOperator{\AgdaDatatype{⟷}}\AgdaSpace{}%
\AgdaBound{t₂}\AgdaSpace{}%
\AgdaOperator{\AgdaInductiveConstructor{+ᵤ}}\AgdaSpace{}%
\AgdaBound{t₁}\<%
\\
\>[4]\AgdaInductiveConstructor{assocl₊}\AgdaSpace{}%
\AgdaSymbol{:}\AgdaSpace{}%
\AgdaSymbol{\{}\AgdaBound{t₁}\AgdaSpace{}%
\AgdaBound{t₂}\AgdaSpace{}%
\AgdaBound{t₃}\AgdaSpace{}%
\AgdaSymbol{:}\AgdaSpace{}%
\AgdaDatatype{𝕌}\AgdaSymbol{\}}\AgdaSpace{}%
\AgdaSymbol{→}\AgdaSpace{}%
\AgdaBound{t₁}\AgdaSpace{}%
\AgdaOperator{\AgdaInductiveConstructor{+ᵤ}}\AgdaSpace{}%
\AgdaSymbol{(}\AgdaBound{t₂}\AgdaSpace{}%
\AgdaOperator{\AgdaInductiveConstructor{+ᵤ}}\AgdaSpace{}%
\AgdaBound{t₃}\AgdaSymbol{)}\AgdaSpace{}%
\AgdaOperator{\AgdaDatatype{⟷}}\AgdaSpace{}%
\AgdaSymbol{(}\AgdaBound{t₁}\AgdaSpace{}%
\AgdaOperator{\AgdaInductiveConstructor{+ᵤ}}\AgdaSpace{}%
\AgdaBound{t₂}\AgdaSymbol{)}\AgdaSpace{}%
\AgdaOperator{\AgdaInductiveConstructor{+ᵤ}}\AgdaSpace{}%
\AgdaBound{t₃}\<%
\\
\>[4]\AgdaInductiveConstructor{assocr₊}\AgdaSpace{}%
\AgdaSymbol{:}\AgdaSpace{}%
\AgdaSymbol{\{}\AgdaBound{t₁}\AgdaSpace{}%
\AgdaBound{t₂}\AgdaSpace{}%
\AgdaBound{t₃}\AgdaSpace{}%
\AgdaSymbol{:}\AgdaSpace{}%
\AgdaDatatype{𝕌}\AgdaSymbol{\}}\AgdaSpace{}%
\AgdaSymbol{→}\AgdaSpace{}%
\AgdaSymbol{(}\AgdaBound{t₁}\AgdaSpace{}%
\AgdaOperator{\AgdaInductiveConstructor{+ᵤ}}\AgdaSpace{}%
\AgdaBound{t₂}\AgdaSymbol{)}\AgdaSpace{}%
\AgdaOperator{\AgdaInductiveConstructor{+ᵤ}}\AgdaSpace{}%
\AgdaBound{t₃}\AgdaSpace{}%
\AgdaOperator{\AgdaDatatype{⟷}}\AgdaSpace{}%
\AgdaBound{t₁}\AgdaSpace{}%
\AgdaOperator{\AgdaInductiveConstructor{+ᵤ}}\AgdaSpace{}%
\AgdaSymbol{(}\AgdaBound{t₂}\AgdaSpace{}%
\AgdaOperator{\AgdaInductiveConstructor{+ᵤ}}\AgdaSpace{}%
\AgdaBound{t₃}\AgdaSymbol{)}\<%
\\
\>[4]\AgdaInductiveConstructor{unite⋆l}\AgdaSpace{}%
\AgdaSymbol{:}\AgdaSpace{}%
\AgdaSymbol{\{}\AgdaBound{t}\AgdaSpace{}%
\AgdaSymbol{:}\AgdaSpace{}%
\AgdaDatatype{𝕌}\AgdaSymbol{\}}\AgdaSpace{}%
\AgdaSymbol{→}\AgdaSpace{}%
\AgdaInductiveConstructor{𝟙}\AgdaSpace{}%
\AgdaOperator{\AgdaInductiveConstructor{×ᵤ}}\AgdaSpace{}%
\AgdaBound{t}\AgdaSpace{}%
\AgdaOperator{\AgdaDatatype{⟷}}\AgdaSpace{}%
\AgdaBound{t}\<%
\\
\>[4]\AgdaInductiveConstructor{uniti⋆l}\AgdaSpace{}%
\AgdaSymbol{:}\AgdaSpace{}%
\AgdaSymbol{\{}\AgdaBound{t}\AgdaSpace{}%
\AgdaSymbol{:}\AgdaSpace{}%
\AgdaDatatype{𝕌}\AgdaSymbol{\}}\AgdaSpace{}%
\AgdaSymbol{→}\AgdaSpace{}%
\AgdaBound{t}\AgdaSpace{}%
\AgdaOperator{\AgdaDatatype{⟷}}\AgdaSpace{}%
\AgdaInductiveConstructor{𝟙}\AgdaSpace{}%
\AgdaOperator{\AgdaInductiveConstructor{×ᵤ}}\AgdaSpace{}%
\AgdaBound{t}\<%
\\
\>[4]\AgdaInductiveConstructor{unite⋆r}\AgdaSpace{}%
\AgdaSymbol{:}\AgdaSpace{}%
\AgdaSymbol{\{}\AgdaBound{t}\AgdaSpace{}%
\AgdaSymbol{:}\AgdaSpace{}%
\AgdaDatatype{𝕌}\AgdaSymbol{\}}\AgdaSpace{}%
\AgdaSymbol{→}\AgdaSpace{}%
\AgdaBound{t}\AgdaSpace{}%
\AgdaOperator{\AgdaInductiveConstructor{×ᵤ}}\AgdaSpace{}%
\AgdaInductiveConstructor{𝟙}\AgdaSpace{}%
\AgdaOperator{\AgdaDatatype{⟷}}\AgdaSpace{}%
\AgdaBound{t}\<%
\\
\>[4]\AgdaInductiveConstructor{uniti⋆r}\AgdaSpace{}%
\AgdaSymbol{:}\AgdaSpace{}%
\AgdaSymbol{\{}\AgdaBound{t}\AgdaSpace{}%
\AgdaSymbol{:}\AgdaSpace{}%
\AgdaDatatype{𝕌}\AgdaSymbol{\}}\AgdaSpace{}%
\AgdaSymbol{→}\AgdaSpace{}%
\AgdaBound{t}\AgdaSpace{}%
\AgdaOperator{\AgdaDatatype{⟷}}\AgdaSpace{}%
\AgdaBound{t}\AgdaSpace{}%
\AgdaOperator{\AgdaInductiveConstructor{×ᵤ}}\AgdaSpace{}%
\AgdaInductiveConstructor{𝟙}\<%
\\
\>[4]\AgdaInductiveConstructor{swap⋆}%
\>[12]\AgdaSymbol{:}\AgdaSpace{}%
\AgdaSymbol{\{}\AgdaBound{t₁}\AgdaSpace{}%
\AgdaBound{t₂}\AgdaSpace{}%
\AgdaSymbol{:}\AgdaSpace{}%
\AgdaDatatype{𝕌}\AgdaSymbol{\}}\AgdaSpace{}%
\AgdaSymbol{→}\AgdaSpace{}%
\AgdaBound{t₁}\AgdaSpace{}%
\AgdaOperator{\AgdaInductiveConstructor{×ᵤ}}\AgdaSpace{}%
\AgdaBound{t₂}\AgdaSpace{}%
\AgdaOperator{\AgdaDatatype{⟷}}\AgdaSpace{}%
\AgdaBound{t₂}\AgdaSpace{}%
\AgdaOperator{\AgdaInductiveConstructor{×ᵤ}}\AgdaSpace{}%
\AgdaBound{t₁}\<%
\\
\>[4]\AgdaInductiveConstructor{assocl⋆}\AgdaSpace{}%
\AgdaSymbol{:}\AgdaSpace{}%
\AgdaSymbol{\{}\AgdaBound{t₁}\AgdaSpace{}%
\AgdaBound{t₂}\AgdaSpace{}%
\AgdaBound{t₃}\AgdaSpace{}%
\AgdaSymbol{:}\AgdaSpace{}%
\AgdaDatatype{𝕌}\AgdaSymbol{\}}\AgdaSpace{}%
\AgdaSymbol{→}\AgdaSpace{}%
\AgdaBound{t₁}\AgdaSpace{}%
\AgdaOperator{\AgdaInductiveConstructor{×ᵤ}}\AgdaSpace{}%
\AgdaSymbol{(}\AgdaBound{t₂}\AgdaSpace{}%
\AgdaOperator{\AgdaInductiveConstructor{×ᵤ}}\AgdaSpace{}%
\AgdaBound{t₃}\AgdaSymbol{)}\AgdaSpace{}%
\AgdaOperator{\AgdaDatatype{⟷}}\AgdaSpace{}%
\AgdaSymbol{(}\AgdaBound{t₁}\AgdaSpace{}%
\AgdaOperator{\AgdaInductiveConstructor{×ᵤ}}\AgdaSpace{}%
\AgdaBound{t₂}\AgdaSymbol{)}\AgdaSpace{}%
\AgdaOperator{\AgdaInductiveConstructor{×ᵤ}}\AgdaSpace{}%
\AgdaBound{t₃}\<%
\\
\>[4]\AgdaInductiveConstructor{assocr⋆}\AgdaSpace{}%
\AgdaSymbol{:}\AgdaSpace{}%
\AgdaSymbol{\{}\AgdaBound{t₁}\AgdaSpace{}%
\AgdaBound{t₂}\AgdaSpace{}%
\AgdaBound{t₃}\AgdaSpace{}%
\AgdaSymbol{:}\AgdaSpace{}%
\AgdaDatatype{𝕌}\AgdaSymbol{\}}\AgdaSpace{}%
\AgdaSymbol{→}\AgdaSpace{}%
\AgdaSymbol{(}\AgdaBound{t₁}\AgdaSpace{}%
\AgdaOperator{\AgdaInductiveConstructor{×ᵤ}}\AgdaSpace{}%
\AgdaBound{t₂}\AgdaSymbol{)}\AgdaSpace{}%
\AgdaOperator{\AgdaInductiveConstructor{×ᵤ}}\AgdaSpace{}%
\AgdaBound{t₃}\AgdaSpace{}%
\AgdaOperator{\AgdaDatatype{⟷}}\AgdaSpace{}%
\AgdaBound{t₁}\AgdaSpace{}%
\AgdaOperator{\AgdaInductiveConstructor{×ᵤ}}\AgdaSpace{}%
\AgdaSymbol{(}\AgdaBound{t₂}\AgdaSpace{}%
\AgdaOperator{\AgdaInductiveConstructor{×ᵤ}}\AgdaSpace{}%
\AgdaBound{t₃}\AgdaSymbol{)}\<%
\\
\>[4]\AgdaInductiveConstructor{absorbr}\AgdaSpace{}%
\AgdaSymbol{:}\AgdaSpace{}%
\AgdaSymbol{\{}\AgdaBound{t}\AgdaSpace{}%
\AgdaSymbol{:}\AgdaSpace{}%
\AgdaDatatype{𝕌}\AgdaSymbol{\}}\AgdaSpace{}%
\AgdaSymbol{→}\AgdaSpace{}%
\AgdaInductiveConstructor{𝟘}\AgdaSpace{}%
\AgdaOperator{\AgdaInductiveConstructor{×ᵤ}}\AgdaSpace{}%
\AgdaBound{t}\AgdaSpace{}%
\AgdaOperator{\AgdaDatatype{⟷}}\AgdaSpace{}%
\AgdaInductiveConstructor{𝟘}\<%
\\
\>[4]\AgdaInductiveConstructor{absorbl}\AgdaSpace{}%
\AgdaSymbol{:}\AgdaSpace{}%
\AgdaSymbol{\{}\AgdaBound{t}\AgdaSpace{}%
\AgdaSymbol{:}\AgdaSpace{}%
\AgdaDatatype{𝕌}\AgdaSymbol{\}}\AgdaSpace{}%
\AgdaSymbol{→}\AgdaSpace{}%
\AgdaBound{t}\AgdaSpace{}%
\AgdaOperator{\AgdaInductiveConstructor{×ᵤ}}\AgdaSpace{}%
\AgdaInductiveConstructor{𝟘}\AgdaSpace{}%
\AgdaOperator{\AgdaDatatype{⟷}}\AgdaSpace{}%
\AgdaInductiveConstructor{𝟘}\<%
\\
\>[4]\AgdaInductiveConstructor{factorzr}\AgdaSpace{}%
\AgdaSymbol{:}\AgdaSpace{}%
\AgdaSymbol{\{}\AgdaBound{t}\AgdaSpace{}%
\AgdaSymbol{:}\AgdaSpace{}%
\AgdaDatatype{𝕌}\AgdaSymbol{\}}\AgdaSpace{}%
\AgdaSymbol{→}\AgdaSpace{}%
\AgdaInductiveConstructor{𝟘}\AgdaSpace{}%
\AgdaOperator{\AgdaDatatype{⟷}}\AgdaSpace{}%
\AgdaBound{t}\AgdaSpace{}%
\AgdaOperator{\AgdaInductiveConstructor{×ᵤ}}\AgdaSpace{}%
\AgdaInductiveConstructor{𝟘}\<%
\\
\>[4]\AgdaInductiveConstructor{factorzl}\AgdaSpace{}%
\AgdaSymbol{:}\AgdaSpace{}%
\AgdaSymbol{\{}\AgdaBound{t}\AgdaSpace{}%
\AgdaSymbol{:}\AgdaSpace{}%
\AgdaDatatype{𝕌}\AgdaSymbol{\}}\AgdaSpace{}%
\AgdaSymbol{→}\AgdaSpace{}%
\AgdaInductiveConstructor{𝟘}\AgdaSpace{}%
\AgdaOperator{\AgdaDatatype{⟷}}\AgdaSpace{}%
\AgdaInductiveConstructor{𝟘}\AgdaSpace{}%
\AgdaOperator{\AgdaInductiveConstructor{×ᵤ}}\AgdaSpace{}%
\AgdaBound{t}\<%
\\
\>[4]\AgdaInductiveConstructor{dist}%
\>[12]\AgdaSymbol{:}\AgdaSpace{}%
\AgdaSymbol{\{}\AgdaBound{t₁}\AgdaSpace{}%
\AgdaBound{t₂}\AgdaSpace{}%
\AgdaBound{t₃}\AgdaSpace{}%
\AgdaSymbol{:}\AgdaSpace{}%
\AgdaDatatype{𝕌}\AgdaSymbol{\}}\AgdaSpace{}%
\AgdaSymbol{→}\AgdaSpace{}%
\AgdaSymbol{(}\AgdaBound{t₁}\AgdaSpace{}%
\AgdaOperator{\AgdaInductiveConstructor{+ᵤ}}\AgdaSpace{}%
\AgdaBound{t₂}\AgdaSymbol{)}\AgdaSpace{}%
\AgdaOperator{\AgdaInductiveConstructor{×ᵤ}}\AgdaSpace{}%
\AgdaBound{t₃}\AgdaSpace{}%
\AgdaOperator{\AgdaDatatype{⟷}}\AgdaSpace{}%
\AgdaSymbol{(}\AgdaBound{t₁}\AgdaSpace{}%
\AgdaOperator{\AgdaInductiveConstructor{×ᵤ}}\AgdaSpace{}%
\AgdaBound{t₃}\AgdaSymbol{)}\AgdaSpace{}%
\AgdaOperator{\AgdaInductiveConstructor{+ᵤ}}\AgdaSpace{}%
\AgdaSymbol{(}\AgdaBound{t₂}\AgdaSpace{}%
\AgdaOperator{\AgdaInductiveConstructor{×ᵤ}}\AgdaSpace{}%
\AgdaBound{t₃}\AgdaSymbol{)}\<%
\\
\>[4]\AgdaInductiveConstructor{factor}%
\>[12]\AgdaSymbol{:}\AgdaSpace{}%
\AgdaSymbol{\{}\AgdaBound{t₁}\AgdaSpace{}%
\AgdaBound{t₂}\AgdaSpace{}%
\AgdaBound{t₃}\AgdaSpace{}%
\AgdaSymbol{:}\AgdaSpace{}%
\AgdaDatatype{𝕌}\AgdaSymbol{\}}\AgdaSpace{}%
\AgdaSymbol{→}\AgdaSpace{}%
\AgdaSymbol{(}\AgdaBound{t₁}\AgdaSpace{}%
\AgdaOperator{\AgdaInductiveConstructor{×ᵤ}}\AgdaSpace{}%
\AgdaBound{t₃}\AgdaSymbol{)}\AgdaSpace{}%
\AgdaOperator{\AgdaInductiveConstructor{+ᵤ}}\AgdaSpace{}%
\AgdaSymbol{(}\AgdaBound{t₂}\AgdaSpace{}%
\AgdaOperator{\AgdaInductiveConstructor{×ᵤ}}\AgdaSpace{}%
\AgdaBound{t₃}\AgdaSymbol{)}\AgdaSpace{}%
\AgdaOperator{\AgdaDatatype{⟷}}\AgdaSpace{}%
\AgdaSymbol{(}\AgdaBound{t₁}\AgdaSpace{}%
\AgdaOperator{\AgdaInductiveConstructor{+ᵤ}}\AgdaSpace{}%
\AgdaBound{t₂}\AgdaSymbol{)}\AgdaSpace{}%
\AgdaOperator{\AgdaInductiveConstructor{×ᵤ}}\AgdaSpace{}%
\AgdaBound{t₃}\<%
\\
\>[4]\AgdaInductiveConstructor{distl}%
\>[12]\AgdaSymbol{:}\AgdaSpace{}%
\AgdaSymbol{\{}\AgdaBound{t₁}\AgdaSpace{}%
\AgdaBound{t₂}\AgdaSpace{}%
\AgdaBound{t₃}\AgdaSpace{}%
\AgdaSymbol{:}\AgdaSpace{}%
\AgdaDatatype{𝕌}\AgdaSymbol{\}}\AgdaSpace{}%
\AgdaSymbol{→}\AgdaSpace{}%
\AgdaBound{t₁}\AgdaSpace{}%
\AgdaOperator{\AgdaInductiveConstructor{×ᵤ}}\AgdaSpace{}%
\AgdaSymbol{(}\AgdaBound{t₂}\AgdaSpace{}%
\AgdaOperator{\AgdaInductiveConstructor{+ᵤ}}\AgdaSpace{}%
\AgdaBound{t₃}\AgdaSymbol{)}\AgdaSpace{}%
\AgdaOperator{\AgdaDatatype{⟷}}\AgdaSpace{}%
\AgdaSymbol{(}\AgdaBound{t₁}\AgdaSpace{}%
\AgdaOperator{\AgdaInductiveConstructor{×ᵤ}}\AgdaSpace{}%
\AgdaBound{t₂}\AgdaSymbol{)}\AgdaSpace{}%
\AgdaOperator{\AgdaInductiveConstructor{+ᵤ}}\AgdaSpace{}%
\AgdaSymbol{(}\AgdaBound{t₁}\AgdaSpace{}%
\AgdaOperator{\AgdaInductiveConstructor{×ᵤ}}\AgdaSpace{}%
\AgdaBound{t₃}\AgdaSymbol{)}\<%
\\
\>[4]\AgdaInductiveConstructor{factorl}\AgdaSpace{}%
\AgdaSymbol{:}\AgdaSpace{}%
\AgdaSymbol{\{}\AgdaBound{t₁}\AgdaSpace{}%
\AgdaBound{t₂}\AgdaSpace{}%
\AgdaBound{t₃}\AgdaSpace{}%
\AgdaSymbol{:}\AgdaSpace{}%
\AgdaDatatype{𝕌}\AgdaSpace{}%
\AgdaSymbol{\}}\AgdaSpace{}%
\AgdaSymbol{→}\AgdaSpace{}%
\AgdaSymbol{(}\AgdaBound{t₁}\AgdaSpace{}%
\AgdaOperator{\AgdaInductiveConstructor{×ᵤ}}\AgdaSpace{}%
\AgdaBound{t₂}\AgdaSymbol{)}\AgdaSpace{}%
\AgdaOperator{\AgdaInductiveConstructor{+ᵤ}}\AgdaSpace{}%
\AgdaSymbol{(}\AgdaBound{t₁}\AgdaSpace{}%
\AgdaOperator{\AgdaInductiveConstructor{×ᵤ}}\AgdaSpace{}%
\AgdaBound{t₃}\AgdaSymbol{)}\AgdaSpace{}%
\AgdaOperator{\AgdaDatatype{⟷}}\AgdaSpace{}%
\AgdaBound{t₁}\AgdaSpace{}%
\AgdaOperator{\AgdaInductiveConstructor{×ᵤ}}\AgdaSpace{}%
\AgdaSymbol{(}\AgdaBound{t₂}\AgdaSpace{}%
\AgdaOperator{\AgdaInductiveConstructor{+ᵤ}}\AgdaSpace{}%
\AgdaBound{t₃}\AgdaSymbol{)}\<%
\\
\>[4]\AgdaInductiveConstructor{id⟷}%
\>[12]\AgdaSymbol{:}\AgdaSpace{}%
\AgdaSymbol{\{}\AgdaBound{t}\AgdaSpace{}%
\AgdaSymbol{:}\AgdaSpace{}%
\AgdaDatatype{𝕌}\AgdaSymbol{\}}\AgdaSpace{}%
\AgdaSymbol{→}\AgdaSpace{}%
\AgdaBound{t}\AgdaSpace{}%
\AgdaOperator{\AgdaDatatype{⟷}}\AgdaSpace{}%
\AgdaBound{t}\<%
\\
\>[4]\AgdaOperator{\AgdaInductiveConstructor{\AgdaUnderscore{}⊚\AgdaUnderscore{}}}%
\>[12]\AgdaSymbol{:}\AgdaSpace{}%
\AgdaSymbol{\{}\AgdaBound{t₁}\AgdaSpace{}%
\AgdaBound{t₂}\AgdaSpace{}%
\AgdaBound{t₃}\AgdaSpace{}%
\AgdaSymbol{:}\AgdaSpace{}%
\AgdaDatatype{𝕌}\AgdaSymbol{\}}\AgdaSpace{}%
\AgdaSymbol{→}\AgdaSpace{}%
\AgdaSymbol{(}\AgdaBound{t₁}\AgdaSpace{}%
\AgdaOperator{\AgdaDatatype{⟷}}\AgdaSpace{}%
\AgdaBound{t₂}\AgdaSymbol{)}\AgdaSpace{}%
\AgdaSymbol{→}\AgdaSpace{}%
\AgdaSymbol{(}\AgdaBound{t₂}\AgdaSpace{}%
\AgdaOperator{\AgdaDatatype{⟷}}\AgdaSpace{}%
\AgdaBound{t₃}\AgdaSymbol{)}\AgdaSpace{}%
\AgdaSymbol{→}\AgdaSpace{}%
\AgdaSymbol{(}\AgdaBound{t₁}\AgdaSpace{}%
\AgdaOperator{\AgdaDatatype{⟷}}\AgdaSpace{}%
\AgdaBound{t₃}\AgdaSymbol{)}\<%
\\
\>[4]\AgdaOperator{\AgdaInductiveConstructor{\AgdaUnderscore{}⊕\AgdaUnderscore{}}}%
\>[12]\AgdaSymbol{:}\AgdaSpace{}%
\AgdaSymbol{\{}\AgdaBound{t₁}\AgdaSpace{}%
\AgdaBound{t₂}\AgdaSpace{}%
\AgdaBound{t₃}\AgdaSpace{}%
\AgdaBound{t₄}\AgdaSpace{}%
\AgdaSymbol{:}\AgdaSpace{}%
\AgdaDatatype{𝕌}\AgdaSymbol{\}}\AgdaSpace{}%
\AgdaSymbol{→}\AgdaSpace{}%
\AgdaSymbol{(}\AgdaBound{t₁}\AgdaSpace{}%
\AgdaOperator{\AgdaDatatype{⟷}}\AgdaSpace{}%
\AgdaBound{t₃}\AgdaSymbol{)}\AgdaSpace{}%
\AgdaSymbol{→}\AgdaSpace{}%
\AgdaSymbol{(}\AgdaBound{t₂}\AgdaSpace{}%
\AgdaOperator{\AgdaDatatype{⟷}}\AgdaSpace{}%
\AgdaBound{t₄}\AgdaSymbol{)}\AgdaSpace{}%
\AgdaSymbol{→}\AgdaSpace{}%
\AgdaSymbol{(}\AgdaBound{t₁}\AgdaSpace{}%
\AgdaOperator{\AgdaInductiveConstructor{+ᵤ}}\AgdaSpace{}%
\AgdaBound{t₂}\AgdaSpace{}%
\AgdaOperator{\AgdaDatatype{⟷}}\AgdaSpace{}%
\AgdaBound{t₃}\AgdaSpace{}%
\AgdaOperator{\AgdaInductiveConstructor{+ᵤ}}\AgdaSpace{}%
\AgdaBound{t₄}\AgdaSymbol{)}\<%
\\
\>[4]\AgdaOperator{\AgdaInductiveConstructor{\AgdaUnderscore{}⊗\AgdaUnderscore{}}}%
\>[12]\AgdaSymbol{:}\AgdaSpace{}%
\AgdaSymbol{\{}\AgdaBound{t₁}\AgdaSpace{}%
\AgdaBound{t₂}\AgdaSpace{}%
\AgdaBound{t₃}\AgdaSpace{}%
\AgdaBound{t₄}\AgdaSpace{}%
\AgdaSymbol{:}\AgdaSpace{}%
\AgdaDatatype{𝕌}\AgdaSymbol{\}}\AgdaSpace{}%
\AgdaSymbol{→}\AgdaSpace{}%
\AgdaSymbol{(}\AgdaBound{t₁}\AgdaSpace{}%
\AgdaOperator{\AgdaDatatype{⟷}}\AgdaSpace{}%
\AgdaBound{t₃}\AgdaSymbol{)}\AgdaSpace{}%
\AgdaSymbol{→}\AgdaSpace{}%
\AgdaSymbol{(}\AgdaBound{t₂}\AgdaSpace{}%
\AgdaOperator{\AgdaDatatype{⟷}}\AgdaSpace{}%
\AgdaBound{t₄}\AgdaSymbol{)}\AgdaSpace{}%
\AgdaSymbol{→}\AgdaSpace{}%
\AgdaSymbol{(}\AgdaBound{t₁}\AgdaSpace{}%
\AgdaOperator{\AgdaInductiveConstructor{×ᵤ}}\AgdaSpace{}%
\AgdaBound{t₂}\AgdaSpace{}%
\AgdaOperator{\AgdaDatatype{⟷}}\AgdaSpace{}%
\AgdaBound{t₃}\AgdaSpace{}%
\AgdaOperator{\AgdaInductiveConstructor{×ᵤ}}\AgdaSpace{}%
\AgdaBound{t₄}\AgdaSymbol{)}\<%
\end{code}
\begin{code}%
\>[4]\AgdaInductiveConstructor{η}\AgdaSpace{}%
\AgdaSymbol{:}\AgdaSpace{}%
\AgdaSymbol{\{}\AgdaBound{t}\AgdaSpace{}%
\AgdaSymbol{:}\AgdaSpace{}%
\AgdaDatatype{𝕌}\AgdaSymbol{\}}\AgdaSpace{}%
\AgdaSymbol{(}\AgdaBound{v}\AgdaSpace{}%
\AgdaSymbol{:}\AgdaSpace{}%
\AgdaOperator{\AgdaFunction{⟦}}\AgdaSpace{}%
\AgdaBound{t}\AgdaSpace{}%
\AgdaOperator{\AgdaFunction{⟧}}\AgdaSymbol{)}\AgdaSpace{}%
\AgdaSymbol{→}\AgdaSpace{}%
\AgdaInductiveConstructor{𝟙}\AgdaSpace{}%
\AgdaOperator{\AgdaDatatype{⟷}}\AgdaSpace{}%
\AgdaBound{t}\AgdaSpace{}%
\AgdaOperator{\AgdaInductiveConstructor{×ᵤ}}\AgdaSpace{}%
\AgdaSymbol{(}\AgdaOperator{\AgdaInductiveConstructor{𝟙/}}\AgdaSpace{}%
\AgdaBound{v}\AgdaSymbol{)}\<%
\\
\>[4]\AgdaInductiveConstructor{ε}\AgdaSpace{}%
\AgdaSymbol{:}\AgdaSpace{}%
\AgdaSymbol{\{}\AgdaBound{t}\AgdaSpace{}%
\AgdaSymbol{:}\AgdaSpace{}%
\AgdaDatatype{𝕌}\AgdaSymbol{\}}\AgdaSpace{}%
\AgdaSymbol{(}\AgdaBound{v}\AgdaSpace{}%
\AgdaSymbol{:}\AgdaSpace{}%
\AgdaOperator{\AgdaFunction{⟦}}\AgdaSpace{}%
\AgdaBound{t}\AgdaSpace{}%
\AgdaOperator{\AgdaFunction{⟧}}\AgdaSymbol{)}\AgdaSpace{}%
\AgdaSymbol{→}\AgdaSpace{}%
\AgdaBound{t}\AgdaSpace{}%
\AgdaOperator{\AgdaInductiveConstructor{×ᵤ}}\AgdaSpace{}%
\AgdaSymbol{(}\AgdaOperator{\AgdaInductiveConstructor{𝟙/}}\AgdaSpace{}%
\AgdaBound{v}\AgdaSymbol{)}\AgdaSpace{}%
\AgdaOperator{\AgdaDatatype{⟷}}\AgdaSpace{}%
\AgdaInductiveConstructor{𝟙}\<%
\end{code}}
\begin{code}[hide]%
\>[0]\AgdaOperator{\AgdaFunction{\AgdaUnderscore{}≟ᵤ\AgdaUnderscore{}}}\AgdaSpace{}%
\AgdaSymbol{:}\AgdaSpace{}%
\AgdaSymbol{\{}\AgdaBound{t}\AgdaSpace{}%
\AgdaSymbol{:}\AgdaSpace{}%
\AgdaDatatype{𝕌}\AgdaSymbol{\}}\AgdaSpace{}%
\AgdaSymbol{→}\AgdaSpace{}%
\AgdaFunction{Decidable}\AgdaSpace{}%
\AgdaSymbol{(}\AgdaOperator{\AgdaDatatype{\AgdaUnderscore{}≡\AgdaUnderscore{}}}\AgdaSpace{}%
\AgdaSymbol{\{}\AgdaArgument{A}\AgdaSpace{}%
\AgdaSymbol{=}\AgdaSpace{}%
\AgdaOperator{\AgdaFunction{⟦}}\AgdaSpace{}%
\AgdaBound{t}\AgdaSpace{}%
\AgdaOperator{\AgdaFunction{⟧}}\AgdaSymbol{\})}\<%
\\
\>[0]\AgdaOperator{\AgdaFunction{\AgdaUnderscore{}≟ᵤ\AgdaUnderscore{}}}\AgdaSpace{}%
\AgdaSymbol{\{}\AgdaInductiveConstructor{𝟙}\AgdaSymbol{\}}\AgdaSpace{}%
\AgdaInductiveConstructor{tt}\AgdaSpace{}%
\AgdaInductiveConstructor{tt}\AgdaSpace{}%
\AgdaSymbol{=}\AgdaSpace{}%
\AgdaInductiveConstructor{yes}\AgdaSpace{}%
\AgdaInductiveConstructor{refl}\<%
\\
\>[0]\AgdaOperator{\AgdaFunction{\AgdaUnderscore{}≟ᵤ\AgdaUnderscore{}}}\AgdaSpace{}%
\AgdaSymbol{\{}\AgdaBound{t₁}\AgdaSpace{}%
\AgdaOperator{\AgdaInductiveConstructor{+ᵤ}}\AgdaSpace{}%
\AgdaBound{t₂}\AgdaSymbol{\}}\AgdaSpace{}%
\AgdaSymbol{(}\AgdaInductiveConstructor{inj₁}\AgdaSpace{}%
\AgdaBound{x}\AgdaSymbol{)}\AgdaSpace{}%
\AgdaSymbol{(}\AgdaInductiveConstructor{inj₁}\AgdaSpace{}%
\AgdaBound{y}\AgdaSymbol{)}\AgdaSpace{}%
\AgdaKeyword{with}\AgdaSpace{}%
\AgdaBound{x}\AgdaSpace{}%
\AgdaOperator{\AgdaFunction{≟ᵤ}}\AgdaSpace{}%
\AgdaBound{y}\<%
\\
\>[0]\AgdaSymbol{...}\AgdaSpace{}%
\AgdaSymbol{|}\AgdaSpace{}%
\AgdaInductiveConstructor{yes}\AgdaSpace{}%
\AgdaInductiveConstructor{refl}\AgdaSpace{}%
\AgdaSymbol{=}\AgdaSpace{}%
\AgdaInductiveConstructor{yes}\AgdaSpace{}%
\AgdaInductiveConstructor{refl}\<%
\\
\>[0]\AgdaSymbol{...}\AgdaSpace{}%
\AgdaSymbol{|}\AgdaSpace{}%
\AgdaInductiveConstructor{no}%
\>[10]\AgdaBound{x≠y}%
\>[15]\AgdaSymbol{=}\AgdaSpace{}%
\AgdaInductiveConstructor{no}\AgdaSpace{}%
\AgdaSymbol{(λ}\AgdaSpace{}%
\AgdaSymbol{\{}\AgdaInductiveConstructor{refl}\AgdaSpace{}%
\AgdaSymbol{→}\AgdaSpace{}%
\AgdaBound{x≠y}\AgdaSpace{}%
\AgdaInductiveConstructor{refl}\AgdaSymbol{\})}\<%
\\
\>[0]\AgdaOperator{\AgdaFunction{\AgdaUnderscore{}≟ᵤ\AgdaUnderscore{}}}\AgdaSpace{}%
\AgdaSymbol{\{}\AgdaBound{t₁}\AgdaSpace{}%
\AgdaOperator{\AgdaInductiveConstructor{+ᵤ}}\AgdaSpace{}%
\AgdaBound{t₂}\AgdaSymbol{\}}\AgdaSpace{}%
\AgdaSymbol{(}\AgdaInductiveConstructor{inj₁}\AgdaSpace{}%
\AgdaBound{x}\AgdaSymbol{)}\AgdaSpace{}%
\AgdaSymbol{(}\AgdaInductiveConstructor{inj₂}\AgdaSpace{}%
\AgdaBound{y}\AgdaSymbol{)}\AgdaSpace{}%
\AgdaSymbol{=}\AgdaSpace{}%
\AgdaInductiveConstructor{no}\AgdaSpace{}%
\AgdaSymbol{(λ}\AgdaSpace{}%
\AgdaSymbol{())}\<%
\\
\>[0]\AgdaOperator{\AgdaFunction{\AgdaUnderscore{}≟ᵤ\AgdaUnderscore{}}}\AgdaSpace{}%
\AgdaSymbol{\{}\AgdaBound{t₁}\AgdaSpace{}%
\AgdaOperator{\AgdaInductiveConstructor{+ᵤ}}\AgdaSpace{}%
\AgdaBound{t₂}\AgdaSymbol{\}}\AgdaSpace{}%
\AgdaSymbol{(}\AgdaInductiveConstructor{inj₂}\AgdaSpace{}%
\AgdaBound{x}\AgdaSymbol{)}\AgdaSpace{}%
\AgdaSymbol{(}\AgdaInductiveConstructor{inj₁}\AgdaSpace{}%
\AgdaBound{y}\AgdaSymbol{)}\AgdaSpace{}%
\AgdaSymbol{=}\AgdaSpace{}%
\AgdaInductiveConstructor{no}\AgdaSpace{}%
\AgdaSymbol{(λ}\AgdaSpace{}%
\AgdaSymbol{())}\<%
\\
\>[0]\AgdaOperator{\AgdaFunction{\AgdaUnderscore{}≟ᵤ\AgdaUnderscore{}}}\AgdaSpace{}%
\AgdaSymbol{\{}\AgdaBound{t₁}\AgdaSpace{}%
\AgdaOperator{\AgdaInductiveConstructor{+ᵤ}}\AgdaSpace{}%
\AgdaBound{t₂}\AgdaSymbol{\}}\AgdaSpace{}%
\AgdaSymbol{(}\AgdaInductiveConstructor{inj₂}\AgdaSpace{}%
\AgdaBound{x}\AgdaSymbol{)}\AgdaSpace{}%
\AgdaSymbol{(}\AgdaInductiveConstructor{inj₂}\AgdaSpace{}%
\AgdaBound{y}\AgdaSymbol{)}\AgdaSpace{}%
\AgdaKeyword{with}\AgdaSpace{}%
\AgdaBound{x}\AgdaSpace{}%
\AgdaOperator{\AgdaFunction{≟ᵤ}}\AgdaSpace{}%
\AgdaBound{y}\<%
\\
\>[0]\AgdaSymbol{...}\AgdaSpace{}%
\AgdaSymbol{|}\AgdaSpace{}%
\AgdaInductiveConstructor{yes}\AgdaSpace{}%
\AgdaInductiveConstructor{refl}\AgdaSpace{}%
\AgdaSymbol{=}\AgdaSpace{}%
\AgdaInductiveConstructor{yes}\AgdaSpace{}%
\AgdaInductiveConstructor{refl}\<%
\\
\>[0]\AgdaSymbol{...}\AgdaSpace{}%
\AgdaSymbol{|}\AgdaSpace{}%
\AgdaInductiveConstructor{no}%
\>[10]\AgdaBound{x≠y}%
\>[15]\AgdaSymbol{=}\AgdaSpace{}%
\AgdaInductiveConstructor{no}\AgdaSpace{}%
\AgdaSymbol{(λ}\AgdaSpace{}%
\AgdaSymbol{\{}\AgdaInductiveConstructor{refl}\AgdaSpace{}%
\AgdaSymbol{→}\AgdaSpace{}%
\AgdaBound{x≠y}\AgdaSpace{}%
\AgdaInductiveConstructor{refl}\AgdaSymbol{\})}\<%
\\
\>[0]\AgdaOperator{\AgdaFunction{\AgdaUnderscore{}≟ᵤ\AgdaUnderscore{}}}\AgdaSpace{}%
\AgdaSymbol{\{}\AgdaBound{t₁}\AgdaSpace{}%
\AgdaOperator{\AgdaInductiveConstructor{×ᵤ}}\AgdaSpace{}%
\AgdaBound{t₂}\AgdaSymbol{\}}\AgdaSpace{}%
\AgdaSymbol{(}\AgdaBound{x₁}\AgdaSpace{}%
\AgdaOperator{\AgdaInductiveConstructor{,}}\AgdaSpace{}%
\AgdaBound{x₂}\AgdaSymbol{)}\AgdaSpace{}%
\AgdaSymbol{(}\AgdaBound{y₁}\AgdaSpace{}%
\AgdaOperator{\AgdaInductiveConstructor{,}}\AgdaSpace{}%
\AgdaBound{y₂}\AgdaSymbol{)}\AgdaSpace{}%
\AgdaKeyword{with}\AgdaSpace{}%
\AgdaBound{x₁}\AgdaSpace{}%
\AgdaOperator{\AgdaFunction{≟ᵤ}}\AgdaSpace{}%
\AgdaBound{y₁}\AgdaSpace{}%
\AgdaSymbol{|}\AgdaSpace{}%
\AgdaBound{x₂}\AgdaSpace{}%
\AgdaOperator{\AgdaFunction{≟ᵤ}}\AgdaSpace{}%
\AgdaBound{y₂}\<%
\\
\>[0]\AgdaSymbol{...}\AgdaSpace{}%
\AgdaSymbol{|}\AgdaSpace{}%
\AgdaInductiveConstructor{yes}\AgdaSpace{}%
\AgdaInductiveConstructor{refl}\AgdaSpace{}%
\AgdaSymbol{|}\AgdaSpace{}%
\AgdaInductiveConstructor{yes}\AgdaSpace{}%
\AgdaInductiveConstructor{refl}\AgdaSpace{}%
\AgdaSymbol{=}\AgdaSpace{}%
\AgdaInductiveConstructor{yes}\AgdaSpace{}%
\AgdaInductiveConstructor{refl}\<%
\\
\>[0]\AgdaSymbol{...}\AgdaSpace{}%
\AgdaSymbol{|}\AgdaSpace{}%
\AgdaInductiveConstructor{yes}\AgdaSpace{}%
\AgdaInductiveConstructor{refl}\AgdaSpace{}%
\AgdaSymbol{|}\AgdaSpace{}%
\AgdaInductiveConstructor{no}\AgdaSpace{}%
\AgdaBound{x₂≠y₂}\AgdaSpace{}%
\AgdaSymbol{=}\AgdaSpace{}%
\AgdaInductiveConstructor{no}%
\>[32]\AgdaSymbol{(}\AgdaBound{x₂≠y₂}\AgdaSpace{}%
\AgdaOperator{\AgdaFunction{∘}}\AgdaSpace{}%
\AgdaFunction{cong}\AgdaSpace{}%
\AgdaField{proj₂}\AgdaSymbol{)}\<%
\\
\>[0]\AgdaSymbol{...}\AgdaSpace{}%
\AgdaSymbol{|}\AgdaSpace{}%
\AgdaInductiveConstructor{no}\AgdaSpace{}%
\AgdaBound{x₁≠y₁}\AgdaSpace{}%
\AgdaSymbol{|}\AgdaSpace{}%
\AgdaInductiveConstructor{yes}\AgdaSpace{}%
\AgdaInductiveConstructor{refl}\AgdaSpace{}%
\AgdaSymbol{=}\AgdaSpace{}%
\AgdaInductiveConstructor{no}\AgdaSpace{}%
\AgdaSymbol{(}\AgdaBound{x₁≠y₁}\AgdaSpace{}%
\AgdaOperator{\AgdaFunction{∘}}\AgdaSpace{}%
\AgdaFunction{cong}\AgdaSpace{}%
\AgdaField{proj₁}\AgdaSymbol{)}\<%
\\
\>[0]\AgdaSymbol{...}\AgdaSpace{}%
\AgdaSymbol{|}\AgdaSpace{}%
\AgdaInductiveConstructor{no}\AgdaSpace{}%
\AgdaBound{x₁≠y₁}\AgdaSpace{}%
\AgdaSymbol{|}\AgdaSpace{}%
\AgdaInductiveConstructor{no}\AgdaSpace{}%
\AgdaBound{x₂≠y₂}\AgdaSpace{}%
\AgdaSymbol{=}\AgdaSpace{}%
\AgdaInductiveConstructor{no}\AgdaSpace{}%
\AgdaSymbol{(}\AgdaBound{x₁≠y₁}\AgdaSpace{}%
\AgdaOperator{\AgdaFunction{∘}}\AgdaSpace{}%
\AgdaFunction{cong}\AgdaSpace{}%
\AgdaField{proj₁}\AgdaSymbol{)}\<%
\\
\>[0]\AgdaOperator{\AgdaFunction{\AgdaUnderscore{}≟ᵤ\AgdaUnderscore{}}}\AgdaSpace{}%
\AgdaSymbol{\{}\AgdaOperator{\AgdaInductiveConstructor{𝟙/}}\AgdaSpace{}%
\AgdaBound{v}\AgdaSymbol{\}}\AgdaSpace{}%
\AgdaInductiveConstructor{↻}\AgdaSpace{}%
\AgdaInductiveConstructor{↻}\AgdaSpace{}%
\AgdaSymbol{=}\AgdaSpace{}%
\AgdaInductiveConstructor{yes}\AgdaSpace{}%
\AgdaInductiveConstructor{refl}\<%
\\
\\[\AgdaEmptyExtraSkip]%
\>[0]\AgdaOperator{\AgdaFunction{!\AgdaUnderscore{}}}\AgdaSpace{}%
\AgdaSymbol{:}\AgdaSpace{}%
\AgdaSymbol{\{}\AgdaBound{t₁}\AgdaSpace{}%
\AgdaBound{t₂}\AgdaSpace{}%
\AgdaSymbol{:}\AgdaSpace{}%
\AgdaDatatype{𝕌}\AgdaSymbol{\}}\AgdaSpace{}%
\AgdaSymbol{→}\AgdaSpace{}%
\AgdaBound{t₁}\AgdaSpace{}%
\AgdaOperator{\AgdaDatatype{⟷}}\AgdaSpace{}%
\AgdaBound{t₂}\AgdaSpace{}%
\AgdaSymbol{→}\AgdaSpace{}%
\AgdaBound{t₂}\AgdaSpace{}%
\AgdaOperator{\AgdaDatatype{⟷}}\AgdaSpace{}%
\AgdaBound{t₁}\<%
\\
\>[0]\AgdaOperator{\AgdaFunction{!}}\AgdaSpace{}%
\AgdaInductiveConstructor{unite₊l}\AgdaSpace{}%
\AgdaSymbol{=}\AgdaSpace{}%
\AgdaInductiveConstructor{uniti₊l}\<%
\\
\>[0]\AgdaOperator{\AgdaFunction{!}}\AgdaSpace{}%
\AgdaInductiveConstructor{uniti₊l}\AgdaSpace{}%
\AgdaSymbol{=}\AgdaSpace{}%
\AgdaInductiveConstructor{unite₊l}\<%
\\
\>[0]\AgdaOperator{\AgdaFunction{!}}\AgdaSpace{}%
\AgdaInductiveConstructor{unite₊r}\AgdaSpace{}%
\AgdaSymbol{=}\AgdaSpace{}%
\AgdaInductiveConstructor{uniti₊r}\<%
\\
\>[0]\AgdaOperator{\AgdaFunction{!}}\AgdaSpace{}%
\AgdaInductiveConstructor{uniti₊r}\AgdaSpace{}%
\AgdaSymbol{=}\AgdaSpace{}%
\AgdaInductiveConstructor{unite₊r}\<%
\\
\>[0]\AgdaOperator{\AgdaFunction{!}}\AgdaSpace{}%
\AgdaInductiveConstructor{swap₊}\AgdaSpace{}%
\AgdaSymbol{=}\AgdaSpace{}%
\AgdaInductiveConstructor{swap₊}\<%
\\
\>[0]\AgdaOperator{\AgdaFunction{!}}\AgdaSpace{}%
\AgdaInductiveConstructor{assocl₊}\AgdaSpace{}%
\AgdaSymbol{=}\AgdaSpace{}%
\AgdaInductiveConstructor{assocr₊}\<%
\\
\>[0]\AgdaOperator{\AgdaFunction{!}}\AgdaSpace{}%
\AgdaInductiveConstructor{assocr₊}\AgdaSpace{}%
\AgdaSymbol{=}\AgdaSpace{}%
\AgdaInductiveConstructor{assocl₊}\<%
\\
\>[0]\AgdaOperator{\AgdaFunction{!}}\AgdaSpace{}%
\AgdaInductiveConstructor{unite⋆l}\AgdaSpace{}%
\AgdaSymbol{=}\AgdaSpace{}%
\AgdaInductiveConstructor{uniti⋆l}\<%
\\
\>[0]\AgdaOperator{\AgdaFunction{!}}\AgdaSpace{}%
\AgdaInductiveConstructor{uniti⋆l}\AgdaSpace{}%
\AgdaSymbol{=}\AgdaSpace{}%
\AgdaInductiveConstructor{unite⋆l}\<%
\\
\>[0]\AgdaOperator{\AgdaFunction{!}}\AgdaSpace{}%
\AgdaInductiveConstructor{unite⋆r}\AgdaSpace{}%
\AgdaSymbol{=}\AgdaSpace{}%
\AgdaInductiveConstructor{uniti⋆r}\<%
\\
\>[0]\AgdaOperator{\AgdaFunction{!}}\AgdaSpace{}%
\AgdaInductiveConstructor{uniti⋆r}\AgdaSpace{}%
\AgdaSymbol{=}\AgdaSpace{}%
\AgdaInductiveConstructor{unite⋆r}\<%
\\
\>[0]\AgdaOperator{\AgdaFunction{!}}\AgdaSpace{}%
\AgdaInductiveConstructor{swap⋆}\AgdaSpace{}%
\AgdaSymbol{=}\AgdaSpace{}%
\AgdaInductiveConstructor{swap⋆}\<%
\\
\>[0]\AgdaOperator{\AgdaFunction{!}}\AgdaSpace{}%
\AgdaInductiveConstructor{assocl⋆}\AgdaSpace{}%
\AgdaSymbol{=}\AgdaSpace{}%
\AgdaInductiveConstructor{assocr⋆}\<%
\\
\>[0]\AgdaOperator{\AgdaFunction{!}}\AgdaSpace{}%
\AgdaInductiveConstructor{assocr⋆}\AgdaSpace{}%
\AgdaSymbol{=}\AgdaSpace{}%
\AgdaInductiveConstructor{assocl⋆}\<%
\\
\>[0]\AgdaOperator{\AgdaFunction{!}}\AgdaSpace{}%
\AgdaInductiveConstructor{absorbr}\AgdaSpace{}%
\AgdaSymbol{=}\AgdaSpace{}%
\AgdaInductiveConstructor{factorzl}\<%
\\
\>[0]\AgdaOperator{\AgdaFunction{!}}\AgdaSpace{}%
\AgdaInductiveConstructor{absorbl}\AgdaSpace{}%
\AgdaSymbol{=}\AgdaSpace{}%
\AgdaInductiveConstructor{factorzr}\<%
\\
\>[0]\AgdaOperator{\AgdaFunction{!}}\AgdaSpace{}%
\AgdaInductiveConstructor{factorzr}\AgdaSpace{}%
\AgdaSymbol{=}\AgdaSpace{}%
\AgdaInductiveConstructor{absorbl}\<%
\\
\>[0]\AgdaOperator{\AgdaFunction{!}}\AgdaSpace{}%
\AgdaInductiveConstructor{factorzl}\AgdaSpace{}%
\AgdaSymbol{=}\AgdaSpace{}%
\AgdaInductiveConstructor{absorbr}\<%
\\
\>[0]\AgdaOperator{\AgdaFunction{!}}\AgdaSpace{}%
\AgdaInductiveConstructor{dist}\AgdaSpace{}%
\AgdaSymbol{=}\AgdaSpace{}%
\AgdaInductiveConstructor{factor}\<%
\\
\>[0]\AgdaOperator{\AgdaFunction{!}}\AgdaSpace{}%
\AgdaInductiveConstructor{factor}\AgdaSpace{}%
\AgdaSymbol{=}\AgdaSpace{}%
\AgdaInductiveConstructor{dist}\<%
\\
\>[0]\AgdaOperator{\AgdaFunction{!}}\AgdaSpace{}%
\AgdaInductiveConstructor{distl}\AgdaSpace{}%
\AgdaSymbol{=}\AgdaSpace{}%
\AgdaInductiveConstructor{factorl}\<%
\\
\>[0]\AgdaOperator{\AgdaFunction{!}}\AgdaSpace{}%
\AgdaInductiveConstructor{factorl}\AgdaSpace{}%
\AgdaSymbol{=}\AgdaSpace{}%
\AgdaInductiveConstructor{distl}\<%
\\
\>[0]\AgdaOperator{\AgdaFunction{!}}\AgdaSpace{}%
\AgdaInductiveConstructor{id⟷}\AgdaSpace{}%
\AgdaSymbol{=}\AgdaSpace{}%
\AgdaInductiveConstructor{id⟷}\<%
\\
\>[0]\AgdaOperator{\AgdaFunction{!}}\AgdaSpace{}%
\AgdaSymbol{(}\AgdaBound{c}\AgdaSpace{}%
\AgdaOperator{\AgdaInductiveConstructor{⊚}}\AgdaSpace{}%
\AgdaBound{c₁}\AgdaSymbol{)}\AgdaSpace{}%
\AgdaSymbol{=}\AgdaSpace{}%
\AgdaOperator{\AgdaFunction{!}}\AgdaSpace{}%
\AgdaBound{c₁}\AgdaSpace{}%
\AgdaOperator{\AgdaInductiveConstructor{⊚}}\AgdaSpace{}%
\AgdaOperator{\AgdaFunction{!}}\AgdaSpace{}%
\AgdaBound{c}\<%
\\
\>[0]\AgdaOperator{\AgdaFunction{!}}\AgdaSpace{}%
\AgdaSymbol{(}\AgdaBound{c}\AgdaSpace{}%
\AgdaOperator{\AgdaInductiveConstructor{⊕}}\AgdaSpace{}%
\AgdaBound{c₁}\AgdaSymbol{)}\AgdaSpace{}%
\AgdaSymbol{=}\AgdaSpace{}%
\AgdaSymbol{(}\AgdaOperator{\AgdaFunction{!}}\AgdaSpace{}%
\AgdaBound{c}\AgdaSymbol{)}\AgdaSpace{}%
\AgdaOperator{\AgdaInductiveConstructor{⊕}}\AgdaSpace{}%
\AgdaSymbol{(}\AgdaOperator{\AgdaFunction{!}}\AgdaSpace{}%
\AgdaBound{c₁}\AgdaSymbol{)}\<%
\\
\>[0]\AgdaOperator{\AgdaFunction{!}}\AgdaSpace{}%
\AgdaSymbol{(}\AgdaBound{c}\AgdaSpace{}%
\AgdaOperator{\AgdaInductiveConstructor{⊗}}\AgdaSpace{}%
\AgdaBound{c₁}\AgdaSymbol{)}\AgdaSpace{}%
\AgdaSymbol{=}\AgdaSpace{}%
\AgdaSymbol{(}\AgdaOperator{\AgdaFunction{!}}\AgdaSpace{}%
\AgdaBound{c}\AgdaSymbol{)}\AgdaSpace{}%
\AgdaOperator{\AgdaInductiveConstructor{⊗}}\AgdaSpace{}%
\AgdaSymbol{(}\AgdaOperator{\AgdaFunction{!}}\AgdaSpace{}%
\AgdaBound{c₁}\AgdaSymbol{)}\<%
\\
\>[0]\AgdaOperator{\AgdaFunction{!}}\AgdaSpace{}%
\AgdaInductiveConstructor{η}\AgdaSpace{}%
\AgdaBound{v}\AgdaSpace{}%
\AgdaSymbol{=}\AgdaSpace{}%
\AgdaInductiveConstructor{ε}\AgdaSpace{}%
\AgdaBound{v}\<%
\\
\>[0]\AgdaOperator{\AgdaFunction{!}}\AgdaSpace{}%
\AgdaInductiveConstructor{ε}\AgdaSpace{}%
\AgdaBound{v}\AgdaSpace{}%
\AgdaSymbol{=}\AgdaSpace{}%
\AgdaInductiveConstructor{η}\AgdaSpace{}%
\AgdaBound{v}\<%
\end{code}
\newcommand{\PIFDinterp}{%
\begin{code}%
\>[0]\AgdaFunction{interp}\AgdaSpace{}%
\AgdaSymbol{:}\AgdaSpace{}%
\AgdaSymbol{\{}\AgdaBound{t₁}\AgdaSpace{}%
\AgdaBound{t₂}\AgdaSpace{}%
\AgdaSymbol{:}\AgdaSpace{}%
\AgdaDatatype{𝕌}\AgdaSymbol{\}}\AgdaSpace{}%
\AgdaSymbol{→}\AgdaSpace{}%
\AgdaSymbol{(}\AgdaBound{t₁}\AgdaSpace{}%
\AgdaOperator{\AgdaDatatype{⟷}}\AgdaSpace{}%
\AgdaBound{t₂}\AgdaSymbol{)}\AgdaSpace{}%
\AgdaSymbol{→}\AgdaSpace{}%
\AgdaOperator{\AgdaFunction{⟦}}\AgdaSpace{}%
\AgdaBound{t₁}\AgdaSpace{}%
\AgdaOperator{\AgdaFunction{⟧}}\AgdaSpace{}%
\AgdaSymbol{→}\AgdaSpace{}%
\AgdaDatatype{Maybe}\AgdaSpace{}%
\AgdaOperator{\AgdaFunction{⟦}}\AgdaSpace{}%
\AgdaBound{t₂}\AgdaSpace{}%
\AgdaOperator{\AgdaFunction{⟧}}\<%
\\
\>[0]\AgdaFunction{interp}\AgdaSpace{}%
\AgdaInductiveConstructor{swap⋆}\AgdaSpace{}%
\AgdaSymbol{(}\AgdaBound{v₁}\AgdaSpace{}%
\AgdaOperator{\AgdaInductiveConstructor{,}}\AgdaSpace{}%
\AgdaBound{v₂}\AgdaSymbol{)}\AgdaSpace{}%
\AgdaSymbol{=}\AgdaSpace{}%
\AgdaInductiveConstructor{just}\AgdaSpace{}%
\AgdaSymbol{(}\AgdaBound{v₂}\AgdaSpace{}%
\AgdaOperator{\AgdaInductiveConstructor{,}}\AgdaSpace{}%
\AgdaBound{v₁}\AgdaSymbol{)}\<%
\\
\>[0]\AgdaFunction{interp}\AgdaSpace{}%
\AgdaSymbol{(}\AgdaBound{c₁}\AgdaSpace{}%
\AgdaOperator{\AgdaInductiveConstructor{⊚}}\AgdaSpace{}%
\AgdaBound{c₂}\AgdaSymbol{)}\AgdaSpace{}%
\AgdaBound{v}\AgdaSpace{}%
\AgdaSymbol{=}\AgdaSpace{}%
\AgdaFunction{interp}\AgdaSpace{}%
\AgdaBound{c₁}\AgdaSpace{}%
\AgdaBound{v}\AgdaSpace{}%
\AgdaOperator{\AgdaFunction{≫=}}\AgdaSpace{}%
\AgdaFunction{interp}\AgdaSpace{}%
\AgdaBound{c₂}\<%
\\
\>[0]\AgdaComment{-- (elided)}\<%
\\
\>[0]\AgdaFunction{interp}\AgdaSpace{}%
\AgdaSymbol{(}\AgdaInductiveConstructor{η}\AgdaSpace{}%
\AgdaBound{v}\AgdaSymbol{)}\AgdaSpace{}%
\AgdaInductiveConstructor{tt}\AgdaSpace{}%
\AgdaSymbol{=}\AgdaSpace{}%
\AgdaInductiveConstructor{just}\AgdaSpace{}%
\AgdaSymbol{(}\AgdaBound{v}\AgdaSpace{}%
\AgdaOperator{\AgdaInductiveConstructor{,}}\AgdaSpace{}%
\AgdaInductiveConstructor{↻}\AgdaSymbol{)}\<%
\\
\>[0]\AgdaFunction{interp}\AgdaSpace{}%
\AgdaSymbol{(}\AgdaInductiveConstructor{ε}\AgdaSpace{}%
\AgdaBound{v}\AgdaSymbol{)}\AgdaSpace{}%
\AgdaSymbol{(}\AgdaBound{v'}\AgdaSpace{}%
\AgdaOperator{\AgdaInductiveConstructor{,}}\AgdaSpace{}%
\AgdaBound{○}\AgdaSymbol{)}\AgdaSpace{}%
\AgdaKeyword{with}\AgdaSpace{}%
\AgdaBound{v}\AgdaSpace{}%
\AgdaOperator{\AgdaFunction{≟ᵤ}}\AgdaSpace{}%
\AgdaBound{v'}\<%
\\
\>[0]\AgdaSymbol{...}\AgdaSpace{}%
\AgdaSymbol{|}\AgdaSpace{}%
\AgdaInductiveConstructor{yes}\AgdaSpace{}%
\AgdaSymbol{\AgdaUnderscore{}}\AgdaSpace{}%
\AgdaSymbol{=}\AgdaSpace{}%
\AgdaInductiveConstructor{just}\AgdaSpace{}%
\AgdaInductiveConstructor{tt}\<%
\\
\>[0]\AgdaSymbol{...}\AgdaSpace{}%
\AgdaSymbol{|}\AgdaSpace{}%
\AgdaInductiveConstructor{no}\AgdaSpace{}%
\AgdaSymbol{\AgdaUnderscore{}}\AgdaSpace{}%
\AgdaSymbol{=}\AgdaSpace{}%
\AgdaInductiveConstructor{nothing}\<%
\end{code}
\begin{code}[hide]%
\>[0]\AgdaFunction{interp}\AgdaSpace{}%
\AgdaInductiveConstructor{unite₊l}\AgdaSpace{}%
\AgdaSymbol{(}\AgdaInductiveConstructor{inj₁}\AgdaSpace{}%
\AgdaSymbol{())}\<%
\\
\>[0]\AgdaFunction{interp}\AgdaSpace{}%
\AgdaInductiveConstructor{unite₊l}\AgdaSpace{}%
\AgdaSymbol{(}\AgdaInductiveConstructor{inj₂}\AgdaSpace{}%
\AgdaBound{v}\AgdaSymbol{)}\AgdaSpace{}%
\AgdaSymbol{=}\AgdaSpace{}%
\AgdaInductiveConstructor{just}\AgdaSpace{}%
\AgdaBound{v}\<%
\\
\>[0]\AgdaFunction{interp}\AgdaSpace{}%
\AgdaInductiveConstructor{uniti₊l}\AgdaSpace{}%
\AgdaBound{v}\AgdaSpace{}%
\AgdaSymbol{=}\AgdaSpace{}%
\AgdaInductiveConstructor{just}\AgdaSpace{}%
\AgdaSymbol{(}\AgdaInductiveConstructor{inj₂}\AgdaSpace{}%
\AgdaBound{v}\AgdaSymbol{)}\<%
\\
\>[0]\AgdaFunction{interp}\AgdaSpace{}%
\AgdaInductiveConstructor{unite₊r}\AgdaSpace{}%
\AgdaSymbol{(}\AgdaInductiveConstructor{inj₁}\AgdaSpace{}%
\AgdaBound{v}\AgdaSymbol{)}\AgdaSpace{}%
\AgdaSymbol{=}\AgdaSpace{}%
\AgdaInductiveConstructor{just}\AgdaSpace{}%
\AgdaBound{v}\<%
\\
\>[0]\AgdaFunction{interp}\AgdaSpace{}%
\AgdaInductiveConstructor{unite₊r}\AgdaSpace{}%
\AgdaSymbol{(}\AgdaInductiveConstructor{inj₂}\AgdaSpace{}%
\AgdaSymbol{())}\<%
\\
\>[0]\AgdaFunction{interp}\AgdaSpace{}%
\AgdaInductiveConstructor{uniti₊r}\AgdaSpace{}%
\AgdaBound{v}\AgdaSpace{}%
\AgdaSymbol{=}\AgdaSpace{}%
\AgdaInductiveConstructor{just}\AgdaSpace{}%
\AgdaSymbol{(}\AgdaInductiveConstructor{inj₁}\AgdaSpace{}%
\AgdaBound{v}\AgdaSymbol{)}\<%
\\
\>[0]\AgdaFunction{interp}\AgdaSpace{}%
\AgdaInductiveConstructor{swap₊}\AgdaSpace{}%
\AgdaSymbol{(}\AgdaInductiveConstructor{inj₁}\AgdaSpace{}%
\AgdaBound{v}\AgdaSymbol{)}\AgdaSpace{}%
\AgdaSymbol{=}\AgdaSpace{}%
\AgdaInductiveConstructor{just}\AgdaSpace{}%
\AgdaSymbol{(}\AgdaInductiveConstructor{inj₂}\AgdaSpace{}%
\AgdaBound{v}\AgdaSymbol{)}\<%
\\
\>[0]\AgdaFunction{interp}\AgdaSpace{}%
\AgdaInductiveConstructor{swap₊}\AgdaSpace{}%
\AgdaSymbol{(}\AgdaInductiveConstructor{inj₂}\AgdaSpace{}%
\AgdaBound{v}\AgdaSymbol{)}\AgdaSpace{}%
\AgdaSymbol{=}\AgdaSpace{}%
\AgdaInductiveConstructor{just}\AgdaSpace{}%
\AgdaSymbol{(}\AgdaInductiveConstructor{inj₁}\AgdaSpace{}%
\AgdaBound{v}\AgdaSymbol{)}\<%
\\
\>[0]\AgdaFunction{interp}\AgdaSpace{}%
\AgdaInductiveConstructor{assocl₊}\AgdaSpace{}%
\AgdaSymbol{(}\AgdaInductiveConstructor{inj₁}\AgdaSpace{}%
\AgdaBound{v}\AgdaSymbol{)}\AgdaSpace{}%
\AgdaSymbol{=}\AgdaSpace{}%
\AgdaInductiveConstructor{just}\AgdaSpace{}%
\AgdaSymbol{(}\AgdaInductiveConstructor{inj₁}\AgdaSpace{}%
\AgdaSymbol{(}\AgdaInductiveConstructor{inj₁}\AgdaSpace{}%
\AgdaBound{v}\AgdaSymbol{))}\<%
\\
\>[0]\AgdaFunction{interp}\AgdaSpace{}%
\AgdaInductiveConstructor{assocl₊}\AgdaSpace{}%
\AgdaSymbol{(}\AgdaInductiveConstructor{inj₂}\AgdaSpace{}%
\AgdaSymbol{(}\AgdaInductiveConstructor{inj₁}\AgdaSpace{}%
\AgdaBound{v}\AgdaSymbol{))}\AgdaSpace{}%
\AgdaSymbol{=}\AgdaSpace{}%
\AgdaInductiveConstructor{just}\AgdaSpace{}%
\AgdaSymbol{(}\AgdaInductiveConstructor{inj₁}\AgdaSpace{}%
\AgdaSymbol{(}\AgdaInductiveConstructor{inj₂}\AgdaSpace{}%
\AgdaBound{v}\AgdaSymbol{))}\<%
\\
\>[0]\AgdaFunction{interp}\AgdaSpace{}%
\AgdaInductiveConstructor{assocl₊}\AgdaSpace{}%
\AgdaSymbol{(}\AgdaInductiveConstructor{inj₂}\AgdaSpace{}%
\AgdaSymbol{(}\AgdaInductiveConstructor{inj₂}\AgdaSpace{}%
\AgdaBound{v}\AgdaSymbol{))}\AgdaSpace{}%
\AgdaSymbol{=}\AgdaSpace{}%
\AgdaInductiveConstructor{just}\AgdaSpace{}%
\AgdaSymbol{(}\AgdaInductiveConstructor{inj₂}\AgdaSpace{}%
\AgdaBound{v}\AgdaSymbol{)}\<%
\\
\>[0]\AgdaFunction{interp}\AgdaSpace{}%
\AgdaInductiveConstructor{assocr₊}\AgdaSpace{}%
\AgdaSymbol{(}\AgdaInductiveConstructor{inj₁}\AgdaSpace{}%
\AgdaSymbol{(}\AgdaInductiveConstructor{inj₁}\AgdaSpace{}%
\AgdaBound{v}\AgdaSymbol{))}\AgdaSpace{}%
\AgdaSymbol{=}\AgdaSpace{}%
\AgdaInductiveConstructor{just}\AgdaSpace{}%
\AgdaSymbol{(}\AgdaInductiveConstructor{inj₁}\AgdaSpace{}%
\AgdaBound{v}\AgdaSymbol{)}\<%
\\
\>[0]\AgdaFunction{interp}\AgdaSpace{}%
\AgdaInductiveConstructor{assocr₊}\AgdaSpace{}%
\AgdaSymbol{(}\AgdaInductiveConstructor{inj₁}\AgdaSpace{}%
\AgdaSymbol{(}\AgdaInductiveConstructor{inj₂}\AgdaSpace{}%
\AgdaBound{v}\AgdaSymbol{))}\AgdaSpace{}%
\AgdaSymbol{=}\AgdaSpace{}%
\AgdaInductiveConstructor{just}\AgdaSpace{}%
\AgdaSymbol{(}\AgdaInductiveConstructor{inj₂}\AgdaSpace{}%
\AgdaSymbol{(}\AgdaInductiveConstructor{inj₁}\AgdaSpace{}%
\AgdaBound{v}\AgdaSymbol{))}\<%
\\
\>[0]\AgdaFunction{interp}\AgdaSpace{}%
\AgdaInductiveConstructor{assocr₊}\AgdaSpace{}%
\AgdaSymbol{(}\AgdaInductiveConstructor{inj₂}\AgdaSpace{}%
\AgdaBound{v}\AgdaSymbol{)}\AgdaSpace{}%
\AgdaSymbol{=}\AgdaSpace{}%
\AgdaInductiveConstructor{just}\AgdaSpace{}%
\AgdaSymbol{(}\AgdaInductiveConstructor{inj₂}\AgdaSpace{}%
\AgdaSymbol{(}\AgdaInductiveConstructor{inj₂}\AgdaSpace{}%
\AgdaBound{v}\AgdaSymbol{))}\<%
\\
\>[0]\AgdaFunction{interp}\AgdaSpace{}%
\AgdaInductiveConstructor{unite⋆l}\AgdaSpace{}%
\AgdaBound{v}\AgdaSpace{}%
\AgdaSymbol{=}\AgdaSpace{}%
\AgdaInductiveConstructor{just}\AgdaSpace{}%
\AgdaSymbol{(}\AgdaField{proj₂}\AgdaSpace{}%
\AgdaBound{v}\AgdaSymbol{)}\<%
\\
\>[0]\AgdaFunction{interp}\AgdaSpace{}%
\AgdaInductiveConstructor{uniti⋆l}\AgdaSpace{}%
\AgdaBound{v}\AgdaSpace{}%
\AgdaSymbol{=}\AgdaSpace{}%
\AgdaInductiveConstructor{just}\AgdaSpace{}%
\AgdaSymbol{(}\AgdaInductiveConstructor{tt}\AgdaSpace{}%
\AgdaOperator{\AgdaInductiveConstructor{,}}\AgdaSpace{}%
\AgdaBound{v}\AgdaSymbol{)}\<%
\\
\>[0]\AgdaFunction{interp}\AgdaSpace{}%
\AgdaInductiveConstructor{unite⋆r}\AgdaSpace{}%
\AgdaBound{v}\AgdaSpace{}%
\AgdaSymbol{=}\AgdaSpace{}%
\AgdaInductiveConstructor{just}\AgdaSpace{}%
\AgdaSymbol{(}\AgdaField{proj₁}\AgdaSpace{}%
\AgdaBound{v}\AgdaSymbol{)}\<%
\\
\>[0]\AgdaFunction{interp}\AgdaSpace{}%
\AgdaInductiveConstructor{uniti⋆r}\AgdaSpace{}%
\AgdaBound{v}\AgdaSpace{}%
\AgdaSymbol{=}\AgdaSpace{}%
\AgdaInductiveConstructor{just}\AgdaSpace{}%
\AgdaSymbol{(}\AgdaBound{v}\AgdaSpace{}%
\AgdaOperator{\AgdaInductiveConstructor{,}}\AgdaSpace{}%
\AgdaInductiveConstructor{tt}\AgdaSymbol{)}\<%
\\
\>[0]\AgdaFunction{interp}\AgdaSpace{}%
\AgdaInductiveConstructor{assocl⋆}\AgdaSpace{}%
\AgdaSymbol{(}\AgdaBound{v₁}\AgdaSpace{}%
\AgdaOperator{\AgdaInductiveConstructor{,}}\AgdaSpace{}%
\AgdaBound{v₂}\AgdaSpace{}%
\AgdaOperator{\AgdaInductiveConstructor{,}}\AgdaSpace{}%
\AgdaBound{v₃}\AgdaSymbol{)}\AgdaSpace{}%
\AgdaSymbol{=}\AgdaSpace{}%
\AgdaInductiveConstructor{just}\AgdaSpace{}%
\AgdaSymbol{((}\AgdaBound{v₁}\AgdaSpace{}%
\AgdaOperator{\AgdaInductiveConstructor{,}}\AgdaSpace{}%
\AgdaBound{v₂}\AgdaSymbol{)}\AgdaSpace{}%
\AgdaOperator{\AgdaInductiveConstructor{,}}\AgdaSpace{}%
\AgdaBound{v₃}\AgdaSymbol{)}\<%
\\
\>[0]\AgdaFunction{interp}\AgdaSpace{}%
\AgdaInductiveConstructor{assocr⋆}\AgdaSpace{}%
\AgdaSymbol{((}\AgdaBound{v₁}\AgdaSpace{}%
\AgdaOperator{\AgdaInductiveConstructor{,}}\AgdaSpace{}%
\AgdaBound{v₂}\AgdaSymbol{)}\AgdaSpace{}%
\AgdaOperator{\AgdaInductiveConstructor{,}}\AgdaSpace{}%
\AgdaBound{v₃}\AgdaSymbol{)}\AgdaSpace{}%
\AgdaSymbol{=}\AgdaSpace{}%
\AgdaInductiveConstructor{just}\AgdaSpace{}%
\AgdaSymbol{(}\AgdaBound{v₁}\AgdaSpace{}%
\AgdaOperator{\AgdaInductiveConstructor{,}}\AgdaSpace{}%
\AgdaBound{v₂}\AgdaSpace{}%
\AgdaOperator{\AgdaInductiveConstructor{,}}\AgdaSpace{}%
\AgdaBound{v₃}\AgdaSymbol{)}\<%
\\
\>[0]\AgdaFunction{interp}\AgdaSpace{}%
\AgdaInductiveConstructor{absorbr}\AgdaSpace{}%
\AgdaSymbol{(()}\AgdaSpace{}%
\AgdaOperator{\AgdaInductiveConstructor{,}}\AgdaSpace{}%
\AgdaBound{v}\AgdaSymbol{)}\<%
\\
\>[0]\AgdaFunction{interp}\AgdaSpace{}%
\AgdaInductiveConstructor{absorbl}\AgdaSpace{}%
\AgdaSymbol{(}\AgdaBound{v}\AgdaSpace{}%
\AgdaOperator{\AgdaInductiveConstructor{,}}\AgdaSpace{}%
\AgdaSymbol{())}\<%
\\
\>[0]\AgdaFunction{interp}\AgdaSpace{}%
\AgdaInductiveConstructor{factorzr}\AgdaSpace{}%
\AgdaSymbol{()}\<%
\\
\>[0]\AgdaFunction{interp}\AgdaSpace{}%
\AgdaInductiveConstructor{factorzl}\AgdaSpace{}%
\AgdaSymbol{()}\<%
\\
\>[0]\AgdaFunction{interp}\AgdaSpace{}%
\AgdaInductiveConstructor{dist}\AgdaSpace{}%
\AgdaSymbol{(}\AgdaInductiveConstructor{inj₁}\AgdaSpace{}%
\AgdaBound{v₁}\AgdaSpace{}%
\AgdaOperator{\AgdaInductiveConstructor{,}}\AgdaSpace{}%
\AgdaBound{v₃}\AgdaSymbol{)}\AgdaSpace{}%
\AgdaSymbol{=}\AgdaSpace{}%
\AgdaInductiveConstructor{just}\AgdaSpace{}%
\AgdaSymbol{(}\AgdaInductiveConstructor{inj₁}\AgdaSpace{}%
\AgdaSymbol{(}\AgdaBound{v₁}\AgdaSpace{}%
\AgdaOperator{\AgdaInductiveConstructor{,}}\AgdaSpace{}%
\AgdaBound{v₃}\AgdaSymbol{))}\<%
\\
\>[0]\AgdaFunction{interp}\AgdaSpace{}%
\AgdaInductiveConstructor{dist}\AgdaSpace{}%
\AgdaSymbol{(}\AgdaInductiveConstructor{inj₂}\AgdaSpace{}%
\AgdaBound{v₂}\AgdaSpace{}%
\AgdaOperator{\AgdaInductiveConstructor{,}}\AgdaSpace{}%
\AgdaBound{v₃}\AgdaSymbol{)}\AgdaSpace{}%
\AgdaSymbol{=}\AgdaSpace{}%
\AgdaInductiveConstructor{just}\AgdaSpace{}%
\AgdaSymbol{(}\AgdaInductiveConstructor{inj₂}\AgdaSpace{}%
\AgdaSymbol{(}\AgdaBound{v₂}\AgdaSpace{}%
\AgdaOperator{\AgdaInductiveConstructor{,}}\AgdaSpace{}%
\AgdaBound{v₃}\AgdaSymbol{))}\<%
\\
\>[0]\AgdaFunction{interp}\AgdaSpace{}%
\AgdaInductiveConstructor{factor}\AgdaSpace{}%
\AgdaSymbol{(}\AgdaInductiveConstructor{inj₁}\AgdaSpace{}%
\AgdaSymbol{(}\AgdaBound{v₁}\AgdaSpace{}%
\AgdaOperator{\AgdaInductiveConstructor{,}}\AgdaSpace{}%
\AgdaBound{v₃}\AgdaSymbol{))}\AgdaSpace{}%
\AgdaSymbol{=}\AgdaSpace{}%
\AgdaInductiveConstructor{just}\AgdaSpace{}%
\AgdaSymbol{(}\AgdaInductiveConstructor{inj₁}\AgdaSpace{}%
\AgdaBound{v₁}\AgdaSpace{}%
\AgdaOperator{\AgdaInductiveConstructor{,}}\AgdaSpace{}%
\AgdaBound{v₃}\AgdaSymbol{)}\<%
\\
\>[0]\AgdaFunction{interp}\AgdaSpace{}%
\AgdaInductiveConstructor{factor}\AgdaSpace{}%
\AgdaSymbol{(}\AgdaInductiveConstructor{inj₂}\AgdaSpace{}%
\AgdaSymbol{(}\AgdaBound{v₂}\AgdaSpace{}%
\AgdaOperator{\AgdaInductiveConstructor{,}}\AgdaSpace{}%
\AgdaBound{v₃}\AgdaSymbol{))}\AgdaSpace{}%
\AgdaSymbol{=}\AgdaSpace{}%
\AgdaInductiveConstructor{just}\AgdaSpace{}%
\AgdaSymbol{(}\AgdaInductiveConstructor{inj₂}\AgdaSpace{}%
\AgdaBound{v₂}\AgdaSpace{}%
\AgdaOperator{\AgdaInductiveConstructor{,}}\AgdaSpace{}%
\AgdaBound{v₃}\AgdaSymbol{)}\<%
\\
\>[0]\AgdaFunction{interp}\AgdaSpace{}%
\AgdaInductiveConstructor{distl}\AgdaSpace{}%
\AgdaSymbol{(}\AgdaBound{v₁}\AgdaSpace{}%
\AgdaOperator{\AgdaInductiveConstructor{,}}\AgdaSpace{}%
\AgdaInductiveConstructor{inj₁}\AgdaSpace{}%
\AgdaBound{v₂}\AgdaSymbol{)}\AgdaSpace{}%
\AgdaSymbol{=}\AgdaSpace{}%
\AgdaInductiveConstructor{just}\AgdaSpace{}%
\AgdaSymbol{(}\AgdaInductiveConstructor{inj₁}\AgdaSpace{}%
\AgdaSymbol{(}\AgdaBound{v₁}\AgdaSpace{}%
\AgdaOperator{\AgdaInductiveConstructor{,}}\AgdaSpace{}%
\AgdaBound{v₂}\AgdaSymbol{))}\<%
\\
\>[0]\AgdaFunction{interp}\AgdaSpace{}%
\AgdaInductiveConstructor{distl}\AgdaSpace{}%
\AgdaSymbol{(}\AgdaBound{v₁}\AgdaSpace{}%
\AgdaOperator{\AgdaInductiveConstructor{,}}\AgdaSpace{}%
\AgdaInductiveConstructor{inj₂}\AgdaSpace{}%
\AgdaBound{v₃}\AgdaSymbol{)}\AgdaSpace{}%
\AgdaSymbol{=}\AgdaSpace{}%
\AgdaInductiveConstructor{just}\AgdaSpace{}%
\AgdaSymbol{(}\AgdaInductiveConstructor{inj₂}\AgdaSpace{}%
\AgdaSymbol{(}\AgdaBound{v₁}\AgdaSpace{}%
\AgdaOperator{\AgdaInductiveConstructor{,}}\AgdaSpace{}%
\AgdaBound{v₃}\AgdaSymbol{))}\<%
\\
\>[0]\AgdaFunction{interp}\AgdaSpace{}%
\AgdaInductiveConstructor{factorl}\AgdaSpace{}%
\AgdaSymbol{(}\AgdaInductiveConstructor{inj₁}\AgdaSpace{}%
\AgdaSymbol{(}\AgdaBound{v₁}\AgdaSpace{}%
\AgdaOperator{\AgdaInductiveConstructor{,}}\AgdaSpace{}%
\AgdaBound{v₂}\AgdaSymbol{))}\AgdaSpace{}%
\AgdaSymbol{=}\AgdaSpace{}%
\AgdaInductiveConstructor{just}\AgdaSpace{}%
\AgdaSymbol{(}\AgdaBound{v₁}\AgdaSpace{}%
\AgdaOperator{\AgdaInductiveConstructor{,}}\AgdaSpace{}%
\AgdaInductiveConstructor{inj₁}\AgdaSpace{}%
\AgdaBound{v₂}\AgdaSymbol{)}\<%
\\
\>[0]\AgdaFunction{interp}\AgdaSpace{}%
\AgdaInductiveConstructor{factorl}\AgdaSpace{}%
\AgdaSymbol{(}\AgdaInductiveConstructor{inj₂}\AgdaSpace{}%
\AgdaSymbol{(}\AgdaBound{v₁}\AgdaSpace{}%
\AgdaOperator{\AgdaInductiveConstructor{,}}\AgdaSpace{}%
\AgdaBound{v₃}\AgdaSymbol{))}\AgdaSpace{}%
\AgdaSymbol{=}\AgdaSpace{}%
\AgdaInductiveConstructor{just}\AgdaSpace{}%
\AgdaSymbol{(}\AgdaBound{v₁}\AgdaSpace{}%
\AgdaOperator{\AgdaInductiveConstructor{,}}\AgdaSpace{}%
\AgdaInductiveConstructor{inj₂}\AgdaSpace{}%
\AgdaBound{v₃}\AgdaSymbol{)}\<%
\\
\>[0]\AgdaFunction{interp}\AgdaSpace{}%
\AgdaInductiveConstructor{id⟷}\AgdaSpace{}%
\AgdaBound{v}\AgdaSpace{}%
\AgdaSymbol{=}\AgdaSpace{}%
\AgdaInductiveConstructor{just}\AgdaSpace{}%
\AgdaBound{v}\<%
\\
\>[0]\AgdaFunction{interp}\AgdaSpace{}%
\AgdaSymbol{(}\AgdaBound{c₁}\AgdaSpace{}%
\AgdaOperator{\AgdaInductiveConstructor{⊕}}\AgdaSpace{}%
\AgdaBound{c₂}\AgdaSymbol{)}\AgdaSpace{}%
\AgdaSymbol{(}\AgdaInductiveConstructor{inj₁}\AgdaSpace{}%
\AgdaBound{v}\AgdaSymbol{)}\AgdaSpace{}%
\AgdaSymbol{=}\AgdaSpace{}%
\AgdaFunction{interp}\AgdaSpace{}%
\AgdaBound{c₁}\AgdaSpace{}%
\AgdaBound{v}\AgdaSpace{}%
\AgdaOperator{\AgdaFunction{≫=}}\AgdaSpace{}%
\AgdaInductiveConstructor{just}\AgdaSpace{}%
\AgdaOperator{\AgdaFunction{∘}}\AgdaSpace{}%
\AgdaInductiveConstructor{inj₁}\<%
\\
\>[0]\AgdaFunction{interp}\AgdaSpace{}%
\AgdaSymbol{(}\AgdaBound{c₁}\AgdaSpace{}%
\AgdaOperator{\AgdaInductiveConstructor{⊕}}\AgdaSpace{}%
\AgdaBound{c₂}\AgdaSymbol{)}\AgdaSpace{}%
\AgdaSymbol{(}\AgdaInductiveConstructor{inj₂}\AgdaSpace{}%
\AgdaBound{v}\AgdaSymbol{)}\AgdaSpace{}%
\AgdaSymbol{=}\AgdaSpace{}%
\AgdaFunction{interp}\AgdaSpace{}%
\AgdaBound{c₂}\AgdaSpace{}%
\AgdaBound{v}\AgdaSpace{}%
\AgdaOperator{\AgdaFunction{≫=}}\AgdaSpace{}%
\AgdaInductiveConstructor{just}\AgdaSpace{}%
\AgdaOperator{\AgdaFunction{∘}}\AgdaSpace{}%
\AgdaInductiveConstructor{inj₂}\<%
\\
\>[0]\AgdaFunction{interp}\AgdaSpace{}%
\AgdaSymbol{(}\AgdaBound{c₁}\AgdaSpace{}%
\AgdaOperator{\AgdaInductiveConstructor{⊗}}\AgdaSpace{}%
\AgdaBound{c₂}\AgdaSymbol{)}\AgdaSpace{}%
\AgdaSymbol{(}\AgdaBound{v₁}\AgdaSpace{}%
\AgdaOperator{\AgdaInductiveConstructor{,}}\AgdaSpace{}%
\AgdaBound{v₂}\AgdaSymbol{)}\AgdaSpace{}%
\AgdaSymbol{=}%
\>[1154I]\AgdaFunction{interp}\AgdaSpace{}%
\AgdaBound{c₁}\AgdaSpace{}%
\AgdaBound{v₁}\AgdaSpace{}%
\AgdaOperator{\AgdaFunction{≫=}}\<%
\\
\>[1154I][@{}l@{\AgdaIndent{0}}]%
\>[31]\AgdaSymbol{(λ}\AgdaSpace{}%
\AgdaBound{v₁'}\AgdaSpace{}%
\AgdaSymbol{→}\AgdaSpace{}%
\AgdaFunction{interp}\AgdaSpace{}%
\AgdaBound{c₂}\AgdaSpace{}%
\AgdaBound{v₂}\AgdaSpace{}%
\AgdaOperator{\AgdaFunction{≫=}}\<%
\\
\>[31][@{}l@{\AgdaIndent{0}}]%
\>[33]\AgdaSymbol{λ}\AgdaSpace{}%
\AgdaBound{v₂'}\AgdaSpace{}%
\AgdaSymbol{→}\AgdaSpace{}%
\AgdaInductiveConstructor{just}\AgdaSpace{}%
\AgdaSymbol{(}\AgdaBound{v₁'}\AgdaSpace{}%
\AgdaOperator{\AgdaInductiveConstructor{,}}\AgdaSpace{}%
\AgdaBound{v₂'}\AgdaSymbol{))}\<%
\end{code}}
\begin{code}[hide]%
\>[0]\AgdaFunction{𝔹}\AgdaSpace{}%
\AgdaSymbol{:}\AgdaSpace{}%
\AgdaDatatype{𝕌}\<%
\\
\>[0]\AgdaFunction{𝔹}\AgdaSpace{}%
\AgdaSymbol{=}\AgdaSpace{}%
\AgdaInductiveConstructor{𝟙}\AgdaSpace{}%
\AgdaOperator{\AgdaInductiveConstructor{+ᵤ}}\AgdaSpace{}%
\AgdaInductiveConstructor{𝟙}\<%
\\
\\[\AgdaEmptyExtraSkip]%
\>[0]\AgdaFunction{𝔹²}\AgdaSpace{}%
\AgdaFunction{𝔹³}\AgdaSpace{}%
\AgdaFunction{𝔹⁴}\AgdaSpace{}%
\AgdaSymbol{:}\AgdaSpace{}%
\AgdaDatatype{𝕌}\<%
\\
\>[0]\AgdaFunction{𝔹²}%
\>[4]\AgdaSymbol{=}\AgdaSpace{}%
\AgdaFunction{𝔹}\AgdaSpace{}%
\AgdaOperator{\AgdaInductiveConstructor{×ᵤ}}\AgdaSpace{}%
\AgdaFunction{𝔹}\<%
\\
\>[0]\AgdaFunction{𝔹³}%
\>[4]\AgdaSymbol{=}\AgdaSpace{}%
\AgdaFunction{𝔹}\AgdaSpace{}%
\AgdaOperator{\AgdaInductiveConstructor{×ᵤ}}\AgdaSpace{}%
\AgdaFunction{𝔹²}\<%
\\
\>[0]\AgdaFunction{𝔹⁴}%
\>[4]\AgdaSymbol{=}\AgdaSpace{}%
\AgdaFunction{𝔹}\AgdaSpace{}%
\AgdaOperator{\AgdaInductiveConstructor{×ᵤ}}\AgdaSpace{}%
\AgdaFunction{𝔹³}\<%
\\
\\[\AgdaEmptyExtraSkip]%
\>[0]\AgdaKeyword{pattern}\AgdaSpace{}%
\AgdaInductiveConstructor{𝔽}\AgdaSpace{}%
\AgdaSymbol{=}\AgdaSpace{}%
\AgdaInductiveConstructor{inj₁}\AgdaSpace{}%
\AgdaInductiveConstructor{tt}\<%
\\
\>[0]\AgdaKeyword{pattern}\AgdaSpace{}%
\AgdaInductiveConstructor{𝕋}\AgdaSpace{}%
\AgdaSymbol{=}\AgdaSpace{}%
\AgdaInductiveConstructor{inj₂}\AgdaSpace{}%
\AgdaInductiveConstructor{tt}\<%
\\
\\[\AgdaEmptyExtraSkip]%
\>[0]\AgdaKeyword{infixr}\AgdaSpace{}%
\AgdaNumber{2}%
\>[10]\AgdaOperator{\AgdaFunction{\AgdaUnderscore{}⟷⟨\AgdaUnderscore{}⟩\AgdaUnderscore{}}}\<%
\\
\>[0]\AgdaKeyword{infix}%
\>[7]\AgdaNumber{3}%
\>[10]\AgdaOperator{\AgdaFunction{\AgdaUnderscore{}□}}\<%
\\
\\[\AgdaEmptyExtraSkip]%
\>[0]\AgdaOperator{\AgdaFunction{\AgdaUnderscore{}⟷⟨\AgdaUnderscore{}⟩\AgdaUnderscore{}}}\AgdaSpace{}%
\AgdaSymbol{:}%
\>[1199I]\AgdaSymbol{(}\AgdaBound{t₁}\AgdaSpace{}%
\AgdaSymbol{:}\AgdaSpace{}%
\AgdaDatatype{𝕌}\AgdaSymbol{)}\AgdaSpace{}%
\AgdaSymbol{\{}\AgdaBound{t₂}\AgdaSpace{}%
\AgdaSymbol{:}\AgdaSpace{}%
\AgdaDatatype{𝕌}\AgdaSymbol{\}}\AgdaSpace{}%
\AgdaSymbol{\{}\AgdaBound{t₃}\AgdaSpace{}%
\AgdaSymbol{:}\AgdaSpace{}%
\AgdaDatatype{𝕌}\AgdaSymbol{\}}\AgdaSpace{}%
\AgdaSymbol{→}\<%
\\
\>[1199I][@{}l@{\AgdaIndent{0}}]%
\>[10]\AgdaSymbol{(}\AgdaBound{t₁}\AgdaSpace{}%
\AgdaOperator{\AgdaDatatype{⟷}}\AgdaSpace{}%
\AgdaBound{t₂}\AgdaSymbol{)}\AgdaSpace{}%
\AgdaSymbol{→}\AgdaSpace{}%
\AgdaSymbol{(}\AgdaBound{t₂}\AgdaSpace{}%
\AgdaOperator{\AgdaDatatype{⟷}}\AgdaSpace{}%
\AgdaBound{t₃}\AgdaSymbol{)}\AgdaSpace{}%
\AgdaSymbol{→}\AgdaSpace{}%
\AgdaSymbol{(}\AgdaBound{t₁}\AgdaSpace{}%
\AgdaOperator{\AgdaDatatype{⟷}}\AgdaSpace{}%
\AgdaBound{t₃}\AgdaSymbol{)}\<%
\\
\>[0]\AgdaSymbol{\AgdaUnderscore{}}\AgdaSpace{}%
\AgdaOperator{\AgdaFunction{⟷⟨}}\AgdaSpace{}%
\AgdaBound{α}\AgdaSpace{}%
\AgdaOperator{\AgdaFunction{⟩}}\AgdaSpace{}%
\AgdaBound{β}\AgdaSpace{}%
\AgdaSymbol{=}\AgdaSpace{}%
\AgdaBound{α}\AgdaSpace{}%
\AgdaOperator{\AgdaInductiveConstructor{⊚}}\AgdaSpace{}%
\AgdaBound{β}\<%
\\
\\[\AgdaEmptyExtraSkip]%
\>[0]\AgdaOperator{\AgdaFunction{\AgdaUnderscore{}□}}\AgdaSpace{}%
\AgdaSymbol{:}\AgdaSpace{}%
\AgdaSymbol{(}\AgdaBound{t}\AgdaSpace{}%
\AgdaSymbol{:}\AgdaSpace{}%
\AgdaDatatype{𝕌}\AgdaSymbol{)}\AgdaSpace{}%
\AgdaSymbol{→}\AgdaSpace{}%
\AgdaSymbol{\{}\AgdaBound{t}\AgdaSpace{}%
\AgdaSymbol{:}\AgdaSpace{}%
\AgdaDatatype{𝕌}\AgdaSymbol{\}}\AgdaSpace{}%
\AgdaSymbol{→}\AgdaSpace{}%
\AgdaSymbol{(}\AgdaBound{t}\AgdaSpace{}%
\AgdaOperator{\AgdaDatatype{⟷}}\AgdaSpace{}%
\AgdaBound{t}\AgdaSymbol{)}\<%
\\
\>[0]\AgdaOperator{\AgdaFunction{\AgdaUnderscore{}□}}\AgdaSpace{}%
\AgdaBound{t}\AgdaSpace{}%
\AgdaSymbol{=}\AgdaSpace{}%
\AgdaInductiveConstructor{id⟷}\<%
\\
\\[\AgdaEmptyExtraSkip]%
\>[0]\AgdaFunction{CNOT}\AgdaSpace{}%
\AgdaFunction{CNOT'}\AgdaSpace{}%
\AgdaSymbol{:}\AgdaSpace{}%
\AgdaFunction{𝔹}\AgdaSpace{}%
\AgdaOperator{\AgdaInductiveConstructor{×ᵤ}}\AgdaSpace{}%
\AgdaFunction{𝔹}\AgdaSpace{}%
\AgdaOperator{\AgdaDatatype{⟷}}\AgdaSpace{}%
\AgdaFunction{𝔹}\AgdaSpace{}%
\AgdaOperator{\AgdaInductiveConstructor{×ᵤ}}\AgdaSpace{}%
\AgdaFunction{𝔹}\<%
\\
\>[0]\AgdaFunction{CNOT}\AgdaSpace{}%
\AgdaSymbol{=}\AgdaSpace{}%
\AgdaInductiveConstructor{dist}\AgdaSpace{}%
\AgdaOperator{\AgdaInductiveConstructor{⊚}}\AgdaSpace{}%
\AgdaSymbol{(}\AgdaInductiveConstructor{id⟷}\AgdaSpace{}%
\AgdaOperator{\AgdaInductiveConstructor{⊕}}\AgdaSpace{}%
\AgdaSymbol{(}\AgdaInductiveConstructor{id⟷}\AgdaSpace{}%
\AgdaOperator{\AgdaInductiveConstructor{⊗}}\AgdaSpace{}%
\AgdaInductiveConstructor{swap₊}\AgdaSymbol{))}\AgdaSpace{}%
\AgdaOperator{\AgdaInductiveConstructor{⊚}}\AgdaSpace{}%
\AgdaInductiveConstructor{factor}\<%
\\
\>[0]\AgdaFunction{CNOT'}\AgdaSpace{}%
\AgdaSymbol{=}\AgdaSpace{}%
\AgdaInductiveConstructor{distl}\AgdaSpace{}%
\AgdaOperator{\AgdaInductiveConstructor{⊚}}\AgdaSpace{}%
\AgdaSymbol{(}\AgdaInductiveConstructor{id⟷}\AgdaSpace{}%
\AgdaOperator{\AgdaInductiveConstructor{⊕}}\AgdaSpace{}%
\AgdaSymbol{(}\AgdaInductiveConstructor{swap₊}\AgdaSpace{}%
\AgdaOperator{\AgdaInductiveConstructor{⊗}}\AgdaSpace{}%
\AgdaInductiveConstructor{id⟷}\AgdaSymbol{))}\AgdaSpace{}%
\AgdaOperator{\AgdaInductiveConstructor{⊚}}\AgdaSpace{}%
\AgdaInductiveConstructor{factorl}\<%
\end{code}
\newcommand{\PIFDExample}{%
\begin{code}%
\>[0]\AgdaFunction{Ex}\AgdaSpace{}%
\AgdaSymbol{:}\AgdaSpace{}%
\AgdaFunction{𝔹}\AgdaSpace{}%
\AgdaOperator{\AgdaDatatype{⟷}}\AgdaSpace{}%
\AgdaFunction{𝔹}\<%
\\
\>[0]\AgdaFunction{Ex}\AgdaSpace{}%
\AgdaSymbol{=}%
\>[1276I]\AgdaInductiveConstructor{uniti⋆r}\AgdaSpace{}%
\AgdaOperator{\AgdaInductiveConstructor{⊚}}\AgdaSpace{}%
\AgdaSymbol{(}\AgdaInductiveConstructor{id⟷}\AgdaSpace{}%
\AgdaOperator{\AgdaInductiveConstructor{⊗}}\AgdaSpace{}%
\AgdaInductiveConstructor{η}\AgdaSpace{}%
\AgdaInductiveConstructor{𝔽}\AgdaSymbol{)}\AgdaSpace{}%
\AgdaOperator{\AgdaInductiveConstructor{⊚}}\<%
\\
\>[.][@{}l@{}]\<[1276I]%
\>[5]\AgdaInductiveConstructor{assocl⋆}\AgdaSpace{}%
\AgdaOperator{\AgdaInductiveConstructor{⊚}}\AgdaSpace{}%
\AgdaSymbol{(}\AgdaFunction{CNOT}\AgdaSpace{}%
\AgdaOperator{\AgdaInductiveConstructor{⊗}}\AgdaSpace{}%
\AgdaInductiveConstructor{id⟷}\AgdaSymbol{)}\AgdaSpace{}%
\AgdaOperator{\AgdaInductiveConstructor{⊚}}\AgdaSpace{}%
\AgdaInductiveConstructor{assocr⋆}\AgdaSpace{}%
\AgdaOperator{\AgdaInductiveConstructor{⊚}}\<%
\\
\>[5]\AgdaSymbol{(}\AgdaInductiveConstructor{id⟷}\AgdaSpace{}%
\AgdaOperator{\AgdaInductiveConstructor{⊗}}\AgdaSpace{}%
\AgdaInductiveConstructor{ε}\AgdaSpace{}%
\AgdaInductiveConstructor{𝔽}\AgdaSymbol{)}\AgdaSpace{}%
\AgdaOperator{\AgdaInductiveConstructor{⊚}}\AgdaSpace{}%
\AgdaInductiveConstructor{unite⋆r}\<%
\\
\\[\AgdaEmptyExtraSkip]%
\>[0]\AgdaFunction{ExTest₁}\AgdaSpace{}%
\AgdaSymbol{:}\AgdaSpace{}%
\AgdaFunction{interp}\AgdaSpace{}%
\AgdaFunction{Ex}\AgdaSpace{}%
\AgdaInductiveConstructor{𝔽}\AgdaSpace{}%
\AgdaOperator{\AgdaDatatype{≡}}\AgdaSpace{}%
\AgdaInductiveConstructor{just}\AgdaSpace{}%
\AgdaInductiveConstructor{𝔽}\<%
\\
\>[0]\AgdaFunction{ExTest₁}\AgdaSpace{}%
\AgdaSymbol{=}\AgdaSpace{}%
\AgdaInductiveConstructor{refl}\<%
\\
\\[\AgdaEmptyExtraSkip]%
\>[0]\AgdaFunction{ExTest₂}\AgdaSpace{}%
\AgdaSymbol{:}\AgdaSpace{}%
\AgdaFunction{interp}\AgdaSpace{}%
\AgdaFunction{Ex}\AgdaSpace{}%
\AgdaInductiveConstructor{𝕋}\AgdaSpace{}%
\AgdaOperator{\AgdaDatatype{≡}}\AgdaSpace{}%
\AgdaInductiveConstructor{nothing}\<%
\\
\>[0]\AgdaFunction{ExTest₂}\AgdaSpace{}%
\AgdaSymbol{=}\AgdaSpace{}%
\AgdaInductiveConstructor{refl}\<%
\end{code}}
\newcommand{\PIFDEtaEpsilonExampleone}{%
\begin{code}%
\>[0]\AgdaFunction{id'}\AgdaSpace{}%
\AgdaSymbol{:}\AgdaSpace{}%
\AgdaFunction{𝔹}\AgdaSpace{}%
\AgdaOperator{\AgdaDatatype{⟷}}\AgdaSpace{}%
\AgdaFunction{𝔹}\<%
\\
\>[0]\AgdaFunction{id'}\AgdaSpace{}%
\AgdaSymbol{=}%
\>[1317I]\AgdaFunction{𝔹}\<%
\\
\>[.][@{}l@{}]\<[1317I]%
\>[6]\AgdaOperator{\AgdaFunction{⟷⟨}}\AgdaSpace{}%
\AgdaInductiveConstructor{uniti⋆r}\AgdaSpace{}%
\AgdaOperator{\AgdaFunction{⟩}}%
\>[42]\AgdaFunction{𝔹}\AgdaSpace{}%
\AgdaOperator{\AgdaInductiveConstructor{×ᵤ}}\AgdaSpace{}%
\AgdaInductiveConstructor{𝟙}\<%
\\
\>[6]\AgdaOperator{\AgdaFunction{⟷⟨}}\AgdaSpace{}%
\AgdaInductiveConstructor{id⟷}\AgdaSpace{}%
\AgdaOperator{\AgdaInductiveConstructor{⊗}}\AgdaSpace{}%
\AgdaInductiveConstructor{η}\AgdaSpace{}%
\AgdaInductiveConstructor{𝔽}\AgdaSpace{}%
\AgdaOperator{\AgdaFunction{⟩}}%
\>[42]\AgdaFunction{𝔹}\AgdaSpace{}%
\AgdaOperator{\AgdaInductiveConstructor{×ᵤ}}\AgdaSpace{}%
\AgdaSymbol{(}\AgdaFunction{𝔹}\AgdaSpace{}%
\AgdaOperator{\AgdaInductiveConstructor{×ᵤ}}\AgdaSpace{}%
\AgdaOperator{\AgdaInductiveConstructor{𝟙/}}\AgdaSpace{}%
\AgdaInductiveConstructor{𝔽}\AgdaSymbol{)}\<%
\\
\>[6]\AgdaOperator{\AgdaFunction{⟷⟨}}\AgdaSpace{}%
\AgdaInductiveConstructor{assocl⋆}\AgdaSpace{}%
\AgdaOperator{\AgdaFunction{⟩}}%
\>[42]\AgdaSymbol{(}\AgdaFunction{𝔹}\AgdaSpace{}%
\AgdaOperator{\AgdaInductiveConstructor{×ᵤ}}\AgdaSpace{}%
\AgdaFunction{𝔹}\AgdaSymbol{)}\AgdaSpace{}%
\AgdaOperator{\AgdaInductiveConstructor{×ᵤ}}\AgdaSpace{}%
\AgdaOperator{\AgdaInductiveConstructor{𝟙/}}\AgdaSpace{}%
\AgdaInductiveConstructor{𝔽}\<%
\\
\>[6]\AgdaOperator{\AgdaFunction{⟷⟨}}\AgdaSpace{}%
\AgdaSymbol{(}\AgdaFunction{CNOT}\AgdaSpace{}%
\AgdaOperator{\AgdaInductiveConstructor{⊚}}\AgdaSpace{}%
\AgdaFunction{CNOT'}\AgdaSpace{}%
\AgdaOperator{\AgdaInductiveConstructor{⊚}}\AgdaSpace{}%
\AgdaInductiveConstructor{swap⋆}\AgdaSymbol{)}\AgdaSpace{}%
\AgdaOperator{\AgdaInductiveConstructor{⊗}}\AgdaSpace{}%
\AgdaInductiveConstructor{id⟷}\AgdaSpace{}%
\AgdaOperator{\AgdaFunction{⟩}}%
\>[42]\AgdaSymbol{(}\AgdaFunction{𝔹}\AgdaSpace{}%
\AgdaOperator{\AgdaInductiveConstructor{×ᵤ}}\AgdaSpace{}%
\AgdaFunction{𝔹}\AgdaSymbol{)}\AgdaSpace{}%
\AgdaOperator{\AgdaInductiveConstructor{×ᵤ}}\AgdaSpace{}%
\AgdaOperator{\AgdaInductiveConstructor{𝟙/}}\AgdaSpace{}%
\AgdaInductiveConstructor{𝔽}\<%
\\
\>[6]\AgdaOperator{\AgdaFunction{⟷⟨}}\AgdaSpace{}%
\AgdaInductiveConstructor{assocr⋆}\AgdaSpace{}%
\AgdaOperator{\AgdaFunction{⟩}}%
\>[42]\AgdaFunction{𝔹}\AgdaSpace{}%
\AgdaOperator{\AgdaInductiveConstructor{×ᵤ}}\AgdaSpace{}%
\AgdaSymbol{(}\AgdaFunction{𝔹}\AgdaSpace{}%
\AgdaOperator{\AgdaInductiveConstructor{×ᵤ}}\AgdaSpace{}%
\AgdaOperator{\AgdaInductiveConstructor{𝟙/}}\AgdaSpace{}%
\AgdaInductiveConstructor{𝔽}\AgdaSymbol{)}\<%
\\
\>[6]\AgdaOperator{\AgdaFunction{⟷⟨}}\AgdaSpace{}%
\AgdaInductiveConstructor{id⟷}\AgdaSpace{}%
\AgdaOperator{\AgdaInductiveConstructor{⊗}}\AgdaSpace{}%
\AgdaInductiveConstructor{ε}\AgdaSpace{}%
\AgdaInductiveConstructor{𝔽}\AgdaSpace{}%
\AgdaOperator{\AgdaFunction{⟩}}%
\>[42]\AgdaFunction{𝔹}\AgdaSpace{}%
\AgdaOperator{\AgdaInductiveConstructor{×ᵤ}}\AgdaSpace{}%
\AgdaInductiveConstructor{𝟙}\<%
\\
\>[6]\AgdaOperator{\AgdaFunction{⟷⟨}}\AgdaSpace{}%
\AgdaInductiveConstructor{unite⋆r}\AgdaSpace{}%
\AgdaOperator{\AgdaFunction{⟩}}%
\>[42]\AgdaFunction{𝔹}\AgdaSpace{}%
\AgdaOperator{\AgdaFunction{□}}\<%
\end{code}}
\begin{code}[hide]%
\>[0]\AgdaFunction{idcheck}\AgdaSpace{}%
\AgdaSymbol{:}\AgdaSpace{}%
\AgdaSymbol{(}\AgdaBound{b}\AgdaSpace{}%
\AgdaSymbol{:}\AgdaSpace{}%
\AgdaOperator{\AgdaFunction{⟦}}\AgdaSpace{}%
\AgdaFunction{𝔹}\AgdaSpace{}%
\AgdaOperator{\AgdaFunction{⟧}}\AgdaSymbol{)}\AgdaSpace{}%
\AgdaSymbol{→}\AgdaSpace{}%
\AgdaFunction{interp}\AgdaSpace{}%
\AgdaFunction{id'}\AgdaSpace{}%
\AgdaBound{b}\AgdaSpace{}%
\AgdaOperator{\AgdaDatatype{≡}}\AgdaSpace{}%
\AgdaInductiveConstructor{just}\AgdaSpace{}%
\AgdaBound{b}\<%
\\
\>[0]\AgdaFunction{idcheck}\AgdaSpace{}%
\AgdaInductiveConstructor{𝔽}\AgdaSpace{}%
\AgdaSymbol{=}\AgdaSpace{}%
\AgdaInductiveConstructor{refl}\<%
\\
\>[0]\AgdaFunction{idcheck}\AgdaSpace{}%
\AgdaInductiveConstructor{𝕋}\AgdaSpace{}%
\AgdaSymbol{=}\AgdaSpace{}%
\AgdaInductiveConstructor{refl}\<%
\end{code}
\newcommand{\PIFDrevprod}{%
\begin{code}%
\>[0]\AgdaFunction{rev×}%
\>[1388I]\AgdaSymbol{:}\AgdaSpace{}%
\AgdaSymbol{\{}\AgdaBound{A}\AgdaSpace{}%
\AgdaBound{B}\AgdaSpace{}%
\AgdaSymbol{:}\AgdaSpace{}%
\AgdaDatatype{𝕌}\AgdaSymbol{\}}\AgdaSpace{}%
\AgdaSymbol{→}\AgdaSpace{}%
\AgdaSymbol{(}\AgdaBound{a}\AgdaSpace{}%
\AgdaSymbol{:}\AgdaSpace{}%
\AgdaOperator{\AgdaFunction{⟦}}\AgdaSpace{}%
\AgdaBound{A}\AgdaSpace{}%
\AgdaOperator{\AgdaFunction{⟧}}\AgdaSymbol{)}\AgdaSpace{}%
\AgdaSymbol{(}\AgdaBound{b}\AgdaSpace{}%
\AgdaSymbol{:}\AgdaSpace{}%
\AgdaOperator{\AgdaFunction{⟦}}\AgdaSpace{}%
\AgdaBound{B}\AgdaSpace{}%
\AgdaOperator{\AgdaFunction{⟧}}\AgdaSymbol{)}\<%
\\
\>[.][@{}l@{}]\<[1388I]%
\>[5]\AgdaSymbol{→}\AgdaSpace{}%
\AgdaOperator{\AgdaInductiveConstructor{𝟙/}}\AgdaSpace{}%
\AgdaSymbol{(}\AgdaBound{a}\AgdaSpace{}%
\AgdaOperator{\AgdaInductiveConstructor{,}}\AgdaSpace{}%
\AgdaBound{b}\AgdaSymbol{)}\AgdaSpace{}%
\AgdaOperator{\AgdaDatatype{⟷}}\AgdaSpace{}%
\AgdaOperator{\AgdaInductiveConstructor{𝟙/}}\AgdaSpace{}%
\AgdaBound{a}\AgdaSpace{}%
\AgdaOperator{\AgdaInductiveConstructor{×ᵤ}}\AgdaSpace{}%
\AgdaOperator{\AgdaInductiveConstructor{𝟙/}}\AgdaSpace{}%
\AgdaBound{b}\<%
\\
\>[0]\AgdaFunction{rev×}%
\>[1414I]\AgdaSymbol{\{}\AgdaBound{A}\AgdaSymbol{\}}\AgdaSpace{}%
\AgdaSymbol{\{}\AgdaBound{B}\AgdaSymbol{\}}\AgdaSpace{}%
\AgdaBound{a}\AgdaSpace{}%
\AgdaBound{b}\AgdaSpace{}%
\AgdaSymbol{=}\<%
\\
\>[.][@{}l@{}]\<[1414I]%
\>[5]\AgdaOperator{\AgdaInductiveConstructor{𝟙/}}\AgdaSpace{}%
\AgdaSymbol{(}\AgdaBound{a}\AgdaSpace{}%
\AgdaOperator{\AgdaInductiveConstructor{,}}\AgdaSpace{}%
\AgdaBound{b}\AgdaSymbol{)}\<%
\\
\>[0][@{}l@{\AgdaIndent{0}}]%
\>[2]\AgdaOperator{\AgdaFunction{⟷⟨}}%
\>[1422I]\AgdaInductiveConstructor{uniti⋆l}\AgdaSpace{}%
\AgdaOperator{\AgdaInductiveConstructor{⊚}}\AgdaSpace{}%
\AgdaInductiveConstructor{uniti⋆l}\AgdaSpace{}%
\AgdaOperator{\AgdaInductiveConstructor{⊚}}\AgdaSpace{}%
\AgdaInductiveConstructor{assocl⋆}\AgdaSpace{}%
\AgdaOperator{\AgdaFunction{⟩}}\<%
\\
\>[.][@{}l@{}]\<[1422I]%
\>[5]\AgdaSymbol{(}\AgdaInductiveConstructor{𝟙}\AgdaSpace{}%
\AgdaOperator{\AgdaInductiveConstructor{×ᵤ}}\AgdaSpace{}%
\AgdaInductiveConstructor{𝟙}\AgdaSymbol{)}\AgdaSpace{}%
\AgdaOperator{\AgdaInductiveConstructor{×ᵤ}}\AgdaSpace{}%
\AgdaOperator{\AgdaInductiveConstructor{𝟙/}}\AgdaSpace{}%
\AgdaSymbol{(}\AgdaBound{a}\AgdaSpace{}%
\AgdaOperator{\AgdaInductiveConstructor{,}}\AgdaSpace{}%
\AgdaBound{b}\AgdaSymbol{)}\<%
\\
\>[2]\AgdaOperator{\AgdaFunction{⟷⟨}}%
\>[1435I]\AgdaSymbol{(}\AgdaInductiveConstructor{η}\AgdaSpace{}%
\AgdaBound{a}\AgdaSpace{}%
\AgdaOperator{\AgdaInductiveConstructor{⊗}}\AgdaSpace{}%
\AgdaInductiveConstructor{η}\AgdaSpace{}%
\AgdaBound{b}\AgdaSymbol{)}\AgdaSpace{}%
\AgdaOperator{\AgdaInductiveConstructor{⊗}}\AgdaSpace{}%
\AgdaInductiveConstructor{id⟷}\AgdaSpace{}%
\AgdaOperator{\AgdaFunction{⟩}}\<%
\\
\>[.][@{}l@{}]\<[1435I]%
\>[5]\AgdaSymbol{((}\AgdaBound{A}\AgdaSpace{}%
\AgdaOperator{\AgdaInductiveConstructor{×ᵤ}}\AgdaSpace{}%
\AgdaOperator{\AgdaInductiveConstructor{𝟙/}}\AgdaSpace{}%
\AgdaBound{a}\AgdaSymbol{)}\AgdaSpace{}%
\AgdaOperator{\AgdaInductiveConstructor{×ᵤ}}\AgdaSpace{}%
\AgdaSymbol{(}\AgdaBound{B}\AgdaSpace{}%
\AgdaOperator{\AgdaInductiveConstructor{×ᵤ}}\AgdaSpace{}%
\AgdaOperator{\AgdaInductiveConstructor{𝟙/}}\AgdaSpace{}%
\AgdaBound{b}\AgdaSymbol{))}\AgdaSpace{}%
\AgdaOperator{\AgdaInductiveConstructor{×ᵤ}}\AgdaSpace{}%
\AgdaOperator{\AgdaInductiveConstructor{𝟙/}}\AgdaSpace{}%
\AgdaSymbol{(}\AgdaBound{a}\AgdaSpace{}%
\AgdaOperator{\AgdaInductiveConstructor{,}}\AgdaSpace{}%
\AgdaBound{b}\AgdaSymbol{)}\<%
\\
\>[2]\AgdaOperator{\AgdaFunction{⟷⟨}}%
\>[1456I]\AgdaSymbol{(}\AgdaFunction{shuffle}\AgdaSpace{}%
\AgdaOperator{\AgdaInductiveConstructor{⊗}}\AgdaSpace{}%
\AgdaInductiveConstructor{id⟷}\AgdaSymbol{)}\AgdaSpace{}%
\AgdaOperator{\AgdaInductiveConstructor{⊚}}\AgdaSpace{}%
\AgdaInductiveConstructor{assocr⋆}\AgdaSpace{}%
\AgdaOperator{\AgdaFunction{⟩}}\<%
\\
\>[.][@{}l@{}]\<[1456I]%
\>[5]\AgdaSymbol{(}\AgdaOperator{\AgdaInductiveConstructor{𝟙/}}\AgdaSpace{}%
\AgdaBound{a}\AgdaSpace{}%
\AgdaOperator{\AgdaInductiveConstructor{×ᵤ}}\AgdaSpace{}%
\AgdaOperator{\AgdaInductiveConstructor{𝟙/}}\AgdaSpace{}%
\AgdaBound{b}\AgdaSymbol{)}\AgdaSpace{}%
\AgdaOperator{\AgdaInductiveConstructor{×ᵤ}}\AgdaSpace{}%
\AgdaSymbol{((}\AgdaBound{A}\AgdaSpace{}%
\AgdaOperator{\AgdaInductiveConstructor{×ᵤ}}\AgdaSpace{}%
\AgdaBound{B}\AgdaSymbol{)}\AgdaSpace{}%
\AgdaOperator{\AgdaInductiveConstructor{×ᵤ}}\AgdaSpace{}%
\AgdaOperator{\AgdaInductiveConstructor{𝟙/}}\AgdaSpace{}%
\AgdaSymbol{(}\AgdaBound{a}\AgdaSpace{}%
\AgdaOperator{\AgdaInductiveConstructor{,}}\AgdaSpace{}%
\AgdaBound{b}\AgdaSymbol{))}\<%
\\
\>[2]\AgdaOperator{\AgdaFunction{⟷⟨}}%
\>[1475I]\AgdaInductiveConstructor{id⟷}\AgdaSpace{}%
\AgdaOperator{\AgdaInductiveConstructor{⊗}}\AgdaSpace{}%
\AgdaInductiveConstructor{ε}\AgdaSpace{}%
\AgdaSymbol{(}\AgdaBound{a}\AgdaSpace{}%
\AgdaOperator{\AgdaInductiveConstructor{,}}\AgdaSpace{}%
\AgdaBound{b}\AgdaSymbol{)}\AgdaSpace{}%
\AgdaOperator{\AgdaFunction{⟩}}\<%
\\
\>[.][@{}l@{}]\<[1475I]%
\>[5]\AgdaSymbol{(}\AgdaOperator{\AgdaInductiveConstructor{𝟙/}}\AgdaSpace{}%
\AgdaBound{a}\AgdaSpace{}%
\AgdaOperator{\AgdaInductiveConstructor{×ᵤ}}\AgdaSpace{}%
\AgdaOperator{\AgdaInductiveConstructor{𝟙/}}\AgdaSpace{}%
\AgdaBound{b}\AgdaSymbol{)}\AgdaSpace{}%
\AgdaOperator{\AgdaInductiveConstructor{×ᵤ}}\AgdaSpace{}%
\AgdaInductiveConstructor{𝟙}\<%
\\
\>[2]\AgdaOperator{\AgdaFunction{⟷⟨}}%
\>[1488I]\AgdaInductiveConstructor{unite⋆r}\AgdaSpace{}%
\AgdaOperator{\AgdaFunction{⟩}}\<%
\\
\>[.][@{}l@{}]\<[1488I]%
\>[5]\AgdaOperator{\AgdaInductiveConstructor{𝟙/}}\AgdaSpace{}%
\AgdaBound{a}\AgdaSpace{}%
\AgdaOperator{\AgdaInductiveConstructor{×ᵤ}}\AgdaSpace{}%
\AgdaOperator{\AgdaInductiveConstructor{𝟙/}}\AgdaSpace{}%
\AgdaBound{b}\AgdaSpace{}%
\AgdaOperator{\AgdaFunction{□}}\<%
\\
\>[2]\AgdaKeyword{where}\<%
\\
\>[2][@{}l@{\AgdaIndent{0}}]%
\>[4]\AgdaFunction{shuffle}\AgdaSpace{}%
\AgdaSymbol{:}\AgdaSpace{}%
\AgdaSymbol{\{}\AgdaBound{A}\AgdaSpace{}%
\AgdaBound{B}\AgdaSpace{}%
\AgdaBound{C}\AgdaSpace{}%
\AgdaBound{D}\AgdaSpace{}%
\AgdaSymbol{:}\AgdaSpace{}%
\AgdaDatatype{𝕌}\AgdaSymbol{\}}\AgdaSpace{}%
\AgdaSymbol{→}\AgdaSpace{}%
\AgdaSymbol{(}\AgdaBound{A}\AgdaSpace{}%
\AgdaOperator{\AgdaInductiveConstructor{×ᵤ}}\AgdaSpace{}%
\AgdaBound{B}\AgdaSymbol{)}\AgdaSpace{}%
\AgdaOperator{\AgdaInductiveConstructor{×ᵤ}}\AgdaSpace{}%
\AgdaSymbol{(}\AgdaBound{C}\AgdaSpace{}%
\AgdaOperator{\AgdaInductiveConstructor{×ᵤ}}\AgdaSpace{}%
\AgdaBound{D}\AgdaSymbol{)}\AgdaSpace{}%
\AgdaOperator{\AgdaDatatype{⟷}}\AgdaSpace{}%
\AgdaSymbol{(}\AgdaBound{B}\AgdaSpace{}%
\AgdaOperator{\AgdaInductiveConstructor{×ᵤ}}\AgdaSpace{}%
\AgdaBound{D}\AgdaSymbol{)}\AgdaSpace{}%
\AgdaOperator{\AgdaInductiveConstructor{×ᵤ}}\AgdaSpace{}%
\AgdaSymbol{(}\AgdaBound{A}\AgdaSpace{}%
\AgdaOperator{\AgdaInductiveConstructor{×ᵤ}}\AgdaSpace{}%
\AgdaBound{C}\AgdaSymbol{)}\<%
\\
\>[4]\AgdaFunction{shuffle}\AgdaSpace{}%
\AgdaSymbol{=}%
\>[1519I]\AgdaSymbol{(}\AgdaInductiveConstructor{swap⋆}\AgdaSpace{}%
\AgdaOperator{\AgdaInductiveConstructor{⊗}}\AgdaSpace{}%
\AgdaInductiveConstructor{swap⋆}\AgdaSymbol{)}\AgdaSpace{}%
\AgdaOperator{\AgdaInductiveConstructor{⊚}}\AgdaSpace{}%
\AgdaInductiveConstructor{assocr⋆}\AgdaSpace{}%
\AgdaOperator{\AgdaInductiveConstructor{⊚}}\<%
\\
\>[.][@{}l@{}]\<[1519I]%
\>[14]\AgdaSymbol{(}\AgdaInductiveConstructor{id⟷}\AgdaSpace{}%
\AgdaOperator{\AgdaInductiveConstructor{⊗}}\AgdaSpace{}%
\AgdaSymbol{(}\AgdaInductiveConstructor{assocl⋆}\AgdaSpace{}%
\AgdaOperator{\AgdaInductiveConstructor{⊚}}\AgdaSpace{}%
\AgdaSymbol{(}\AgdaInductiveConstructor{swap⋆}\AgdaSpace{}%
\AgdaOperator{\AgdaInductiveConstructor{⊗}}\AgdaSpace{}%
\AgdaInductiveConstructor{id⟷}\AgdaSymbol{)}\AgdaSpace{}%
\AgdaOperator{\AgdaInductiveConstructor{⊚}}\AgdaSpace{}%
\AgdaInductiveConstructor{assocr⋆}\AgdaSymbol{))}\AgdaSpace{}%
\AgdaOperator{\AgdaInductiveConstructor{⊚}}\<%
\\
\>[14]\AgdaInductiveConstructor{assocl⋆}\<%
\end{code}}
\begin{code}[hide]%
\>[0]\AgdaFunction{NOT}\AgdaSpace{}%
\AgdaSymbol{:}\AgdaSpace{}%
\AgdaFunction{𝔹}\AgdaSpace{}%
\AgdaOperator{\AgdaDatatype{⟷}}\AgdaSpace{}%
\AgdaFunction{𝔹}\<%
\\
\>[0]\AgdaFunction{NOT}\AgdaSpace{}%
\AgdaSymbol{=}\AgdaSpace{}%
\AgdaInductiveConstructor{swap₊}\<%
\\
\\[\AgdaEmptyExtraSkip]%
\>[0]\AgdaFunction{CNOT13}\AgdaSpace{}%
\AgdaSymbol{:}\AgdaSpace{}%
\AgdaSymbol{\{}\AgdaBound{b}\AgdaSpace{}%
\AgdaSymbol{:}\AgdaSpace{}%
\AgdaOperator{\AgdaFunction{⟦}}\AgdaSpace{}%
\AgdaFunction{𝔹}\AgdaSpace{}%
\AgdaOperator{\AgdaFunction{⟧}}\AgdaSymbol{\}}\AgdaSpace{}%
\AgdaSymbol{→}\AgdaSpace{}%
\AgdaSymbol{(}\AgdaFunction{𝔹}\AgdaSpace{}%
\AgdaOperator{\AgdaInductiveConstructor{×ᵤ}}\AgdaSpace{}%
\AgdaFunction{𝔹}\AgdaSymbol{)}\AgdaSpace{}%
\AgdaOperator{\AgdaInductiveConstructor{×ᵤ}}\AgdaSpace{}%
\AgdaSymbol{(}\AgdaFunction{𝔹}\AgdaSpace{}%
\AgdaOperator{\AgdaInductiveConstructor{×ᵤ}}\AgdaSpace{}%
\AgdaOperator{\AgdaInductiveConstructor{𝟙/}}\AgdaSpace{}%
\AgdaBound{b}\AgdaSymbol{)}\AgdaSpace{}%
\AgdaOperator{\AgdaDatatype{⟷}}%
\>[50]\AgdaSymbol{(}\AgdaFunction{𝔹}\AgdaSpace{}%
\AgdaOperator{\AgdaInductiveConstructor{×ᵤ}}\AgdaSpace{}%
\AgdaFunction{𝔹}\AgdaSymbol{)}\AgdaSpace{}%
\AgdaOperator{\AgdaInductiveConstructor{×ᵤ}}\AgdaSpace{}%
\AgdaSymbol{(}\AgdaFunction{𝔹}\AgdaSpace{}%
\AgdaOperator{\AgdaInductiveConstructor{×ᵤ}}\AgdaSpace{}%
\AgdaOperator{\AgdaInductiveConstructor{𝟙/}}\AgdaSpace{}%
\AgdaBound{b}\AgdaSymbol{)}\<%
\\
\>[0]\AgdaFunction{CNOT13}\AgdaSpace{}%
\AgdaSymbol{=}\AgdaSpace{}%
\AgdaInductiveConstructor{assocl⋆}\AgdaSpace{}%
\AgdaOperator{\AgdaInductiveConstructor{⊚}}\AgdaSpace{}%
\AgdaSymbol{((}\AgdaInductiveConstructor{assocr⋆}\AgdaSpace{}%
\AgdaOperator{\AgdaInductiveConstructor{⊚}}\AgdaSpace{}%
\AgdaSymbol{(}\AgdaInductiveConstructor{id⟷}\AgdaSpace{}%
\AgdaOperator{\AgdaInductiveConstructor{⊗}}\AgdaSpace{}%
\AgdaInductiveConstructor{swap⋆}\AgdaSymbol{)}\AgdaSpace{}%
\AgdaOperator{\AgdaInductiveConstructor{⊚}}\AgdaSpace{}%
\AgdaInductiveConstructor{assocl⋆}\AgdaSpace{}%
\AgdaOperator{\AgdaInductiveConstructor{⊚}}\AgdaSpace{}%
\AgdaSymbol{(}\AgdaFunction{CNOT}\AgdaSpace{}%
\AgdaOperator{\AgdaInductiveConstructor{⊗}}\AgdaSpace{}%
\AgdaInductiveConstructor{id⟷}\AgdaSymbol{)}\AgdaSpace{}%
\AgdaOperator{\AgdaInductiveConstructor{⊚}}\AgdaSpace{}%
\AgdaInductiveConstructor{assocr⋆}\AgdaSpace{}%
\AgdaOperator{\AgdaInductiveConstructor{⊚}}\AgdaSpace{}%
\AgdaSymbol{(}\AgdaInductiveConstructor{id⟷}\AgdaSpace{}%
\AgdaOperator{\AgdaInductiveConstructor{⊗}}\AgdaSpace{}%
\AgdaInductiveConstructor{swap⋆}\AgdaSymbol{)}\AgdaSpace{}%
\AgdaOperator{\AgdaInductiveConstructor{⊚}}\AgdaSpace{}%
\AgdaInductiveConstructor{assocl⋆}\AgdaSymbol{)}\AgdaSpace{}%
\AgdaOperator{\AgdaInductiveConstructor{⊗}}\AgdaSpace{}%
\AgdaInductiveConstructor{id⟷}\AgdaSymbol{)}\AgdaSpace{}%
\AgdaOperator{\AgdaInductiveConstructor{⊚}}\AgdaSpace{}%
\AgdaInductiveConstructor{assocr⋆}\<%
\\
\\[\AgdaEmptyExtraSkip]%
\>[0]\AgdaFunction{CNOT23}\AgdaSpace{}%
\AgdaSymbol{:}\AgdaSpace{}%
\AgdaSymbol{\{}\AgdaBound{b}\AgdaSpace{}%
\AgdaSymbol{:}\AgdaSpace{}%
\AgdaOperator{\AgdaFunction{⟦}}\AgdaSpace{}%
\AgdaFunction{𝔹}\AgdaSpace{}%
\AgdaOperator{\AgdaFunction{⟧}}\AgdaSymbol{\}}\AgdaSpace{}%
\AgdaSymbol{→}\AgdaSpace{}%
\AgdaSymbol{(}\AgdaFunction{𝔹}\AgdaSpace{}%
\AgdaOperator{\AgdaInductiveConstructor{×ᵤ}}\AgdaSpace{}%
\AgdaFunction{𝔹}\AgdaSymbol{)}\AgdaSpace{}%
\AgdaOperator{\AgdaInductiveConstructor{×ᵤ}}\AgdaSpace{}%
\AgdaSymbol{(}\AgdaFunction{𝔹}\AgdaSpace{}%
\AgdaOperator{\AgdaInductiveConstructor{×ᵤ}}\AgdaSpace{}%
\AgdaOperator{\AgdaInductiveConstructor{𝟙/}}\AgdaSpace{}%
\AgdaBound{b}\AgdaSymbol{)}\AgdaSpace{}%
\AgdaOperator{\AgdaDatatype{⟷}}%
\>[50]\AgdaSymbol{(}\AgdaFunction{𝔹}\AgdaSpace{}%
\AgdaOperator{\AgdaInductiveConstructor{×ᵤ}}\AgdaSpace{}%
\AgdaFunction{𝔹}\AgdaSymbol{)}\AgdaSpace{}%
\AgdaOperator{\AgdaInductiveConstructor{×ᵤ}}\AgdaSpace{}%
\AgdaSymbol{(}\AgdaFunction{𝔹}\AgdaSpace{}%
\AgdaOperator{\AgdaInductiveConstructor{×ᵤ}}\AgdaSpace{}%
\AgdaOperator{\AgdaInductiveConstructor{𝟙/}}\AgdaSpace{}%
\AgdaBound{b}\AgdaSymbol{)}\<%
\\
\>[0]\AgdaFunction{CNOT23}\AgdaSpace{}%
\AgdaSymbol{=}\AgdaSpace{}%
\AgdaInductiveConstructor{assocl⋆}\AgdaSpace{}%
\AgdaOperator{\AgdaInductiveConstructor{⊚}}\AgdaSpace{}%
\AgdaSymbol{((}\AgdaInductiveConstructor{assocr⋆}\AgdaSpace{}%
\AgdaOperator{\AgdaInductiveConstructor{⊚}}\AgdaSpace{}%
\AgdaSymbol{(}\AgdaInductiveConstructor{id⟷}\AgdaSpace{}%
\AgdaOperator{\AgdaInductiveConstructor{⊗}}\AgdaSpace{}%
\AgdaFunction{CNOT}\AgdaSymbol{)}\AgdaSpace{}%
\AgdaOperator{\AgdaInductiveConstructor{⊚}}\AgdaSpace{}%
\AgdaInductiveConstructor{assocl⋆}\AgdaSymbol{)}\AgdaSpace{}%
\AgdaOperator{\AgdaInductiveConstructor{⊗}}\AgdaSpace{}%
\AgdaInductiveConstructor{id⟷}\AgdaSymbol{)}\AgdaSpace{}%
\AgdaOperator{\AgdaInductiveConstructor{⊚}}\AgdaSpace{}%
\AgdaInductiveConstructor{assocr⋆}\<%
\end{code}

\begin{code}[hide]%
\>[0]\AgdaFunction{test₁}\AgdaSpace{}%
\AgdaSymbol{:}\AgdaSpace{}%
\AgdaFunction{interp}\AgdaSpace{}%
\AgdaSymbol{(}\AgdaFunction{θ}\AgdaSpace{}%
\AgdaNumber{0}\AgdaSymbol{)}\AgdaSpace{}%
\AgdaInductiveConstructor{𝔽}\AgdaSpace{}%
\AgdaOperator{\AgdaDatatype{≡}}\AgdaSpace{}%
\AgdaInductiveConstructor{just}\AgdaSpace{}%
\AgdaInductiveConstructor{𝕋}\<%
\\
\>[0]\AgdaFunction{test₁}\AgdaSpace{}%
\AgdaSymbol{=}\AgdaSpace{}%
\AgdaInductiveConstructor{refl}\<%
\\
\\[\AgdaEmptyExtraSkip]%
\>[0]\AgdaFunction{test₂}\AgdaSpace{}%
\AgdaSymbol{:}\AgdaSpace{}%
\AgdaFunction{interp}\AgdaSpace{}%
\AgdaSymbol{(}\AgdaFunction{θ}\AgdaSpace{}%
\AgdaNumber{1}\AgdaSymbol{)}\AgdaSpace{}%
\AgdaSymbol{(}\AgdaInductiveConstructor{𝕋}\AgdaSpace{}%
\AgdaOperator{\AgdaInductiveConstructor{,}}\AgdaSpace{}%
\AgdaInductiveConstructor{𝔽}\AgdaSymbol{)}\AgdaSpace{}%
\AgdaOperator{\AgdaDatatype{≡}}\AgdaSpace{}%
\AgdaInductiveConstructor{just}\AgdaSpace{}%
\AgdaSymbol{(}\AgdaInductiveConstructor{𝕋}\AgdaSpace{}%
\AgdaOperator{\AgdaInductiveConstructor{,}}\AgdaSpace{}%
\AgdaInductiveConstructor{𝕋}\AgdaSymbol{)}\<%
\\
\>[0]\AgdaFunction{test₂}\AgdaSpace{}%
\AgdaSymbol{=}\AgdaSpace{}%
\AgdaInductiveConstructor{refl}\<%
\\
\\[\AgdaEmptyExtraSkip]%
\>[0]\AgdaFunction{test₃}\AgdaSpace{}%
\AgdaSymbol{:}\AgdaSpace{}%
\AgdaFunction{interp}\AgdaSpace{}%
\AgdaSymbol{(}\AgdaFunction{θ}\AgdaSpace{}%
\AgdaNumber{2}\AgdaSymbol{)}\AgdaSpace{}%
\AgdaSymbol{((}\AgdaInductiveConstructor{𝕋}\AgdaSpace{}%
\AgdaOperator{\AgdaInductiveConstructor{,}}\AgdaSpace{}%
\AgdaInductiveConstructor{𝕋}\AgdaSymbol{)}\AgdaSpace{}%
\AgdaOperator{\AgdaInductiveConstructor{,}}\AgdaSpace{}%
\AgdaInductiveConstructor{𝕋}\AgdaSymbol{)}\AgdaSpace{}%
\AgdaOperator{\AgdaDatatype{≡}}\AgdaSpace{}%
\AgdaInductiveConstructor{just}\AgdaSpace{}%
\AgdaSymbol{((}\AgdaInductiveConstructor{𝕋}\AgdaSpace{}%
\AgdaOperator{\AgdaInductiveConstructor{,}}\AgdaSpace{}%
\AgdaInductiveConstructor{𝕋}\AgdaSymbol{)}\AgdaSpace{}%
\AgdaOperator{\AgdaInductiveConstructor{,}}\AgdaSpace{}%
\AgdaInductiveConstructor{𝔽}\AgdaSymbol{)}\<%
\\
\>[0]\AgdaFunction{test₃}\AgdaSpace{}%
\AgdaSymbol{=}\AgdaSpace{}%
\AgdaInductiveConstructor{refl}\<%
\\
\\[\AgdaEmptyExtraSkip]%
\>[0]\AgdaFunction{test₄}\AgdaSpace{}%
\AgdaSymbol{:}\AgdaSpace{}%
\AgdaFunction{interp}\AgdaSpace{}%
\AgdaSymbol{(}\AgdaFunction{θ}\AgdaSpace{}%
\AgdaNumber{3}\AgdaSymbol{)}\AgdaSpace{}%
\AgdaSymbol{(((}\AgdaInductiveConstructor{𝕋}\AgdaSpace{}%
\AgdaOperator{\AgdaInductiveConstructor{,}}\AgdaSpace{}%
\AgdaInductiveConstructor{𝕋}\AgdaSymbol{)}\AgdaSpace{}%
\AgdaOperator{\AgdaInductiveConstructor{,}}\AgdaSpace{}%
\AgdaInductiveConstructor{𝕋}\AgdaSymbol{)}\AgdaSpace{}%
\AgdaOperator{\AgdaInductiveConstructor{,}}\AgdaSpace{}%
\AgdaInductiveConstructor{𝔽}\AgdaSymbol{)}\AgdaSpace{}%
\AgdaOperator{\AgdaDatatype{≡}}\AgdaSpace{}%
\AgdaInductiveConstructor{just}\AgdaSpace{}%
\AgdaSymbol{(((}\AgdaInductiveConstructor{𝕋}\AgdaSpace{}%
\AgdaOperator{\AgdaInductiveConstructor{,}}\AgdaSpace{}%
\AgdaInductiveConstructor{𝕋}\AgdaSymbol{)}\AgdaSpace{}%
\AgdaOperator{\AgdaInductiveConstructor{,}}\AgdaSpace{}%
\AgdaInductiveConstructor{𝕋}\AgdaSymbol{)}\AgdaSpace{}%
\AgdaOperator{\AgdaInductiveConstructor{,}}\AgdaSpace{}%
\AgdaInductiveConstructor{𝕋}\AgdaSymbol{)}\<%
\\
\>[0]\AgdaFunction{test₄}\AgdaSpace{}%
\AgdaSymbol{=}\AgdaSpace{}%
\AgdaInductiveConstructor{refl}\<%
\\
\\[\AgdaEmptyExtraSkip]%
\>[0]\AgdaFunction{test₅}\AgdaSpace{}%
\AgdaSymbol{:}\AgdaSpace{}%
\AgdaFunction{interp}\AgdaSpace{}%
\AgdaSymbol{(}\AgdaFunction{θ}\AgdaSpace{}%
\AgdaNumber{4}\AgdaSymbol{)}\AgdaSpace{}%
\AgdaSymbol{((((}\AgdaInductiveConstructor{𝕋}\AgdaSpace{}%
\AgdaOperator{\AgdaInductiveConstructor{,}}\AgdaSpace{}%
\AgdaInductiveConstructor{𝕋}\AgdaSymbol{)}\AgdaSpace{}%
\AgdaOperator{\AgdaInductiveConstructor{,}}\AgdaSpace{}%
\AgdaInductiveConstructor{𝕋}\AgdaSymbol{)}\AgdaSpace{}%
\AgdaOperator{\AgdaInductiveConstructor{,}}\AgdaSpace{}%
\AgdaInductiveConstructor{𝕋}\AgdaSymbol{)}\AgdaSpace{}%
\AgdaOperator{\AgdaInductiveConstructor{,}}\AgdaSpace{}%
\AgdaInductiveConstructor{𝕋}\AgdaSymbol{)}\AgdaSpace{}%
\AgdaOperator{\AgdaDatatype{≡}}\AgdaSpace{}%
\AgdaInductiveConstructor{just}\AgdaSpace{}%
\AgdaSymbol{((((}\AgdaInductiveConstructor{𝕋}\AgdaSpace{}%
\AgdaOperator{\AgdaInductiveConstructor{,}}\AgdaSpace{}%
\AgdaInductiveConstructor{𝕋}\AgdaSymbol{)}\AgdaSpace{}%
\AgdaOperator{\AgdaInductiveConstructor{,}}\AgdaSpace{}%
\AgdaInductiveConstructor{𝕋}\AgdaSymbol{)}\AgdaSpace{}%
\AgdaOperator{\AgdaInductiveConstructor{,}}\AgdaSpace{}%
\AgdaInductiveConstructor{𝕋}\AgdaSymbol{)}\AgdaSpace{}%
\AgdaOperator{\AgdaInductiveConstructor{,}}\AgdaSpace{}%
\AgdaInductiveConstructor{𝔽}\AgdaSymbol{)}\<%
\\
\>[0]\AgdaFunction{test₅}\AgdaSpace{}%
\AgdaSymbol{=}\AgdaSpace{}%
\AgdaInductiveConstructor{refl}\<%
\end{code}

\begin{code}[hide]%
\>[0]\AgdaSymbol{\{-\#}\AgdaSpace{}%
\AgdaKeyword{OPTIONS}\AgdaSpace{}%
\AgdaPragma{--without-K}\AgdaSpace{}%
\AgdaPragma{--safe}\AgdaSpace{}%
\AgdaSymbol{\#-\}}\<%
\\
\>[0]\AgdaKeyword{module}\AgdaSpace{}%
\AgdaModule{PiMemSem}\AgdaSpace{}%
\AgdaKeyword{where}\<%
\\
\>[0]\AgdaKeyword{open}\AgdaSpace{}%
\AgdaKeyword{import}\AgdaSpace{}%
\AgdaModule{Function}\AgdaSpace{}%
\AgdaKeyword{using}\AgdaSpace{}%
\AgdaSymbol{(}\AgdaOperator{\AgdaFunction{\AgdaUnderscore{}∘\AgdaUnderscore{}}}\AgdaSymbol{;}\AgdaSpace{}%
\AgdaOperator{\AgdaFunction{\AgdaUnderscore{}\$\AgdaUnderscore{}}}\AgdaSymbol{)}\<%
\\
\>[0]\AgdaKeyword{open}\AgdaSpace{}%
\AgdaKeyword{import}\AgdaSpace{}%
\AgdaModule{Data.Nat}\<%
\\
\>[0]\AgdaKeyword{open}\AgdaSpace{}%
\AgdaKeyword{import}\AgdaSpace{}%
\AgdaModule{Data.Nat.Properties}\<%
\\
\>[0]\AgdaKeyword{open}\AgdaSpace{}%
\AgdaKeyword{import}\AgdaSpace{}%
\AgdaModule{Data.Fin}\AgdaSpace{}%
\AgdaKeyword{hiding}\AgdaSpace{}%
\AgdaSymbol{(}\AgdaOperator{\AgdaFunction{\AgdaUnderscore{}+\AgdaUnderscore{}}}\AgdaSymbol{)}\<%
\\
\>[0]\AgdaKeyword{open}\AgdaSpace{}%
\AgdaKeyword{import}\AgdaSpace{}%
\AgdaModule{Data.Empty}\AgdaSpace{}%
\AgdaKeyword{using}\AgdaSpace{}%
\AgdaSymbol{(}\AgdaDatatype{⊥}\AgdaSymbol{)}\<%
\\
\>[0]\AgdaKeyword{open}\AgdaSpace{}%
\AgdaKeyword{import}\AgdaSpace{}%
\AgdaModule{Data.Unit}\AgdaSpace{}%
\AgdaKeyword{using}\AgdaSpace{}%
\AgdaSymbol{(}\AgdaRecord{⊤}\AgdaSymbol{;}\AgdaSpace{}%
\AgdaInductiveConstructor{tt}\AgdaSymbol{)}\<%
\\
\>[0]\AgdaKeyword{open}\AgdaSpace{}%
\AgdaKeyword{import}\AgdaSpace{}%
\AgdaModule{Data.Sum}\AgdaSpace{}%
\AgdaKeyword{using}\AgdaSpace{}%
\AgdaSymbol{(}\AgdaOperator{\AgdaDatatype{\AgdaUnderscore{}⊎\AgdaUnderscore{}}}\AgdaSymbol{;}\AgdaSpace{}%
\AgdaInductiveConstructor{inj₁}\AgdaSymbol{;}\AgdaSpace{}%
\AgdaInductiveConstructor{inj₂}\AgdaSymbol{)}\<%
\\
\>[0]\AgdaKeyword{open}\AgdaSpace{}%
\AgdaKeyword{import}\AgdaSpace{}%
\AgdaModule{Data.Product}\AgdaSpace{}%
\AgdaKeyword{using}\AgdaSpace{}%
\AgdaSymbol{(}\AgdaOperator{\AgdaFunction{\AgdaUnderscore{}×\AgdaUnderscore{}}}\AgdaSymbol{;}\AgdaSpace{}%
\AgdaOperator{\AgdaInductiveConstructor{\AgdaUnderscore{},\AgdaUnderscore{}}}\AgdaSymbol{;}\AgdaSpace{}%
\AgdaField{proj₁}\AgdaSymbol{;}\AgdaSpace{}%
\AgdaField{proj₂}\AgdaSymbol{;}\AgdaSpace{}%
\AgdaFunction{Σ-syntax}\AgdaSymbol{)}\<%
\\
\>[0]\AgdaKeyword{open}\AgdaSpace{}%
\AgdaKeyword{import}\AgdaSpace{}%
\AgdaModule{Data.Vec}\<%
\\
\>[0]\AgdaKeyword{open}\AgdaSpace{}%
\AgdaKeyword{import}\AgdaSpace{}%
\AgdaModule{Data.Vec.Relation.Unary.Any.Properties}\<%
\\
\>[0]\AgdaKeyword{open}\AgdaSpace{}%
\AgdaKeyword{import}\AgdaSpace{}%
\AgdaModule{Data.Vec.Any}\AgdaSpace{}%
\AgdaKeyword{using}\AgdaSpace{}%
\AgdaSymbol{(}\AgdaDatatype{Any}\AgdaSymbol{;}\AgdaSpace{}%
\AgdaInductiveConstructor{here}\AgdaSymbol{;}\AgdaSpace{}%
\AgdaInductiveConstructor{there}\AgdaSymbol{;}\AgdaSpace{}%
\AgdaFunction{index}\AgdaSymbol{)}\<%
\\
\>[0]\AgdaKeyword{open}\AgdaSpace{}%
\AgdaKeyword{import}\AgdaSpace{}%
\AgdaModule{Data.Vec.Membership.Propositional}\<%
\\
\>[0]\AgdaKeyword{open}\AgdaSpace{}%
\AgdaKeyword{import}\AgdaSpace{}%
\AgdaModule{Data.Vec.Membership.Propositional.Properties}\<%
\\
\>[0]\AgdaKeyword{open}\AgdaSpace{}%
\AgdaKeyword{import}\AgdaSpace{}%
\AgdaModule{Relation.Binary.PropositionalEquality}\AgdaSpace{}%
\AgdaKeyword{hiding}\AgdaSpace{}%
\AgdaSymbol{(}\AgdaOperator{\AgdaInductiveConstructor{[\AgdaUnderscore{}]}}\AgdaSymbol{)}\<%
\\
\>[0]\AgdaKeyword{open}\AgdaSpace{}%
\AgdaKeyword{import}\AgdaSpace{}%
\AgdaModule{Pi}\<%
\end{code}

\begin{code}[hide]%
\>[0]\AgdaFunction{Enum}\AgdaSpace{}%
\AgdaSymbol{:}\AgdaSpace{}%
\AgdaSymbol{(}\AgdaBound{A}\AgdaSpace{}%
\AgdaSymbol{:}\AgdaSpace{}%
\AgdaDatatype{𝕌}\AgdaSymbol{)}\AgdaSpace{}%
\AgdaSymbol{→}\AgdaSpace{}%
\AgdaDatatype{Vec}\AgdaSpace{}%
\AgdaOperator{\AgdaFunction{⟦}}\AgdaSpace{}%
\AgdaBound{A}\AgdaSpace{}%
\AgdaOperator{\AgdaFunction{⟧}}\AgdaSpace{}%
\AgdaOperator{\AgdaFunction{∣}}\AgdaSpace{}%
\AgdaBound{A}\AgdaSpace{}%
\AgdaOperator{\AgdaFunction{∣}}\<%
\\
\>[0]\AgdaFunction{Enum}\AgdaSpace{}%
\AgdaInductiveConstructor{𝟘}%
\>[15]\AgdaSymbol{=}\AgdaSpace{}%
\AgdaInductiveConstructor{[]}\<%
\\
\>[0]\AgdaFunction{Enum}\AgdaSpace{}%
\AgdaInductiveConstructor{𝟙}%
\>[16]\AgdaSymbol{=}\AgdaSpace{}%
\AgdaInductiveConstructor{tt}\AgdaSpace{}%
\AgdaOperator{\AgdaInductiveConstructor{∷}}\AgdaSpace{}%
\AgdaInductiveConstructor{[]}\<%
\\
\>[0]\AgdaFunction{Enum}\AgdaSpace{}%
\AgdaSymbol{(}\AgdaBound{A₁}\AgdaSpace{}%
\AgdaOperator{\AgdaInductiveConstructor{+ᵤ}}\AgdaSpace{}%
\AgdaBound{A₂}\AgdaSymbol{)}\AgdaSpace{}%
\AgdaSymbol{=}\AgdaSpace{}%
\AgdaFunction{map}\AgdaSpace{}%
\AgdaInductiveConstructor{inj₁}\AgdaSpace{}%
\AgdaSymbol{(}\AgdaFunction{Enum}\AgdaSpace{}%
\AgdaBound{A₁}\AgdaSymbol{)}\AgdaSpace{}%
\AgdaOperator{\AgdaFunction{++}}\AgdaSpace{}%
\AgdaFunction{map}\AgdaSpace{}%
\AgdaInductiveConstructor{inj₂}\AgdaSpace{}%
\AgdaSymbol{(}\AgdaFunction{Enum}\AgdaSpace{}%
\AgdaBound{A₂}\AgdaSymbol{)}\<%
\\
\>[0]\AgdaFunction{Enum}\AgdaSpace{}%
\AgdaSymbol{(}\AgdaBound{A₁}\AgdaSpace{}%
\AgdaOperator{\AgdaInductiveConstructor{×ᵤ}}\AgdaSpace{}%
\AgdaBound{A₂}\AgdaSymbol{)}\AgdaSpace{}%
\AgdaSymbol{=}\AgdaSpace{}%
\AgdaFunction{allPairs}\AgdaSpace{}%
\AgdaSymbol{(}\AgdaFunction{Enum}\AgdaSpace{}%
\AgdaBound{A₁}\AgdaSymbol{)}\AgdaSpace{}%
\AgdaSymbol{(}\AgdaFunction{Enum}\AgdaSpace{}%
\AgdaBound{A₂}\AgdaSymbol{)}\<%
\\
\\[\AgdaEmptyExtraSkip]%
\>[0]\AgdaFunction{Find}\AgdaSpace{}%
\AgdaSymbol{:}\AgdaSpace{}%
\AgdaSymbol{\{}\AgdaBound{A}\AgdaSpace{}%
\AgdaSymbol{:}\AgdaSpace{}%
\AgdaDatatype{𝕌}\AgdaSymbol{\}}\AgdaSpace{}%
\AgdaSymbol{(}\AgdaBound{x}\AgdaSpace{}%
\AgdaSymbol{:}\AgdaSpace{}%
\AgdaOperator{\AgdaFunction{⟦}}\AgdaSpace{}%
\AgdaBound{A}\AgdaSpace{}%
\AgdaOperator{\AgdaFunction{⟧}}\AgdaSymbol{)}\AgdaSpace{}%
\AgdaSymbol{→}\AgdaSpace{}%
\AgdaBound{x}\AgdaSpace{}%
\AgdaOperator{\AgdaFunction{∈}}\AgdaSpace{}%
\AgdaFunction{Enum}\AgdaSpace{}%
\AgdaBound{A}\<%
\\
\>[0]\AgdaFunction{Find}\AgdaSpace{}%
\AgdaSymbol{\{}\AgdaInductiveConstructor{𝟙}\AgdaSymbol{\}}\AgdaSpace{}%
\AgdaInductiveConstructor{tt}\AgdaSpace{}%
\AgdaSymbol{=}\AgdaSpace{}%
\AgdaInductiveConstructor{here}\AgdaSpace{}%
\AgdaInductiveConstructor{refl}\<%
\\
\>[0]\AgdaFunction{Find}\AgdaSpace{}%
\AgdaSymbol{\{}\AgdaBound{A₁}\AgdaSpace{}%
\AgdaOperator{\AgdaInductiveConstructor{+ᵤ}}\AgdaSpace{}%
\AgdaBound{A₂}\AgdaSymbol{\}}\AgdaSpace{}%
\AgdaSymbol{(}\AgdaInductiveConstructor{inj₁}\AgdaSpace{}%
\AgdaBound{x}\AgdaSymbol{)}\AgdaSpace{}%
\AgdaSymbol{=}\AgdaSpace{}%
\AgdaFunction{++⁺ˡ}\AgdaSpace{}%
\AgdaSymbol{\{}\AgdaArgument{xs}\AgdaSpace{}%
\AgdaSymbol{=}\AgdaSpace{}%
\AgdaFunction{map}\AgdaSpace{}%
\AgdaInductiveConstructor{inj₁}\AgdaSpace{}%
\AgdaSymbol{(}\AgdaFunction{Enum}\AgdaSpace{}%
\AgdaBound{A₁}\AgdaSymbol{)\}}\AgdaSpace{}%
\AgdaSymbol{(}\AgdaFunction{∈-map⁺}\AgdaSpace{}%
\AgdaInductiveConstructor{inj₁}\AgdaSpace{}%
\AgdaSymbol{(}\AgdaFunction{Find}\AgdaSpace{}%
\AgdaBound{x}\AgdaSymbol{))}\<%
\\
\>[0]\AgdaFunction{Find}\AgdaSpace{}%
\AgdaSymbol{\{}\AgdaBound{A₁}\AgdaSpace{}%
\AgdaOperator{\AgdaInductiveConstructor{+ᵤ}}\AgdaSpace{}%
\AgdaBound{A₂}\AgdaSymbol{\}}\AgdaSpace{}%
\AgdaSymbol{(}\AgdaInductiveConstructor{inj₂}\AgdaSpace{}%
\AgdaBound{y}\AgdaSymbol{)}\AgdaSpace{}%
\AgdaSymbol{=}\AgdaSpace{}%
\AgdaFunction{++⁺ʳ}\AgdaSpace{}%
\AgdaSymbol{(}\AgdaFunction{map}\AgdaSpace{}%
\AgdaInductiveConstructor{inj₁}\AgdaSpace{}%
\AgdaSymbol{(}\AgdaFunction{Enum}\AgdaSpace{}%
\AgdaBound{A₁}\AgdaSymbol{))}\AgdaSpace{}%
\AgdaSymbol{(}\AgdaFunction{∈-map⁺}\AgdaSpace{}%
\AgdaInductiveConstructor{inj₂}\AgdaSpace{}%
\AgdaSymbol{(}\AgdaFunction{Find}\AgdaSpace{}%
\AgdaBound{y}\AgdaSymbol{))}\<%
\\
\>[0]\AgdaFunction{Find}\AgdaSpace{}%
\AgdaSymbol{\{}\AgdaBound{A₁}\AgdaSpace{}%
\AgdaOperator{\AgdaInductiveConstructor{×ᵤ}}\AgdaSpace{}%
\AgdaBound{A₂}\AgdaSymbol{\}}\AgdaSpace{}%
\AgdaSymbol{(}\AgdaBound{x}\AgdaSpace{}%
\AgdaOperator{\AgdaInductiveConstructor{,}}\AgdaSpace{}%
\AgdaBound{y}\AgdaSymbol{)}\AgdaSpace{}%
\AgdaSymbol{=}\AgdaSpace{}%
\AgdaFunction{∈-allPairs⁺}\AgdaSpace{}%
\AgdaSymbol{(}\AgdaFunction{Find}\AgdaSpace{}%
\AgdaBound{x}\AgdaSymbol{)}\AgdaSpace{}%
\AgdaSymbol{(}\AgdaFunction{Find}\AgdaSpace{}%
\AgdaBound{y}\AgdaSymbol{)}\<%
\\
\\[\AgdaEmptyExtraSkip]%
\>[0]\AgdaFunction{Find'}\AgdaSpace{}%
\AgdaSymbol{:}\AgdaSpace{}%
\AgdaSymbol{\{}\AgdaBound{A}\AgdaSpace{}%
\AgdaSymbol{:}\AgdaSpace{}%
\AgdaDatatype{𝕌}\AgdaSymbol{\}}\AgdaSpace{}%
\AgdaSymbol{(}\AgdaBound{x}\AgdaSpace{}%
\AgdaSymbol{:}\AgdaSpace{}%
\AgdaOperator{\AgdaFunction{⟦}}\AgdaSpace{}%
\AgdaBound{A}\AgdaSpace{}%
\AgdaOperator{\AgdaFunction{⟧}}\AgdaSymbol{)}\AgdaSpace{}%
\AgdaSymbol{→}\AgdaSpace{}%
\AgdaDatatype{Fin}\AgdaSpace{}%
\AgdaOperator{\AgdaFunction{∣}}\AgdaSpace{}%
\AgdaBound{A}\AgdaSpace{}%
\AgdaOperator{\AgdaFunction{∣}}\<%
\\
\>[0]\AgdaFunction{Find'}\AgdaSpace{}%
\AgdaSymbol{=}\AgdaSpace{}%
\AgdaFunction{index}\AgdaSpace{}%
\AgdaOperator{\AgdaFunction{∘}}\AgdaSpace{}%
\AgdaFunction{Find}\<%
\end{code}

\begin{code}[hide]%
\>[0]\AgdaKeyword{data}\AgdaSpace{}%
\AgdaDatatype{State}\AgdaSpace{}%
\AgdaSymbol{(}\AgdaBound{A}\AgdaSpace{}%
\AgdaSymbol{:}\AgdaSpace{}%
\AgdaDatatype{𝕌}\AgdaSymbol{)}\AgdaSpace{}%
\AgdaSymbol{:}\AgdaSpace{}%
\AgdaPrimitiveType{Set}\AgdaSpace{}%
\AgdaKeyword{where}\<%
\\
\>[0][@{}l@{\AgdaIndent{0}}]%
\>[2]\AgdaOperator{\AgdaInductiveConstructor{⟪\AgdaUnderscore{}[\AgdaUnderscore{}]⟫}}\AgdaSpace{}%
\AgdaSymbol{:}\AgdaSpace{}%
\AgdaDatatype{Vec}\AgdaSpace{}%
\AgdaOperator{\AgdaFunction{⟦}}\AgdaSpace{}%
\AgdaBound{A}\AgdaSpace{}%
\AgdaOperator{\AgdaFunction{⟧}}\AgdaSpace{}%
\AgdaOperator{\AgdaFunction{∣}}\AgdaSpace{}%
\AgdaBound{A}\AgdaSpace{}%
\AgdaOperator{\AgdaFunction{∣}}\AgdaSpace{}%
\AgdaSymbol{→}\AgdaSpace{}%
\AgdaDatatype{Fin}\AgdaSpace{}%
\AgdaOperator{\AgdaFunction{∣}}\AgdaSpace{}%
\AgdaBound{A}\AgdaSpace{}%
\AgdaOperator{\AgdaFunction{∣}}\AgdaSpace{}%
\AgdaSymbol{→}\AgdaSpace{}%
\AgdaDatatype{State}\AgdaSpace{}%
\AgdaBound{A}\<%
\\
\\[\AgdaEmptyExtraSkip]%
\>[0]\AgdaFunction{resolve}\AgdaSpace{}%
\AgdaSymbol{:}\AgdaSpace{}%
\AgdaSymbol{\{}\AgdaBound{A}\AgdaSpace{}%
\AgdaSymbol{:}\AgdaSpace{}%
\AgdaDatatype{𝕌}\AgdaSymbol{\}}\AgdaSpace{}%
\AgdaSymbol{→}\AgdaSpace{}%
\AgdaDatatype{State}\AgdaSpace{}%
\AgdaBound{A}\AgdaSpace{}%
\AgdaSymbol{→}\AgdaSpace{}%
\AgdaOperator{\AgdaFunction{⟦}}\AgdaSpace{}%
\AgdaBound{A}\AgdaSpace{}%
\AgdaOperator{\AgdaFunction{⟧}}\<%
\\
\>[0]\AgdaFunction{resolve}\AgdaSpace{}%
\AgdaOperator{\AgdaInductiveConstructor{⟪}}\AgdaSpace{}%
\AgdaBound{v}\AgdaSpace{}%
\AgdaOperator{\AgdaInductiveConstructor{[}}\AgdaSpace{}%
\AgdaBound{i}\AgdaSpace{}%
\AgdaOperator{\AgdaInductiveConstructor{]⟫}}\AgdaSpace{}%
\AgdaSymbol{=}\AgdaSpace{}%
\AgdaFunction{lookup}\AgdaSpace{}%
\AgdaBound{v}\AgdaSpace{}%
\AgdaBound{i}\<%
\\
\\[\AgdaEmptyExtraSkip]%
\>[0]\AgdaFunction{st}\AgdaSpace{}%
\AgdaSymbol{:}\AgdaSpace{}%
\AgdaSymbol{\{}\AgdaBound{A}\AgdaSpace{}%
\AgdaBound{B}\AgdaSpace{}%
\AgdaSymbol{:}\AgdaSpace{}%
\AgdaDatatype{𝕌}\AgdaSymbol{\}}\AgdaSpace{}%
\AgdaSymbol{→}\AgdaSpace{}%
\AgdaOperator{\AgdaFunction{⟦}}\AgdaSpace{}%
\AgdaBound{A}\AgdaSpace{}%
\AgdaOperator{\AgdaFunction{⟧}}\AgdaSpace{}%
\AgdaSymbol{→}\AgdaSpace{}%
\AgdaSymbol{(}\AgdaBound{c}\AgdaSpace{}%
\AgdaSymbol{:}\AgdaSpace{}%
\AgdaBound{A}\AgdaSpace{}%
\AgdaOperator{\AgdaDatatype{⟷}}\AgdaSpace{}%
\AgdaBound{B}\AgdaSymbol{)}\AgdaSpace{}%
\AgdaSymbol{→}\AgdaSpace{}%
\AgdaFunction{Σ[}\AgdaSpace{}%
\AgdaBound{C}\AgdaSpace{}%
\AgdaFunction{∈}\AgdaSpace{}%
\AgdaDatatype{𝕌}\AgdaSpace{}%
\AgdaFunction{]}\AgdaSpace{}%
\AgdaSymbol{(}\AgdaBound{C}\AgdaSpace{}%
\AgdaOperator{\AgdaDatatype{⟷}}\AgdaSpace{}%
\AgdaBound{B}\AgdaSpace{}%
\AgdaOperator{\AgdaFunction{×}}\AgdaSpace{}%
\AgdaDatatype{State}\AgdaSpace{}%
\AgdaBound{C}\AgdaSymbol{)}\<%
\\
\>[0]\AgdaFunction{st}\AgdaSpace{}%
\AgdaSymbol{(}\AgdaInductiveConstructor{inj₂}\AgdaSpace{}%
\AgdaBound{y}\AgdaSymbol{)}\AgdaSpace{}%
\AgdaSymbol{(}\AgdaInductiveConstructor{unite₊l}\AgdaSpace{}%
\AgdaSymbol{\{}\AgdaBound{t}\AgdaSymbol{\})}%
\>[44]\AgdaSymbol{=}\AgdaSpace{}%
\AgdaSymbol{\AgdaUnderscore{}}\AgdaSpace{}%
\AgdaOperator{\AgdaInductiveConstructor{,}}\AgdaSpace{}%
\AgdaInductiveConstructor{id⟷}\AgdaSpace{}%
\AgdaOperator{\AgdaInductiveConstructor{,}}\AgdaSpace{}%
\AgdaOperator{\AgdaInductiveConstructor{⟪}}\AgdaSpace{}%
\AgdaFunction{Enum}\AgdaSpace{}%
\AgdaBound{t}\AgdaSpace{}%
\AgdaOperator{\AgdaInductiveConstructor{[}}\AgdaSpace{}%
\AgdaFunction{Find'}\AgdaSpace{}%
\AgdaBound{y}\AgdaSpace{}%
\AgdaOperator{\AgdaInductiveConstructor{]⟫}}\<%
\\
\>[0]\AgdaFunction{st}\AgdaSpace{}%
\AgdaBound{a}\AgdaSpace{}%
\AgdaSymbol{(}\AgdaInductiveConstructor{uniti₊l}\AgdaSpace{}%
\AgdaSymbol{\{}\AgdaBound{t}\AgdaSymbol{\})}%
\>[44]\AgdaSymbol{=}\AgdaSpace{}%
\AgdaSymbol{\AgdaUnderscore{}}\AgdaSpace{}%
\AgdaOperator{\AgdaInductiveConstructor{,}}\AgdaSpace{}%
\AgdaInductiveConstructor{id⟷}\AgdaSpace{}%
\AgdaOperator{\AgdaInductiveConstructor{,}}\AgdaSpace{}%
\AgdaOperator{\AgdaInductiveConstructor{⟪}}\AgdaSpace{}%
\AgdaSymbol{(}\AgdaFunction{Enum}\AgdaSpace{}%
\AgdaSymbol{(}\AgdaInductiveConstructor{𝟘}\AgdaSpace{}%
\AgdaOperator{\AgdaInductiveConstructor{+ᵤ}}\AgdaSpace{}%
\AgdaBound{t}\AgdaSymbol{))}\AgdaSpace{}%
\AgdaOperator{\AgdaInductiveConstructor{[}}\AgdaSpace{}%
\AgdaFunction{Find'}\AgdaSpace{}%
\AgdaBound{a}\AgdaSpace{}%
\AgdaOperator{\AgdaInductiveConstructor{]⟫}}\<%
\\
\>[0]\AgdaFunction{st}\AgdaSpace{}%
\AgdaSymbol{(}\AgdaInductiveConstructor{inj₁}\AgdaSpace{}%
\AgdaBound{x}\AgdaSymbol{)}\AgdaSpace{}%
\AgdaSymbol{(}\AgdaInductiveConstructor{unite₊r}\AgdaSpace{}%
\AgdaSymbol{\{}\AgdaBound{t}\AgdaSymbol{\})}%
\>[44]\AgdaSymbol{=}\AgdaSpace{}%
\AgdaSymbol{\AgdaUnderscore{}}\AgdaSpace{}%
\AgdaOperator{\AgdaInductiveConstructor{,}}\AgdaSpace{}%
\AgdaInductiveConstructor{id⟷}\AgdaSpace{}%
\AgdaOperator{\AgdaInductiveConstructor{,}}\AgdaSpace{}%
\AgdaOperator{\AgdaInductiveConstructor{⟪}}\AgdaSpace{}%
\AgdaFunction{Enum}\AgdaSpace{}%
\AgdaBound{t}\AgdaSpace{}%
\AgdaOperator{\AgdaInductiveConstructor{[}}\AgdaSpace{}%
\AgdaFunction{Find'}\AgdaSpace{}%
\AgdaBound{x}\AgdaSpace{}%
\AgdaOperator{\AgdaInductiveConstructor{]⟫}}\<%
\\
\>[0]\AgdaFunction{st}\AgdaSpace{}%
\AgdaBound{a}\AgdaSpace{}%
\AgdaSymbol{(}\AgdaInductiveConstructor{uniti₊r}\AgdaSpace{}%
\AgdaSymbol{\{}\AgdaBound{t}\AgdaSymbol{\})}%
\>[44]\AgdaSymbol{=}\AgdaSpace{}%
\AgdaSymbol{\AgdaUnderscore{}}\AgdaSpace{}%
\AgdaOperator{\AgdaInductiveConstructor{,}}\AgdaSpace{}%
\AgdaInductiveConstructor{id⟷}\AgdaSpace{}%
\AgdaOperator{\AgdaInductiveConstructor{,}}\AgdaSpace{}%
\AgdaOperator{\AgdaInductiveConstructor{⟪}}\AgdaSpace{}%
\AgdaSymbol{(}\AgdaFunction{Enum}\AgdaSpace{}%
\AgdaSymbol{(}\AgdaBound{t}\AgdaSpace{}%
\AgdaOperator{\AgdaInductiveConstructor{+ᵤ}}\AgdaSpace{}%
\AgdaInductiveConstructor{𝟘}\AgdaSymbol{))}\AgdaSpace{}%
\AgdaOperator{\AgdaInductiveConstructor{[}}\AgdaSpace{}%
\AgdaFunction{Find'}\AgdaSpace{}%
\AgdaSymbol{\{}\AgdaBound{t}\AgdaSpace{}%
\AgdaOperator{\AgdaInductiveConstructor{+ᵤ}}\AgdaSpace{}%
\AgdaInductiveConstructor{𝟘}\AgdaSymbol{\}}\AgdaSpace{}%
\AgdaSymbol{(}\AgdaInductiveConstructor{inj₁}\AgdaSpace{}%
\AgdaBound{a}\AgdaSymbol{)}\AgdaSpace{}%
\AgdaOperator{\AgdaInductiveConstructor{]⟫}}\<%
\\
\>[0]\AgdaFunction{st}\AgdaSpace{}%
\AgdaSymbol{(}\AgdaInductiveConstructor{inj₁}\AgdaSpace{}%
\AgdaBound{x}\AgdaSymbol{)}\AgdaSpace{}%
\AgdaSymbol{(}\AgdaInductiveConstructor{swap₊}\AgdaSpace{}%
\AgdaSymbol{\{}\AgdaBound{t₁}\AgdaSymbol{\}}\AgdaSpace{}%
\AgdaSymbol{\{}\AgdaBound{t₂}\AgdaSymbol{\})}%
\>[44]\AgdaSymbol{=}\AgdaSpace{}%
\AgdaSymbol{\AgdaUnderscore{}}\AgdaSpace{}%
\AgdaOperator{\AgdaInductiveConstructor{,}}\AgdaSpace{}%
\AgdaInductiveConstructor{id⟷}\AgdaSpace{}%
\AgdaOperator{\AgdaInductiveConstructor{,}}\AgdaSpace{}%
\AgdaOperator{\AgdaInductiveConstructor{⟪}}\AgdaSpace{}%
\AgdaFunction{Enum}\AgdaSpace{}%
\AgdaSymbol{\AgdaUnderscore{}}\AgdaSpace{}%
\AgdaOperator{\AgdaInductiveConstructor{[}}\AgdaSpace{}%
\AgdaFunction{Find'}\AgdaSpace{}%
\AgdaSymbol{\{}\AgdaBound{t₂}\AgdaSpace{}%
\AgdaOperator{\AgdaInductiveConstructor{+ᵤ}}\AgdaSpace{}%
\AgdaBound{t₁}\AgdaSymbol{\}}\AgdaSpace{}%
\AgdaSymbol{(}\AgdaInductiveConstructor{inj₂}\AgdaSpace{}%
\AgdaBound{x}\AgdaSymbol{)}\AgdaSpace{}%
\AgdaOperator{\AgdaInductiveConstructor{]⟫}}\<%
\\
\>[0]\AgdaFunction{st}\AgdaSpace{}%
\AgdaSymbol{(}\AgdaInductiveConstructor{inj₂}\AgdaSpace{}%
\AgdaBound{y}\AgdaSymbol{)}\AgdaSpace{}%
\AgdaSymbol{(}\AgdaInductiveConstructor{swap₊}\AgdaSpace{}%
\AgdaSymbol{\{}\AgdaBound{t₁}\AgdaSymbol{\}}\AgdaSpace{}%
\AgdaSymbol{\{}\AgdaBound{t₂}\AgdaSymbol{\})}%
\>[44]\AgdaSymbol{=}\AgdaSpace{}%
\AgdaSymbol{\AgdaUnderscore{}}\AgdaSpace{}%
\AgdaOperator{\AgdaInductiveConstructor{,}}\AgdaSpace{}%
\AgdaInductiveConstructor{id⟷}\AgdaSpace{}%
\AgdaOperator{\AgdaInductiveConstructor{,}}\AgdaSpace{}%
\AgdaOperator{\AgdaInductiveConstructor{⟪}}\AgdaSpace{}%
\AgdaFunction{Enum}\AgdaSpace{}%
\AgdaSymbol{\AgdaUnderscore{}}\AgdaSpace{}%
\AgdaOperator{\AgdaInductiveConstructor{[}}\AgdaSpace{}%
\AgdaFunction{Find'}\AgdaSpace{}%
\AgdaSymbol{\{}\AgdaBound{t₂}\AgdaSpace{}%
\AgdaOperator{\AgdaInductiveConstructor{+ᵤ}}\AgdaSpace{}%
\AgdaBound{t₁}\AgdaSymbol{\}}\AgdaSpace{}%
\AgdaSymbol{(}\AgdaInductiveConstructor{inj₁}\AgdaSpace{}%
\AgdaBound{y}\AgdaSymbol{)}\AgdaSpace{}%
\AgdaOperator{\AgdaInductiveConstructor{]⟫}}\<%
\\
\>[0]\AgdaFunction{st}\AgdaSpace{}%
\AgdaSymbol{(}\AgdaInductiveConstructor{inj₁}\AgdaSpace{}%
\AgdaBound{x}\AgdaSymbol{)}\AgdaSpace{}%
\AgdaSymbol{(}\AgdaInductiveConstructor{assocl₊}\AgdaSpace{}%
\AgdaSymbol{\{}\AgdaBound{t₁}\AgdaSymbol{\}}\AgdaSpace{}%
\AgdaSymbol{\{}\AgdaBound{t₂}\AgdaSymbol{\}}\AgdaSpace{}%
\AgdaSymbol{\{}\AgdaBound{t₃}\AgdaSymbol{\})}%
\>[44]\AgdaSymbol{=}\AgdaSpace{}%
\AgdaSymbol{\AgdaUnderscore{}}\AgdaSpace{}%
\AgdaOperator{\AgdaInductiveConstructor{,}}\AgdaSpace{}%
\AgdaInductiveConstructor{id⟷}\AgdaSpace{}%
\AgdaOperator{\AgdaInductiveConstructor{,}}\AgdaSpace{}%
\AgdaOperator{\AgdaInductiveConstructor{⟪}}\AgdaSpace{}%
\AgdaFunction{Enum}\AgdaSpace{}%
\AgdaSymbol{\AgdaUnderscore{}}\AgdaSpace{}%
\AgdaOperator{\AgdaInductiveConstructor{[}}\AgdaSpace{}%
\AgdaFunction{Find'}\AgdaSpace{}%
\AgdaSymbol{\{(}\AgdaBound{t₁}\AgdaSpace{}%
\AgdaOperator{\AgdaInductiveConstructor{+ᵤ}}\AgdaSpace{}%
\AgdaBound{t₂}\AgdaSymbol{)}\AgdaSpace{}%
\AgdaOperator{\AgdaInductiveConstructor{+ᵤ}}\AgdaSpace{}%
\AgdaBound{t₃}\AgdaSymbol{\}}\AgdaSpace{}%
\AgdaSymbol{(}\AgdaInductiveConstructor{inj₁}\AgdaSpace{}%
\AgdaSymbol{(}\AgdaInductiveConstructor{inj₁}\AgdaSpace{}%
\AgdaBound{x}\AgdaSymbol{))}\AgdaSpace{}%
\AgdaOperator{\AgdaInductiveConstructor{]⟫}}\<%
\\
\>[0]\AgdaFunction{st}\AgdaSpace{}%
\AgdaSymbol{(}\AgdaInductiveConstructor{inj₂}\AgdaSpace{}%
\AgdaSymbol{(}\AgdaInductiveConstructor{inj₁}\AgdaSpace{}%
\AgdaBound{x}\AgdaSymbol{))}\AgdaSpace{}%
\AgdaSymbol{(}\AgdaInductiveConstructor{assocl₊}\AgdaSpace{}%
\AgdaSymbol{\{}\AgdaBound{t₁}\AgdaSymbol{\}}\AgdaSpace{}%
\AgdaSymbol{\{}\AgdaBound{t₂}\AgdaSymbol{\}}\AgdaSpace{}%
\AgdaSymbol{\{}\AgdaBound{t₃}\AgdaSymbol{\})}\AgdaSpace{}%
\AgdaSymbol{=}\AgdaSpace{}%
\AgdaSymbol{\AgdaUnderscore{}}\AgdaSpace{}%
\AgdaOperator{\AgdaInductiveConstructor{,}}\AgdaSpace{}%
\AgdaInductiveConstructor{id⟷}\AgdaSpace{}%
\AgdaOperator{\AgdaInductiveConstructor{,}}\AgdaSpace{}%
\AgdaOperator{\AgdaInductiveConstructor{⟪}}\AgdaSpace{}%
\AgdaFunction{Enum}\AgdaSpace{}%
\AgdaSymbol{\AgdaUnderscore{}}\AgdaSpace{}%
\AgdaOperator{\AgdaInductiveConstructor{[}}\AgdaSpace{}%
\AgdaFunction{Find'}\AgdaSpace{}%
\AgdaSymbol{\{(}\AgdaBound{t₁}\AgdaSpace{}%
\AgdaOperator{\AgdaInductiveConstructor{+ᵤ}}\AgdaSpace{}%
\AgdaBound{t₂}\AgdaSymbol{)}\AgdaSpace{}%
\AgdaOperator{\AgdaInductiveConstructor{+ᵤ}}\AgdaSpace{}%
\AgdaBound{t₃}\AgdaSymbol{\}}\AgdaSpace{}%
\AgdaSymbol{(}\AgdaInductiveConstructor{inj₁}\AgdaSpace{}%
\AgdaSymbol{(}\AgdaInductiveConstructor{inj₂}\AgdaSpace{}%
\AgdaBound{x}\AgdaSymbol{))}\AgdaSpace{}%
\AgdaOperator{\AgdaInductiveConstructor{]⟫}}\<%
\\
\>[0]\AgdaFunction{st}\AgdaSpace{}%
\AgdaSymbol{(}\AgdaInductiveConstructor{inj₂}\AgdaSpace{}%
\AgdaSymbol{(}\AgdaInductiveConstructor{inj₂}\AgdaSpace{}%
\AgdaBound{y}\AgdaSymbol{))}\AgdaSpace{}%
\AgdaSymbol{(}\AgdaInductiveConstructor{assocl₊}\AgdaSpace{}%
\AgdaSymbol{\{}\AgdaBound{t₁}\AgdaSymbol{\}}\AgdaSpace{}%
\AgdaSymbol{\{}\AgdaBound{t₂}\AgdaSymbol{\}}\AgdaSpace{}%
\AgdaSymbol{\{}\AgdaBound{t₃}\AgdaSymbol{\})}\AgdaSpace{}%
\AgdaSymbol{=}\AgdaSpace{}%
\AgdaSymbol{\AgdaUnderscore{}}\AgdaSpace{}%
\AgdaOperator{\AgdaInductiveConstructor{,}}\AgdaSpace{}%
\AgdaInductiveConstructor{id⟷}\AgdaSpace{}%
\AgdaOperator{\AgdaInductiveConstructor{,}}\AgdaSpace{}%
\AgdaOperator{\AgdaInductiveConstructor{⟪}}\AgdaSpace{}%
\AgdaFunction{Enum}\AgdaSpace{}%
\AgdaSymbol{\AgdaUnderscore{}}\AgdaSpace{}%
\AgdaOperator{\AgdaInductiveConstructor{[}}\AgdaSpace{}%
\AgdaFunction{Find'}\AgdaSpace{}%
\AgdaSymbol{\{(}\AgdaBound{t₁}\AgdaSpace{}%
\AgdaOperator{\AgdaInductiveConstructor{+ᵤ}}\AgdaSpace{}%
\AgdaBound{t₂}\AgdaSymbol{)}\AgdaSpace{}%
\AgdaOperator{\AgdaInductiveConstructor{+ᵤ}}\AgdaSpace{}%
\AgdaBound{t₃}\AgdaSymbol{\}}\AgdaSpace{}%
\AgdaSymbol{(}\AgdaInductiveConstructor{inj₂}\AgdaSpace{}%
\AgdaBound{y}\AgdaSymbol{)}\AgdaSpace{}%
\AgdaOperator{\AgdaInductiveConstructor{]⟫}}\<%
\\
\>[0]\AgdaFunction{st}\AgdaSpace{}%
\AgdaSymbol{(}\AgdaInductiveConstructor{inj₁}\AgdaSpace{}%
\AgdaSymbol{(}\AgdaInductiveConstructor{inj₁}\AgdaSpace{}%
\AgdaBound{x}\AgdaSymbol{))}\AgdaSpace{}%
\AgdaSymbol{(}\AgdaInductiveConstructor{assocr₊}\AgdaSpace{}%
\AgdaSymbol{\{}\AgdaBound{t₁}\AgdaSymbol{\}}\AgdaSpace{}%
\AgdaSymbol{\{}\AgdaBound{t₂}\AgdaSymbol{\}}\AgdaSpace{}%
\AgdaSymbol{\{}\AgdaBound{t₃}\AgdaSymbol{\})}\AgdaSpace{}%
\AgdaSymbol{=}\AgdaSpace{}%
\AgdaSymbol{\AgdaUnderscore{}}\AgdaSpace{}%
\AgdaOperator{\AgdaInductiveConstructor{,}}\AgdaSpace{}%
\AgdaInductiveConstructor{id⟷}\AgdaSpace{}%
\AgdaOperator{\AgdaInductiveConstructor{,}}\AgdaSpace{}%
\AgdaOperator{\AgdaInductiveConstructor{⟪}}\AgdaSpace{}%
\AgdaFunction{Enum}\AgdaSpace{}%
\AgdaSymbol{\AgdaUnderscore{}}\AgdaSpace{}%
\AgdaOperator{\AgdaInductiveConstructor{[}}\AgdaSpace{}%
\AgdaFunction{Find'}\AgdaSpace{}%
\AgdaSymbol{\{}\AgdaBound{t₁}\AgdaSpace{}%
\AgdaOperator{\AgdaInductiveConstructor{+ᵤ}}\AgdaSpace{}%
\AgdaBound{t₂}\AgdaSpace{}%
\AgdaOperator{\AgdaInductiveConstructor{+ᵤ}}\AgdaSpace{}%
\AgdaBound{t₃}\AgdaSymbol{\}}\AgdaSpace{}%
\AgdaSymbol{(}\AgdaInductiveConstructor{inj₁}\AgdaSpace{}%
\AgdaBound{x}\AgdaSymbol{)}\AgdaSpace{}%
\AgdaOperator{\AgdaInductiveConstructor{]⟫}}\<%
\\
\>[0]\AgdaFunction{st}\AgdaSpace{}%
\AgdaSymbol{(}\AgdaInductiveConstructor{inj₁}\AgdaSpace{}%
\AgdaSymbol{(}\AgdaInductiveConstructor{inj₂}\AgdaSpace{}%
\AgdaBound{y}\AgdaSymbol{))}\AgdaSpace{}%
\AgdaSymbol{(}\AgdaInductiveConstructor{assocr₊}\AgdaSpace{}%
\AgdaSymbol{\{}\AgdaBound{t₁}\AgdaSymbol{\}}\AgdaSpace{}%
\AgdaSymbol{\{}\AgdaBound{t₂}\AgdaSymbol{\}}\AgdaSpace{}%
\AgdaSymbol{\{}\AgdaBound{t₃}\AgdaSymbol{\})}\AgdaSpace{}%
\AgdaSymbol{=}\AgdaSpace{}%
\AgdaSymbol{\AgdaUnderscore{}}\AgdaSpace{}%
\AgdaOperator{\AgdaInductiveConstructor{,}}\AgdaSpace{}%
\AgdaInductiveConstructor{id⟷}\AgdaSpace{}%
\AgdaOperator{\AgdaInductiveConstructor{,}}\AgdaSpace{}%
\AgdaOperator{\AgdaInductiveConstructor{⟪}}\AgdaSpace{}%
\AgdaFunction{Enum}\AgdaSpace{}%
\AgdaSymbol{\AgdaUnderscore{}}\AgdaSpace{}%
\AgdaOperator{\AgdaInductiveConstructor{[}}\AgdaSpace{}%
\AgdaFunction{Find'}\AgdaSpace{}%
\AgdaSymbol{\{}\AgdaBound{t₁}\AgdaSpace{}%
\AgdaOperator{\AgdaInductiveConstructor{+ᵤ}}\AgdaSpace{}%
\AgdaBound{t₂}\AgdaSpace{}%
\AgdaOperator{\AgdaInductiveConstructor{+ᵤ}}\AgdaSpace{}%
\AgdaBound{t₃}\AgdaSymbol{\}}\AgdaSpace{}%
\AgdaSymbol{(}\AgdaInductiveConstructor{inj₂}\AgdaSpace{}%
\AgdaSymbol{(}\AgdaInductiveConstructor{inj₁}\AgdaSpace{}%
\AgdaBound{y}\AgdaSymbol{))}\AgdaSpace{}%
\AgdaOperator{\AgdaInductiveConstructor{]⟫}}\<%
\\
\>[0]\AgdaFunction{st}\AgdaSpace{}%
\AgdaSymbol{(}\AgdaInductiveConstructor{inj₂}\AgdaSpace{}%
\AgdaBound{y}\AgdaSymbol{)}\AgdaSpace{}%
\AgdaSymbol{(}\AgdaInductiveConstructor{assocr₊}\AgdaSpace{}%
\AgdaSymbol{\{}\AgdaBound{t₁}\AgdaSymbol{\}}\AgdaSpace{}%
\AgdaSymbol{\{}\AgdaBound{t₂}\AgdaSymbol{\}}\AgdaSpace{}%
\AgdaSymbol{\{}\AgdaBound{t₃}\AgdaSymbol{\})}%
\>[44]\AgdaSymbol{=}\AgdaSpace{}%
\AgdaSymbol{\AgdaUnderscore{}}\AgdaSpace{}%
\AgdaOperator{\AgdaInductiveConstructor{,}}\AgdaSpace{}%
\AgdaInductiveConstructor{id⟷}\AgdaSpace{}%
\AgdaOperator{\AgdaInductiveConstructor{,}}\AgdaSpace{}%
\AgdaOperator{\AgdaInductiveConstructor{⟪}}\AgdaSpace{}%
\AgdaFunction{Enum}\AgdaSpace{}%
\AgdaSymbol{\AgdaUnderscore{}}\AgdaSpace{}%
\AgdaOperator{\AgdaInductiveConstructor{[}}\AgdaSpace{}%
\AgdaFunction{Find'}\AgdaSpace{}%
\AgdaSymbol{\{}\AgdaBound{t₁}\AgdaSpace{}%
\AgdaOperator{\AgdaInductiveConstructor{+ᵤ}}\AgdaSpace{}%
\AgdaBound{t₂}\AgdaSpace{}%
\AgdaOperator{\AgdaInductiveConstructor{+ᵤ}}\AgdaSpace{}%
\AgdaBound{t₃}\AgdaSymbol{\}}\AgdaSpace{}%
\AgdaSymbol{(}\AgdaInductiveConstructor{inj₂}\AgdaSpace{}%
\AgdaSymbol{(}\AgdaInductiveConstructor{inj₂}\AgdaSpace{}%
\AgdaBound{y}\AgdaSymbol{))}\AgdaSpace{}%
\AgdaOperator{\AgdaInductiveConstructor{]⟫}}\<%
\\
\>[0]\AgdaFunction{st}\AgdaSpace{}%
\AgdaSymbol{(}\AgdaInductiveConstructor{tt}\AgdaSpace{}%
\AgdaOperator{\AgdaInductiveConstructor{,}}\AgdaSpace{}%
\AgdaBound{y}\AgdaSymbol{)}\AgdaSpace{}%
\AgdaInductiveConstructor{unite⋆l}%
\>[44]\AgdaSymbol{=}\AgdaSpace{}%
\AgdaSymbol{\AgdaUnderscore{}}\AgdaSpace{}%
\AgdaOperator{\AgdaInductiveConstructor{,}}\AgdaSpace{}%
\AgdaInductiveConstructor{id⟷}\AgdaSpace{}%
\AgdaOperator{\AgdaInductiveConstructor{,}}\AgdaSpace{}%
\AgdaOperator{\AgdaInductiveConstructor{⟪}}\AgdaSpace{}%
\AgdaFunction{Enum}\AgdaSpace{}%
\AgdaSymbol{\AgdaUnderscore{}}\AgdaSpace{}%
\AgdaOperator{\AgdaInductiveConstructor{[}}\AgdaSpace{}%
\AgdaFunction{Find'}\AgdaSpace{}%
\AgdaBound{y}\AgdaSpace{}%
\AgdaOperator{\AgdaInductiveConstructor{]⟫}}\<%
\\
\>[0]\AgdaFunction{st}\AgdaSpace{}%
\AgdaBound{a}\AgdaSpace{}%
\AgdaInductiveConstructor{uniti⋆l}%
\>[44]\AgdaSymbol{=}\AgdaSpace{}%
\AgdaSymbol{\AgdaUnderscore{}}\AgdaSpace{}%
\AgdaOperator{\AgdaInductiveConstructor{,}}\AgdaSpace{}%
\AgdaInductiveConstructor{id⟷}\AgdaSpace{}%
\AgdaOperator{\AgdaInductiveConstructor{,}}\AgdaSpace{}%
\AgdaOperator{\AgdaInductiveConstructor{⟪}}\AgdaSpace{}%
\AgdaFunction{Enum}\AgdaSpace{}%
\AgdaSymbol{\AgdaUnderscore{}}\AgdaSpace{}%
\AgdaOperator{\AgdaInductiveConstructor{[}}\AgdaSpace{}%
\AgdaFunction{Find'}\AgdaSpace{}%
\AgdaSymbol{(}\AgdaInductiveConstructor{tt}\AgdaSpace{}%
\AgdaOperator{\AgdaInductiveConstructor{,}}\AgdaSpace{}%
\AgdaBound{a}\AgdaSymbol{)}\AgdaSpace{}%
\AgdaOperator{\AgdaInductiveConstructor{]⟫}}\<%
\\
\>[0]\AgdaFunction{st}\AgdaSpace{}%
\AgdaSymbol{(}\AgdaBound{x}\AgdaSpace{}%
\AgdaOperator{\AgdaInductiveConstructor{,}}\AgdaSpace{}%
\AgdaInductiveConstructor{tt}\AgdaSymbol{)}\AgdaSpace{}%
\AgdaInductiveConstructor{unite⋆r}%
\>[44]\AgdaSymbol{=}\AgdaSpace{}%
\AgdaSymbol{\AgdaUnderscore{}}\AgdaSpace{}%
\AgdaOperator{\AgdaInductiveConstructor{,}}\AgdaSpace{}%
\AgdaInductiveConstructor{id⟷}\AgdaSpace{}%
\AgdaOperator{\AgdaInductiveConstructor{,}}\AgdaSpace{}%
\AgdaOperator{\AgdaInductiveConstructor{⟪}}\AgdaSpace{}%
\AgdaFunction{Enum}\AgdaSpace{}%
\AgdaSymbol{\AgdaUnderscore{}}\AgdaSpace{}%
\AgdaOperator{\AgdaInductiveConstructor{[}}\AgdaSpace{}%
\AgdaFunction{Find'}\AgdaSpace{}%
\AgdaBound{x}\AgdaSpace{}%
\AgdaOperator{\AgdaInductiveConstructor{]⟫}}\<%
\\
\>[0]\AgdaFunction{st}\AgdaSpace{}%
\AgdaBound{a}\AgdaSpace{}%
\AgdaInductiveConstructor{uniti⋆r}%
\>[44]\AgdaSymbol{=}\AgdaSpace{}%
\AgdaSymbol{\AgdaUnderscore{}}\AgdaSpace{}%
\AgdaOperator{\AgdaInductiveConstructor{,}}\AgdaSpace{}%
\AgdaInductiveConstructor{id⟷}\AgdaSpace{}%
\AgdaOperator{\AgdaInductiveConstructor{,}}\AgdaSpace{}%
\AgdaOperator{\AgdaInductiveConstructor{⟪}}\AgdaSpace{}%
\AgdaFunction{Enum}\AgdaSpace{}%
\AgdaSymbol{\AgdaUnderscore{}}\AgdaSpace{}%
\AgdaOperator{\AgdaInductiveConstructor{[}}\AgdaSpace{}%
\AgdaFunction{Find'}\AgdaSpace{}%
\AgdaSymbol{(}\AgdaBound{a}\AgdaSpace{}%
\AgdaOperator{\AgdaInductiveConstructor{,}}\AgdaSpace{}%
\AgdaInductiveConstructor{tt}\AgdaSymbol{)}\AgdaSpace{}%
\AgdaOperator{\AgdaInductiveConstructor{]⟫}}\<%
\\
\>[0]\AgdaFunction{st}\AgdaSpace{}%
\AgdaSymbol{(}\AgdaBound{x}\AgdaSpace{}%
\AgdaOperator{\AgdaInductiveConstructor{,}}\AgdaSpace{}%
\AgdaBound{y}\AgdaSymbol{)}\AgdaSpace{}%
\AgdaInductiveConstructor{swap⋆}%
\>[44]\AgdaSymbol{=}\AgdaSpace{}%
\AgdaSymbol{\AgdaUnderscore{}}\AgdaSpace{}%
\AgdaOperator{\AgdaInductiveConstructor{,}}\AgdaSpace{}%
\AgdaInductiveConstructor{id⟷}\AgdaSpace{}%
\AgdaOperator{\AgdaInductiveConstructor{,}}\AgdaSpace{}%
\AgdaOperator{\AgdaInductiveConstructor{⟪}}\AgdaSpace{}%
\AgdaFunction{Enum}\AgdaSpace{}%
\AgdaSymbol{\AgdaUnderscore{}}\AgdaSpace{}%
\AgdaOperator{\AgdaInductiveConstructor{[}}\AgdaSpace{}%
\AgdaFunction{Find'}\AgdaSpace{}%
\AgdaSymbol{(}\AgdaBound{y}\AgdaSpace{}%
\AgdaOperator{\AgdaInductiveConstructor{,}}\AgdaSpace{}%
\AgdaBound{x}\AgdaSymbol{)}\AgdaSpace{}%
\AgdaOperator{\AgdaInductiveConstructor{]⟫}}\<%
\\
\>[0]\AgdaFunction{st}\AgdaSpace{}%
\AgdaSymbol{(}\AgdaBound{x}\AgdaSpace{}%
\AgdaOperator{\AgdaInductiveConstructor{,}}\AgdaSpace{}%
\AgdaBound{y}\AgdaSpace{}%
\AgdaOperator{\AgdaInductiveConstructor{,}}\AgdaSpace{}%
\AgdaBound{z}\AgdaSymbol{)}\AgdaSpace{}%
\AgdaInductiveConstructor{assocl⋆}%
\>[44]\AgdaSymbol{=}\AgdaSpace{}%
\AgdaSymbol{\AgdaUnderscore{}}\AgdaSpace{}%
\AgdaOperator{\AgdaInductiveConstructor{,}}\AgdaSpace{}%
\AgdaInductiveConstructor{id⟷}\AgdaSpace{}%
\AgdaOperator{\AgdaInductiveConstructor{,}}\AgdaSpace{}%
\AgdaOperator{\AgdaInductiveConstructor{⟪}}\AgdaSpace{}%
\AgdaFunction{Enum}\AgdaSpace{}%
\AgdaSymbol{\AgdaUnderscore{}}\AgdaSpace{}%
\AgdaOperator{\AgdaInductiveConstructor{[}}\AgdaSpace{}%
\AgdaFunction{Find'}\AgdaSpace{}%
\AgdaSymbol{((}\AgdaBound{x}\AgdaSpace{}%
\AgdaOperator{\AgdaInductiveConstructor{,}}\AgdaSpace{}%
\AgdaBound{y}\AgdaSymbol{)}\AgdaSpace{}%
\AgdaOperator{\AgdaInductiveConstructor{,}}\AgdaSpace{}%
\AgdaBound{z}\AgdaSymbol{)}\AgdaSpace{}%
\AgdaOperator{\AgdaInductiveConstructor{]⟫}}\<%
\\
\>[0]\AgdaFunction{st}\AgdaSpace{}%
\AgdaSymbol{((}\AgdaBound{x}\AgdaSpace{}%
\AgdaOperator{\AgdaInductiveConstructor{,}}\AgdaSpace{}%
\AgdaBound{y}\AgdaSymbol{)}\AgdaSpace{}%
\AgdaOperator{\AgdaInductiveConstructor{,}}\AgdaSpace{}%
\AgdaBound{z}\AgdaSymbol{)}\AgdaSpace{}%
\AgdaInductiveConstructor{assocr⋆}%
\>[44]\AgdaSymbol{=}\AgdaSpace{}%
\AgdaSymbol{\AgdaUnderscore{}}\AgdaSpace{}%
\AgdaOperator{\AgdaInductiveConstructor{,}}\AgdaSpace{}%
\AgdaInductiveConstructor{id⟷}\AgdaSpace{}%
\AgdaOperator{\AgdaInductiveConstructor{,}}\AgdaSpace{}%
\AgdaOperator{\AgdaInductiveConstructor{⟪}}\AgdaSpace{}%
\AgdaFunction{Enum}\AgdaSpace{}%
\AgdaSymbol{\AgdaUnderscore{}}\AgdaSpace{}%
\AgdaOperator{\AgdaInductiveConstructor{[}}\AgdaSpace{}%
\AgdaFunction{Find'}\AgdaSpace{}%
\AgdaSymbol{(}\AgdaBound{x}\AgdaSpace{}%
\AgdaOperator{\AgdaInductiveConstructor{,}}\AgdaSpace{}%
\AgdaBound{y}\AgdaSpace{}%
\AgdaOperator{\AgdaInductiveConstructor{,}}\AgdaSpace{}%
\AgdaBound{z}\AgdaSymbol{)}\AgdaSpace{}%
\AgdaOperator{\AgdaInductiveConstructor{]⟫}}\<%
\\
\>[0]\AgdaFunction{st}\AgdaSpace{}%
\AgdaSymbol{(}\AgdaInductiveConstructor{inj₁}\AgdaSpace{}%
\AgdaBound{x}\AgdaSpace{}%
\AgdaOperator{\AgdaInductiveConstructor{,}}\AgdaSpace{}%
\AgdaBound{y}\AgdaSymbol{)}\AgdaSpace{}%
\AgdaSymbol{(}\AgdaInductiveConstructor{dist}\AgdaSpace{}%
\AgdaSymbol{\{}\AgdaBound{t₁}\AgdaSymbol{\}}\AgdaSpace{}%
\AgdaSymbol{\{}\AgdaBound{t₂}\AgdaSymbol{\}}\AgdaSpace{}%
\AgdaSymbol{\{}\AgdaBound{t₃}\AgdaSymbol{\})}%
\>[44]\AgdaSymbol{=}\AgdaSpace{}%
\AgdaSymbol{\AgdaUnderscore{}}\AgdaSpace{}%
\AgdaOperator{\AgdaInductiveConstructor{,}}\AgdaSpace{}%
\AgdaInductiveConstructor{id⟷}\AgdaSpace{}%
\AgdaOperator{\AgdaInductiveConstructor{,}}\AgdaSpace{}%
\AgdaOperator{\AgdaInductiveConstructor{⟪}}\AgdaSpace{}%
\AgdaFunction{Enum}\AgdaSpace{}%
\AgdaSymbol{\AgdaUnderscore{}}\AgdaSpace{}%
\AgdaOperator{\AgdaInductiveConstructor{[}}\AgdaSpace{}%
\AgdaFunction{Find'}\AgdaSpace{}%
\AgdaSymbol{\{}\AgdaBound{t₁}\AgdaSpace{}%
\AgdaOperator{\AgdaInductiveConstructor{×ᵤ}}\AgdaSpace{}%
\AgdaBound{t₃}\AgdaSpace{}%
\AgdaOperator{\AgdaInductiveConstructor{+ᵤ}}\AgdaSpace{}%
\AgdaBound{t₂}\AgdaSpace{}%
\AgdaOperator{\AgdaInductiveConstructor{×ᵤ}}\AgdaSpace{}%
\AgdaBound{t₃}\AgdaSymbol{\}}\AgdaSpace{}%
\AgdaSymbol{(}\AgdaInductiveConstructor{inj₁}\AgdaSpace{}%
\AgdaSymbol{(}\AgdaBound{x}\AgdaSpace{}%
\AgdaOperator{\AgdaInductiveConstructor{,}}\AgdaSpace{}%
\AgdaBound{y}\AgdaSymbol{))}\AgdaSpace{}%
\AgdaOperator{\AgdaInductiveConstructor{]⟫}}\<%
\\
\>[0]\AgdaFunction{st}\AgdaSpace{}%
\AgdaSymbol{(}\AgdaInductiveConstructor{inj₂}\AgdaSpace{}%
\AgdaBound{x}\AgdaSpace{}%
\AgdaOperator{\AgdaInductiveConstructor{,}}\AgdaSpace{}%
\AgdaBound{y}\AgdaSymbol{)}\AgdaSpace{}%
\AgdaSymbol{(}\AgdaInductiveConstructor{dist}\AgdaSpace{}%
\AgdaSymbol{\{}\AgdaBound{t₁}\AgdaSymbol{\}}\AgdaSpace{}%
\AgdaSymbol{\{}\AgdaBound{t₂}\AgdaSymbol{\}}\AgdaSpace{}%
\AgdaSymbol{\{}\AgdaBound{t₃}\AgdaSymbol{\})}%
\>[44]\AgdaSymbol{=}\AgdaSpace{}%
\AgdaSymbol{\AgdaUnderscore{}}\AgdaSpace{}%
\AgdaOperator{\AgdaInductiveConstructor{,}}\AgdaSpace{}%
\AgdaInductiveConstructor{id⟷}\AgdaSpace{}%
\AgdaOperator{\AgdaInductiveConstructor{,}}\AgdaSpace{}%
\AgdaOperator{\AgdaInductiveConstructor{⟪}}\AgdaSpace{}%
\AgdaFunction{Enum}\AgdaSpace{}%
\AgdaSymbol{\AgdaUnderscore{}}\AgdaSpace{}%
\AgdaOperator{\AgdaInductiveConstructor{[}}\AgdaSpace{}%
\AgdaFunction{Find'}\AgdaSpace{}%
\AgdaSymbol{\{}\AgdaBound{t₁}\AgdaSpace{}%
\AgdaOperator{\AgdaInductiveConstructor{×ᵤ}}\AgdaSpace{}%
\AgdaBound{t₃}\AgdaSpace{}%
\AgdaOperator{\AgdaInductiveConstructor{+ᵤ}}\AgdaSpace{}%
\AgdaBound{t₂}\AgdaSpace{}%
\AgdaOperator{\AgdaInductiveConstructor{×ᵤ}}\AgdaSpace{}%
\AgdaBound{t₃}\AgdaSymbol{\}}\AgdaSpace{}%
\AgdaSymbol{(}\AgdaInductiveConstructor{inj₂}\AgdaSpace{}%
\AgdaSymbol{(}\AgdaBound{x}\AgdaSpace{}%
\AgdaOperator{\AgdaInductiveConstructor{,}}\AgdaSpace{}%
\AgdaBound{y}\AgdaSymbol{))}\AgdaSpace{}%
\AgdaOperator{\AgdaInductiveConstructor{]⟫}}\<%
\\
\>[0]\AgdaFunction{st}\AgdaSpace{}%
\AgdaSymbol{(}\AgdaInductiveConstructor{inj₁}\AgdaSpace{}%
\AgdaSymbol{(}\AgdaBound{x}\AgdaSpace{}%
\AgdaOperator{\AgdaInductiveConstructor{,}}\AgdaSpace{}%
\AgdaBound{y}\AgdaSymbol{))}\AgdaSpace{}%
\AgdaSymbol{(}\AgdaInductiveConstructor{factor}\AgdaSpace{}%
\AgdaSymbol{\{}\AgdaBound{t₁}\AgdaSymbol{\}}\AgdaSpace{}%
\AgdaSymbol{\{}\AgdaBound{t₂}\AgdaSymbol{\}}\AgdaSpace{}%
\AgdaSymbol{\{}\AgdaBound{t₃}\AgdaSymbol{\})}%
\>[44]\AgdaSymbol{=}\AgdaSpace{}%
\AgdaSymbol{\AgdaUnderscore{}}\AgdaSpace{}%
\AgdaOperator{\AgdaInductiveConstructor{,}}\AgdaSpace{}%
\AgdaInductiveConstructor{id⟷}\AgdaSpace{}%
\AgdaOperator{\AgdaInductiveConstructor{,}}\AgdaSpace{}%
\AgdaOperator{\AgdaInductiveConstructor{⟪}}\AgdaSpace{}%
\AgdaFunction{Enum}\AgdaSpace{}%
\AgdaSymbol{\AgdaUnderscore{}}\AgdaSpace{}%
\AgdaOperator{\AgdaInductiveConstructor{[}}\AgdaSpace{}%
\AgdaFunction{Find'}\AgdaSpace{}%
\AgdaSymbol{\{(}\AgdaBound{t₁}\AgdaSpace{}%
\AgdaOperator{\AgdaInductiveConstructor{+ᵤ}}\AgdaSpace{}%
\AgdaBound{t₂}\AgdaSymbol{)}\AgdaSpace{}%
\AgdaOperator{\AgdaInductiveConstructor{×ᵤ}}\AgdaSpace{}%
\AgdaBound{t₃}\AgdaSymbol{\}}\AgdaSpace{}%
\AgdaSymbol{(}\AgdaInductiveConstructor{inj₁}\AgdaSpace{}%
\AgdaBound{x}\AgdaSpace{}%
\AgdaOperator{\AgdaInductiveConstructor{,}}\AgdaSpace{}%
\AgdaBound{y}\AgdaSymbol{)}\AgdaSpace{}%
\AgdaOperator{\AgdaInductiveConstructor{]⟫}}\<%
\\
\>[0]\AgdaFunction{st}\AgdaSpace{}%
\AgdaSymbol{(}\AgdaInductiveConstructor{inj₂}\AgdaSpace{}%
\AgdaSymbol{(}\AgdaBound{y}\AgdaSpace{}%
\AgdaOperator{\AgdaInductiveConstructor{,}}\AgdaSpace{}%
\AgdaBound{z}\AgdaSymbol{))}\AgdaSpace{}%
\AgdaSymbol{(}\AgdaInductiveConstructor{factor}\AgdaSpace{}%
\AgdaSymbol{\{}\AgdaBound{t₁}\AgdaSymbol{\}}\AgdaSpace{}%
\AgdaSymbol{\{}\AgdaBound{t₂}\AgdaSymbol{\}}\AgdaSpace{}%
\AgdaSymbol{\{}\AgdaBound{t₃}\AgdaSymbol{\})}%
\>[44]\AgdaSymbol{=}\AgdaSpace{}%
\AgdaSymbol{\AgdaUnderscore{}}\AgdaSpace{}%
\AgdaOperator{\AgdaInductiveConstructor{,}}\AgdaSpace{}%
\AgdaInductiveConstructor{id⟷}\AgdaSpace{}%
\AgdaOperator{\AgdaInductiveConstructor{,}}\AgdaSpace{}%
\AgdaOperator{\AgdaInductiveConstructor{⟪}}\AgdaSpace{}%
\AgdaFunction{Enum}\AgdaSpace{}%
\AgdaSymbol{\AgdaUnderscore{}}\AgdaSpace{}%
\AgdaOperator{\AgdaInductiveConstructor{[}}\AgdaSpace{}%
\AgdaFunction{Find'}\AgdaSpace{}%
\AgdaSymbol{\{(}\AgdaBound{t₁}\AgdaSpace{}%
\AgdaOperator{\AgdaInductiveConstructor{+ᵤ}}\AgdaSpace{}%
\AgdaBound{t₂}\AgdaSymbol{)}\AgdaSpace{}%
\AgdaOperator{\AgdaInductiveConstructor{×ᵤ}}\AgdaSpace{}%
\AgdaBound{t₃}\AgdaSymbol{\}}\AgdaSpace{}%
\AgdaSymbol{(}\AgdaInductiveConstructor{inj₂}\AgdaSpace{}%
\AgdaBound{y}\AgdaSpace{}%
\AgdaOperator{\AgdaInductiveConstructor{,}}\AgdaSpace{}%
\AgdaBound{z}\AgdaSymbol{)}\AgdaSpace{}%
\AgdaOperator{\AgdaInductiveConstructor{]⟫}}\<%
\\
\>[0]\AgdaFunction{st}\AgdaSpace{}%
\AgdaSymbol{(}\AgdaBound{x}\AgdaSpace{}%
\AgdaOperator{\AgdaInductiveConstructor{,}}\AgdaSpace{}%
\AgdaInductiveConstructor{inj₁}\AgdaSpace{}%
\AgdaBound{y}\AgdaSymbol{)}\AgdaSpace{}%
\AgdaSymbol{(}\AgdaInductiveConstructor{distl}\AgdaSpace{}%
\AgdaSymbol{\{}\AgdaBound{t₁}\AgdaSymbol{\}}\AgdaSpace{}%
\AgdaSymbol{\{}\AgdaBound{t₂}\AgdaSymbol{\}}\AgdaSpace{}%
\AgdaSymbol{\{}\AgdaBound{t₃}\AgdaSymbol{\})}%
\>[44]\AgdaSymbol{=}\AgdaSpace{}%
\AgdaSymbol{\AgdaUnderscore{}}\AgdaSpace{}%
\AgdaOperator{\AgdaInductiveConstructor{,}}\AgdaSpace{}%
\AgdaInductiveConstructor{id⟷}\AgdaSpace{}%
\AgdaOperator{\AgdaInductiveConstructor{,}}\AgdaSpace{}%
\AgdaOperator{\AgdaInductiveConstructor{⟪}}\AgdaSpace{}%
\AgdaFunction{Enum}\AgdaSpace{}%
\AgdaSymbol{\AgdaUnderscore{}}\AgdaSpace{}%
\AgdaOperator{\AgdaInductiveConstructor{[}}\AgdaSpace{}%
\AgdaFunction{Find'}\AgdaSpace{}%
\AgdaSymbol{\{(}\AgdaBound{t₁}\AgdaSpace{}%
\AgdaOperator{\AgdaInductiveConstructor{×ᵤ}}\AgdaSpace{}%
\AgdaBound{t₂}\AgdaSymbol{)}\AgdaSpace{}%
\AgdaOperator{\AgdaInductiveConstructor{+ᵤ}}\AgdaSpace{}%
\AgdaSymbol{(}\AgdaBound{t₁}\AgdaSpace{}%
\AgdaOperator{\AgdaInductiveConstructor{×ᵤ}}\AgdaSpace{}%
\AgdaBound{t₃}\AgdaSymbol{)\}}\AgdaSpace{}%
\AgdaSymbol{(}\AgdaInductiveConstructor{inj₁}\AgdaSpace{}%
\AgdaSymbol{(}\AgdaBound{x}\AgdaSpace{}%
\AgdaOperator{\AgdaInductiveConstructor{,}}\AgdaSpace{}%
\AgdaBound{y}\AgdaSymbol{))}\AgdaSpace{}%
\AgdaOperator{\AgdaInductiveConstructor{]⟫}}\<%
\\
\>[0]\AgdaFunction{st}\AgdaSpace{}%
\AgdaSymbol{(}\AgdaBound{x}\AgdaSpace{}%
\AgdaOperator{\AgdaInductiveConstructor{,}}\AgdaSpace{}%
\AgdaInductiveConstructor{inj₂}\AgdaSpace{}%
\AgdaBound{y}\AgdaSymbol{)}\AgdaSpace{}%
\AgdaSymbol{(}\AgdaInductiveConstructor{distl}\AgdaSpace{}%
\AgdaSymbol{\{}\AgdaBound{t₁}\AgdaSymbol{\}}\AgdaSpace{}%
\AgdaSymbol{\{}\AgdaBound{t₂}\AgdaSymbol{\}}\AgdaSpace{}%
\AgdaSymbol{\{}\AgdaBound{t₃}\AgdaSymbol{\})}%
\>[44]\AgdaSymbol{=}\AgdaSpace{}%
\AgdaSymbol{\AgdaUnderscore{}}\AgdaSpace{}%
\AgdaOperator{\AgdaInductiveConstructor{,}}\AgdaSpace{}%
\AgdaInductiveConstructor{id⟷}\AgdaSpace{}%
\AgdaOperator{\AgdaInductiveConstructor{,}}\AgdaSpace{}%
\AgdaOperator{\AgdaInductiveConstructor{⟪}}\AgdaSpace{}%
\AgdaFunction{Enum}\AgdaSpace{}%
\AgdaSymbol{\AgdaUnderscore{}}\AgdaSpace{}%
\AgdaOperator{\AgdaInductiveConstructor{[}}\AgdaSpace{}%
\AgdaFunction{Find'}\AgdaSpace{}%
\AgdaSymbol{\{(}\AgdaBound{t₁}\AgdaSpace{}%
\AgdaOperator{\AgdaInductiveConstructor{×ᵤ}}\AgdaSpace{}%
\AgdaBound{t₂}\AgdaSymbol{)}\AgdaSpace{}%
\AgdaOperator{\AgdaInductiveConstructor{+ᵤ}}\AgdaSpace{}%
\AgdaSymbol{(}\AgdaBound{t₁}\AgdaSpace{}%
\AgdaOperator{\AgdaInductiveConstructor{×ᵤ}}\AgdaSpace{}%
\AgdaBound{t₃}\AgdaSymbol{)\}}\AgdaSpace{}%
\AgdaSymbol{(}\AgdaInductiveConstructor{inj₂}\AgdaSpace{}%
\AgdaSymbol{(}\AgdaBound{x}\AgdaSpace{}%
\AgdaOperator{\AgdaInductiveConstructor{,}}\AgdaSpace{}%
\AgdaBound{y}\AgdaSymbol{))}\AgdaSpace{}%
\AgdaOperator{\AgdaInductiveConstructor{]⟫}}\<%
\\
\>[0]\AgdaFunction{st}\AgdaSpace{}%
\AgdaSymbol{(}\AgdaInductiveConstructor{inj₁}\AgdaSpace{}%
\AgdaSymbol{(}\AgdaBound{x}\AgdaSpace{}%
\AgdaOperator{\AgdaInductiveConstructor{,}}\AgdaSpace{}%
\AgdaBound{y}\AgdaSymbol{))}\AgdaSpace{}%
\AgdaSymbol{(}\AgdaInductiveConstructor{factorl}\AgdaSpace{}%
\AgdaSymbol{\{}\AgdaBound{t₁}\AgdaSymbol{\}}\AgdaSpace{}%
\AgdaSymbol{\{}\AgdaBound{t₂}\AgdaSymbol{\}}\AgdaSpace{}%
\AgdaSymbol{\{}\AgdaBound{t₃}\AgdaSymbol{\})}%
\>[44]\AgdaSymbol{=}\AgdaSpace{}%
\AgdaSymbol{\AgdaUnderscore{}}\AgdaSpace{}%
\AgdaOperator{\AgdaInductiveConstructor{,}}\AgdaSpace{}%
\AgdaInductiveConstructor{id⟷}\AgdaSpace{}%
\AgdaOperator{\AgdaInductiveConstructor{,}}\AgdaSpace{}%
\AgdaOperator{\AgdaInductiveConstructor{⟪}}\AgdaSpace{}%
\AgdaFunction{Enum}\AgdaSpace{}%
\AgdaSymbol{\AgdaUnderscore{}}\AgdaSpace{}%
\AgdaOperator{\AgdaInductiveConstructor{[}}\AgdaSpace{}%
\AgdaFunction{Find'}\AgdaSpace{}%
\AgdaSymbol{\{}\AgdaBound{t₁}\AgdaSpace{}%
\AgdaOperator{\AgdaInductiveConstructor{×ᵤ}}\AgdaSpace{}%
\AgdaSymbol{(}\AgdaBound{t₂}\AgdaSpace{}%
\AgdaOperator{\AgdaInductiveConstructor{+ᵤ}}\AgdaSpace{}%
\AgdaBound{t₃}\AgdaSymbol{)\}}\AgdaSpace{}%
\AgdaSymbol{(}\AgdaBound{x}\AgdaSpace{}%
\AgdaOperator{\AgdaInductiveConstructor{,}}\AgdaSpace{}%
\AgdaInductiveConstructor{inj₁}\AgdaSpace{}%
\AgdaBound{y}\AgdaSymbol{)}\AgdaSpace{}%
\AgdaOperator{\AgdaInductiveConstructor{]⟫}}\<%
\\
\>[0]\AgdaFunction{st}\AgdaSpace{}%
\AgdaSymbol{(}\AgdaInductiveConstructor{inj₂}\AgdaSpace{}%
\AgdaSymbol{(}\AgdaBound{x}\AgdaSpace{}%
\AgdaOperator{\AgdaInductiveConstructor{,}}\AgdaSpace{}%
\AgdaBound{z}\AgdaSymbol{))}\AgdaSpace{}%
\AgdaSymbol{(}\AgdaInductiveConstructor{factorl}\AgdaSpace{}%
\AgdaSymbol{\{}\AgdaBound{t₁}\AgdaSymbol{\}}\AgdaSpace{}%
\AgdaSymbol{\{}\AgdaBound{t₂}\AgdaSymbol{\}}\AgdaSpace{}%
\AgdaSymbol{\{}\AgdaBound{t₃}\AgdaSymbol{\})}%
\>[44]\AgdaSymbol{=}\AgdaSpace{}%
\AgdaSymbol{\AgdaUnderscore{}}\AgdaSpace{}%
\AgdaOperator{\AgdaInductiveConstructor{,}}\AgdaSpace{}%
\AgdaInductiveConstructor{id⟷}\AgdaSpace{}%
\AgdaOperator{\AgdaInductiveConstructor{,}}\AgdaSpace{}%
\AgdaOperator{\AgdaInductiveConstructor{⟪}}\AgdaSpace{}%
\AgdaFunction{Enum}\AgdaSpace{}%
\AgdaSymbol{\AgdaUnderscore{}}\AgdaSpace{}%
\AgdaOperator{\AgdaInductiveConstructor{[}}\AgdaSpace{}%
\AgdaFunction{Find'}\AgdaSpace{}%
\AgdaSymbol{\{}\AgdaBound{t₁}\AgdaSpace{}%
\AgdaOperator{\AgdaInductiveConstructor{×ᵤ}}\AgdaSpace{}%
\AgdaSymbol{(}\AgdaBound{t₂}\AgdaSpace{}%
\AgdaOperator{\AgdaInductiveConstructor{+ᵤ}}\AgdaSpace{}%
\AgdaBound{t₃}\AgdaSymbol{)\}}\AgdaSpace{}%
\AgdaSymbol{(}\AgdaBound{x}\AgdaSpace{}%
\AgdaOperator{\AgdaInductiveConstructor{,}}\AgdaSpace{}%
\AgdaInductiveConstructor{inj₂}\AgdaSpace{}%
\AgdaBound{z}\AgdaSymbol{)}\AgdaSpace{}%
\AgdaOperator{\AgdaInductiveConstructor{]⟫}}\<%
\\
\>[0]\AgdaFunction{st}\AgdaSpace{}%
\AgdaBound{a}\AgdaSpace{}%
\AgdaInductiveConstructor{id⟷}%
\>[44]\AgdaSymbol{=}\AgdaSpace{}%
\AgdaSymbol{\AgdaUnderscore{}}\AgdaSpace{}%
\AgdaOperator{\AgdaInductiveConstructor{,}}\AgdaSpace{}%
\AgdaInductiveConstructor{id⟷}\AgdaSpace{}%
\AgdaOperator{\AgdaInductiveConstructor{,}}\AgdaSpace{}%
\AgdaOperator{\AgdaInductiveConstructor{⟪}}\AgdaSpace{}%
\AgdaFunction{Enum}\AgdaSpace{}%
\AgdaSymbol{\AgdaUnderscore{}}\AgdaSpace{}%
\AgdaOperator{\AgdaInductiveConstructor{[}}\AgdaSpace{}%
\AgdaFunction{Find'}\AgdaSpace{}%
\AgdaBound{a}\AgdaSpace{}%
\AgdaOperator{\AgdaInductiveConstructor{]⟫}}\<%
\\
\>[0]\AgdaFunction{st}\AgdaSpace{}%
\AgdaBound{a}\AgdaSpace{}%
\AgdaSymbol{(}\AgdaInductiveConstructor{id⟷}\AgdaSpace{}%
\AgdaOperator{\AgdaInductiveConstructor{⊚}}\AgdaSpace{}%
\AgdaBound{c}\AgdaSymbol{)}%
\>[44]\AgdaSymbol{=}\AgdaSpace{}%
\AgdaSymbol{\AgdaUnderscore{}}\AgdaSpace{}%
\AgdaOperator{\AgdaInductiveConstructor{,}}\AgdaSpace{}%
\AgdaBound{c}\AgdaSpace{}%
\AgdaOperator{\AgdaInductiveConstructor{,}}\AgdaSpace{}%
\AgdaOperator{\AgdaInductiveConstructor{⟪}}\AgdaSpace{}%
\AgdaFunction{Enum}\AgdaSpace{}%
\AgdaSymbol{\AgdaUnderscore{}}\AgdaSpace{}%
\AgdaOperator{\AgdaInductiveConstructor{[}}\AgdaSpace{}%
\AgdaFunction{Find'}\AgdaSpace{}%
\AgdaBound{a}\AgdaSpace{}%
\AgdaOperator{\AgdaInductiveConstructor{]⟫}}\<%
\\
\>[0]\AgdaCatchallClause{\AgdaFunction{st}}\AgdaSpace{}%
\AgdaCatchallClause{\AgdaBound{a}}\AgdaSpace{}%
\AgdaCatchallClause{\AgdaSymbol{(}}\AgdaCatchallClause{\AgdaBound{c₁}}\AgdaSpace{}%
\AgdaCatchallClause{\AgdaOperator{\AgdaInductiveConstructor{⊚}}}\AgdaSpace{}%
\AgdaCatchallClause{\AgdaBound{c₂}}\AgdaCatchallClause{\AgdaSymbol{)}}%
\>[44]\AgdaSymbol{=}%
\>[1226I]\AgdaKeyword{let}\AgdaSpace{}%
\AgdaSymbol{\AgdaUnderscore{}}\AgdaSpace{}%
\AgdaOperator{\AgdaInductiveConstructor{,}}\AgdaSpace{}%
\AgdaBound{c}\AgdaSpace{}%
\AgdaOperator{\AgdaInductiveConstructor{,}}\AgdaSpace{}%
\AgdaBound{st'}\AgdaSpace{}%
\AgdaSymbol{=}\AgdaSpace{}%
\AgdaFunction{st}\AgdaSpace{}%
\AgdaBound{a}\AgdaSpace{}%
\AgdaBound{c₁}\AgdaSpace{}%
\AgdaKeyword{in}\<%
\\
\>[.][@{}l@{}]\<[1226I]%
\>[46]\AgdaSymbol{\AgdaUnderscore{}}\AgdaSpace{}%
\AgdaOperator{\AgdaInductiveConstructor{,}}\AgdaSpace{}%
\AgdaBound{c}\AgdaSpace{}%
\AgdaOperator{\AgdaInductiveConstructor{⊚}}\AgdaSpace{}%
\AgdaBound{c₂}\AgdaSpace{}%
\AgdaOperator{\AgdaInductiveConstructor{,}}\AgdaSpace{}%
\AgdaBound{st'}\<%
\\
\>[0]\AgdaFunction{st}\AgdaSpace{}%
\AgdaSymbol{(}\AgdaInductiveConstructor{inj₁}\AgdaSpace{}%
\AgdaBound{x}\AgdaSymbol{)}\AgdaSpace{}%
\AgdaSymbol{(}\AgdaOperator{\AgdaInductiveConstructor{\AgdaUnderscore{}⊕\AgdaUnderscore{}}}\AgdaSpace{}%
\AgdaSymbol{\{}\AgdaBound{t₁}\AgdaSymbol{\}}\AgdaSpace{}%
\AgdaSymbol{\{}\AgdaBound{t₂}\AgdaSymbol{\}}\AgdaSpace{}%
\AgdaSymbol{\{\AgdaUnderscore{}\}}\AgdaSpace{}%
\AgdaSymbol{\{}\AgdaBound{t₄}\AgdaSymbol{\}}\AgdaSpace{}%
\AgdaInductiveConstructor{id⟷}\AgdaSpace{}%
\AgdaBound{c₂}\AgdaSymbol{)}\AgdaSpace{}%
\AgdaSymbol{=}\AgdaSpace{}%
\AgdaSymbol{\AgdaUnderscore{}}\AgdaSpace{}%
\AgdaOperator{\AgdaInductiveConstructor{,}}\AgdaSpace{}%
\AgdaInductiveConstructor{id⟷}\AgdaSpace{}%
\AgdaOperator{\AgdaInductiveConstructor{,}}\AgdaSpace{}%
\AgdaOperator{\AgdaInductiveConstructor{⟪}}\AgdaSpace{}%
\AgdaFunction{Enum}\AgdaSpace{}%
\AgdaSymbol{\AgdaUnderscore{}}\AgdaSpace{}%
\AgdaOperator{\AgdaInductiveConstructor{[}}\AgdaSpace{}%
\AgdaFunction{Find'}\AgdaSpace{}%
\AgdaSymbol{\{\AgdaUnderscore{}}\AgdaSpace{}%
\AgdaOperator{\AgdaInductiveConstructor{+ᵤ}}\AgdaSpace{}%
\AgdaBound{t₄}\AgdaSymbol{\}}\AgdaSpace{}%
\AgdaSymbol{(}\AgdaInductiveConstructor{inj₁}\AgdaSpace{}%
\AgdaBound{x}\AgdaSymbol{)}\AgdaSpace{}%
\AgdaOperator{\AgdaInductiveConstructor{]⟫}}\<%
\\
\>[0]\AgdaCatchallClause{\AgdaFunction{st}}\AgdaSpace{}%
\AgdaCatchallClause{\AgdaSymbol{(}}\AgdaCatchallClause{\AgdaInductiveConstructor{inj₁}}\AgdaSpace{}%
\AgdaCatchallClause{\AgdaBound{x}}\AgdaCatchallClause{\AgdaSymbol{)}}\AgdaSpace{}%
\AgdaCatchallClause{\AgdaSymbol{(}}\AgdaCatchallClause{\AgdaOperator{\AgdaInductiveConstructor{\AgdaUnderscore{}⊕\AgdaUnderscore{}}}}\AgdaSpace{}%
\AgdaCatchallClause{\AgdaSymbol{\{}}\AgdaCatchallClause{\AgdaBound{t₁}}\AgdaCatchallClause{\AgdaSymbol{\}}}\AgdaSpace{}%
\AgdaCatchallClause{\AgdaSymbol{\{}}\AgdaCatchallClause{\AgdaBound{t₂}}\AgdaCatchallClause{\AgdaSymbol{\}}}\AgdaSpace{}%
\AgdaCatchallClause{\AgdaBound{c₁}}\AgdaSpace{}%
\AgdaCatchallClause{\AgdaBound{c₂}}\AgdaCatchallClause{\AgdaSymbol{)}}%
\>[44]\AgdaSymbol{=}%
\>[1275I]\AgdaKeyword{let}\AgdaSpace{}%
\AgdaSymbol{\AgdaUnderscore{}}\AgdaSpace{}%
\AgdaOperator{\AgdaInductiveConstructor{,}}\AgdaSpace{}%
\AgdaBound{c}\AgdaSpace{}%
\AgdaOperator{\AgdaInductiveConstructor{,}}\AgdaSpace{}%
\AgdaBound{st'}\AgdaSpace{}%
\AgdaSymbol{=}\AgdaSpace{}%
\AgdaFunction{st}\AgdaSpace{}%
\AgdaBound{x}\AgdaSpace{}%
\AgdaBound{c₁}\AgdaSpace{}%
\AgdaKeyword{in}\<%
\\
\>[.][@{}l@{}]\<[1275I]%
\>[46]\AgdaSymbol{\AgdaUnderscore{}}\AgdaSpace{}%
\AgdaOperator{\AgdaInductiveConstructor{,}}\AgdaSpace{}%
\AgdaBound{c}\AgdaSpace{}%
\AgdaOperator{\AgdaInductiveConstructor{⊕}}\AgdaSpace{}%
\AgdaBound{c₂}\AgdaSpace{}%
\AgdaOperator{\AgdaInductiveConstructor{,}}\AgdaSpace{}%
\AgdaOperator{\AgdaInductiveConstructor{⟪}}\AgdaSpace{}%
\AgdaFunction{Enum}\AgdaSpace{}%
\AgdaSymbol{\AgdaUnderscore{}}\AgdaSpace{}%
\AgdaOperator{\AgdaInductiveConstructor{[}}\AgdaSpace{}%
\AgdaFunction{Find'}\AgdaSpace{}%
\AgdaSymbol{\{\AgdaUnderscore{}}\AgdaSpace{}%
\AgdaOperator{\AgdaInductiveConstructor{+ᵤ}}\AgdaSpace{}%
\AgdaBound{t₂}\AgdaSymbol{\}}\AgdaSpace{}%
\AgdaSymbol{(}\AgdaInductiveConstructor{inj₁}\AgdaSpace{}%
\AgdaOperator{\AgdaFunction{\$}}\AgdaSpace{}%
\AgdaFunction{resolve}\AgdaSpace{}%
\AgdaBound{st'}\AgdaSymbol{)}\AgdaSpace{}%
\AgdaOperator{\AgdaInductiveConstructor{]⟫}}\<%
\\
\>[0]\AgdaFunction{st}\AgdaSpace{}%
\AgdaSymbol{(}\AgdaInductiveConstructor{inj₂}\AgdaSpace{}%
\AgdaBound{y}\AgdaSymbol{)}\AgdaSpace{}%
\AgdaSymbol{(}\AgdaOperator{\AgdaInductiveConstructor{\AgdaUnderscore{}⊕\AgdaUnderscore{}}}\AgdaSpace{}%
\AgdaSymbol{\{}\AgdaBound{t₁}\AgdaSymbol{\}}\AgdaSpace{}%
\AgdaSymbol{\{}\AgdaBound{t₂}\AgdaSymbol{\}}\AgdaSpace{}%
\AgdaSymbol{\{}\AgdaBound{t₃}\AgdaSymbol{\}}\AgdaSpace{}%
\AgdaSymbol{\{\AgdaUnderscore{}\}}\AgdaSpace{}%
\AgdaBound{c₁}\AgdaSpace{}%
\AgdaInductiveConstructor{id⟷}\AgdaSymbol{)}\AgdaSpace{}%
\AgdaSymbol{=}\AgdaSpace{}%
\AgdaSymbol{\AgdaUnderscore{}}\AgdaSpace{}%
\AgdaOperator{\AgdaInductiveConstructor{,}}\AgdaSpace{}%
\AgdaInductiveConstructor{id⟷}\AgdaSpace{}%
\AgdaOperator{\AgdaInductiveConstructor{,}}\AgdaSpace{}%
\AgdaOperator{\AgdaInductiveConstructor{⟪}}\AgdaSpace{}%
\AgdaFunction{Enum}\AgdaSpace{}%
\AgdaSymbol{\AgdaUnderscore{}}\AgdaSpace{}%
\AgdaOperator{\AgdaInductiveConstructor{[}}\AgdaSpace{}%
\AgdaFunction{Find'}\AgdaSpace{}%
\AgdaSymbol{\{}\AgdaBound{t₃}\AgdaSpace{}%
\AgdaOperator{\AgdaInductiveConstructor{+ᵤ}}\AgdaSpace{}%
\AgdaSymbol{\AgdaUnderscore{}\}}\AgdaSpace{}%
\AgdaSymbol{(}\AgdaInductiveConstructor{inj₂}\AgdaSpace{}%
\AgdaBound{y}\AgdaSymbol{)}\AgdaSpace{}%
\AgdaOperator{\AgdaInductiveConstructor{]⟫}}\<%
\\
\>[0]\AgdaCatchallClause{\AgdaFunction{st}}\AgdaSpace{}%
\AgdaCatchallClause{\AgdaSymbol{(}}\AgdaCatchallClause{\AgdaInductiveConstructor{inj₂}}\AgdaSpace{}%
\AgdaCatchallClause{\AgdaBound{y}}\AgdaCatchallClause{\AgdaSymbol{)}}\AgdaSpace{}%
\AgdaCatchallClause{\AgdaSymbol{(}}\AgdaCatchallClause{\AgdaOperator{\AgdaInductiveConstructor{\AgdaUnderscore{}⊕\AgdaUnderscore{}}}}\AgdaSpace{}%
\AgdaCatchallClause{\AgdaSymbol{\{}}\AgdaCatchallClause{\AgdaBound{t₁}}\AgdaCatchallClause{\AgdaSymbol{\}}}\AgdaSpace{}%
\AgdaCatchallClause{\AgdaBound{c₁}}\AgdaSpace{}%
\AgdaCatchallClause{\AgdaBound{c₂}}\AgdaCatchallClause{\AgdaSymbol{)}}%
\>[44]\AgdaSymbol{=}%
\>[1335I]\AgdaKeyword{let}\AgdaSpace{}%
\AgdaSymbol{\AgdaUnderscore{}}\AgdaSpace{}%
\AgdaOperator{\AgdaInductiveConstructor{,}}\AgdaSpace{}%
\AgdaBound{c}\AgdaSpace{}%
\AgdaOperator{\AgdaInductiveConstructor{,}}\AgdaSpace{}%
\AgdaBound{st'}\AgdaSpace{}%
\AgdaSymbol{=}\AgdaSpace{}%
\AgdaFunction{st}\AgdaSpace{}%
\AgdaBound{y}\AgdaSpace{}%
\AgdaBound{c₂}\AgdaSpace{}%
\AgdaKeyword{in}\<%
\\
\>[.][@{}l@{}]\<[1335I]%
\>[46]\AgdaSymbol{\AgdaUnderscore{}}\AgdaSpace{}%
\AgdaOperator{\AgdaInductiveConstructor{,}}\AgdaSpace{}%
\AgdaBound{c₁}\AgdaSpace{}%
\AgdaOperator{\AgdaInductiveConstructor{⊕}}\AgdaSpace{}%
\AgdaBound{c}\AgdaSpace{}%
\AgdaOperator{\AgdaInductiveConstructor{,}}\AgdaSpace{}%
\AgdaOperator{\AgdaInductiveConstructor{⟪}}\AgdaSpace{}%
\AgdaFunction{Enum}\AgdaSpace{}%
\AgdaSymbol{\AgdaUnderscore{}}\AgdaSpace{}%
\AgdaOperator{\AgdaInductiveConstructor{[}}\AgdaSpace{}%
\AgdaFunction{Find'}\AgdaSpace{}%
\AgdaSymbol{\{}\AgdaBound{t₁}\AgdaSpace{}%
\AgdaOperator{\AgdaInductiveConstructor{+ᵤ}}\AgdaSpace{}%
\AgdaSymbol{\AgdaUnderscore{}\}}\AgdaSpace{}%
\AgdaSymbol{(}\AgdaInductiveConstructor{inj₂}\AgdaSpace{}%
\AgdaOperator{\AgdaFunction{\$}}\AgdaSpace{}%
\AgdaFunction{resolve}\AgdaSpace{}%
\AgdaBound{st'}\AgdaSymbol{)}\AgdaSpace{}%
\AgdaOperator{\AgdaInductiveConstructor{]⟫}}\<%
\\
\>[0]\AgdaFunction{st}\AgdaSpace{}%
\AgdaSymbol{(}\AgdaBound{x}\AgdaSpace{}%
\AgdaOperator{\AgdaInductiveConstructor{,}}\AgdaSpace{}%
\AgdaBound{y}\AgdaSymbol{)}\AgdaSpace{}%
\AgdaSymbol{(}\AgdaInductiveConstructor{id⟷}\AgdaSpace{}%
\AgdaOperator{\AgdaInductiveConstructor{⊗}}\AgdaSpace{}%
\AgdaInductiveConstructor{id⟷}\AgdaSymbol{)}%
\>[44]\AgdaSymbol{=}\AgdaSpace{}%
\AgdaSymbol{\AgdaUnderscore{}}\AgdaSpace{}%
\AgdaOperator{\AgdaInductiveConstructor{,}}\AgdaSpace{}%
\AgdaInductiveConstructor{id⟷}\AgdaSpace{}%
\AgdaOperator{\AgdaInductiveConstructor{,}}\AgdaSpace{}%
\AgdaOperator{\AgdaInductiveConstructor{⟪}}\AgdaSpace{}%
\AgdaFunction{Enum}\AgdaSpace{}%
\AgdaSymbol{\AgdaUnderscore{}}\AgdaSpace{}%
\AgdaOperator{\AgdaInductiveConstructor{[}}\AgdaSpace{}%
\AgdaFunction{Find'}\AgdaSpace{}%
\AgdaSymbol{(}\AgdaBound{x}\AgdaSpace{}%
\AgdaOperator{\AgdaInductiveConstructor{,}}\AgdaSpace{}%
\AgdaBound{y}\AgdaSymbol{)}\AgdaSpace{}%
\AgdaOperator{\AgdaInductiveConstructor{]⟫}}\<%
\\
\>[0]\AgdaCatchallClause{\AgdaFunction{st}}\AgdaSpace{}%
\AgdaCatchallClause{\AgdaSymbol{(}}\AgdaCatchallClause{\AgdaBound{x}}\AgdaSpace{}%
\AgdaCatchallClause{\AgdaOperator{\AgdaInductiveConstructor{,}}}\AgdaSpace{}%
\AgdaCatchallClause{\AgdaBound{y}}\AgdaCatchallClause{\AgdaSymbol{)}}\AgdaSpace{}%
\AgdaCatchallClause{\AgdaSymbol{(}}\AgdaCatchallClause{\AgdaInductiveConstructor{id⟷}}\AgdaSpace{}%
\AgdaCatchallClause{\AgdaOperator{\AgdaInductiveConstructor{⊗}}}\AgdaSpace{}%
\AgdaCatchallClause{\AgdaBound{c₂}}\AgdaCatchallClause{\AgdaSymbol{)}}%
\>[44]\AgdaSymbol{=}%
\>[1389I]\AgdaKeyword{let}\AgdaSpace{}%
\AgdaSymbol{\AgdaUnderscore{}}\AgdaSpace{}%
\AgdaOperator{\AgdaInductiveConstructor{,}}\AgdaSpace{}%
\AgdaBound{c}\AgdaSpace{}%
\AgdaOperator{\AgdaInductiveConstructor{,}}\AgdaSpace{}%
\AgdaBound{st'}\AgdaSpace{}%
\AgdaSymbol{=}\AgdaSpace{}%
\AgdaFunction{st}\AgdaSpace{}%
\AgdaBound{y}\AgdaSpace{}%
\AgdaBound{c₂}\AgdaSpace{}%
\AgdaKeyword{in}\<%
\\
\>[1389I][@{}l@{\AgdaIndent{0}}]%
\>[47]\AgdaSymbol{\AgdaUnderscore{}}\AgdaSpace{}%
\AgdaOperator{\AgdaInductiveConstructor{,}}\AgdaSpace{}%
\AgdaInductiveConstructor{id⟷}\AgdaSpace{}%
\AgdaOperator{\AgdaInductiveConstructor{⊗}}\AgdaSpace{}%
\AgdaBound{c}\AgdaSpace{}%
\AgdaOperator{\AgdaInductiveConstructor{,}}\AgdaSpace{}%
\AgdaOperator{\AgdaInductiveConstructor{⟪}}\AgdaSpace{}%
\AgdaFunction{Enum}\AgdaSpace{}%
\AgdaSymbol{\AgdaUnderscore{}}\AgdaSpace{}%
\AgdaOperator{\AgdaInductiveConstructor{[}}\AgdaSpace{}%
\AgdaFunction{Find'}\AgdaSpace{}%
\AgdaSymbol{(}\AgdaBound{x}\AgdaSpace{}%
\AgdaOperator{\AgdaInductiveConstructor{,}}\AgdaSpace{}%
\AgdaFunction{resolve}\AgdaSpace{}%
\AgdaBound{st'}\AgdaSymbol{)}\AgdaSpace{}%
\AgdaOperator{\AgdaInductiveConstructor{]⟫}}\<%
\\
\>[0]\AgdaCatchallClause{\AgdaFunction{st}}\AgdaSpace{}%
\AgdaCatchallClause{\AgdaSymbol{(}}\AgdaCatchallClause{\AgdaBound{x}}\AgdaSpace{}%
\AgdaCatchallClause{\AgdaOperator{\AgdaInductiveConstructor{,}}}\AgdaSpace{}%
\AgdaCatchallClause{\AgdaBound{y}}\AgdaCatchallClause{\AgdaSymbol{)}}\AgdaSpace{}%
\AgdaCatchallClause{\AgdaSymbol{(}}\AgdaCatchallClause{\AgdaBound{c₁}}\AgdaSpace{}%
\AgdaCatchallClause{\AgdaOperator{\AgdaInductiveConstructor{⊗}}}\AgdaSpace{}%
\AgdaCatchallClause{\AgdaBound{c₂}}\AgdaCatchallClause{\AgdaSymbol{)}}%
\>[44]\AgdaSymbol{=}%
\>[1421I]\AgdaKeyword{let}\AgdaSpace{}%
\AgdaSymbol{\AgdaUnderscore{}}\AgdaSpace{}%
\AgdaOperator{\AgdaInductiveConstructor{,}}\AgdaSpace{}%
\AgdaBound{c}\AgdaSpace{}%
\AgdaOperator{\AgdaInductiveConstructor{,}}\AgdaSpace{}%
\AgdaBound{st'}\AgdaSpace{}%
\AgdaSymbol{=}\AgdaSpace{}%
\AgdaFunction{st}\AgdaSpace{}%
\AgdaBound{x}\AgdaSpace{}%
\AgdaBound{c₁}\AgdaSpace{}%
\AgdaKeyword{in}\<%
\\
\>[.][@{}l@{}]\<[1421I]%
\>[46]\AgdaSymbol{\AgdaUnderscore{}}\AgdaSpace{}%
\AgdaOperator{\AgdaInductiveConstructor{,}}\AgdaSpace{}%
\AgdaBound{c}\AgdaSpace{}%
\AgdaOperator{\AgdaInductiveConstructor{⊗}}\AgdaSpace{}%
\AgdaBound{c₂}\AgdaSpace{}%
\AgdaOperator{\AgdaInductiveConstructor{,}}\AgdaSpace{}%
\AgdaOperator{\AgdaInductiveConstructor{⟪}}\AgdaSpace{}%
\AgdaFunction{Enum}\AgdaSpace{}%
\AgdaSymbol{\AgdaUnderscore{}}\AgdaSpace{}%
\AgdaOperator{\AgdaInductiveConstructor{[}}\AgdaSpace{}%
\AgdaFunction{Find'}\AgdaSpace{}%
\AgdaSymbol{(}\AgdaFunction{resolve}\AgdaSpace{}%
\AgdaBound{st'}\AgdaSpace{}%
\AgdaOperator{\AgdaInductiveConstructor{,}}\AgdaSpace{}%
\AgdaBound{y}\AgdaSymbol{)}\AgdaSpace{}%
\AgdaOperator{\AgdaInductiveConstructor{]⟫}}\<%
\\
\\[\AgdaEmptyExtraSkip]%
\>[0]\AgdaFunction{step}\AgdaSpace{}%
\AgdaSymbol{:}\AgdaSpace{}%
\AgdaSymbol{\{}\AgdaBound{A}\AgdaSpace{}%
\AgdaBound{B}\AgdaSpace{}%
\AgdaSymbol{:}\AgdaSpace{}%
\AgdaDatatype{𝕌}\AgdaSymbol{\}}\AgdaSpace{}%
\AgdaSymbol{(}\AgdaBound{c}\AgdaSpace{}%
\AgdaSymbol{:}\AgdaSpace{}%
\AgdaBound{A}\AgdaSpace{}%
\AgdaOperator{\AgdaDatatype{⟷}}\AgdaSpace{}%
\AgdaBound{B}\AgdaSymbol{)}\AgdaSpace{}%
\AgdaSymbol{→}\AgdaSpace{}%
\AgdaDatatype{State}\AgdaSpace{}%
\AgdaBound{A}\AgdaSpace{}%
\AgdaSymbol{→}\AgdaSpace{}%
\AgdaFunction{Σ[}\AgdaSpace{}%
\AgdaBound{C}\AgdaSpace{}%
\AgdaFunction{∈}\AgdaSpace{}%
\AgdaDatatype{𝕌}\AgdaSpace{}%
\AgdaFunction{]}\AgdaSpace{}%
\AgdaSymbol{(}\AgdaBound{C}\AgdaSpace{}%
\AgdaOperator{\AgdaDatatype{⟷}}\AgdaSpace{}%
\AgdaBound{B}\AgdaSpace{}%
\AgdaOperator{\AgdaFunction{×}}\AgdaSpace{}%
\AgdaDatatype{State}\AgdaSpace{}%
\AgdaBound{C}\AgdaSymbol{)}\<%
\\
\>[0]\AgdaFunction{step}\AgdaSpace{}%
\AgdaBound{c}\AgdaSpace{}%
\AgdaOperator{\AgdaInductiveConstructor{⟪}}\AgdaSpace{}%
\AgdaBound{v}\AgdaSpace{}%
\AgdaOperator{\AgdaInductiveConstructor{[}}\AgdaSpace{}%
\AgdaBound{i}\AgdaSpace{}%
\AgdaOperator{\AgdaInductiveConstructor{]⟫}}\AgdaSpace{}%
\AgdaSymbol{=}\AgdaSpace{}%
\AgdaFunction{st}\AgdaSpace{}%
\AgdaSymbol{(}\AgdaFunction{lookup}\AgdaSpace{}%
\AgdaBound{v}\AgdaSpace{}%
\AgdaBound{i}\AgdaSymbol{)}\AgdaSpace{}%
\AgdaBound{c}\<%
\end{code}

\begin{code}[hide]%
\>[0]\AgdaFunction{run}\AgdaSpace{}%
\AgdaSymbol{:}\AgdaSpace{}%
\AgdaSymbol{(}\AgdaBound{sz}\AgdaSpace{}%
\AgdaBound{n}\AgdaSpace{}%
\AgdaSymbol{:}\AgdaSpace{}%
\AgdaDatatype{ℕ}\AgdaSymbol{)}\AgdaSpace{}%
\AgdaSymbol{→}\AgdaSpace{}%
\AgdaSymbol{(}\AgdaBound{st}\AgdaSpace{}%
\AgdaSymbol{:}\AgdaSpace{}%
\AgdaDatatype{State'}\AgdaSpace{}%
\AgdaBound{sz}\AgdaSymbol{)}\AgdaSpace{}%
\AgdaSymbol{→}\AgdaSpace{}%
\AgdaDatatype{Vec}\AgdaSpace{}%
\AgdaSymbol{(}\AgdaDatatype{State'}\AgdaSpace{}%
\AgdaBound{sz}\AgdaSymbol{)}\AgdaSpace{}%
\AgdaSymbol{(}\AgdaInductiveConstructor{suc}\AgdaSpace{}%
\AgdaBound{n}\AgdaSymbol{)}\<%
\\
\>[0]\AgdaFunction{run}\AgdaSpace{}%
\AgdaBound{sz}\AgdaSpace{}%
\AgdaNumber{0}\AgdaSpace{}%
\AgdaBound{st}\AgdaSpace{}%
\AgdaSymbol{=}\AgdaSpace{}%
\AgdaOperator{\AgdaFunction{[}}\AgdaSpace{}%
\AgdaBound{st}\AgdaSpace{}%
\AgdaOperator{\AgdaFunction{]}}\<%
\\
\>[0]\AgdaFunction{run}\AgdaSpace{}%
\AgdaBound{sz}\AgdaSpace{}%
\AgdaSymbol{(}\AgdaInductiveConstructor{suc}\AgdaSpace{}%
\AgdaBound{n}\AgdaSymbol{)}\AgdaSpace{}%
\AgdaBound{st}\AgdaSpace{}%
\AgdaKeyword{with}\AgdaSpace{}%
\AgdaFunction{run}\AgdaSpace{}%
\AgdaBound{sz}\AgdaSpace{}%
\AgdaBound{n}\AgdaSpace{}%
\AgdaBound{st}\<%
\\
\>[0]\AgdaSymbol{...}\AgdaSpace{}%
\AgdaSymbol{|}\AgdaSpace{}%
\AgdaBound{sts}\AgdaSpace{}%
\AgdaKeyword{with}\AgdaSpace{}%
\AgdaFunction{last}\AgdaSpace{}%
\AgdaBound{sts}\<%
\\
\>[0]\AgdaSymbol{...}\AgdaSpace{}%
\AgdaSymbol{|}\AgdaSpace{}%
\AgdaOperator{\AgdaInductiveConstructor{⟪\AgdaUnderscore{}∥\AgdaUnderscore{}[\AgdaUnderscore{}]⟫}}\AgdaSpace{}%
\AgdaSymbol{\{}\AgdaBound{A}\AgdaSymbol{\}}\AgdaSpace{}%
\AgdaSymbol{\{}\AgdaBound{B}\AgdaSymbol{\}}\AgdaSpace{}%
\AgdaBound{cx}\AgdaSpace{}%
\AgdaBound{vx}\AgdaSpace{}%
\AgdaBound{ix}\AgdaSpace{}%
\AgdaKeyword{rewrite}\AgdaSpace{}%
\AgdaFunction{+-comm}\AgdaSpace{}%
\AgdaNumber{1}\AgdaSpace{}%
\AgdaSymbol{(}\AgdaInductiveConstructor{suc}\AgdaSpace{}%
\AgdaBound{n}\AgdaSymbol{)}\AgdaSpace{}%
\AgdaSymbol{=}\AgdaSpace{}%
\AgdaBound{sts}\AgdaSpace{}%
\AgdaOperator{\AgdaFunction{++}}\AgdaSpace{}%
\AgdaOperator{\AgdaFunction{[}}\AgdaSpace{}%
\AgdaFunction{step'}\AgdaSpace{}%
\AgdaOperator{\AgdaInductiveConstructor{⟪}}\AgdaSpace{}%
\AgdaBound{cx}\AgdaSpace{}%
\AgdaOperator{\AgdaInductiveConstructor{∥}}\AgdaSpace{}%
\AgdaBound{vx}\AgdaSpace{}%
\AgdaOperator{\AgdaInductiveConstructor{[}}\AgdaSpace{}%
\AgdaBound{ix}\AgdaSpace{}%
\AgdaOperator{\AgdaInductiveConstructor{]⟫}}\AgdaSpace{}%
\AgdaOperator{\AgdaFunction{]}}\<%
\\
\\[\AgdaEmptyExtraSkip]%
\>[0]\AgdaFunction{trace₁}\AgdaSpace{}%
\AgdaSymbol{:}\AgdaSpace{}%
\AgdaDatatype{Vec}\AgdaSpace{}%
\AgdaSymbol{(}\AgdaDatatype{State'}\AgdaSpace{}%
\AgdaOperator{\AgdaFunction{∣}}\AgdaSpace{}%
\AgdaFunction{𝔹}\AgdaSpace{}%
\AgdaOperator{\AgdaInductiveConstructor{×ᵤ}}\AgdaSpace{}%
\AgdaFunction{𝔹}\AgdaSpace{}%
\AgdaOperator{\AgdaFunction{∣}}\AgdaSymbol{)}\AgdaSpace{}%
\AgdaNumber{8}\<%
\\
\>[0]\AgdaFunction{trace₁}\AgdaSpace{}%
\AgdaSymbol{=}\AgdaSpace{}%
\AgdaFunction{run}\AgdaSpace{}%
\AgdaOperator{\AgdaFunction{∣}}\AgdaSpace{}%
\AgdaFunction{𝔹}\AgdaSpace{}%
\AgdaOperator{\AgdaInductiveConstructor{×ᵤ}}\AgdaSpace{}%
\AgdaFunction{𝔹}\AgdaSpace{}%
\AgdaOperator{\AgdaFunction{∣}}\AgdaSpace{}%
\AgdaSymbol{\AgdaUnderscore{}}\AgdaSpace{}%
\AgdaOperator{\AgdaInductiveConstructor{⟪}}\AgdaSpace{}%
\AgdaFunction{CNOT}\AgdaSpace{}%
\AgdaOperator{\AgdaInductiveConstructor{∥}}\AgdaSpace{}%
\AgdaFunction{Enum}\AgdaSpace{}%
\AgdaSymbol{(}\AgdaFunction{𝔹}\AgdaSpace{}%
\AgdaOperator{\AgdaInductiveConstructor{×ᵤ}}\AgdaSpace{}%
\AgdaFunction{𝔹}\AgdaSymbol{)}\AgdaSpace{}%
\AgdaOperator{\AgdaInductiveConstructor{[}}\AgdaSpace{}%
\AgdaFunction{fromℕ}\AgdaSpace{}%
\AgdaNumber{3}\AgdaSpace{}%
\AgdaOperator{\AgdaInductiveConstructor{]⟫}}\<%
\end{code}

\begin{code}[hide]%
\>[0]\AgdaSymbol{\{-\#}\AgdaSpace{}%
\AgdaKeyword{OPTIONS}\AgdaSpace{}%
\AgdaPragma{--without-K}\AgdaSpace{}%
\AgdaPragma{--safe}\AgdaSpace{}%
\AgdaSymbol{\#-\}}\<%
\\
\>[0]\AgdaKeyword{module}\AgdaSpace{}%
\AgdaModule{Pi.DMemSem}\AgdaSpace{}%
\AgdaKeyword{where}\<%
\\
\>[0]\AgdaKeyword{open}\AgdaSpace{}%
\AgdaKeyword{import}\AgdaSpace{}%
\AgdaModule{Relation.Binary.Core}\<%
\\
\>[0]\AgdaKeyword{open}\AgdaSpace{}%
\AgdaKeyword{import}\AgdaSpace{}%
\AgdaModule{Data.Empty}\<%
\\
\>[0]\AgdaKeyword{open}\AgdaSpace{}%
\AgdaKeyword{import}\AgdaSpace{}%
\AgdaModule{Function}\AgdaSpace{}%
\AgdaKeyword{using}\AgdaSpace{}%
\AgdaSymbol{(}\AgdaOperator{\AgdaFunction{\AgdaUnderscore{}∘\AgdaUnderscore{}}}\AgdaSymbol{;}\AgdaSpace{}%
\AgdaOperator{\AgdaFunction{\AgdaUnderscore{}\$\AgdaUnderscore{}}}\AgdaSymbol{)}\<%
\\
\>[0]\AgdaKeyword{open}\AgdaSpace{}%
\AgdaKeyword{import}\AgdaSpace{}%
\AgdaModule{Data.Nat}\<%
\\
\>[0]\AgdaKeyword{open}\AgdaSpace{}%
\AgdaKeyword{import}\AgdaSpace{}%
\AgdaModule{Data.Nat.Properties}\<%
\\
\>[0]\AgdaKeyword{open}\AgdaSpace{}%
\AgdaKeyword{import}\AgdaSpace{}%
\AgdaModule{Data.Fin}\AgdaSpace{}%
\AgdaKeyword{hiding}\AgdaSpace{}%
\AgdaSymbol{(}\AgdaOperator{\AgdaFunction{\AgdaUnderscore{}+\AgdaUnderscore{}}}\AgdaSymbol{)}\<%
\\
\>[0]\AgdaKeyword{open}\AgdaSpace{}%
\AgdaKeyword{import}\AgdaSpace{}%
\AgdaModule{Data.Empty}\AgdaSpace{}%
\AgdaKeyword{using}\AgdaSpace{}%
\AgdaSymbol{(}\AgdaDatatype{⊥}\AgdaSymbol{)}\<%
\\
\>[0]\AgdaKeyword{open}\AgdaSpace{}%
\AgdaKeyword{import}\AgdaSpace{}%
\AgdaModule{Data.Unit}\AgdaSpace{}%
\AgdaKeyword{using}\AgdaSpace{}%
\AgdaSymbol{(}\AgdaRecord{⊤}\AgdaSymbol{;}\AgdaSpace{}%
\AgdaInductiveConstructor{tt}\AgdaSymbol{)}\<%
\\
\>[0]\AgdaKeyword{open}\AgdaSpace{}%
\AgdaKeyword{import}\AgdaSpace{}%
\AgdaModule{Data.Sum}\AgdaSpace{}%
\AgdaKeyword{using}\AgdaSpace{}%
\AgdaSymbol{(}\AgdaOperator{\AgdaDatatype{\AgdaUnderscore{}⊎\AgdaUnderscore{}}}\AgdaSymbol{;}\AgdaSpace{}%
\AgdaInductiveConstructor{inj₁}\AgdaSymbol{;}\AgdaSpace{}%
\AgdaInductiveConstructor{inj₂}\AgdaSymbol{)}\<%
\\
\>[0]\AgdaKeyword{open}\AgdaSpace{}%
\AgdaKeyword{import}\AgdaSpace{}%
\AgdaModule{Data.Product}\AgdaSpace{}%
\AgdaKeyword{using}\AgdaSpace{}%
\AgdaSymbol{(}\AgdaOperator{\AgdaFunction{\AgdaUnderscore{}×\AgdaUnderscore{}}}\AgdaSymbol{;}\AgdaSpace{}%
\AgdaOperator{\AgdaInductiveConstructor{\AgdaUnderscore{},\AgdaUnderscore{}}}\AgdaSymbol{;}\AgdaSpace{}%
\AgdaField{proj₁}\AgdaSymbol{;}\AgdaSpace{}%
\AgdaField{proj₂}\AgdaSymbol{;}\AgdaSpace{}%
\AgdaFunction{Σ-syntax}\AgdaSymbol{)}\<%
\\
\>[0]\AgdaKeyword{open}\AgdaSpace{}%
\AgdaKeyword{import}\AgdaSpace{}%
\AgdaModule{Data.Vec}\<%
\\
\>[0]\AgdaKeyword{open}\AgdaSpace{}%
\AgdaKeyword{import}\AgdaSpace{}%
\AgdaModule{Data.Vec.Relation.Unary.Any.Properties}\<%
\\
\>[0]\AgdaKeyword{open}\AgdaSpace{}%
\AgdaKeyword{import}\AgdaSpace{}%
\AgdaModule{Data.Vec.Any}\AgdaSpace{}%
\AgdaKeyword{using}\AgdaSpace{}%
\AgdaSymbol{(}\AgdaDatatype{Any}\AgdaSymbol{;}\AgdaSpace{}%
\AgdaInductiveConstructor{here}\AgdaSymbol{;}\AgdaSpace{}%
\AgdaInductiveConstructor{there}\AgdaSymbol{;}\AgdaSpace{}%
\AgdaFunction{index}\AgdaSymbol{)}\<%
\\
\>[0]\AgdaKeyword{open}\AgdaSpace{}%
\AgdaKeyword{import}\AgdaSpace{}%
\AgdaModule{Data.Vec.Membership.Propositional}\<%
\\
\>[0]\AgdaKeyword{open}\AgdaSpace{}%
\AgdaKeyword{import}\AgdaSpace{}%
\AgdaModule{Data.Vec.Membership.Propositional.Properties}\<%
\\
\>[0]\AgdaKeyword{open}\AgdaSpace{}%
\AgdaKeyword{import}\AgdaSpace{}%
\AgdaModule{Relation.Nullary}\<%
\\
\>[0]\AgdaKeyword{open}\AgdaSpace{}%
\AgdaKeyword{import}\AgdaSpace{}%
\AgdaModule{Relation.Binary.PropositionalEquality}\AgdaSpace{}%
\AgdaKeyword{hiding}\AgdaSpace{}%
\AgdaSymbol{(}\AgdaOperator{\AgdaInductiveConstructor{[\AgdaUnderscore{}]}}\AgdaSymbol{)}\<%
\\
\>[0]\AgdaKeyword{open}\AgdaSpace{}%
\AgdaKeyword{import}\AgdaSpace{}%
\AgdaModule{Pi.D}\<%
\\
\\[\AgdaEmptyExtraSkip]%
\>[0]\AgdaKeyword{infix}%
\>[7]\AgdaNumber{80}\AgdaSpace{}%
\AgdaOperator{\AgdaFunction{∣\AgdaUnderscore{}∣}}\<%
\\
\\[\AgdaEmptyExtraSkip]%
\>[0]\AgdaOperator{\AgdaFunction{∣\AgdaUnderscore{}∣}}\AgdaSpace{}%
\AgdaSymbol{:}\AgdaSpace{}%
\AgdaSymbol{(}\AgdaBound{A}\AgdaSpace{}%
\AgdaSymbol{:}\AgdaSpace{}%
\AgdaDatatype{𝕌}\AgdaSymbol{)}\AgdaSpace{}%
\AgdaSymbol{→}\AgdaSpace{}%
\AgdaDatatype{ℕ}\<%
\\
\>[0]\AgdaOperator{\AgdaFunction{∣}}\AgdaSpace{}%
\AgdaInductiveConstructor{𝟘}\AgdaSpace{}%
\AgdaOperator{\AgdaFunction{∣}}%
\>[14]\AgdaSymbol{=}\AgdaSpace{}%
\AgdaNumber{0}\<%
\\
\>[0]\AgdaOperator{\AgdaFunction{∣}}\AgdaSpace{}%
\AgdaInductiveConstructor{𝟙}\AgdaSpace{}%
\AgdaOperator{\AgdaFunction{∣}}%
\>[14]\AgdaSymbol{=}\AgdaSpace{}%
\AgdaNumber{1}\<%
\\
\>[0]\AgdaOperator{\AgdaFunction{∣}}\AgdaSpace{}%
\AgdaBound{A₁}\AgdaSpace{}%
\AgdaOperator{\AgdaInductiveConstructor{+ᵤ}}\AgdaSpace{}%
\AgdaBound{A₂}\AgdaSpace{}%
\AgdaOperator{\AgdaFunction{∣}}%
\>[14]\AgdaSymbol{=}\AgdaSpace{}%
\AgdaOperator{\AgdaFunction{∣}}\AgdaSpace{}%
\AgdaBound{A₁}\AgdaSpace{}%
\AgdaOperator{\AgdaFunction{∣}}\AgdaSpace{}%
\AgdaOperator{\AgdaPrimitive{+}}\AgdaSpace{}%
\AgdaOperator{\AgdaFunction{∣}}\AgdaSpace{}%
\AgdaBound{A₂}\AgdaSpace{}%
\AgdaOperator{\AgdaFunction{∣}}\<%
\\
\>[0]\AgdaOperator{\AgdaFunction{∣}}\AgdaSpace{}%
\AgdaBound{A₁}\AgdaSpace{}%
\AgdaOperator{\AgdaInductiveConstructor{×ᵤ}}\AgdaSpace{}%
\AgdaBound{A₂}\AgdaSpace{}%
\AgdaOperator{\AgdaFunction{∣}}%
\>[14]\AgdaSymbol{=}\AgdaSpace{}%
\AgdaOperator{\AgdaFunction{∣}}\AgdaSpace{}%
\AgdaBound{A₁}\AgdaSpace{}%
\AgdaOperator{\AgdaFunction{∣}}\AgdaSpace{}%
\AgdaOperator{\AgdaPrimitive{*}}\AgdaSpace{}%
\AgdaOperator{\AgdaFunction{∣}}\AgdaSpace{}%
\AgdaBound{A₂}\AgdaSpace{}%
\AgdaOperator{\AgdaFunction{∣}}\<%
\\
\>[0]\AgdaOperator{\AgdaFunction{∣}}\AgdaSpace{}%
\AgdaOperator{\AgdaInductiveConstructor{𝟙/}}\AgdaSpace{}%
\AgdaBound{v}\AgdaSpace{}%
\AgdaOperator{\AgdaFunction{∣}}%
\>[14]\AgdaSymbol{=}\AgdaSpace{}%
\AgdaNumber{1}\<%
\\
\\[\AgdaEmptyExtraSkip]%
\>[0]\AgdaFunction{Enum}\AgdaSpace{}%
\AgdaSymbol{:}\AgdaSpace{}%
\AgdaSymbol{(}\AgdaBound{A}\AgdaSpace{}%
\AgdaSymbol{:}\AgdaSpace{}%
\AgdaDatatype{𝕌}\AgdaSymbol{)}\AgdaSpace{}%
\AgdaSymbol{→}\AgdaSpace{}%
\AgdaDatatype{Vec}\AgdaSpace{}%
\AgdaOperator{\AgdaFunction{⟦}}\AgdaSpace{}%
\AgdaBound{A}\AgdaSpace{}%
\AgdaOperator{\AgdaFunction{⟧}}\AgdaSpace{}%
\AgdaOperator{\AgdaFunction{∣}}\AgdaSpace{}%
\AgdaBound{A}\AgdaSpace{}%
\AgdaOperator{\AgdaFunction{∣}}\<%
\\
\>[0]\AgdaFunction{Enum}\AgdaSpace{}%
\AgdaInductiveConstructor{𝟘}%
\>[16]\AgdaSymbol{=}\AgdaSpace{}%
\AgdaInductiveConstructor{[]}\<%
\\
\>[0]\AgdaFunction{Enum}\AgdaSpace{}%
\AgdaInductiveConstructor{𝟙}%
\>[16]\AgdaSymbol{=}\AgdaSpace{}%
\AgdaInductiveConstructor{tt}\AgdaSpace{}%
\AgdaOperator{\AgdaInductiveConstructor{∷}}\AgdaSpace{}%
\AgdaInductiveConstructor{[]}\<%
\\
\>[0]\AgdaFunction{Enum}\AgdaSpace{}%
\AgdaSymbol{(}\AgdaBound{A₁}\AgdaSpace{}%
\AgdaOperator{\AgdaInductiveConstructor{+ᵤ}}\AgdaSpace{}%
\AgdaBound{A₂}\AgdaSymbol{)}\AgdaSpace{}%
\AgdaSymbol{=}\AgdaSpace{}%
\AgdaFunction{map}\AgdaSpace{}%
\AgdaInductiveConstructor{inj₁}\AgdaSpace{}%
\AgdaSymbol{(}\AgdaFunction{Enum}\AgdaSpace{}%
\AgdaBound{A₁}\AgdaSymbol{)}\AgdaSpace{}%
\AgdaOperator{\AgdaFunction{++}}\AgdaSpace{}%
\AgdaFunction{map}\AgdaSpace{}%
\AgdaInductiveConstructor{inj₂}\AgdaSpace{}%
\AgdaSymbol{(}\AgdaFunction{Enum}\AgdaSpace{}%
\AgdaBound{A₂}\AgdaSymbol{)}\<%
\\
\>[0]\AgdaFunction{Enum}\AgdaSpace{}%
\AgdaSymbol{(}\AgdaBound{A₁}\AgdaSpace{}%
\AgdaOperator{\AgdaInductiveConstructor{×ᵤ}}\AgdaSpace{}%
\AgdaBound{A₂}\AgdaSymbol{)}\AgdaSpace{}%
\AgdaSymbol{=}\AgdaSpace{}%
\AgdaFunction{allPairs}\AgdaSpace{}%
\AgdaSymbol{(}\AgdaFunction{Enum}\AgdaSpace{}%
\AgdaBound{A₁}\AgdaSymbol{)}\AgdaSpace{}%
\AgdaSymbol{(}\AgdaFunction{Enum}\AgdaSpace{}%
\AgdaBound{A₂}\AgdaSymbol{)}\<%
\\
\>[0]\AgdaFunction{Enum}\AgdaSpace{}%
\AgdaSymbol{(}\AgdaOperator{\AgdaInductiveConstructor{𝟙/}}\AgdaSpace{}%
\AgdaBound{A}\AgdaSymbol{)}%
\>[16]\AgdaSymbol{=}\AgdaSpace{}%
\AgdaInductiveConstructor{↻}\AgdaSpace{}%
\AgdaOperator{\AgdaInductiveConstructor{∷}}\AgdaSpace{}%
\AgdaInductiveConstructor{[]}\<%
\\
\\[\AgdaEmptyExtraSkip]%
\>[0]\AgdaFunction{Find}\AgdaSpace{}%
\AgdaSymbol{:}\AgdaSpace{}%
\AgdaSymbol{\{}\AgdaBound{A}\AgdaSpace{}%
\AgdaSymbol{:}\AgdaSpace{}%
\AgdaDatatype{𝕌}\AgdaSymbol{\}}\AgdaSpace{}%
\AgdaSymbol{(}\AgdaBound{x}\AgdaSpace{}%
\AgdaSymbol{:}\AgdaSpace{}%
\AgdaOperator{\AgdaFunction{⟦}}\AgdaSpace{}%
\AgdaBound{A}\AgdaSpace{}%
\AgdaOperator{\AgdaFunction{⟧}}\AgdaSymbol{)}\AgdaSpace{}%
\AgdaSymbol{→}\AgdaSpace{}%
\AgdaBound{x}\AgdaSpace{}%
\AgdaOperator{\AgdaFunction{∈}}\AgdaSpace{}%
\AgdaFunction{Enum}\AgdaSpace{}%
\AgdaBound{A}\<%
\\
\>[0]\AgdaFunction{Find}\AgdaSpace{}%
\AgdaSymbol{\{}\AgdaInductiveConstructor{𝟙}\AgdaSymbol{\}}\AgdaSpace{}%
\AgdaInductiveConstructor{tt}\AgdaSpace{}%
\AgdaSymbol{=}\AgdaSpace{}%
\AgdaInductiveConstructor{here}\AgdaSpace{}%
\AgdaInductiveConstructor{refl}\<%
\\
\>[0]\AgdaFunction{Find}\AgdaSpace{}%
\AgdaSymbol{\{}\AgdaBound{A₁}\AgdaSpace{}%
\AgdaOperator{\AgdaInductiveConstructor{+ᵤ}}\AgdaSpace{}%
\AgdaBound{A₂}\AgdaSymbol{\}}\AgdaSpace{}%
\AgdaSymbol{(}\AgdaInductiveConstructor{inj₁}\AgdaSpace{}%
\AgdaBound{x}\AgdaSymbol{)}\AgdaSpace{}%
\AgdaSymbol{=}\AgdaSpace{}%
\AgdaFunction{++⁺ˡ}\AgdaSpace{}%
\AgdaSymbol{\{}\AgdaArgument{xs}\AgdaSpace{}%
\AgdaSymbol{=}\AgdaSpace{}%
\AgdaFunction{map}\AgdaSpace{}%
\AgdaInductiveConstructor{inj₁}\AgdaSpace{}%
\AgdaSymbol{(}\AgdaFunction{Enum}\AgdaSpace{}%
\AgdaBound{A₁}\AgdaSymbol{)\}}\AgdaSpace{}%
\AgdaSymbol{(}\AgdaFunction{∈-map⁺}\AgdaSpace{}%
\AgdaInductiveConstructor{inj₁}\AgdaSpace{}%
\AgdaSymbol{(}\AgdaFunction{Find}\AgdaSpace{}%
\AgdaBound{x}\AgdaSymbol{))}\<%
\\
\>[0]\AgdaFunction{Find}\AgdaSpace{}%
\AgdaSymbol{\{}\AgdaBound{A₁}\AgdaSpace{}%
\AgdaOperator{\AgdaInductiveConstructor{+ᵤ}}\AgdaSpace{}%
\AgdaBound{A₂}\AgdaSymbol{\}}\AgdaSpace{}%
\AgdaSymbol{(}\AgdaInductiveConstructor{inj₂}\AgdaSpace{}%
\AgdaBound{y}\AgdaSymbol{)}\AgdaSpace{}%
\AgdaSymbol{=}\AgdaSpace{}%
\AgdaFunction{++⁺ʳ}\AgdaSpace{}%
\AgdaSymbol{(}\AgdaFunction{map}\AgdaSpace{}%
\AgdaInductiveConstructor{inj₁}\AgdaSpace{}%
\AgdaSymbol{(}\AgdaFunction{Enum}\AgdaSpace{}%
\AgdaBound{A₁}\AgdaSymbol{))}\AgdaSpace{}%
\AgdaSymbol{(}\AgdaFunction{∈-map⁺}\AgdaSpace{}%
\AgdaInductiveConstructor{inj₂}\AgdaSpace{}%
\AgdaSymbol{(}\AgdaFunction{Find}\AgdaSpace{}%
\AgdaBound{y}\AgdaSymbol{))}\<%
\\
\>[0]\AgdaFunction{Find}\AgdaSpace{}%
\AgdaSymbol{\{}\AgdaBound{A₁}\AgdaSpace{}%
\AgdaOperator{\AgdaInductiveConstructor{×ᵤ}}\AgdaSpace{}%
\AgdaBound{A₂}\AgdaSymbol{\}}\AgdaSpace{}%
\AgdaSymbol{(}\AgdaBound{x}\AgdaSpace{}%
\AgdaOperator{\AgdaInductiveConstructor{,}}\AgdaSpace{}%
\AgdaBound{y}\AgdaSymbol{)}\AgdaSpace{}%
\AgdaSymbol{=}\AgdaSpace{}%
\AgdaFunction{∈-allPairs⁺}\AgdaSpace{}%
\AgdaSymbol{(}\AgdaFunction{Find}\AgdaSpace{}%
\AgdaBound{x}\AgdaSymbol{)}\AgdaSpace{}%
\AgdaSymbol{(}\AgdaFunction{Find}\AgdaSpace{}%
\AgdaBound{y}\AgdaSymbol{)}\<%
\\
\>[0]\AgdaFunction{Find}\AgdaSpace{}%
\AgdaSymbol{\{}\AgdaOperator{\AgdaInductiveConstructor{𝟙/}}\AgdaSpace{}%
\AgdaBound{t}\AgdaSymbol{\}}\AgdaSpace{}%
\AgdaInductiveConstructor{↻}\AgdaSpace{}%
\AgdaSymbol{=}\AgdaSpace{}%
\AgdaInductiveConstructor{here}\AgdaSpace{}%
\AgdaInductiveConstructor{refl}\<%
\\
\\[\AgdaEmptyExtraSkip]%
\>[0]\AgdaFunction{Find'}\AgdaSpace{}%
\AgdaSymbol{:}\AgdaSpace{}%
\AgdaSymbol{\{}\AgdaBound{A}\AgdaSpace{}%
\AgdaSymbol{:}\AgdaSpace{}%
\AgdaDatatype{𝕌}\AgdaSymbol{\}}\AgdaSpace{}%
\AgdaSymbol{(}\AgdaBound{x}\AgdaSpace{}%
\AgdaSymbol{:}\AgdaSpace{}%
\AgdaOperator{\AgdaFunction{⟦}}\AgdaSpace{}%
\AgdaBound{A}\AgdaSpace{}%
\AgdaOperator{\AgdaFunction{⟧}}\AgdaSymbol{)}\AgdaSpace{}%
\AgdaSymbol{→}\AgdaSpace{}%
\AgdaDatatype{Fin}\AgdaSpace{}%
\AgdaOperator{\AgdaFunction{∣}}\AgdaSpace{}%
\AgdaBound{A}\AgdaSpace{}%
\AgdaOperator{\AgdaFunction{∣}}\<%
\\
\>[0]\AgdaFunction{Find'}\AgdaSpace{}%
\AgdaSymbol{=}\AgdaSpace{}%
\AgdaFunction{index}\AgdaSpace{}%
\AgdaOperator{\AgdaFunction{∘}}\AgdaSpace{}%
\AgdaFunction{Find}\<%
\\
\\[\AgdaEmptyExtraSkip]%
\>[0]\AgdaKeyword{data}\AgdaSpace{}%
\AgdaDatatype{State}\AgdaSpace{}%
\AgdaSymbol{(}\AgdaBound{A}\AgdaSpace{}%
\AgdaSymbol{:}\AgdaSpace{}%
\AgdaDatatype{𝕌}\AgdaSymbol{)}\AgdaSpace{}%
\AgdaSymbol{:}\AgdaSpace{}%
\AgdaPrimitiveType{Set}\AgdaSpace{}%
\AgdaKeyword{where}\<%
\\
\>[0][@{}l@{\AgdaIndent{0}}]%
\>[2]\AgdaOperator{\AgdaInductiveConstructor{⟪\AgdaUnderscore{}[\AgdaUnderscore{}]⟫}}\AgdaSpace{}%
\AgdaSymbol{:}\AgdaSpace{}%
\AgdaDatatype{Vec}\AgdaSpace{}%
\AgdaOperator{\AgdaFunction{⟦}}\AgdaSpace{}%
\AgdaBound{A}\AgdaSpace{}%
\AgdaOperator{\AgdaFunction{⟧}}\AgdaSpace{}%
\AgdaOperator{\AgdaFunction{∣}}\AgdaSpace{}%
\AgdaBound{A}\AgdaSpace{}%
\AgdaOperator{\AgdaFunction{∣}}\AgdaSpace{}%
\AgdaSymbol{→}\AgdaSpace{}%
\AgdaDatatype{Fin}\AgdaSpace{}%
\AgdaOperator{\AgdaFunction{∣}}\AgdaSpace{}%
\AgdaBound{A}\AgdaSpace{}%
\AgdaOperator{\AgdaFunction{∣}}\AgdaSpace{}%
\AgdaSymbol{→}\AgdaSpace{}%
\AgdaDatatype{State}\AgdaSpace{}%
\AgdaBound{A}\<%
\\
\\[\AgdaEmptyExtraSkip]%
\>[0]\AgdaFunction{resolve}\AgdaSpace{}%
\AgdaSymbol{:}\AgdaSpace{}%
\AgdaSymbol{\{}\AgdaBound{A}\AgdaSpace{}%
\AgdaSymbol{:}\AgdaSpace{}%
\AgdaDatatype{𝕌}\AgdaSymbol{\}}\AgdaSpace{}%
\AgdaSymbol{→}\AgdaSpace{}%
\AgdaDatatype{State}\AgdaSpace{}%
\AgdaBound{A}\AgdaSpace{}%
\AgdaSymbol{→}\AgdaSpace{}%
\AgdaOperator{\AgdaFunction{⟦}}\AgdaSpace{}%
\AgdaBound{A}\AgdaSpace{}%
\AgdaOperator{\AgdaFunction{⟧}}\<%
\\
\>[0]\AgdaFunction{resolve}\AgdaSpace{}%
\AgdaOperator{\AgdaInductiveConstructor{⟪}}\AgdaSpace{}%
\AgdaBound{v}\AgdaSpace{}%
\AgdaOperator{\AgdaInductiveConstructor{[}}\AgdaSpace{}%
\AgdaBound{i}\AgdaSpace{}%
\AgdaOperator{\AgdaInductiveConstructor{]⟫}}\AgdaSpace{}%
\AgdaSymbol{=}\AgdaSpace{}%
\AgdaFunction{lookup}\AgdaSpace{}%
\AgdaBound{v}\AgdaSpace{}%
\AgdaBound{i}\<%
\\
\\[\AgdaEmptyExtraSkip]%
\>[0]\AgdaFunction{st}\AgdaSpace{}%
\AgdaSymbol{:}\AgdaSpace{}%
\AgdaSymbol{\{}\AgdaBound{A}\AgdaSpace{}%
\AgdaBound{B}\AgdaSpace{}%
\AgdaSymbol{:}\AgdaSpace{}%
\AgdaDatatype{𝕌}\AgdaSymbol{\}}\AgdaSpace{}%
\AgdaSymbol{→}\AgdaSpace{}%
\AgdaOperator{\AgdaFunction{⟦}}\AgdaSpace{}%
\AgdaBound{A}\AgdaSpace{}%
\AgdaOperator{\AgdaFunction{⟧}}\AgdaSpace{}%
\AgdaSymbol{→}\AgdaSpace{}%
\AgdaSymbol{(}\AgdaBound{c}\AgdaSpace{}%
\AgdaSymbol{:}\AgdaSpace{}%
\AgdaBound{A}\AgdaSpace{}%
\AgdaOperator{\AgdaDatatype{⟷}}\AgdaSpace{}%
\AgdaBound{B}\AgdaSymbol{)}\AgdaSpace{}%
\AgdaSymbol{→}\AgdaSpace{}%
\AgdaFunction{Σ[}\AgdaSpace{}%
\AgdaBound{C}\AgdaSpace{}%
\AgdaFunction{∈}\AgdaSpace{}%
\AgdaDatatype{𝕌}\AgdaSpace{}%
\AgdaFunction{]}\AgdaSpace{}%
\AgdaSymbol{(}\AgdaBound{C}\AgdaSpace{}%
\AgdaOperator{\AgdaDatatype{⟷}}\AgdaSpace{}%
\AgdaBound{B}\AgdaSpace{}%
\AgdaOperator{\AgdaFunction{×}}\AgdaSpace{}%
\AgdaDatatype{State}\AgdaSpace{}%
\AgdaBound{C}\AgdaSymbol{)}\<%
\\
\>[0]\AgdaFunction{st}\AgdaSpace{}%
\AgdaSymbol{(}\AgdaInductiveConstructor{inj₂}\AgdaSpace{}%
\AgdaBound{y}\AgdaSymbol{)}\AgdaSpace{}%
\AgdaSymbol{(}\AgdaInductiveConstructor{unite₊l}\AgdaSpace{}%
\AgdaSymbol{\{}\AgdaBound{t}\AgdaSymbol{\})}%
\>[44]\AgdaSymbol{=}\AgdaSpace{}%
\AgdaSymbol{\AgdaUnderscore{}}\AgdaSpace{}%
\AgdaOperator{\AgdaInductiveConstructor{,}}\AgdaSpace{}%
\AgdaInductiveConstructor{id⟷}\AgdaSpace{}%
\AgdaOperator{\AgdaInductiveConstructor{,}}\AgdaSpace{}%
\AgdaOperator{\AgdaInductiveConstructor{⟪}}\AgdaSpace{}%
\AgdaFunction{Enum}\AgdaSpace{}%
\AgdaBound{t}\AgdaSpace{}%
\AgdaOperator{\AgdaInductiveConstructor{[}}\AgdaSpace{}%
\AgdaFunction{Find'}\AgdaSpace{}%
\AgdaBound{y}\AgdaSpace{}%
\AgdaOperator{\AgdaInductiveConstructor{]⟫}}\<%
\\
\>[0]\AgdaFunction{st}\AgdaSpace{}%
\AgdaBound{a}\AgdaSpace{}%
\AgdaSymbol{(}\AgdaInductiveConstructor{uniti₊l}\AgdaSpace{}%
\AgdaSymbol{\{}\AgdaBound{t}\AgdaSymbol{\})}%
\>[44]\AgdaSymbol{=}\AgdaSpace{}%
\AgdaSymbol{\AgdaUnderscore{}}\AgdaSpace{}%
\AgdaOperator{\AgdaInductiveConstructor{,}}\AgdaSpace{}%
\AgdaInductiveConstructor{id⟷}\AgdaSpace{}%
\AgdaOperator{\AgdaInductiveConstructor{,}}\AgdaSpace{}%
\AgdaOperator{\AgdaInductiveConstructor{⟪}}\AgdaSpace{}%
\AgdaSymbol{(}\AgdaFunction{Enum}\AgdaSpace{}%
\AgdaSymbol{(}\AgdaInductiveConstructor{𝟘}\AgdaSpace{}%
\AgdaOperator{\AgdaInductiveConstructor{+ᵤ}}\AgdaSpace{}%
\AgdaBound{t}\AgdaSymbol{))}\AgdaSpace{}%
\AgdaOperator{\AgdaInductiveConstructor{[}}\AgdaSpace{}%
\AgdaFunction{Find'}\AgdaSpace{}%
\AgdaBound{a}\AgdaSpace{}%
\AgdaOperator{\AgdaInductiveConstructor{]⟫}}\<%
\\
\>[0]\AgdaFunction{st}\AgdaSpace{}%
\AgdaSymbol{(}\AgdaInductiveConstructor{inj₁}\AgdaSpace{}%
\AgdaBound{x}\AgdaSymbol{)}\AgdaSpace{}%
\AgdaSymbol{(}\AgdaInductiveConstructor{unite₊r}\AgdaSpace{}%
\AgdaSymbol{\{}\AgdaBound{t}\AgdaSymbol{\})}%
\>[44]\AgdaSymbol{=}\AgdaSpace{}%
\AgdaSymbol{\AgdaUnderscore{}}\AgdaSpace{}%
\AgdaOperator{\AgdaInductiveConstructor{,}}\AgdaSpace{}%
\AgdaInductiveConstructor{id⟷}\AgdaSpace{}%
\AgdaOperator{\AgdaInductiveConstructor{,}}\AgdaSpace{}%
\AgdaOperator{\AgdaInductiveConstructor{⟪}}\AgdaSpace{}%
\AgdaFunction{Enum}\AgdaSpace{}%
\AgdaBound{t}\AgdaSpace{}%
\AgdaOperator{\AgdaInductiveConstructor{[}}\AgdaSpace{}%
\AgdaFunction{Find'}\AgdaSpace{}%
\AgdaBound{x}\AgdaSpace{}%
\AgdaOperator{\AgdaInductiveConstructor{]⟫}}\<%
\\
\>[0]\AgdaFunction{st}\AgdaSpace{}%
\AgdaBound{a}\AgdaSpace{}%
\AgdaSymbol{(}\AgdaInductiveConstructor{uniti₊r}\AgdaSpace{}%
\AgdaSymbol{\{}\AgdaBound{t}\AgdaSymbol{\})}%
\>[44]\AgdaSymbol{=}\AgdaSpace{}%
\AgdaSymbol{\AgdaUnderscore{}}\AgdaSpace{}%
\AgdaOperator{\AgdaInductiveConstructor{,}}\AgdaSpace{}%
\AgdaInductiveConstructor{id⟷}\AgdaSpace{}%
\AgdaOperator{\AgdaInductiveConstructor{,}}\AgdaSpace{}%
\AgdaOperator{\AgdaInductiveConstructor{⟪}}\AgdaSpace{}%
\AgdaSymbol{(}\AgdaFunction{Enum}\AgdaSpace{}%
\AgdaSymbol{(}\AgdaBound{t}\AgdaSpace{}%
\AgdaOperator{\AgdaInductiveConstructor{+ᵤ}}\AgdaSpace{}%
\AgdaInductiveConstructor{𝟘}\AgdaSymbol{))}\AgdaSpace{}%
\AgdaOperator{\AgdaInductiveConstructor{[}}\AgdaSpace{}%
\AgdaFunction{Find'}\AgdaSpace{}%
\AgdaSymbol{\{}\AgdaBound{t}\AgdaSpace{}%
\AgdaOperator{\AgdaInductiveConstructor{+ᵤ}}\AgdaSpace{}%
\AgdaInductiveConstructor{𝟘}\AgdaSymbol{\}}\AgdaSpace{}%
\AgdaSymbol{(}\AgdaInductiveConstructor{inj₁}\AgdaSpace{}%
\AgdaBound{a}\AgdaSymbol{)}\AgdaSpace{}%
\AgdaOperator{\AgdaInductiveConstructor{]⟫}}\<%
\\
\>[0]\AgdaFunction{st}\AgdaSpace{}%
\AgdaSymbol{(}\AgdaInductiveConstructor{inj₁}\AgdaSpace{}%
\AgdaBound{x}\AgdaSymbol{)}\AgdaSpace{}%
\AgdaSymbol{(}\AgdaInductiveConstructor{swap₊}\AgdaSpace{}%
\AgdaSymbol{\{}\AgdaBound{t₁}\AgdaSymbol{\}}\AgdaSpace{}%
\AgdaSymbol{\{}\AgdaBound{t₂}\AgdaSymbol{\})}%
\>[44]\AgdaSymbol{=}\AgdaSpace{}%
\AgdaSymbol{\AgdaUnderscore{}}\AgdaSpace{}%
\AgdaOperator{\AgdaInductiveConstructor{,}}\AgdaSpace{}%
\AgdaInductiveConstructor{id⟷}\AgdaSpace{}%
\AgdaOperator{\AgdaInductiveConstructor{,}}\AgdaSpace{}%
\AgdaOperator{\AgdaInductiveConstructor{⟪}}\AgdaSpace{}%
\AgdaFunction{Enum}\AgdaSpace{}%
\AgdaSymbol{\AgdaUnderscore{}}\AgdaSpace{}%
\AgdaOperator{\AgdaInductiveConstructor{[}}\AgdaSpace{}%
\AgdaFunction{Find'}\AgdaSpace{}%
\AgdaSymbol{\{}\AgdaBound{t₂}\AgdaSpace{}%
\AgdaOperator{\AgdaInductiveConstructor{+ᵤ}}\AgdaSpace{}%
\AgdaBound{t₁}\AgdaSymbol{\}}\AgdaSpace{}%
\AgdaSymbol{(}\AgdaInductiveConstructor{inj₂}\AgdaSpace{}%
\AgdaBound{x}\AgdaSymbol{)}\AgdaSpace{}%
\AgdaOperator{\AgdaInductiveConstructor{]⟫}}\<%
\\
\>[0]\AgdaFunction{st}\AgdaSpace{}%
\AgdaSymbol{(}\AgdaInductiveConstructor{inj₂}\AgdaSpace{}%
\AgdaBound{y}\AgdaSymbol{)}\AgdaSpace{}%
\AgdaSymbol{(}\AgdaInductiveConstructor{swap₊}\AgdaSpace{}%
\AgdaSymbol{\{}\AgdaBound{t₁}\AgdaSymbol{\}}\AgdaSpace{}%
\AgdaSymbol{\{}\AgdaBound{t₂}\AgdaSymbol{\})}%
\>[44]\AgdaSymbol{=}\AgdaSpace{}%
\AgdaSymbol{\AgdaUnderscore{}}\AgdaSpace{}%
\AgdaOperator{\AgdaInductiveConstructor{,}}\AgdaSpace{}%
\AgdaInductiveConstructor{id⟷}\AgdaSpace{}%
\AgdaOperator{\AgdaInductiveConstructor{,}}\AgdaSpace{}%
\AgdaOperator{\AgdaInductiveConstructor{⟪}}\AgdaSpace{}%
\AgdaFunction{Enum}\AgdaSpace{}%
\AgdaSymbol{\AgdaUnderscore{}}\AgdaSpace{}%
\AgdaOperator{\AgdaInductiveConstructor{[}}\AgdaSpace{}%
\AgdaFunction{Find'}\AgdaSpace{}%
\AgdaSymbol{\{}\AgdaBound{t₂}\AgdaSpace{}%
\AgdaOperator{\AgdaInductiveConstructor{+ᵤ}}\AgdaSpace{}%
\AgdaBound{t₁}\AgdaSymbol{\}}\AgdaSpace{}%
\AgdaSymbol{(}\AgdaInductiveConstructor{inj₁}\AgdaSpace{}%
\AgdaBound{y}\AgdaSymbol{)}\AgdaSpace{}%
\AgdaOperator{\AgdaInductiveConstructor{]⟫}}\<%
\\
\>[0]\AgdaFunction{st}\AgdaSpace{}%
\AgdaSymbol{(}\AgdaInductiveConstructor{inj₁}\AgdaSpace{}%
\AgdaBound{x}\AgdaSymbol{)}\AgdaSpace{}%
\AgdaSymbol{(}\AgdaInductiveConstructor{assocl₊}\AgdaSpace{}%
\AgdaSymbol{\{}\AgdaBound{t₁}\AgdaSymbol{\}}\AgdaSpace{}%
\AgdaSymbol{\{}\AgdaBound{t₂}\AgdaSymbol{\}}\AgdaSpace{}%
\AgdaSymbol{\{}\AgdaBound{t₃}\AgdaSymbol{\})}%
\>[44]\AgdaSymbol{=}\AgdaSpace{}%
\AgdaSymbol{\AgdaUnderscore{}}\AgdaSpace{}%
\AgdaOperator{\AgdaInductiveConstructor{,}}\AgdaSpace{}%
\AgdaInductiveConstructor{id⟷}\AgdaSpace{}%
\AgdaOperator{\AgdaInductiveConstructor{,}}\AgdaSpace{}%
\AgdaOperator{\AgdaInductiveConstructor{⟪}}\AgdaSpace{}%
\AgdaFunction{Enum}\AgdaSpace{}%
\AgdaSymbol{\AgdaUnderscore{}}\AgdaSpace{}%
\AgdaOperator{\AgdaInductiveConstructor{[}}\AgdaSpace{}%
\AgdaFunction{Find'}\AgdaSpace{}%
\AgdaSymbol{\{(}\AgdaBound{t₁}\AgdaSpace{}%
\AgdaOperator{\AgdaInductiveConstructor{+ᵤ}}\AgdaSpace{}%
\AgdaBound{t₂}\AgdaSymbol{)}\AgdaSpace{}%
\AgdaOperator{\AgdaInductiveConstructor{+ᵤ}}\AgdaSpace{}%
\AgdaBound{t₃}\AgdaSymbol{\}}\AgdaSpace{}%
\AgdaSymbol{(}\AgdaInductiveConstructor{inj₁}\AgdaSpace{}%
\AgdaSymbol{(}\AgdaInductiveConstructor{inj₁}\AgdaSpace{}%
\AgdaBound{x}\AgdaSymbol{))}\AgdaSpace{}%
\AgdaOperator{\AgdaInductiveConstructor{]⟫}}\<%
\\
\>[0]\AgdaFunction{st}\AgdaSpace{}%
\AgdaSymbol{(}\AgdaInductiveConstructor{inj₂}\AgdaSpace{}%
\AgdaSymbol{(}\AgdaInductiveConstructor{inj₁}\AgdaSpace{}%
\AgdaBound{x}\AgdaSymbol{))}\AgdaSpace{}%
\AgdaSymbol{(}\AgdaInductiveConstructor{assocl₊}\AgdaSpace{}%
\AgdaSymbol{\{}\AgdaBound{t₁}\AgdaSymbol{\}}\AgdaSpace{}%
\AgdaSymbol{\{}\AgdaBound{t₂}\AgdaSymbol{\}}\AgdaSpace{}%
\AgdaSymbol{\{}\AgdaBound{t₃}\AgdaSymbol{\})}\AgdaSpace{}%
\AgdaSymbol{=}\AgdaSpace{}%
\AgdaSymbol{\AgdaUnderscore{}}\AgdaSpace{}%
\AgdaOperator{\AgdaInductiveConstructor{,}}\AgdaSpace{}%
\AgdaInductiveConstructor{id⟷}\AgdaSpace{}%
\AgdaOperator{\AgdaInductiveConstructor{,}}\AgdaSpace{}%
\AgdaOperator{\AgdaInductiveConstructor{⟪}}\AgdaSpace{}%
\AgdaFunction{Enum}\AgdaSpace{}%
\AgdaSymbol{\AgdaUnderscore{}}\AgdaSpace{}%
\AgdaOperator{\AgdaInductiveConstructor{[}}\AgdaSpace{}%
\AgdaFunction{Find'}\AgdaSpace{}%
\AgdaSymbol{\{(}\AgdaBound{t₁}\AgdaSpace{}%
\AgdaOperator{\AgdaInductiveConstructor{+ᵤ}}\AgdaSpace{}%
\AgdaBound{t₂}\AgdaSymbol{)}\AgdaSpace{}%
\AgdaOperator{\AgdaInductiveConstructor{+ᵤ}}\AgdaSpace{}%
\AgdaBound{t₃}\AgdaSymbol{\}}\AgdaSpace{}%
\AgdaSymbol{(}\AgdaInductiveConstructor{inj₁}\AgdaSpace{}%
\AgdaSymbol{(}\AgdaInductiveConstructor{inj₂}\AgdaSpace{}%
\AgdaBound{x}\AgdaSymbol{))}\AgdaSpace{}%
\AgdaOperator{\AgdaInductiveConstructor{]⟫}}\<%
\\
\>[0]\AgdaFunction{st}\AgdaSpace{}%
\AgdaSymbol{(}\AgdaInductiveConstructor{inj₂}\AgdaSpace{}%
\AgdaSymbol{(}\AgdaInductiveConstructor{inj₂}\AgdaSpace{}%
\AgdaBound{y}\AgdaSymbol{))}\AgdaSpace{}%
\AgdaSymbol{(}\AgdaInductiveConstructor{assocl₊}\AgdaSpace{}%
\AgdaSymbol{\{}\AgdaBound{t₁}\AgdaSymbol{\}}\AgdaSpace{}%
\AgdaSymbol{\{}\AgdaBound{t₂}\AgdaSymbol{\}}\AgdaSpace{}%
\AgdaSymbol{\{}\AgdaBound{t₃}\AgdaSymbol{\})}\AgdaSpace{}%
\AgdaSymbol{=}\AgdaSpace{}%
\AgdaSymbol{\AgdaUnderscore{}}\AgdaSpace{}%
\AgdaOperator{\AgdaInductiveConstructor{,}}\AgdaSpace{}%
\AgdaInductiveConstructor{id⟷}\AgdaSpace{}%
\AgdaOperator{\AgdaInductiveConstructor{,}}\AgdaSpace{}%
\AgdaOperator{\AgdaInductiveConstructor{⟪}}\AgdaSpace{}%
\AgdaFunction{Enum}\AgdaSpace{}%
\AgdaSymbol{\AgdaUnderscore{}}\AgdaSpace{}%
\AgdaOperator{\AgdaInductiveConstructor{[}}\AgdaSpace{}%
\AgdaFunction{Find'}\AgdaSpace{}%
\AgdaSymbol{\{(}\AgdaBound{t₁}\AgdaSpace{}%
\AgdaOperator{\AgdaInductiveConstructor{+ᵤ}}\AgdaSpace{}%
\AgdaBound{t₂}\AgdaSymbol{)}\AgdaSpace{}%
\AgdaOperator{\AgdaInductiveConstructor{+ᵤ}}\AgdaSpace{}%
\AgdaBound{t₃}\AgdaSymbol{\}}\AgdaSpace{}%
\AgdaSymbol{(}\AgdaInductiveConstructor{inj₂}\AgdaSpace{}%
\AgdaBound{y}\AgdaSymbol{)}\AgdaSpace{}%
\AgdaOperator{\AgdaInductiveConstructor{]⟫}}\<%
\\
\>[0]\AgdaFunction{st}\AgdaSpace{}%
\AgdaSymbol{(}\AgdaInductiveConstructor{inj₁}\AgdaSpace{}%
\AgdaSymbol{(}\AgdaInductiveConstructor{inj₁}\AgdaSpace{}%
\AgdaBound{x}\AgdaSymbol{))}\AgdaSpace{}%
\AgdaSymbol{(}\AgdaInductiveConstructor{assocr₊}\AgdaSpace{}%
\AgdaSymbol{\{}\AgdaBound{t₁}\AgdaSymbol{\}}\AgdaSpace{}%
\AgdaSymbol{\{}\AgdaBound{t₂}\AgdaSymbol{\}}\AgdaSpace{}%
\AgdaSymbol{\{}\AgdaBound{t₃}\AgdaSymbol{\})}\AgdaSpace{}%
\AgdaSymbol{=}\AgdaSpace{}%
\AgdaSymbol{\AgdaUnderscore{}}\AgdaSpace{}%
\AgdaOperator{\AgdaInductiveConstructor{,}}\AgdaSpace{}%
\AgdaInductiveConstructor{id⟷}\AgdaSpace{}%
\AgdaOperator{\AgdaInductiveConstructor{,}}\AgdaSpace{}%
\AgdaOperator{\AgdaInductiveConstructor{⟪}}\AgdaSpace{}%
\AgdaFunction{Enum}\AgdaSpace{}%
\AgdaSymbol{\AgdaUnderscore{}}\AgdaSpace{}%
\AgdaOperator{\AgdaInductiveConstructor{[}}\AgdaSpace{}%
\AgdaFunction{Find'}\AgdaSpace{}%
\AgdaSymbol{\{}\AgdaBound{t₁}\AgdaSpace{}%
\AgdaOperator{\AgdaInductiveConstructor{+ᵤ}}\AgdaSpace{}%
\AgdaBound{t₂}\AgdaSpace{}%
\AgdaOperator{\AgdaInductiveConstructor{+ᵤ}}\AgdaSpace{}%
\AgdaBound{t₃}\AgdaSymbol{\}}\AgdaSpace{}%
\AgdaSymbol{(}\AgdaInductiveConstructor{inj₁}\AgdaSpace{}%
\AgdaBound{x}\AgdaSymbol{)}\AgdaSpace{}%
\AgdaOperator{\AgdaInductiveConstructor{]⟫}}\<%
\\
\>[0]\AgdaFunction{st}\AgdaSpace{}%
\AgdaSymbol{(}\AgdaInductiveConstructor{inj₁}\AgdaSpace{}%
\AgdaSymbol{(}\AgdaInductiveConstructor{inj₂}\AgdaSpace{}%
\AgdaBound{y}\AgdaSymbol{))}\AgdaSpace{}%
\AgdaSymbol{(}\AgdaInductiveConstructor{assocr₊}\AgdaSpace{}%
\AgdaSymbol{\{}\AgdaBound{t₁}\AgdaSymbol{\}}\AgdaSpace{}%
\AgdaSymbol{\{}\AgdaBound{t₂}\AgdaSymbol{\}}\AgdaSpace{}%
\AgdaSymbol{\{}\AgdaBound{t₃}\AgdaSymbol{\})}\AgdaSpace{}%
\AgdaSymbol{=}\AgdaSpace{}%
\AgdaSymbol{\AgdaUnderscore{}}\AgdaSpace{}%
\AgdaOperator{\AgdaInductiveConstructor{,}}\AgdaSpace{}%
\AgdaInductiveConstructor{id⟷}\AgdaSpace{}%
\AgdaOperator{\AgdaInductiveConstructor{,}}\AgdaSpace{}%
\AgdaOperator{\AgdaInductiveConstructor{⟪}}\AgdaSpace{}%
\AgdaFunction{Enum}\AgdaSpace{}%
\AgdaSymbol{\AgdaUnderscore{}}\AgdaSpace{}%
\AgdaOperator{\AgdaInductiveConstructor{[}}\AgdaSpace{}%
\AgdaFunction{Find'}\AgdaSpace{}%
\AgdaSymbol{\{}\AgdaBound{t₁}\AgdaSpace{}%
\AgdaOperator{\AgdaInductiveConstructor{+ᵤ}}\AgdaSpace{}%
\AgdaBound{t₂}\AgdaSpace{}%
\AgdaOperator{\AgdaInductiveConstructor{+ᵤ}}\AgdaSpace{}%
\AgdaBound{t₃}\AgdaSymbol{\}}\AgdaSpace{}%
\AgdaSymbol{(}\AgdaInductiveConstructor{inj₂}\AgdaSpace{}%
\AgdaSymbol{(}\AgdaInductiveConstructor{inj₁}\AgdaSpace{}%
\AgdaBound{y}\AgdaSymbol{))}\AgdaSpace{}%
\AgdaOperator{\AgdaInductiveConstructor{]⟫}}\<%
\\
\>[0]\AgdaFunction{st}\AgdaSpace{}%
\AgdaSymbol{(}\AgdaInductiveConstructor{inj₂}\AgdaSpace{}%
\AgdaBound{y}\AgdaSymbol{)}\AgdaSpace{}%
\AgdaSymbol{(}\AgdaInductiveConstructor{assocr₊}\AgdaSpace{}%
\AgdaSymbol{\{}\AgdaBound{t₁}\AgdaSymbol{\}}\AgdaSpace{}%
\AgdaSymbol{\{}\AgdaBound{t₂}\AgdaSymbol{\}}\AgdaSpace{}%
\AgdaSymbol{\{}\AgdaBound{t₃}\AgdaSymbol{\})}%
\>[44]\AgdaSymbol{=}\AgdaSpace{}%
\AgdaSymbol{\AgdaUnderscore{}}\AgdaSpace{}%
\AgdaOperator{\AgdaInductiveConstructor{,}}\AgdaSpace{}%
\AgdaInductiveConstructor{id⟷}\AgdaSpace{}%
\AgdaOperator{\AgdaInductiveConstructor{,}}\AgdaSpace{}%
\AgdaOperator{\AgdaInductiveConstructor{⟪}}\AgdaSpace{}%
\AgdaFunction{Enum}\AgdaSpace{}%
\AgdaSymbol{\AgdaUnderscore{}}\AgdaSpace{}%
\AgdaOperator{\AgdaInductiveConstructor{[}}\AgdaSpace{}%
\AgdaFunction{Find'}\AgdaSpace{}%
\AgdaSymbol{\{}\AgdaBound{t₁}\AgdaSpace{}%
\AgdaOperator{\AgdaInductiveConstructor{+ᵤ}}\AgdaSpace{}%
\AgdaBound{t₂}\AgdaSpace{}%
\AgdaOperator{\AgdaInductiveConstructor{+ᵤ}}\AgdaSpace{}%
\AgdaBound{t₃}\AgdaSymbol{\}}\AgdaSpace{}%
\AgdaSymbol{(}\AgdaInductiveConstructor{inj₂}\AgdaSpace{}%
\AgdaSymbol{(}\AgdaInductiveConstructor{inj₂}\AgdaSpace{}%
\AgdaBound{y}\AgdaSymbol{))}\AgdaSpace{}%
\AgdaOperator{\AgdaInductiveConstructor{]⟫}}\<%
\\
\>[0]\AgdaFunction{st}\AgdaSpace{}%
\AgdaSymbol{(}\AgdaInductiveConstructor{tt}\AgdaSpace{}%
\AgdaOperator{\AgdaInductiveConstructor{,}}\AgdaSpace{}%
\AgdaBound{y}\AgdaSymbol{)}\AgdaSpace{}%
\AgdaInductiveConstructor{unite⋆l}%
\>[44]\AgdaSymbol{=}\AgdaSpace{}%
\AgdaSymbol{\AgdaUnderscore{}}\AgdaSpace{}%
\AgdaOperator{\AgdaInductiveConstructor{,}}\AgdaSpace{}%
\AgdaInductiveConstructor{id⟷}\AgdaSpace{}%
\AgdaOperator{\AgdaInductiveConstructor{,}}\AgdaSpace{}%
\AgdaOperator{\AgdaInductiveConstructor{⟪}}\AgdaSpace{}%
\AgdaFunction{Enum}\AgdaSpace{}%
\AgdaSymbol{\AgdaUnderscore{}}\AgdaSpace{}%
\AgdaOperator{\AgdaInductiveConstructor{[}}\AgdaSpace{}%
\AgdaFunction{Find'}\AgdaSpace{}%
\AgdaBound{y}\AgdaSpace{}%
\AgdaOperator{\AgdaInductiveConstructor{]⟫}}\<%
\\
\>[0]\AgdaFunction{st}\AgdaSpace{}%
\AgdaBound{a}\AgdaSpace{}%
\AgdaInductiveConstructor{uniti⋆l}%
\>[44]\AgdaSymbol{=}\AgdaSpace{}%
\AgdaSymbol{\AgdaUnderscore{}}\AgdaSpace{}%
\AgdaOperator{\AgdaInductiveConstructor{,}}\AgdaSpace{}%
\AgdaInductiveConstructor{id⟷}\AgdaSpace{}%
\AgdaOperator{\AgdaInductiveConstructor{,}}\AgdaSpace{}%
\AgdaOperator{\AgdaInductiveConstructor{⟪}}\AgdaSpace{}%
\AgdaFunction{Enum}\AgdaSpace{}%
\AgdaSymbol{\AgdaUnderscore{}}\AgdaSpace{}%
\AgdaOperator{\AgdaInductiveConstructor{[}}\AgdaSpace{}%
\AgdaFunction{Find'}\AgdaSpace{}%
\AgdaSymbol{(}\AgdaInductiveConstructor{tt}\AgdaSpace{}%
\AgdaOperator{\AgdaInductiveConstructor{,}}\AgdaSpace{}%
\AgdaBound{a}\AgdaSymbol{)}\AgdaSpace{}%
\AgdaOperator{\AgdaInductiveConstructor{]⟫}}\<%
\\
\>[0]\AgdaFunction{st}\AgdaSpace{}%
\AgdaSymbol{(}\AgdaBound{x}\AgdaSpace{}%
\AgdaOperator{\AgdaInductiveConstructor{,}}\AgdaSpace{}%
\AgdaInductiveConstructor{tt}\AgdaSymbol{)}\AgdaSpace{}%
\AgdaInductiveConstructor{unite⋆r}%
\>[44]\AgdaSymbol{=}\AgdaSpace{}%
\AgdaSymbol{\AgdaUnderscore{}}\AgdaSpace{}%
\AgdaOperator{\AgdaInductiveConstructor{,}}\AgdaSpace{}%
\AgdaInductiveConstructor{id⟷}\AgdaSpace{}%
\AgdaOperator{\AgdaInductiveConstructor{,}}\AgdaSpace{}%
\AgdaOperator{\AgdaInductiveConstructor{⟪}}\AgdaSpace{}%
\AgdaFunction{Enum}\AgdaSpace{}%
\AgdaSymbol{\AgdaUnderscore{}}\AgdaSpace{}%
\AgdaOperator{\AgdaInductiveConstructor{[}}\AgdaSpace{}%
\AgdaFunction{Find'}\AgdaSpace{}%
\AgdaBound{x}\AgdaSpace{}%
\AgdaOperator{\AgdaInductiveConstructor{]⟫}}\<%
\\
\>[0]\AgdaFunction{st}\AgdaSpace{}%
\AgdaBound{a}\AgdaSpace{}%
\AgdaInductiveConstructor{uniti⋆r}%
\>[44]\AgdaSymbol{=}\AgdaSpace{}%
\AgdaSymbol{\AgdaUnderscore{}}\AgdaSpace{}%
\AgdaOperator{\AgdaInductiveConstructor{,}}\AgdaSpace{}%
\AgdaInductiveConstructor{id⟷}\AgdaSpace{}%
\AgdaOperator{\AgdaInductiveConstructor{,}}\AgdaSpace{}%
\AgdaOperator{\AgdaInductiveConstructor{⟪}}\AgdaSpace{}%
\AgdaFunction{Enum}\AgdaSpace{}%
\AgdaSymbol{\AgdaUnderscore{}}\AgdaSpace{}%
\AgdaOperator{\AgdaInductiveConstructor{[}}\AgdaSpace{}%
\AgdaFunction{Find'}\AgdaSpace{}%
\AgdaSymbol{(}\AgdaBound{a}\AgdaSpace{}%
\AgdaOperator{\AgdaInductiveConstructor{,}}\AgdaSpace{}%
\AgdaInductiveConstructor{tt}\AgdaSymbol{)}\AgdaSpace{}%
\AgdaOperator{\AgdaInductiveConstructor{]⟫}}\<%
\\
\>[0]\AgdaFunction{st}\AgdaSpace{}%
\AgdaSymbol{(}\AgdaBound{x}\AgdaSpace{}%
\AgdaOperator{\AgdaInductiveConstructor{,}}\AgdaSpace{}%
\AgdaBound{y}\AgdaSymbol{)}\AgdaSpace{}%
\AgdaInductiveConstructor{swap⋆}%
\>[44]\AgdaSymbol{=}\AgdaSpace{}%
\AgdaSymbol{\AgdaUnderscore{}}\AgdaSpace{}%
\AgdaOperator{\AgdaInductiveConstructor{,}}\AgdaSpace{}%
\AgdaInductiveConstructor{id⟷}\AgdaSpace{}%
\AgdaOperator{\AgdaInductiveConstructor{,}}\AgdaSpace{}%
\AgdaOperator{\AgdaInductiveConstructor{⟪}}\AgdaSpace{}%
\AgdaFunction{Enum}\AgdaSpace{}%
\AgdaSymbol{\AgdaUnderscore{}}\AgdaSpace{}%
\AgdaOperator{\AgdaInductiveConstructor{[}}\AgdaSpace{}%
\AgdaFunction{Find'}\AgdaSpace{}%
\AgdaSymbol{(}\AgdaBound{y}\AgdaSpace{}%
\AgdaOperator{\AgdaInductiveConstructor{,}}\AgdaSpace{}%
\AgdaBound{x}\AgdaSymbol{)}\AgdaSpace{}%
\AgdaOperator{\AgdaInductiveConstructor{]⟫}}\<%
\\
\>[0]\AgdaFunction{st}\AgdaSpace{}%
\AgdaSymbol{(}\AgdaBound{x}\AgdaSpace{}%
\AgdaOperator{\AgdaInductiveConstructor{,}}\AgdaSpace{}%
\AgdaBound{y}\AgdaSpace{}%
\AgdaOperator{\AgdaInductiveConstructor{,}}\AgdaSpace{}%
\AgdaBound{z}\AgdaSymbol{)}\AgdaSpace{}%
\AgdaInductiveConstructor{assocl⋆}%
\>[44]\AgdaSymbol{=}\AgdaSpace{}%
\AgdaSymbol{\AgdaUnderscore{}}\AgdaSpace{}%
\AgdaOperator{\AgdaInductiveConstructor{,}}\AgdaSpace{}%
\AgdaInductiveConstructor{id⟷}\AgdaSpace{}%
\AgdaOperator{\AgdaInductiveConstructor{,}}\AgdaSpace{}%
\AgdaOperator{\AgdaInductiveConstructor{⟪}}\AgdaSpace{}%
\AgdaFunction{Enum}\AgdaSpace{}%
\AgdaSymbol{\AgdaUnderscore{}}\AgdaSpace{}%
\AgdaOperator{\AgdaInductiveConstructor{[}}\AgdaSpace{}%
\AgdaFunction{Find'}\AgdaSpace{}%
\AgdaSymbol{((}\AgdaBound{x}\AgdaSpace{}%
\AgdaOperator{\AgdaInductiveConstructor{,}}\AgdaSpace{}%
\AgdaBound{y}\AgdaSymbol{)}\AgdaSpace{}%
\AgdaOperator{\AgdaInductiveConstructor{,}}\AgdaSpace{}%
\AgdaBound{z}\AgdaSymbol{)}\AgdaSpace{}%
\AgdaOperator{\AgdaInductiveConstructor{]⟫}}\<%
\\
\>[0]\AgdaFunction{st}\AgdaSpace{}%
\AgdaSymbol{((}\AgdaBound{x}\AgdaSpace{}%
\AgdaOperator{\AgdaInductiveConstructor{,}}\AgdaSpace{}%
\AgdaBound{y}\AgdaSymbol{)}\AgdaSpace{}%
\AgdaOperator{\AgdaInductiveConstructor{,}}\AgdaSpace{}%
\AgdaBound{z}\AgdaSymbol{)}\AgdaSpace{}%
\AgdaInductiveConstructor{assocr⋆}%
\>[44]\AgdaSymbol{=}\AgdaSpace{}%
\AgdaSymbol{\AgdaUnderscore{}}\AgdaSpace{}%
\AgdaOperator{\AgdaInductiveConstructor{,}}\AgdaSpace{}%
\AgdaInductiveConstructor{id⟷}\AgdaSpace{}%
\AgdaOperator{\AgdaInductiveConstructor{,}}\AgdaSpace{}%
\AgdaOperator{\AgdaInductiveConstructor{⟪}}\AgdaSpace{}%
\AgdaFunction{Enum}\AgdaSpace{}%
\AgdaSymbol{\AgdaUnderscore{}}\AgdaSpace{}%
\AgdaOperator{\AgdaInductiveConstructor{[}}\AgdaSpace{}%
\AgdaFunction{Find'}\AgdaSpace{}%
\AgdaSymbol{(}\AgdaBound{x}\AgdaSpace{}%
\AgdaOperator{\AgdaInductiveConstructor{,}}\AgdaSpace{}%
\AgdaBound{y}\AgdaSpace{}%
\AgdaOperator{\AgdaInductiveConstructor{,}}\AgdaSpace{}%
\AgdaBound{z}\AgdaSymbol{)}\AgdaSpace{}%
\AgdaOperator{\AgdaInductiveConstructor{]⟫}}\<%
\\
\>[0]\AgdaFunction{st}\AgdaSpace{}%
\AgdaSymbol{(}\AgdaInductiveConstructor{inj₁}\AgdaSpace{}%
\AgdaBound{x}\AgdaSpace{}%
\AgdaOperator{\AgdaInductiveConstructor{,}}\AgdaSpace{}%
\AgdaBound{y}\AgdaSymbol{)}\AgdaSpace{}%
\AgdaSymbol{(}\AgdaInductiveConstructor{dist}\AgdaSpace{}%
\AgdaSymbol{\{}\AgdaBound{t₁}\AgdaSymbol{\}}\AgdaSpace{}%
\AgdaSymbol{\{}\AgdaBound{t₂}\AgdaSymbol{\}}\AgdaSpace{}%
\AgdaSymbol{\{}\AgdaBound{t₃}\AgdaSymbol{\})}%
\>[44]\AgdaSymbol{=}\AgdaSpace{}%
\AgdaSymbol{\AgdaUnderscore{}}\AgdaSpace{}%
\AgdaOperator{\AgdaInductiveConstructor{,}}\AgdaSpace{}%
\AgdaInductiveConstructor{id⟷}\AgdaSpace{}%
\AgdaOperator{\AgdaInductiveConstructor{,}}\AgdaSpace{}%
\AgdaOperator{\AgdaInductiveConstructor{⟪}}\AgdaSpace{}%
\AgdaFunction{Enum}\AgdaSpace{}%
\AgdaSymbol{\AgdaUnderscore{}}\AgdaSpace{}%
\AgdaOperator{\AgdaInductiveConstructor{[}}\AgdaSpace{}%
\AgdaFunction{Find'}\AgdaSpace{}%
\AgdaSymbol{\{}\AgdaBound{t₁}\AgdaSpace{}%
\AgdaOperator{\AgdaInductiveConstructor{×ᵤ}}\AgdaSpace{}%
\AgdaBound{t₃}\AgdaSpace{}%
\AgdaOperator{\AgdaInductiveConstructor{+ᵤ}}\AgdaSpace{}%
\AgdaBound{t₂}\AgdaSpace{}%
\AgdaOperator{\AgdaInductiveConstructor{×ᵤ}}\AgdaSpace{}%
\AgdaBound{t₃}\AgdaSymbol{\}}\AgdaSpace{}%
\AgdaSymbol{(}\AgdaInductiveConstructor{inj₁}\AgdaSpace{}%
\AgdaSymbol{(}\AgdaBound{x}\AgdaSpace{}%
\AgdaOperator{\AgdaInductiveConstructor{,}}\AgdaSpace{}%
\AgdaBound{y}\AgdaSymbol{))}\AgdaSpace{}%
\AgdaOperator{\AgdaInductiveConstructor{]⟫}}\<%
\\
\>[0]\AgdaFunction{st}\AgdaSpace{}%
\AgdaSymbol{(}\AgdaInductiveConstructor{inj₂}\AgdaSpace{}%
\AgdaBound{x}\AgdaSpace{}%
\AgdaOperator{\AgdaInductiveConstructor{,}}\AgdaSpace{}%
\AgdaBound{y}\AgdaSymbol{)}\AgdaSpace{}%
\AgdaSymbol{(}\AgdaInductiveConstructor{dist}\AgdaSpace{}%
\AgdaSymbol{\{}\AgdaBound{t₁}\AgdaSymbol{\}}\AgdaSpace{}%
\AgdaSymbol{\{}\AgdaBound{t₂}\AgdaSymbol{\}}\AgdaSpace{}%
\AgdaSymbol{\{}\AgdaBound{t₃}\AgdaSymbol{\})}%
\>[44]\AgdaSymbol{=}\AgdaSpace{}%
\AgdaSymbol{\AgdaUnderscore{}}\AgdaSpace{}%
\AgdaOperator{\AgdaInductiveConstructor{,}}\AgdaSpace{}%
\AgdaInductiveConstructor{id⟷}\AgdaSpace{}%
\AgdaOperator{\AgdaInductiveConstructor{,}}\AgdaSpace{}%
\AgdaOperator{\AgdaInductiveConstructor{⟪}}\AgdaSpace{}%
\AgdaFunction{Enum}\AgdaSpace{}%
\AgdaSymbol{\AgdaUnderscore{}}\AgdaSpace{}%
\AgdaOperator{\AgdaInductiveConstructor{[}}\AgdaSpace{}%
\AgdaFunction{Find'}\AgdaSpace{}%
\AgdaSymbol{\{}\AgdaBound{t₁}\AgdaSpace{}%
\AgdaOperator{\AgdaInductiveConstructor{×ᵤ}}\AgdaSpace{}%
\AgdaBound{t₃}\AgdaSpace{}%
\AgdaOperator{\AgdaInductiveConstructor{+ᵤ}}\AgdaSpace{}%
\AgdaBound{t₂}\AgdaSpace{}%
\AgdaOperator{\AgdaInductiveConstructor{×ᵤ}}\AgdaSpace{}%
\AgdaBound{t₃}\AgdaSymbol{\}}\AgdaSpace{}%
\AgdaSymbol{(}\AgdaInductiveConstructor{inj₂}\AgdaSpace{}%
\AgdaSymbol{(}\AgdaBound{x}\AgdaSpace{}%
\AgdaOperator{\AgdaInductiveConstructor{,}}\AgdaSpace{}%
\AgdaBound{y}\AgdaSymbol{))}\AgdaSpace{}%
\AgdaOperator{\AgdaInductiveConstructor{]⟫}}\<%
\\
\>[0]\AgdaFunction{st}\AgdaSpace{}%
\AgdaSymbol{(}\AgdaInductiveConstructor{inj₁}\AgdaSpace{}%
\AgdaSymbol{(}\AgdaBound{x}\AgdaSpace{}%
\AgdaOperator{\AgdaInductiveConstructor{,}}\AgdaSpace{}%
\AgdaBound{y}\AgdaSymbol{))}\AgdaSpace{}%
\AgdaSymbol{(}\AgdaInductiveConstructor{factor}\AgdaSpace{}%
\AgdaSymbol{\{}\AgdaBound{t₁}\AgdaSymbol{\}}\AgdaSpace{}%
\AgdaSymbol{\{}\AgdaBound{t₂}\AgdaSymbol{\}}\AgdaSpace{}%
\AgdaSymbol{\{}\AgdaBound{t₃}\AgdaSymbol{\})}%
\>[44]\AgdaSymbol{=}\AgdaSpace{}%
\AgdaSymbol{\AgdaUnderscore{}}\AgdaSpace{}%
\AgdaOperator{\AgdaInductiveConstructor{,}}\AgdaSpace{}%
\AgdaInductiveConstructor{id⟷}\AgdaSpace{}%
\AgdaOperator{\AgdaInductiveConstructor{,}}\AgdaSpace{}%
\AgdaOperator{\AgdaInductiveConstructor{⟪}}\AgdaSpace{}%
\AgdaFunction{Enum}\AgdaSpace{}%
\AgdaSymbol{\AgdaUnderscore{}}\AgdaSpace{}%
\AgdaOperator{\AgdaInductiveConstructor{[}}\AgdaSpace{}%
\AgdaFunction{Find'}\AgdaSpace{}%
\AgdaSymbol{\{(}\AgdaBound{t₁}\AgdaSpace{}%
\AgdaOperator{\AgdaInductiveConstructor{+ᵤ}}\AgdaSpace{}%
\AgdaBound{t₂}\AgdaSymbol{)}\AgdaSpace{}%
\AgdaOperator{\AgdaInductiveConstructor{×ᵤ}}\AgdaSpace{}%
\AgdaBound{t₃}\AgdaSymbol{\}}\AgdaSpace{}%
\AgdaSymbol{(}\AgdaInductiveConstructor{inj₁}\AgdaSpace{}%
\AgdaBound{x}\AgdaSpace{}%
\AgdaOperator{\AgdaInductiveConstructor{,}}\AgdaSpace{}%
\AgdaBound{y}\AgdaSymbol{)}\AgdaSpace{}%
\AgdaOperator{\AgdaInductiveConstructor{]⟫}}\<%
\\
\>[0]\AgdaFunction{st}\AgdaSpace{}%
\AgdaSymbol{(}\AgdaInductiveConstructor{inj₂}\AgdaSpace{}%
\AgdaSymbol{(}\AgdaBound{y}\AgdaSpace{}%
\AgdaOperator{\AgdaInductiveConstructor{,}}\AgdaSpace{}%
\AgdaBound{z}\AgdaSymbol{))}\AgdaSpace{}%
\AgdaSymbol{(}\AgdaInductiveConstructor{factor}\AgdaSpace{}%
\AgdaSymbol{\{}\AgdaBound{t₁}\AgdaSymbol{\}}\AgdaSpace{}%
\AgdaSymbol{\{}\AgdaBound{t₂}\AgdaSymbol{\}}\AgdaSpace{}%
\AgdaSymbol{\{}\AgdaBound{t₃}\AgdaSymbol{\})}%
\>[44]\AgdaSymbol{=}\AgdaSpace{}%
\AgdaSymbol{\AgdaUnderscore{}}\AgdaSpace{}%
\AgdaOperator{\AgdaInductiveConstructor{,}}\AgdaSpace{}%
\AgdaInductiveConstructor{id⟷}\AgdaSpace{}%
\AgdaOperator{\AgdaInductiveConstructor{,}}\AgdaSpace{}%
\AgdaOperator{\AgdaInductiveConstructor{⟪}}\AgdaSpace{}%
\AgdaFunction{Enum}\AgdaSpace{}%
\AgdaSymbol{\AgdaUnderscore{}}\AgdaSpace{}%
\AgdaOperator{\AgdaInductiveConstructor{[}}\AgdaSpace{}%
\AgdaFunction{Find'}\AgdaSpace{}%
\AgdaSymbol{\{(}\AgdaBound{t₁}\AgdaSpace{}%
\AgdaOperator{\AgdaInductiveConstructor{+ᵤ}}\AgdaSpace{}%
\AgdaBound{t₂}\AgdaSymbol{)}\AgdaSpace{}%
\AgdaOperator{\AgdaInductiveConstructor{×ᵤ}}\AgdaSpace{}%
\AgdaBound{t₃}\AgdaSymbol{\}}\AgdaSpace{}%
\AgdaSymbol{(}\AgdaInductiveConstructor{inj₂}\AgdaSpace{}%
\AgdaBound{y}\AgdaSpace{}%
\AgdaOperator{\AgdaInductiveConstructor{,}}\AgdaSpace{}%
\AgdaBound{z}\AgdaSymbol{)}\AgdaSpace{}%
\AgdaOperator{\AgdaInductiveConstructor{]⟫}}\<%
\\
\>[0]\AgdaFunction{st}\AgdaSpace{}%
\AgdaSymbol{(}\AgdaBound{x}\AgdaSpace{}%
\AgdaOperator{\AgdaInductiveConstructor{,}}\AgdaSpace{}%
\AgdaInductiveConstructor{inj₁}\AgdaSpace{}%
\AgdaBound{y}\AgdaSymbol{)}\AgdaSpace{}%
\AgdaSymbol{(}\AgdaInductiveConstructor{distl}\AgdaSpace{}%
\AgdaSymbol{\{}\AgdaBound{t₁}\AgdaSymbol{\}}\AgdaSpace{}%
\AgdaSymbol{\{}\AgdaBound{t₂}\AgdaSymbol{\}}\AgdaSpace{}%
\AgdaSymbol{\{}\AgdaBound{t₃}\AgdaSymbol{\})}%
\>[44]\AgdaSymbol{=}\AgdaSpace{}%
\AgdaSymbol{\AgdaUnderscore{}}\AgdaSpace{}%
\AgdaOperator{\AgdaInductiveConstructor{,}}\AgdaSpace{}%
\AgdaInductiveConstructor{id⟷}\AgdaSpace{}%
\AgdaOperator{\AgdaInductiveConstructor{,}}\AgdaSpace{}%
\AgdaOperator{\AgdaInductiveConstructor{⟪}}\AgdaSpace{}%
\AgdaFunction{Enum}\AgdaSpace{}%
\AgdaSymbol{\AgdaUnderscore{}}\AgdaSpace{}%
\AgdaOperator{\AgdaInductiveConstructor{[}}\AgdaSpace{}%
\AgdaFunction{Find'}\AgdaSpace{}%
\AgdaSymbol{\{(}\AgdaBound{t₁}\AgdaSpace{}%
\AgdaOperator{\AgdaInductiveConstructor{×ᵤ}}\AgdaSpace{}%
\AgdaBound{t₂}\AgdaSymbol{)}\AgdaSpace{}%
\AgdaOperator{\AgdaInductiveConstructor{+ᵤ}}\AgdaSpace{}%
\AgdaSymbol{(}\AgdaBound{t₁}\AgdaSpace{}%
\AgdaOperator{\AgdaInductiveConstructor{×ᵤ}}\AgdaSpace{}%
\AgdaBound{t₃}\AgdaSymbol{)\}}\AgdaSpace{}%
\AgdaSymbol{(}\AgdaInductiveConstructor{inj₁}\AgdaSpace{}%
\AgdaSymbol{(}\AgdaBound{x}\AgdaSpace{}%
\AgdaOperator{\AgdaInductiveConstructor{,}}\AgdaSpace{}%
\AgdaBound{y}\AgdaSymbol{))}\AgdaSpace{}%
\AgdaOperator{\AgdaInductiveConstructor{]⟫}}\<%
\\
\>[0]\AgdaFunction{st}\AgdaSpace{}%
\AgdaSymbol{(}\AgdaBound{x}\AgdaSpace{}%
\AgdaOperator{\AgdaInductiveConstructor{,}}\AgdaSpace{}%
\AgdaInductiveConstructor{inj₂}\AgdaSpace{}%
\AgdaBound{y}\AgdaSymbol{)}\AgdaSpace{}%
\AgdaSymbol{(}\AgdaInductiveConstructor{distl}\AgdaSpace{}%
\AgdaSymbol{\{}\AgdaBound{t₁}\AgdaSymbol{\}}\AgdaSpace{}%
\AgdaSymbol{\{}\AgdaBound{t₂}\AgdaSymbol{\}}\AgdaSpace{}%
\AgdaSymbol{\{}\AgdaBound{t₃}\AgdaSymbol{\})}%
\>[44]\AgdaSymbol{=}\AgdaSpace{}%
\AgdaSymbol{\AgdaUnderscore{}}\AgdaSpace{}%
\AgdaOperator{\AgdaInductiveConstructor{,}}\AgdaSpace{}%
\AgdaInductiveConstructor{id⟷}\AgdaSpace{}%
\AgdaOperator{\AgdaInductiveConstructor{,}}\AgdaSpace{}%
\AgdaOperator{\AgdaInductiveConstructor{⟪}}\AgdaSpace{}%
\AgdaFunction{Enum}\AgdaSpace{}%
\AgdaSymbol{\AgdaUnderscore{}}\AgdaSpace{}%
\AgdaOperator{\AgdaInductiveConstructor{[}}\AgdaSpace{}%
\AgdaFunction{Find'}\AgdaSpace{}%
\AgdaSymbol{\{(}\AgdaBound{t₁}\AgdaSpace{}%
\AgdaOperator{\AgdaInductiveConstructor{×ᵤ}}\AgdaSpace{}%
\AgdaBound{t₂}\AgdaSymbol{)}\AgdaSpace{}%
\AgdaOperator{\AgdaInductiveConstructor{+ᵤ}}\AgdaSpace{}%
\AgdaSymbol{(}\AgdaBound{t₁}\AgdaSpace{}%
\AgdaOperator{\AgdaInductiveConstructor{×ᵤ}}\AgdaSpace{}%
\AgdaBound{t₃}\AgdaSymbol{)\}}\AgdaSpace{}%
\AgdaSymbol{(}\AgdaInductiveConstructor{inj₂}\AgdaSpace{}%
\AgdaSymbol{(}\AgdaBound{x}\AgdaSpace{}%
\AgdaOperator{\AgdaInductiveConstructor{,}}\AgdaSpace{}%
\AgdaBound{y}\AgdaSymbol{))}\AgdaSpace{}%
\AgdaOperator{\AgdaInductiveConstructor{]⟫}}\<%
\\
\>[0]\AgdaFunction{st}\AgdaSpace{}%
\AgdaSymbol{(}\AgdaInductiveConstructor{inj₁}\AgdaSpace{}%
\AgdaSymbol{(}\AgdaBound{x}\AgdaSpace{}%
\AgdaOperator{\AgdaInductiveConstructor{,}}\AgdaSpace{}%
\AgdaBound{y}\AgdaSymbol{))}\AgdaSpace{}%
\AgdaSymbol{(}\AgdaInductiveConstructor{factorl}\AgdaSpace{}%
\AgdaSymbol{\{}\AgdaBound{t₁}\AgdaSymbol{\}}\AgdaSpace{}%
\AgdaSymbol{\{}\AgdaBound{t₂}\AgdaSymbol{\}}\AgdaSpace{}%
\AgdaSymbol{\{}\AgdaBound{t₃}\AgdaSymbol{\})}%
\>[44]\AgdaSymbol{=}\AgdaSpace{}%
\AgdaSymbol{\AgdaUnderscore{}}\AgdaSpace{}%
\AgdaOperator{\AgdaInductiveConstructor{,}}\AgdaSpace{}%
\AgdaInductiveConstructor{id⟷}\AgdaSpace{}%
\AgdaOperator{\AgdaInductiveConstructor{,}}\AgdaSpace{}%
\AgdaOperator{\AgdaInductiveConstructor{⟪}}\AgdaSpace{}%
\AgdaFunction{Enum}\AgdaSpace{}%
\AgdaSymbol{\AgdaUnderscore{}}\AgdaSpace{}%
\AgdaOperator{\AgdaInductiveConstructor{[}}\AgdaSpace{}%
\AgdaFunction{Find'}\AgdaSpace{}%
\AgdaSymbol{\{}\AgdaBound{t₁}\AgdaSpace{}%
\AgdaOperator{\AgdaInductiveConstructor{×ᵤ}}\AgdaSpace{}%
\AgdaSymbol{(}\AgdaBound{t₂}\AgdaSpace{}%
\AgdaOperator{\AgdaInductiveConstructor{+ᵤ}}\AgdaSpace{}%
\AgdaBound{t₃}\AgdaSymbol{)\}}\AgdaSpace{}%
\AgdaSymbol{(}\AgdaBound{x}\AgdaSpace{}%
\AgdaOperator{\AgdaInductiveConstructor{,}}\AgdaSpace{}%
\AgdaInductiveConstructor{inj₁}\AgdaSpace{}%
\AgdaBound{y}\AgdaSymbol{)}\AgdaSpace{}%
\AgdaOperator{\AgdaInductiveConstructor{]⟫}}\<%
\\
\>[0]\AgdaFunction{st}\AgdaSpace{}%
\AgdaSymbol{(}\AgdaInductiveConstructor{inj₂}\AgdaSpace{}%
\AgdaSymbol{(}\AgdaBound{x}\AgdaSpace{}%
\AgdaOperator{\AgdaInductiveConstructor{,}}\AgdaSpace{}%
\AgdaBound{z}\AgdaSymbol{))}\AgdaSpace{}%
\AgdaSymbol{(}\AgdaInductiveConstructor{factorl}\AgdaSpace{}%
\AgdaSymbol{\{}\AgdaBound{t₁}\AgdaSymbol{\}}\AgdaSpace{}%
\AgdaSymbol{\{}\AgdaBound{t₂}\AgdaSymbol{\}}\AgdaSpace{}%
\AgdaSymbol{\{}\AgdaBound{t₃}\AgdaSymbol{\})}%
\>[44]\AgdaSymbol{=}\AgdaSpace{}%
\AgdaSymbol{\AgdaUnderscore{}}\AgdaSpace{}%
\AgdaOperator{\AgdaInductiveConstructor{,}}\AgdaSpace{}%
\AgdaInductiveConstructor{id⟷}\AgdaSpace{}%
\AgdaOperator{\AgdaInductiveConstructor{,}}\AgdaSpace{}%
\AgdaOperator{\AgdaInductiveConstructor{⟪}}\AgdaSpace{}%
\AgdaFunction{Enum}\AgdaSpace{}%
\AgdaSymbol{\AgdaUnderscore{}}\AgdaSpace{}%
\AgdaOperator{\AgdaInductiveConstructor{[}}\AgdaSpace{}%
\AgdaFunction{Find'}\AgdaSpace{}%
\AgdaSymbol{\{}\AgdaBound{t₁}\AgdaSpace{}%
\AgdaOperator{\AgdaInductiveConstructor{×ᵤ}}\AgdaSpace{}%
\AgdaSymbol{(}\AgdaBound{t₂}\AgdaSpace{}%
\AgdaOperator{\AgdaInductiveConstructor{+ᵤ}}\AgdaSpace{}%
\AgdaBound{t₃}\AgdaSymbol{)\}}\AgdaSpace{}%
\AgdaSymbol{(}\AgdaBound{x}\AgdaSpace{}%
\AgdaOperator{\AgdaInductiveConstructor{,}}\AgdaSpace{}%
\AgdaInductiveConstructor{inj₂}\AgdaSpace{}%
\AgdaBound{z}\AgdaSymbol{)}\AgdaSpace{}%
\AgdaOperator{\AgdaInductiveConstructor{]⟫}}\<%
\\
\>[0]\AgdaFunction{st}\AgdaSpace{}%
\AgdaInductiveConstructor{tt}\AgdaSpace{}%
\AgdaSymbol{(}\AgdaInductiveConstructor{η}\AgdaSpace{}%
\AgdaSymbol{\{}\AgdaBound{t}\AgdaSymbol{\}}\AgdaSpace{}%
\AgdaBound{v}\AgdaSymbol{)}%
\>[42]\AgdaSymbol{=}\AgdaSpace{}%
\AgdaSymbol{\AgdaUnderscore{}}\AgdaSpace{}%
\AgdaOperator{\AgdaInductiveConstructor{,}}\AgdaSpace{}%
\AgdaInductiveConstructor{id⟷}\AgdaSpace{}%
\AgdaOperator{\AgdaInductiveConstructor{,}}\AgdaSpace{}%
\AgdaOperator{\AgdaInductiveConstructor{⟪}}\AgdaSpace{}%
\AgdaFunction{Enum}\AgdaSpace{}%
\AgdaSymbol{(}\AgdaBound{t}\AgdaSpace{}%
\AgdaOperator{\AgdaInductiveConstructor{×ᵤ}}\AgdaSpace{}%
\AgdaSymbol{(}\AgdaOperator{\AgdaInductiveConstructor{𝟙/}}\AgdaSpace{}%
\AgdaBound{v}\AgdaSymbol{))}\AgdaSpace{}%
\AgdaOperator{\AgdaInductiveConstructor{[}}\AgdaSpace{}%
\AgdaFunction{Find'}\AgdaSpace{}%
\AgdaSymbol{\{}\AgdaBound{t}\AgdaSpace{}%
\AgdaOperator{\AgdaInductiveConstructor{×ᵤ}}\AgdaSpace{}%
\AgdaSymbol{(}\AgdaOperator{\AgdaInductiveConstructor{𝟙/}}\AgdaSpace{}%
\AgdaBound{v}\AgdaSymbol{)\}}\AgdaSpace{}%
\AgdaSymbol{(}\AgdaBound{v}\AgdaSpace{}%
\AgdaOperator{\AgdaInductiveConstructor{,}}\AgdaSpace{}%
\AgdaInductiveConstructor{↻}\AgdaSymbol{)}\AgdaSpace{}%
\AgdaOperator{\AgdaInductiveConstructor{]⟫}}\<%
\\
\>[0]\AgdaFunction{st}\AgdaSpace{}%
\AgdaSymbol{(}\AgdaBound{x}\AgdaSpace{}%
\AgdaOperator{\AgdaInductiveConstructor{,}}\AgdaSpace{}%
\AgdaBound{○}\AgdaSymbol{)}\AgdaSpace{}%
\AgdaSymbol{(}\AgdaInductiveConstructor{ε}\AgdaSpace{}%
\AgdaSymbol{\{}\AgdaBound{t}\AgdaSymbol{\}}\AgdaSpace{}%
\AgdaBound{v}\AgdaSymbol{)}\AgdaSpace{}%
\AgdaKeyword{with}\AgdaSpace{}%
\AgdaBound{v}\AgdaSpace{}%
\AgdaOperator{\AgdaFunction{≟ᵤ}}\AgdaSpace{}%
\AgdaBound{x}\<%
\\
\>[0]\AgdaFunction{st}\AgdaSpace{}%
\AgdaSymbol{(}\AgdaBound{x}\AgdaSpace{}%
\AgdaOperator{\AgdaInductiveConstructor{,}}\AgdaSpace{}%
\AgdaBound{○}\AgdaSymbol{)}\AgdaSpace{}%
\AgdaSymbol{(}\AgdaInductiveConstructor{ε}\AgdaSpace{}%
\AgdaSymbol{\{}\AgdaBound{t}\AgdaSymbol{\}}\AgdaSpace{}%
\AgdaBound{v}\AgdaSymbol{)}\AgdaSpace{}%
\AgdaSymbol{|}\AgdaSpace{}%
\AgdaInductiveConstructor{yes}\AgdaSpace{}%
\AgdaSymbol{\AgdaUnderscore{}}\AgdaSpace{}%
\AgdaSymbol{=}\AgdaSpace{}%
\AgdaSymbol{\AgdaUnderscore{}}\AgdaSpace{}%
\AgdaOperator{\AgdaInductiveConstructor{,}}\AgdaSpace{}%
\AgdaInductiveConstructor{id⟷}\AgdaSpace{}%
\AgdaOperator{\AgdaInductiveConstructor{,}}\AgdaSpace{}%
\AgdaOperator{\AgdaInductiveConstructor{⟪}}\AgdaSpace{}%
\AgdaFunction{Enum}\AgdaSpace{}%
\AgdaSymbol{\AgdaUnderscore{}}\AgdaSpace{}%
\AgdaOperator{\AgdaInductiveConstructor{[}}\AgdaSpace{}%
\AgdaFunction{Find'}\AgdaSpace{}%
\AgdaInductiveConstructor{tt}\AgdaSpace{}%
\AgdaOperator{\AgdaInductiveConstructor{]⟫}}\<%
\\
\>[0]\AgdaFunction{st}\AgdaSpace{}%
\AgdaSymbol{(}\AgdaBound{x}\AgdaSpace{}%
\AgdaOperator{\AgdaInductiveConstructor{,}}\AgdaSpace{}%
\AgdaBound{○}\AgdaSymbol{)}\AgdaSpace{}%
\AgdaSymbol{(}\AgdaInductiveConstructor{ε}\AgdaSpace{}%
\AgdaSymbol{\{}\AgdaBound{t}\AgdaSymbol{\}}\AgdaSpace{}%
\AgdaBound{v}\AgdaSymbol{)}\AgdaSpace{}%
\AgdaSymbol{|}\AgdaSpace{}%
\AgdaInductiveConstructor{no}%
\>[27]\AgdaSymbol{\AgdaUnderscore{}}\AgdaSpace{}%
\AgdaSymbol{=}\AgdaSpace{}%
\AgdaSymbol{\AgdaUnderscore{}}\AgdaSpace{}%
\AgdaOperator{\AgdaInductiveConstructor{,}}\AgdaSpace{}%
\AgdaSymbol{(}\AgdaInductiveConstructor{ε}\AgdaSpace{}%
\AgdaSymbol{\{}\AgdaBound{t}\AgdaSymbol{\}}\AgdaSpace{}%
\AgdaBound{v}\AgdaSymbol{)}\AgdaSpace{}%
\AgdaOperator{\AgdaInductiveConstructor{,}}\AgdaSpace{}%
\AgdaOperator{\AgdaInductiveConstructor{⟪}}\AgdaSpace{}%
\AgdaFunction{Enum}\AgdaSpace{}%
\AgdaSymbol{(}\AgdaBound{t}\AgdaSpace{}%
\AgdaOperator{\AgdaInductiveConstructor{×ᵤ}}\AgdaSpace{}%
\AgdaSymbol{(}\AgdaOperator{\AgdaInductiveConstructor{𝟙/}}\AgdaSpace{}%
\AgdaBound{v}\AgdaSymbol{))}\AgdaSpace{}%
\AgdaOperator{\AgdaInductiveConstructor{[}}\AgdaSpace{}%
\AgdaFunction{Find'}\AgdaSpace{}%
\AgdaSymbol{\{}\AgdaBound{t}\AgdaSpace{}%
\AgdaOperator{\AgdaInductiveConstructor{×ᵤ}}\AgdaSpace{}%
\AgdaSymbol{(}\AgdaOperator{\AgdaInductiveConstructor{𝟙/}}\AgdaSpace{}%
\AgdaBound{v}\AgdaSymbol{)\}}\AgdaSpace{}%
\AgdaSymbol{(}\AgdaBound{x}\AgdaSpace{}%
\AgdaOperator{\AgdaInductiveConstructor{,}}\AgdaSpace{}%
\AgdaBound{○}\AgdaSymbol{)}\AgdaSpace{}%
\AgdaOperator{\AgdaInductiveConstructor{]⟫}}\<%
\\
\>[0]\AgdaFunction{st}\AgdaSpace{}%
\AgdaBound{a}\AgdaSpace{}%
\AgdaInductiveConstructor{id⟷}%
\>[44]\AgdaSymbol{=}\AgdaSpace{}%
\AgdaSymbol{\AgdaUnderscore{}}\AgdaSpace{}%
\AgdaOperator{\AgdaInductiveConstructor{,}}\AgdaSpace{}%
\AgdaInductiveConstructor{id⟷}\AgdaSpace{}%
\AgdaOperator{\AgdaInductiveConstructor{,}}\AgdaSpace{}%
\AgdaOperator{\AgdaInductiveConstructor{⟪}}\AgdaSpace{}%
\AgdaFunction{Enum}\AgdaSpace{}%
\AgdaSymbol{\AgdaUnderscore{}}\AgdaSpace{}%
\AgdaOperator{\AgdaInductiveConstructor{[}}\AgdaSpace{}%
\AgdaFunction{Find'}\AgdaSpace{}%
\AgdaBound{a}\AgdaSpace{}%
\AgdaOperator{\AgdaInductiveConstructor{]⟫}}\<%
\\
\>[0]\AgdaFunction{st}\AgdaSpace{}%
\AgdaBound{a}\AgdaSpace{}%
\AgdaSymbol{(}\AgdaInductiveConstructor{id⟷}\AgdaSpace{}%
\AgdaOperator{\AgdaInductiveConstructor{⊚}}\AgdaSpace{}%
\AgdaBound{c}\AgdaSymbol{)}%
\>[44]\AgdaSymbol{=}\AgdaSpace{}%
\AgdaSymbol{\AgdaUnderscore{}}\AgdaSpace{}%
\AgdaOperator{\AgdaInductiveConstructor{,}}\AgdaSpace{}%
\AgdaBound{c}\AgdaSpace{}%
\AgdaOperator{\AgdaInductiveConstructor{,}}\AgdaSpace{}%
\AgdaOperator{\AgdaInductiveConstructor{⟪}}\AgdaSpace{}%
\AgdaFunction{Enum}\AgdaSpace{}%
\AgdaSymbol{\AgdaUnderscore{}}\AgdaSpace{}%
\AgdaOperator{\AgdaInductiveConstructor{[}}\AgdaSpace{}%
\AgdaFunction{Find'}\AgdaSpace{}%
\AgdaBound{a}\AgdaSpace{}%
\AgdaOperator{\AgdaInductiveConstructor{]⟫}}\<%
\\
\>[0]\AgdaCatchallClause{\AgdaFunction{st}}\AgdaSpace{}%
\AgdaCatchallClause{\AgdaBound{a}}\AgdaSpace{}%
\AgdaCatchallClause{\AgdaSymbol{(}}\AgdaCatchallClause{\AgdaBound{c₁}}\AgdaSpace{}%
\AgdaCatchallClause{\AgdaOperator{\AgdaInductiveConstructor{⊚}}}\AgdaSpace{}%
\AgdaCatchallClause{\AgdaBound{c₂}}\AgdaCatchallClause{\AgdaSymbol{)}}%
\>[44]\AgdaSymbol{=}%
\>[1027I]\AgdaKeyword{let}\AgdaSpace{}%
\AgdaSymbol{\AgdaUnderscore{}}\AgdaSpace{}%
\AgdaOperator{\AgdaInductiveConstructor{,}}\AgdaSpace{}%
\AgdaBound{c}\AgdaSpace{}%
\AgdaOperator{\AgdaInductiveConstructor{,}}\AgdaSpace{}%
\AgdaBound{st'}\AgdaSpace{}%
\AgdaSymbol{=}\AgdaSpace{}%
\AgdaFunction{st}\AgdaSpace{}%
\AgdaBound{a}\AgdaSpace{}%
\AgdaBound{c₁}\AgdaSpace{}%
\AgdaKeyword{in}\<%
\\
\>[.][@{}l@{}]\<[1027I]%
\>[46]\AgdaSymbol{\AgdaUnderscore{}}\AgdaSpace{}%
\AgdaOperator{\AgdaInductiveConstructor{,}}\AgdaSpace{}%
\AgdaBound{c}\AgdaSpace{}%
\AgdaOperator{\AgdaInductiveConstructor{⊚}}\AgdaSpace{}%
\AgdaBound{c₂}\AgdaSpace{}%
\AgdaOperator{\AgdaInductiveConstructor{,}}\AgdaSpace{}%
\AgdaBound{st'}\<%
\\
\>[0]\AgdaFunction{st}\AgdaSpace{}%
\AgdaSymbol{(}\AgdaInductiveConstructor{inj₁}\AgdaSpace{}%
\AgdaBound{x}\AgdaSymbol{)}\AgdaSpace{}%
\AgdaSymbol{(}\AgdaOperator{\AgdaInductiveConstructor{\AgdaUnderscore{}⊕\AgdaUnderscore{}}}\AgdaSpace{}%
\AgdaSymbol{\{}\AgdaBound{t₁}\AgdaSymbol{\}}\AgdaSpace{}%
\AgdaSymbol{\{}\AgdaBound{t₂}\AgdaSymbol{\}}\AgdaSpace{}%
\AgdaSymbol{\{\AgdaUnderscore{}\}}\AgdaSpace{}%
\AgdaSymbol{\{}\AgdaBound{t₄}\AgdaSymbol{\}}\AgdaSpace{}%
\AgdaInductiveConstructor{id⟷}\AgdaSpace{}%
\AgdaBound{c₂}\AgdaSymbol{)}\AgdaSpace{}%
\AgdaSymbol{=}\AgdaSpace{}%
\AgdaSymbol{\AgdaUnderscore{}}\AgdaSpace{}%
\AgdaOperator{\AgdaInductiveConstructor{,}}\AgdaSpace{}%
\AgdaInductiveConstructor{id⟷}\AgdaSpace{}%
\AgdaOperator{\AgdaInductiveConstructor{,}}\AgdaSpace{}%
\AgdaOperator{\AgdaInductiveConstructor{⟪}}\AgdaSpace{}%
\AgdaFunction{Enum}\AgdaSpace{}%
\AgdaSymbol{\AgdaUnderscore{}}\AgdaSpace{}%
\AgdaOperator{\AgdaInductiveConstructor{[}}\AgdaSpace{}%
\AgdaFunction{Find'}\AgdaSpace{}%
\AgdaSymbol{\{\AgdaUnderscore{}}\AgdaSpace{}%
\AgdaOperator{\AgdaInductiveConstructor{+ᵤ}}\AgdaSpace{}%
\AgdaBound{t₄}\AgdaSymbol{\}}\AgdaSpace{}%
\AgdaSymbol{(}\AgdaInductiveConstructor{inj₁}\AgdaSpace{}%
\AgdaBound{x}\AgdaSymbol{)}\AgdaSpace{}%
\AgdaOperator{\AgdaInductiveConstructor{]⟫}}\<%
\\
\>[0]\AgdaCatchallClause{\AgdaFunction{st}}\AgdaSpace{}%
\AgdaCatchallClause{\AgdaSymbol{(}}\AgdaCatchallClause{\AgdaInductiveConstructor{inj₁}}\AgdaSpace{}%
\AgdaCatchallClause{\AgdaBound{x}}\AgdaCatchallClause{\AgdaSymbol{)}}\AgdaSpace{}%
\AgdaCatchallClause{\AgdaSymbol{(}}\AgdaCatchallClause{\AgdaOperator{\AgdaInductiveConstructor{\AgdaUnderscore{}⊕\AgdaUnderscore{}}}}\AgdaSpace{}%
\AgdaCatchallClause{\AgdaSymbol{\{}}\AgdaCatchallClause{\AgdaBound{t₁}}\AgdaCatchallClause{\AgdaSymbol{\}}}\AgdaSpace{}%
\AgdaCatchallClause{\AgdaSymbol{\{}}\AgdaCatchallClause{\AgdaBound{t₂}}\AgdaCatchallClause{\AgdaSymbol{\}}}\AgdaSpace{}%
\AgdaCatchallClause{\AgdaBound{c₁}}\AgdaSpace{}%
\AgdaCatchallClause{\AgdaBound{c₂}}\AgdaCatchallClause{\AgdaSymbol{)}}%
\>[44]\AgdaSymbol{=}%
\>[1076I]\AgdaKeyword{let}\AgdaSpace{}%
\AgdaSymbol{\AgdaUnderscore{}}\AgdaSpace{}%
\AgdaOperator{\AgdaInductiveConstructor{,}}\AgdaSpace{}%
\AgdaBound{c}\AgdaSpace{}%
\AgdaOperator{\AgdaInductiveConstructor{,}}\AgdaSpace{}%
\AgdaBound{st'}\AgdaSpace{}%
\AgdaSymbol{=}\AgdaSpace{}%
\AgdaFunction{st}\AgdaSpace{}%
\AgdaBound{x}\AgdaSpace{}%
\AgdaBound{c₁}\AgdaSpace{}%
\AgdaKeyword{in}\<%
\\
\>[.][@{}l@{}]\<[1076I]%
\>[46]\AgdaSymbol{\AgdaUnderscore{}}\AgdaSpace{}%
\AgdaOperator{\AgdaInductiveConstructor{,}}\AgdaSpace{}%
\AgdaBound{c}\AgdaSpace{}%
\AgdaOperator{\AgdaInductiveConstructor{⊕}}\AgdaSpace{}%
\AgdaBound{c₂}\AgdaSpace{}%
\AgdaOperator{\AgdaInductiveConstructor{,}}\AgdaSpace{}%
\AgdaOperator{\AgdaInductiveConstructor{⟪}}\AgdaSpace{}%
\AgdaFunction{Enum}\AgdaSpace{}%
\AgdaSymbol{\AgdaUnderscore{}}\AgdaSpace{}%
\AgdaOperator{\AgdaInductiveConstructor{[}}\AgdaSpace{}%
\AgdaFunction{Find'}\AgdaSpace{}%
\AgdaSymbol{\{\AgdaUnderscore{}}\AgdaSpace{}%
\AgdaOperator{\AgdaInductiveConstructor{+ᵤ}}\AgdaSpace{}%
\AgdaBound{t₂}\AgdaSymbol{\}}\AgdaSpace{}%
\AgdaSymbol{(}\AgdaInductiveConstructor{inj₁}\AgdaSpace{}%
\AgdaOperator{\AgdaFunction{\$}}\AgdaSpace{}%
\AgdaFunction{resolve}\AgdaSpace{}%
\AgdaBound{st'}\AgdaSymbol{)}\AgdaSpace{}%
\AgdaOperator{\AgdaInductiveConstructor{]⟫}}\<%
\\
\>[0]\AgdaFunction{st}\AgdaSpace{}%
\AgdaSymbol{(}\AgdaInductiveConstructor{inj₂}\AgdaSpace{}%
\AgdaBound{y}\AgdaSymbol{)}\AgdaSpace{}%
\AgdaSymbol{(}\AgdaOperator{\AgdaInductiveConstructor{\AgdaUnderscore{}⊕\AgdaUnderscore{}}}\AgdaSpace{}%
\AgdaSymbol{\{}\AgdaBound{t₁}\AgdaSymbol{\}}\AgdaSpace{}%
\AgdaSymbol{\{}\AgdaBound{t₂}\AgdaSymbol{\}}\AgdaSpace{}%
\AgdaSymbol{\{}\AgdaBound{t₃}\AgdaSymbol{\}}\AgdaSpace{}%
\AgdaSymbol{\{\AgdaUnderscore{}\}}\AgdaSpace{}%
\AgdaBound{c₁}\AgdaSpace{}%
\AgdaInductiveConstructor{id⟷}\AgdaSymbol{)}\AgdaSpace{}%
\AgdaSymbol{=}\AgdaSpace{}%
\AgdaSymbol{\AgdaUnderscore{}}\AgdaSpace{}%
\AgdaOperator{\AgdaInductiveConstructor{,}}\AgdaSpace{}%
\AgdaInductiveConstructor{id⟷}\AgdaSpace{}%
\AgdaOperator{\AgdaInductiveConstructor{,}}\AgdaSpace{}%
\AgdaOperator{\AgdaInductiveConstructor{⟪}}\AgdaSpace{}%
\AgdaFunction{Enum}\AgdaSpace{}%
\AgdaSymbol{\AgdaUnderscore{}}\AgdaSpace{}%
\AgdaOperator{\AgdaInductiveConstructor{[}}\AgdaSpace{}%
\AgdaFunction{Find'}\AgdaSpace{}%
\AgdaSymbol{\{}\AgdaBound{t₃}\AgdaSpace{}%
\AgdaOperator{\AgdaInductiveConstructor{+ᵤ}}\AgdaSpace{}%
\AgdaSymbol{\AgdaUnderscore{}\}}\AgdaSpace{}%
\AgdaSymbol{(}\AgdaInductiveConstructor{inj₂}\AgdaSpace{}%
\AgdaBound{y}\AgdaSymbol{)}\AgdaSpace{}%
\AgdaOperator{\AgdaInductiveConstructor{]⟫}}\<%
\\
\>[0]\AgdaCatchallClause{\AgdaFunction{st}}\AgdaSpace{}%
\AgdaCatchallClause{\AgdaSymbol{(}}\AgdaCatchallClause{\AgdaInductiveConstructor{inj₂}}\AgdaSpace{}%
\AgdaCatchallClause{\AgdaBound{y}}\AgdaCatchallClause{\AgdaSymbol{)}}\AgdaSpace{}%
\AgdaCatchallClause{\AgdaSymbol{(}}\AgdaCatchallClause{\AgdaOperator{\AgdaInductiveConstructor{\AgdaUnderscore{}⊕\AgdaUnderscore{}}}}\AgdaSpace{}%
\AgdaCatchallClause{\AgdaSymbol{\{}}\AgdaCatchallClause{\AgdaBound{t₁}}\AgdaCatchallClause{\AgdaSymbol{\}}}\AgdaSpace{}%
\AgdaCatchallClause{\AgdaBound{c₁}}\AgdaSpace{}%
\AgdaCatchallClause{\AgdaBound{c₂}}\AgdaCatchallClause{\AgdaSymbol{)}}%
\>[44]\AgdaSymbol{=}%
\>[1136I]\AgdaKeyword{let}\AgdaSpace{}%
\AgdaSymbol{\AgdaUnderscore{}}\AgdaSpace{}%
\AgdaOperator{\AgdaInductiveConstructor{,}}\AgdaSpace{}%
\AgdaBound{c}\AgdaSpace{}%
\AgdaOperator{\AgdaInductiveConstructor{,}}\AgdaSpace{}%
\AgdaBound{st'}\AgdaSpace{}%
\AgdaSymbol{=}\AgdaSpace{}%
\AgdaFunction{st}\AgdaSpace{}%
\AgdaBound{y}\AgdaSpace{}%
\AgdaBound{c₂}\AgdaSpace{}%
\AgdaKeyword{in}\<%
\\
\>[.][@{}l@{}]\<[1136I]%
\>[46]\AgdaSymbol{\AgdaUnderscore{}}\AgdaSpace{}%
\AgdaOperator{\AgdaInductiveConstructor{,}}\AgdaSpace{}%
\AgdaBound{c₁}\AgdaSpace{}%
\AgdaOperator{\AgdaInductiveConstructor{⊕}}\AgdaSpace{}%
\AgdaBound{c}\AgdaSpace{}%
\AgdaOperator{\AgdaInductiveConstructor{,}}\AgdaSpace{}%
\AgdaOperator{\AgdaInductiveConstructor{⟪}}\AgdaSpace{}%
\AgdaFunction{Enum}\AgdaSpace{}%
\AgdaSymbol{\AgdaUnderscore{}}\AgdaSpace{}%
\AgdaOperator{\AgdaInductiveConstructor{[}}\AgdaSpace{}%
\AgdaFunction{Find'}\AgdaSpace{}%
\AgdaSymbol{\{}\AgdaBound{t₁}\AgdaSpace{}%
\AgdaOperator{\AgdaInductiveConstructor{+ᵤ}}\AgdaSpace{}%
\AgdaSymbol{\AgdaUnderscore{}\}}\AgdaSpace{}%
\AgdaSymbol{(}\AgdaInductiveConstructor{inj₂}\AgdaSpace{}%
\AgdaOperator{\AgdaFunction{\$}}\AgdaSpace{}%
\AgdaFunction{resolve}\AgdaSpace{}%
\AgdaBound{st'}\AgdaSymbol{)}\AgdaSpace{}%
\AgdaOperator{\AgdaInductiveConstructor{]⟫}}\<%
\\
\>[0]\AgdaFunction{st}\AgdaSpace{}%
\AgdaSymbol{(}\AgdaBound{x}\AgdaSpace{}%
\AgdaOperator{\AgdaInductiveConstructor{,}}\AgdaSpace{}%
\AgdaBound{y}\AgdaSymbol{)}\AgdaSpace{}%
\AgdaSymbol{(}\AgdaInductiveConstructor{id⟷}\AgdaSpace{}%
\AgdaOperator{\AgdaInductiveConstructor{⊗}}\AgdaSpace{}%
\AgdaInductiveConstructor{id⟷}\AgdaSymbol{)}%
\>[44]\AgdaSymbol{=}\AgdaSpace{}%
\AgdaSymbol{\AgdaUnderscore{}}\AgdaSpace{}%
\AgdaOperator{\AgdaInductiveConstructor{,}}\AgdaSpace{}%
\AgdaInductiveConstructor{id⟷}\AgdaSpace{}%
\AgdaOperator{\AgdaInductiveConstructor{,}}\AgdaSpace{}%
\AgdaOperator{\AgdaInductiveConstructor{⟪}}\AgdaSpace{}%
\AgdaFunction{Enum}\AgdaSpace{}%
\AgdaSymbol{\AgdaUnderscore{}}\AgdaSpace{}%
\AgdaOperator{\AgdaInductiveConstructor{[}}\AgdaSpace{}%
\AgdaFunction{Find'}\AgdaSpace{}%
\AgdaSymbol{(}\AgdaBound{x}\AgdaSpace{}%
\AgdaOperator{\AgdaInductiveConstructor{,}}\AgdaSpace{}%
\AgdaBound{y}\AgdaSymbol{)}\AgdaSpace{}%
\AgdaOperator{\AgdaInductiveConstructor{]⟫}}\<%
\\
\>[0]\AgdaCatchallClause{\AgdaFunction{st}}\AgdaSpace{}%
\AgdaCatchallClause{\AgdaSymbol{(}}\AgdaCatchallClause{\AgdaBound{x}}\AgdaSpace{}%
\AgdaCatchallClause{\AgdaOperator{\AgdaInductiveConstructor{,}}}\AgdaSpace{}%
\AgdaCatchallClause{\AgdaBound{y}}\AgdaCatchallClause{\AgdaSymbol{)}}\AgdaSpace{}%
\AgdaCatchallClause{\AgdaSymbol{(}}\AgdaCatchallClause{\AgdaInductiveConstructor{id⟷}}\AgdaSpace{}%
\AgdaCatchallClause{\AgdaOperator{\AgdaInductiveConstructor{⊗}}}\AgdaSpace{}%
\AgdaCatchallClause{\AgdaBound{c₂}}\AgdaCatchallClause{\AgdaSymbol{)}}%
\>[44]\AgdaSymbol{=}%
\>[1190I]\AgdaKeyword{let}\AgdaSpace{}%
\AgdaSymbol{\AgdaUnderscore{}}\AgdaSpace{}%
\AgdaOperator{\AgdaInductiveConstructor{,}}\AgdaSpace{}%
\AgdaBound{c}\AgdaSpace{}%
\AgdaOperator{\AgdaInductiveConstructor{,}}\AgdaSpace{}%
\AgdaBound{st'}\AgdaSpace{}%
\AgdaSymbol{=}\AgdaSpace{}%
\AgdaFunction{st}\AgdaSpace{}%
\AgdaBound{y}\AgdaSpace{}%
\AgdaBound{c₂}\AgdaSpace{}%
\AgdaKeyword{in}\<%
\\
\>[1190I][@{}l@{\AgdaIndent{0}}]%
\>[47]\AgdaSymbol{\AgdaUnderscore{}}\AgdaSpace{}%
\AgdaOperator{\AgdaInductiveConstructor{,}}\AgdaSpace{}%
\AgdaInductiveConstructor{id⟷}\AgdaSpace{}%
\AgdaOperator{\AgdaInductiveConstructor{⊗}}\AgdaSpace{}%
\AgdaBound{c}\AgdaSpace{}%
\AgdaOperator{\AgdaInductiveConstructor{,}}\AgdaSpace{}%
\AgdaOperator{\AgdaInductiveConstructor{⟪}}\AgdaSpace{}%
\AgdaFunction{Enum}\AgdaSpace{}%
\AgdaSymbol{\AgdaUnderscore{}}\AgdaSpace{}%
\AgdaOperator{\AgdaInductiveConstructor{[}}\AgdaSpace{}%
\AgdaFunction{Find'}\AgdaSpace{}%
\AgdaSymbol{(}\AgdaBound{x}\AgdaSpace{}%
\AgdaOperator{\AgdaInductiveConstructor{,}}\AgdaSpace{}%
\AgdaFunction{resolve}\AgdaSpace{}%
\AgdaBound{st'}\AgdaSymbol{)}\AgdaSpace{}%
\AgdaOperator{\AgdaInductiveConstructor{]⟫}}\<%
\\
\>[0]\AgdaCatchallClause{\AgdaFunction{st}}\AgdaSpace{}%
\AgdaCatchallClause{\AgdaSymbol{(}}\AgdaCatchallClause{\AgdaBound{x}}\AgdaSpace{}%
\AgdaCatchallClause{\AgdaOperator{\AgdaInductiveConstructor{,}}}\AgdaSpace{}%
\AgdaCatchallClause{\AgdaBound{y}}\AgdaCatchallClause{\AgdaSymbol{)}}\AgdaSpace{}%
\AgdaCatchallClause{\AgdaSymbol{(}}\AgdaCatchallClause{\AgdaBound{c₁}}\AgdaSpace{}%
\AgdaCatchallClause{\AgdaOperator{\AgdaInductiveConstructor{⊗}}}\AgdaSpace{}%
\AgdaCatchallClause{\AgdaBound{c₂}}\AgdaCatchallClause{\AgdaSymbol{)}}%
\>[44]\AgdaSymbol{=}%
\>[1222I]\AgdaKeyword{let}\AgdaSpace{}%
\AgdaSymbol{\AgdaUnderscore{}}\AgdaSpace{}%
\AgdaOperator{\AgdaInductiveConstructor{,}}\AgdaSpace{}%
\AgdaBound{c}\AgdaSpace{}%
\AgdaOperator{\AgdaInductiveConstructor{,}}\AgdaSpace{}%
\AgdaBound{st'}\AgdaSpace{}%
\AgdaSymbol{=}\AgdaSpace{}%
\AgdaFunction{st}\AgdaSpace{}%
\AgdaBound{x}\AgdaSpace{}%
\AgdaBound{c₁}\AgdaSpace{}%
\AgdaKeyword{in}\<%
\\
\>[.][@{}l@{}]\<[1222I]%
\>[46]\AgdaSymbol{\AgdaUnderscore{}}\AgdaSpace{}%
\AgdaOperator{\AgdaInductiveConstructor{,}}\AgdaSpace{}%
\AgdaBound{c}\AgdaSpace{}%
\AgdaOperator{\AgdaInductiveConstructor{⊗}}\AgdaSpace{}%
\AgdaBound{c₂}\AgdaSpace{}%
\AgdaOperator{\AgdaInductiveConstructor{,}}\AgdaSpace{}%
\AgdaOperator{\AgdaInductiveConstructor{⟪}}\AgdaSpace{}%
\AgdaFunction{Enum}\AgdaSpace{}%
\AgdaSymbol{\AgdaUnderscore{}}\AgdaSpace{}%
\AgdaOperator{\AgdaInductiveConstructor{[}}\AgdaSpace{}%
\AgdaFunction{Find'}\AgdaSpace{}%
\AgdaSymbol{(}\AgdaFunction{resolve}\AgdaSpace{}%
\AgdaBound{st'}\AgdaSpace{}%
\AgdaOperator{\AgdaInductiveConstructor{,}}\AgdaSpace{}%
\AgdaBound{y}\AgdaSymbol{)}\AgdaSpace{}%
\AgdaOperator{\AgdaInductiveConstructor{]⟫}}\<%
\\
\\[\AgdaEmptyExtraSkip]%
\>[0]\AgdaFunction{step}\AgdaSpace{}%
\AgdaSymbol{:}\AgdaSpace{}%
\AgdaSymbol{\{}\AgdaBound{A}\AgdaSpace{}%
\AgdaBound{B}\AgdaSpace{}%
\AgdaSymbol{:}\AgdaSpace{}%
\AgdaDatatype{𝕌}\AgdaSymbol{\}}\AgdaSpace{}%
\AgdaSymbol{(}\AgdaBound{c}\AgdaSpace{}%
\AgdaSymbol{:}\AgdaSpace{}%
\AgdaBound{A}\AgdaSpace{}%
\AgdaOperator{\AgdaDatatype{⟷}}\AgdaSpace{}%
\AgdaBound{B}\AgdaSymbol{)}\AgdaSpace{}%
\AgdaSymbol{→}\AgdaSpace{}%
\AgdaDatatype{State}\AgdaSpace{}%
\AgdaBound{A}\AgdaSpace{}%
\AgdaSymbol{→}\AgdaSpace{}%
\AgdaFunction{Σ[}\AgdaSpace{}%
\AgdaBound{C}\AgdaSpace{}%
\AgdaFunction{∈}\AgdaSpace{}%
\AgdaDatatype{𝕌}\AgdaSpace{}%
\AgdaFunction{]}\AgdaSpace{}%
\AgdaSymbol{(}\AgdaBound{C}\AgdaSpace{}%
\AgdaOperator{\AgdaDatatype{⟷}}\AgdaSpace{}%
\AgdaBound{B}\AgdaSpace{}%
\AgdaOperator{\AgdaFunction{×}}\AgdaSpace{}%
\AgdaDatatype{State}\AgdaSpace{}%
\AgdaBound{C}\AgdaSymbol{)}\<%
\\
\>[0]\AgdaFunction{step}\AgdaSpace{}%
\AgdaBound{c}\AgdaSpace{}%
\AgdaOperator{\AgdaInductiveConstructor{⟪}}\AgdaSpace{}%
\AgdaBound{v}\AgdaSpace{}%
\AgdaOperator{\AgdaInductiveConstructor{[}}\AgdaSpace{}%
\AgdaBound{i}\AgdaSpace{}%
\AgdaOperator{\AgdaInductiveConstructor{]⟫}}\AgdaSpace{}%
\AgdaSymbol{=}\AgdaSpace{}%
\AgdaFunction{st}\AgdaSpace{}%
\AgdaSymbol{(}\AgdaFunction{lookup}\AgdaSpace{}%
\AgdaBound{v}\AgdaSpace{}%
\AgdaBound{i}\AgdaSymbol{)}\AgdaSpace{}%
\AgdaBound{c}\<%
\end{code}

\begin{code}[hide]%
\>[0]\AgdaFunction{run}\AgdaSpace{}%
\AgdaSymbol{:}\AgdaSpace{}%
\AgdaSymbol{(}\AgdaBound{n}\AgdaSpace{}%
\AgdaSymbol{:}\AgdaSpace{}%
\AgdaDatatype{ℕ}\AgdaSymbol{)}\AgdaSpace{}%
\AgdaSymbol{→}\AgdaSpace{}%
\AgdaDatatype{State'}\AgdaSpace{}%
\AgdaSymbol{→}\AgdaSpace{}%
\AgdaDatatype{Vec}\AgdaSpace{}%
\AgdaDatatype{State'}\AgdaSpace{}%
\AgdaSymbol{(}\AgdaInductiveConstructor{suc}\AgdaSpace{}%
\AgdaBound{n}\AgdaSymbol{)}\<%
\\
\>[0]\AgdaFunction{run}\AgdaSpace{}%
\AgdaNumber{0}\AgdaSpace{}%
\AgdaBound{st}\AgdaSpace{}%
\AgdaSymbol{=}\AgdaSpace{}%
\AgdaOperator{\AgdaFunction{[}}\AgdaSpace{}%
\AgdaBound{st}\AgdaSpace{}%
\AgdaOperator{\AgdaFunction{]}}\<%
\\
\>[0]\AgdaFunction{run}\AgdaSpace{}%
\AgdaSymbol{(}\AgdaInductiveConstructor{suc}\AgdaSpace{}%
\AgdaBound{n}\AgdaSymbol{)}\AgdaSpace{}%
\AgdaBound{st}\AgdaSpace{}%
\AgdaKeyword{with}\AgdaSpace{}%
\AgdaFunction{run}\AgdaSpace{}%
\AgdaBound{n}\AgdaSpace{}%
\AgdaBound{st}\<%
\\
\>[0]\AgdaSymbol{...}\AgdaSpace{}%
\AgdaSymbol{|}\AgdaSpace{}%
\AgdaBound{sts}\AgdaSpace{}%
\AgdaKeyword{with}\AgdaSpace{}%
\AgdaFunction{last}\AgdaSpace{}%
\AgdaBound{sts}\<%
\\
\>[0]\AgdaSymbol{...}\AgdaSpace{}%
\AgdaSymbol{|}\AgdaSpace{}%
\AgdaOperator{\AgdaInductiveConstructor{⟪}}\AgdaSpace{}%
\AgdaBound{cx}\AgdaSpace{}%
\AgdaOperator{\AgdaInductiveConstructor{∥}}\AgdaSpace{}%
\AgdaBound{vx}\AgdaSpace{}%
\AgdaOperator{\AgdaInductiveConstructor{[}}\AgdaSpace{}%
\AgdaBound{ix}\AgdaSpace{}%
\AgdaOperator{\AgdaInductiveConstructor{]⟫}}\AgdaSpace{}%
\AgdaKeyword{rewrite}\AgdaSpace{}%
\AgdaFunction{+-comm}\AgdaSpace{}%
\AgdaNumber{1}\AgdaSpace{}%
\AgdaSymbol{(}\AgdaInductiveConstructor{suc}\AgdaSpace{}%
\AgdaBound{n}\AgdaSymbol{)}\AgdaSpace{}%
\AgdaSymbol{=}\AgdaSpace{}%
\AgdaBound{sts}\AgdaSpace{}%
\AgdaOperator{\AgdaFunction{++}}\AgdaSpace{}%
\AgdaOperator{\AgdaFunction{[}}\AgdaSpace{}%
\AgdaFunction{step'}\AgdaSpace{}%
\AgdaOperator{\AgdaInductiveConstructor{⟪}}\AgdaSpace{}%
\AgdaBound{cx}\AgdaSpace{}%
\AgdaOperator{\AgdaInductiveConstructor{∥}}\AgdaSpace{}%
\AgdaBound{vx}\AgdaSpace{}%
\AgdaOperator{\AgdaInductiveConstructor{[}}\AgdaSpace{}%
\AgdaBound{ix}\AgdaSpace{}%
\AgdaOperator{\AgdaInductiveConstructor{]⟫}}\AgdaSpace{}%
\AgdaOperator{\AgdaFunction{]}}\<%
\\
\\[\AgdaEmptyExtraSkip]%
\>[0]\AgdaFunction{trace₂}\AgdaSpace{}%
\AgdaSymbol{:}\AgdaSpace{}%
\AgdaDatatype{Vec}\AgdaSpace{}%
\AgdaDatatype{State'}\AgdaSpace{}%
\AgdaNumber{34}\<%
\\
\>[0]\AgdaFunction{trace₂}\AgdaSpace{}%
\AgdaSymbol{=}\AgdaSpace{}%
\AgdaFunction{run}\AgdaSpace{}%
\AgdaSymbol{\AgdaUnderscore{}}\AgdaSpace{}%
\AgdaOperator{\AgdaInductiveConstructor{⟪}}\AgdaSpace{}%
\AgdaFunction{id'}\AgdaSpace{}%
\AgdaOperator{\AgdaInductiveConstructor{∥}}\AgdaSpace{}%
\AgdaFunction{Enum}\AgdaSpace{}%
\AgdaFunction{𝔹}\AgdaSpace{}%
\AgdaOperator{\AgdaInductiveConstructor{[}}\AgdaSpace{}%
\AgdaFunction{fromℕ}\AgdaSpace{}%
\AgdaNumber{1}\AgdaSpace{}%
\AgdaOperator{\AgdaInductiveConstructor{]⟫}}\<%
\end{code}

\begin{code}[hide]%
\>[0]\AgdaSymbol{\{-\#}\AgdaSpace{}%
\AgdaKeyword{OPTIONS}\AgdaSpace{}%
\AgdaPragma{--without-K}\AgdaSpace{}%
\AgdaPragma{--safe}\AgdaSpace{}%
\AgdaSymbol{\#-\}}\<%
\\
\>[0]\AgdaKeyword{module}\AgdaSpace{}%
\AgdaModule{T}\AgdaSpace{}%
\AgdaKeyword{where}\<%
\\
\>[0]\AgdaKeyword{open}\AgdaSpace{}%
\AgdaKeyword{import}\AgdaSpace{}%
\AgdaModule{Data.Bool}\AgdaSpace{}%
\AgdaKeyword{hiding}\AgdaSpace{}%
\AgdaSymbol{(}\AgdaFunction{T}\AgdaSymbol{)}\<%
\\
\>[0]\AgdaKeyword{open}\AgdaSpace{}%
\AgdaKeyword{import}\AgdaSpace{}%
\AgdaModule{Data.Nat}\<%
\\
\>[0]\AgdaKeyword{open}\AgdaSpace{}%
\AgdaKeyword{import}\AgdaSpace{}%
\AgdaModule{Data.Sum}\<%
\\
\>[0]\AgdaKeyword{open}\AgdaSpace{}%
\AgdaKeyword{import}\AgdaSpace{}%
\AgdaModule{Data.Product}\<%
\end{code}
\newcommand{\Texample}{%
\begin{code}%
\>[0]\AgdaKeyword{data}\AgdaSpace{}%
\AgdaDatatype{T}\AgdaSpace{}%
\AgdaSymbol{:}\AgdaSpace{}%
\AgdaPrimitiveType{Set}\AgdaSpace{}%
\AgdaKeyword{where}\<%
\\
\>[0][@{}l@{\AgdaIndent{0}}]%
\>[2]\AgdaInductiveConstructor{N}%
\>[5]\AgdaSymbol{:}\AgdaSpace{}%
\AgdaDatatype{T}\<%
\\
\>[2]\AgdaInductiveConstructor{B}%
\>[5]\AgdaSymbol{:}\AgdaSpace{}%
\AgdaDatatype{T}\<%
\\
\\[\AgdaEmptyExtraSkip]%
\>[0]\AgdaOperator{\AgdaFunction{⟦\AgdaUnderscore{}⟧}}\AgdaSpace{}%
\AgdaSymbol{:}\AgdaSpace{}%
\AgdaDatatype{T}\AgdaSpace{}%
\AgdaSymbol{→}\AgdaSpace{}%
\AgdaPrimitiveType{Set}\<%
\\
\>[0]\AgdaOperator{\AgdaFunction{⟦}}\AgdaSpace{}%
\AgdaInductiveConstructor{N}\AgdaSpace{}%
\AgdaOperator{\AgdaFunction{⟧}}%
\>[7]\AgdaSymbol{=}\AgdaSpace{}%
\AgdaDatatype{ℕ}\<%
\\
\>[0]\AgdaOperator{\AgdaFunction{⟦}}\AgdaSpace{}%
\AgdaInductiveConstructor{B}\AgdaSpace{}%
\AgdaOperator{\AgdaFunction{⟧}}%
\>[7]\AgdaSymbol{=}\AgdaSpace{}%
\AgdaDatatype{Bool}\<%
\\
\\[\AgdaEmptyExtraSkip]%
\>[0]\AgdaKeyword{data}\AgdaSpace{}%
\AgdaDatatype{Fun}\AgdaSpace{}%
\AgdaSymbol{:}\AgdaSpace{}%
\AgdaDatatype{T}\AgdaSpace{}%
\AgdaSymbol{→}\AgdaSpace{}%
\AgdaDatatype{T}\AgdaSpace{}%
\AgdaSymbol{→}\AgdaSpace{}%
\AgdaPrimitiveType{Set}\AgdaSpace{}%
\AgdaKeyword{where}\<%
\\
\>[0][@{}l@{\AgdaIndent{0}}]%
\>[2]\AgdaInductiveConstructor{square}%
\>[11]\AgdaSymbol{:}\AgdaSpace{}%
\AgdaDatatype{Fun}\AgdaSpace{}%
\AgdaInductiveConstructor{N}\AgdaSpace{}%
\AgdaInductiveConstructor{N}\<%
\\
\>[2]\AgdaInductiveConstructor{isZero}%
\>[11]\AgdaSymbol{:}\AgdaSpace{}%
\AgdaDatatype{Fun}\AgdaSpace{}%
\AgdaInductiveConstructor{N}\AgdaSpace{}%
\AgdaInductiveConstructor{B}\<%
\\
\>[2]\AgdaInductiveConstructor{compose}%
\>[11]\AgdaSymbol{:}\AgdaSpace{}%
\AgdaSymbol{\{}\AgdaBound{a}\AgdaSpace{}%
\AgdaBound{b}\AgdaSpace{}%
\AgdaBound{c}\AgdaSpace{}%
\AgdaSymbol{:}\AgdaSpace{}%
\AgdaDatatype{T}\AgdaSymbol{\}}\AgdaSpace{}%
\AgdaSymbol{→}\AgdaSpace{}%
\AgdaDatatype{Fun}\AgdaSpace{}%
\AgdaBound{b}\AgdaSpace{}%
\AgdaBound{c}\AgdaSpace{}%
\AgdaSymbol{→}\AgdaSpace{}%
\AgdaDatatype{Fun}\AgdaSpace{}%
\AgdaBound{a}\AgdaSpace{}%
\AgdaBound{b}\AgdaSpace{}%
\AgdaSymbol{→}\AgdaSpace{}%
\AgdaDatatype{Fun}\AgdaSpace{}%
\AgdaBound{a}\AgdaSpace{}%
\AgdaBound{c}\<%
\\
\\[\AgdaEmptyExtraSkip]%
\>[0]\AgdaFunction{eval}\AgdaSpace{}%
\AgdaSymbol{:}\AgdaSpace{}%
\AgdaSymbol{\{}\AgdaBound{a}\AgdaSpace{}%
\AgdaBound{b}\AgdaSpace{}%
\AgdaSymbol{:}\AgdaSpace{}%
\AgdaDatatype{T}\AgdaSymbol{\}}\AgdaSpace{}%
\AgdaSymbol{→}\AgdaSpace{}%
\AgdaDatatype{Fun}\AgdaSpace{}%
\AgdaBound{a}\AgdaSpace{}%
\AgdaBound{b}\AgdaSpace{}%
\AgdaSymbol{→}\AgdaSpace{}%
\AgdaOperator{\AgdaFunction{⟦}}\AgdaSpace{}%
\AgdaBound{a}\AgdaSpace{}%
\AgdaOperator{\AgdaFunction{⟧}}\AgdaSpace{}%
\AgdaSymbol{→}\AgdaSpace{}%
\AgdaOperator{\AgdaFunction{⟦}}\AgdaSpace{}%
\AgdaBound{b}\AgdaSpace{}%
\AgdaOperator{\AgdaFunction{⟧}}\<%
\\
\>[0]\AgdaFunction{eval}\AgdaSpace{}%
\AgdaInductiveConstructor{square}\AgdaSpace{}%
\AgdaBound{n}\AgdaSpace{}%
\AgdaSymbol{=}\AgdaSpace{}%
\AgdaBound{n}\AgdaSpace{}%
\AgdaOperator{\AgdaPrimitive{*}}\AgdaSpace{}%
\AgdaBound{n}\<%
\\
\>[0]\AgdaFunction{eval}\AgdaSpace{}%
\AgdaInductiveConstructor{isZero}\AgdaSpace{}%
\AgdaNumber{0}\AgdaSpace{}%
\AgdaSymbol{=}\AgdaSpace{}%
\AgdaInductiveConstructor{true}\<%
\\
\>[0]\AgdaFunction{eval}\AgdaSpace{}%
\AgdaInductiveConstructor{isZero}\AgdaSpace{}%
\AgdaSymbol{(}\AgdaInductiveConstructor{suc}\AgdaSpace{}%
\AgdaSymbol{\AgdaUnderscore{})}\AgdaSpace{}%
\AgdaSymbol{=}\AgdaSpace{}%
\AgdaInductiveConstructor{false}\<%
\\
\>[0]\AgdaFunction{eval}\AgdaSpace{}%
\AgdaSymbol{(}\AgdaInductiveConstructor{compose}\AgdaSpace{}%
\AgdaBound{g}\AgdaSpace{}%
\AgdaBound{f}\AgdaSymbol{)}\AgdaSpace{}%
\AgdaBound{v}\AgdaSpace{}%
\AgdaSymbol{=}\AgdaSpace{}%
\AgdaFunction{eval}\AgdaSpace{}%
\AgdaBound{g}\AgdaSpace{}%
\AgdaSymbol{(}\AgdaFunction{eval}\AgdaSpace{}%
\AgdaBound{f}\AgdaSpace{}%
\AgdaBound{v}\AgdaSymbol{)}\<%
\end{code}}
\newcommand{\Texamplecont}{%
\begin{code}%
\>[0]\AgdaKeyword{data}\AgdaSpace{}%
\AgdaDatatype{T∙}\AgdaSpace{}%
\AgdaSymbol{:}\AgdaSpace{}%
\AgdaPrimitiveType{Set}\AgdaSpace{}%
\AgdaKeyword{where}\<%
\\
\>[0][@{}l@{\AgdaIndent{0}}]%
\>[2]\AgdaOperator{\AgdaInductiveConstructor{\AgdaUnderscore{}\#\AgdaUnderscore{}}}\AgdaSpace{}%
\AgdaSymbol{:}\AgdaSpace{}%
\AgdaSymbol{(}\AgdaBound{a}\AgdaSpace{}%
\AgdaSymbol{:}\AgdaSpace{}%
\AgdaDatatype{T}\AgdaSymbol{)}\AgdaSpace{}%
\AgdaSymbol{→}\AgdaSpace{}%
\AgdaSymbol{(}\AgdaBound{v}\AgdaSpace{}%
\AgdaSymbol{:}\AgdaSpace{}%
\AgdaOperator{\AgdaFunction{⟦}}\AgdaSpace{}%
\AgdaBound{a}\AgdaSpace{}%
\AgdaOperator{\AgdaFunction{⟧}}\AgdaSymbol{)}\AgdaSpace{}%
\AgdaSymbol{→}\AgdaSpace{}%
\AgdaDatatype{T∙}\<%
\\
\\[\AgdaEmptyExtraSkip]%
\>[0]\AgdaOperator{\AgdaFunction{⟦\AgdaUnderscore{}⟧∙}}\AgdaSpace{}%
\AgdaSymbol{:}\AgdaSpace{}%
\AgdaDatatype{T∙}\AgdaSpace{}%
\AgdaSymbol{→}\AgdaSpace{}%
\AgdaFunction{Σ[}\AgdaSpace{}%
\AgdaBound{A}\AgdaSpace{}%
\AgdaFunction{∈}\AgdaSpace{}%
\AgdaPrimitiveType{Set}\AgdaSpace{}%
\AgdaFunction{]}\AgdaSpace{}%
\AgdaBound{A}\<%
\\
\>[0]\AgdaOperator{\AgdaFunction{⟦}}\AgdaSpace{}%
\AgdaBound{T}\AgdaSpace{}%
\AgdaOperator{\AgdaInductiveConstructor{\#}}\AgdaSpace{}%
\AgdaBound{v}\AgdaSpace{}%
\AgdaOperator{\AgdaFunction{⟧∙}}\AgdaSpace{}%
\AgdaSymbol{=}\AgdaSpace{}%
\AgdaOperator{\AgdaFunction{⟦}}\AgdaSpace{}%
\AgdaBound{T}\AgdaSpace{}%
\AgdaOperator{\AgdaFunction{⟧}}\AgdaSpace{}%
\AgdaOperator{\AgdaInductiveConstructor{,}}\AgdaSpace{}%
\AgdaBound{v}\<%
\\
\\[\AgdaEmptyExtraSkip]%
\>[0]\AgdaKeyword{data}\AgdaSpace{}%
\AgdaDatatype{Fun∙}\AgdaSpace{}%
\AgdaSymbol{:}\AgdaSpace{}%
\AgdaDatatype{T∙}\AgdaSpace{}%
\AgdaSymbol{→}\AgdaSpace{}%
\AgdaDatatype{T∙}\AgdaSpace{}%
\AgdaSymbol{→}\AgdaSpace{}%
\AgdaPrimitiveType{Set}\AgdaSpace{}%
\AgdaKeyword{where}\<%
\\
\>[0][@{}l@{\AgdaIndent{0}}]%
\>[2]\AgdaInductiveConstructor{lift}\AgdaSpace{}%
\AgdaSymbol{:}\AgdaSpace{}%
\AgdaSymbol{\{}\AgdaBound{a}\AgdaSpace{}%
\AgdaBound{b}\AgdaSpace{}%
\AgdaSymbol{:}\AgdaSpace{}%
\AgdaDatatype{T}\AgdaSymbol{\}}\AgdaSpace{}%
\AgdaSymbol{\{}\AgdaBound{v}\AgdaSpace{}%
\AgdaSymbol{:}\AgdaSpace{}%
\AgdaOperator{\AgdaFunction{⟦}}\AgdaSpace{}%
\AgdaBound{a}\AgdaSpace{}%
\AgdaOperator{\AgdaFunction{⟧}}\AgdaSymbol{\}}\AgdaSpace{}%
\AgdaSymbol{→}\AgdaSpace{}%
\AgdaSymbol{(}\AgdaBound{f}\AgdaSpace{}%
\AgdaSymbol{:}\AgdaSpace{}%
\AgdaDatatype{Fun}\AgdaSpace{}%
\AgdaBound{a}\AgdaSpace{}%
\AgdaBound{b}\AgdaSymbol{)}\AgdaSpace{}%
\AgdaSymbol{→}\AgdaSpace{}%
\AgdaDatatype{Fun∙}\AgdaSpace{}%
\AgdaSymbol{(}\AgdaBound{a}\AgdaSpace{}%
\AgdaOperator{\AgdaInductiveConstructor{\#}}\AgdaSpace{}%
\AgdaBound{v}\AgdaSymbol{)}\AgdaSpace{}%
\AgdaSymbol{(}\AgdaBound{b}\AgdaSpace{}%
\AgdaOperator{\AgdaInductiveConstructor{\#}}\AgdaSpace{}%
\AgdaSymbol{(}\AgdaFunction{eval}\AgdaSpace{}%
\AgdaBound{f}\AgdaSpace{}%
\AgdaBound{v}\AgdaSymbol{))}\<%
\end{code}}
\newcommand{\Texampletest}{%
\begin{code}%
\>[0]\AgdaFunction{test1}\AgdaSpace{}%
\AgdaSymbol{:}\AgdaSpace{}%
\AgdaDatatype{Fun∙}\AgdaSpace{}%
\AgdaSymbol{(}\AgdaInductiveConstructor{N}\AgdaSpace{}%
\AgdaOperator{\AgdaInductiveConstructor{\#}}\AgdaSpace{}%
\AgdaNumber{3}\AgdaSymbol{)}\AgdaSpace{}%
\AgdaSymbol{(}\AgdaInductiveConstructor{B}\AgdaSpace{}%
\AgdaOperator{\AgdaInductiveConstructor{\#}}\AgdaSpace{}%
\AgdaInductiveConstructor{false}\AgdaSymbol{)}\<%
\\
\>[0]\AgdaFunction{test1}\AgdaSpace{}%
\AgdaSymbol{=}\AgdaSpace{}%
\AgdaInductiveConstructor{lift}\AgdaSpace{}%
\AgdaSymbol{(}\AgdaInductiveConstructor{compose}\AgdaSpace{}%
\AgdaInductiveConstructor{isZero}\AgdaSpace{}%
\AgdaInductiveConstructor{square}\AgdaSymbol{)}\<%
\\
\\[\AgdaEmptyExtraSkip]%
\>[0]\AgdaFunction{test2}\AgdaSpace{}%
\AgdaSymbol{:}\AgdaSpace{}%
\AgdaDatatype{Fun∙}\AgdaSpace{}%
\AgdaSymbol{(}\AgdaInductiveConstructor{N}\AgdaSpace{}%
\AgdaOperator{\AgdaInductiveConstructor{\#}}\AgdaSpace{}%
\AgdaNumber{0}\AgdaSymbol{)}\AgdaSpace{}%
\AgdaSymbol{(}\AgdaInductiveConstructor{B}\AgdaSpace{}%
\AgdaOperator{\AgdaInductiveConstructor{\#}}\AgdaSpace{}%
\AgdaInductiveConstructor{true}\AgdaSymbol{)}\<%
\\
\>[0]\AgdaFunction{test2}\AgdaSpace{}%
\AgdaSymbol{=}\AgdaSpace{}%
\AgdaInductiveConstructor{lift}\AgdaSpace{}%
\AgdaSymbol{(}\AgdaInductiveConstructor{compose}\AgdaSpace{}%
\AgdaInductiveConstructor{isZero}\AgdaSpace{}%
\AgdaInductiveConstructor{square}\AgdaSymbol{)}\<%
\\
\\[\AgdaEmptyExtraSkip]%
\>[0]\AgdaFunction{test3}\AgdaSpace{}%
\AgdaSymbol{:}\AgdaSpace{}%
\AgdaSymbol{∀}\AgdaSpace{}%
\AgdaSymbol{\{}\AgdaBound{n}\AgdaSymbol{\}}\AgdaSpace{}%
\AgdaSymbol{→}\AgdaSpace{}%
\AgdaDatatype{Fun∙}\AgdaSpace{}%
\AgdaSymbol{(}\AgdaInductiveConstructor{N}\AgdaSpace{}%
\AgdaOperator{\AgdaInductiveConstructor{\#}}\AgdaSpace{}%
\AgdaSymbol{(}\AgdaInductiveConstructor{suc}\AgdaSpace{}%
\AgdaBound{n}\AgdaSymbol{))}\AgdaSpace{}%
\AgdaSymbol{(}\AgdaInductiveConstructor{B}\AgdaSpace{}%
\AgdaOperator{\AgdaInductiveConstructor{\#}}\AgdaSpace{}%
\AgdaInductiveConstructor{false}\AgdaSymbol{)}\<%
\\
\>[0]\AgdaFunction{test3}\AgdaSpace{}%
\AgdaSymbol{=}\AgdaSpace{}%
\AgdaInductiveConstructor{lift}\AgdaSpace{}%
\AgdaSymbol{(}\AgdaInductiveConstructor{compose}\AgdaSpace{}%
\AgdaInductiveConstructor{isZero}\AgdaSpace{}%
\AgdaInductiveConstructor{square}\AgdaSymbol{)}\<%
\end{code}}

\begin{code}[hide]%
\>[0]\AgdaSymbol{\{-\#}\AgdaSpace{}%
\AgdaKeyword{OPTIONS}\AgdaSpace{}%
\AgdaPragma{--without-K}\AgdaSpace{}%
\AgdaPragma{--safe}\AgdaSpace{}%
\AgdaSymbol{\#-\}}\<%
\\
\>[0]\AgdaKeyword{module}\AgdaSpace{}%
\AgdaModule{Pi.Pointed}\AgdaSpace{}%
\AgdaKeyword{where}\<%
\\
\>[0]\AgdaKeyword{open}\AgdaSpace{}%
\AgdaKeyword{import}\AgdaSpace{}%
\AgdaModule{Data.Empty}\<%
\\
\>[0]\AgdaKeyword{open}\AgdaSpace{}%
\AgdaKeyword{import}\AgdaSpace{}%
\AgdaModule{Data.Unit}\<%
\\
\>[0]\AgdaKeyword{open}\AgdaSpace{}%
\AgdaKeyword{import}\AgdaSpace{}%
\AgdaModule{Data.Sum}\<%
\\
\>[0]\AgdaKeyword{open}\AgdaSpace{}%
\AgdaKeyword{import}\AgdaSpace{}%
\AgdaModule{Data.Product}\<%
\\
\>[0]\AgdaKeyword{open}\AgdaSpace{}%
\AgdaKeyword{import}\AgdaSpace{}%
\AgdaModule{Relation.Binary.PropositionalEquality}\<%
\\
\>[0]\AgdaKeyword{open}\AgdaSpace{}%
\AgdaKeyword{import}\AgdaSpace{}%
\AgdaModule{Function}\AgdaSpace{}%
\AgdaKeyword{using}\AgdaSpace{}%
\AgdaSymbol{(}\AgdaOperator{\AgdaFunction{\AgdaUnderscore{}∘\AgdaUnderscore{}}}\AgdaSymbol{)}\<%
\\
\>[0]\AgdaKeyword{open}\AgdaSpace{}%
\AgdaKeyword{import}\AgdaSpace{}%
\AgdaModule{Pi}\<%
\\
\\[\AgdaEmptyExtraSkip]%
\>[0]\AgdaKeyword{infixr}\AgdaSpace{}%
\AgdaNumber{90}\AgdaSpace{}%
\AgdaOperator{\AgdaInductiveConstructor{\AgdaUnderscore{}\#\AgdaUnderscore{}}}\<%
\\
\>[0]\AgdaKeyword{infixr}\AgdaSpace{}%
\AgdaNumber{70}\AgdaSpace{}%
\AgdaOperator{\AgdaInductiveConstructor{\AgdaUnderscore{}∙×ᵤ\AgdaUnderscore{}}}\<%
\\
\>[0]\AgdaKeyword{infixr}\AgdaSpace{}%
\AgdaNumber{70}\AgdaSpace{}%
\AgdaOperator{\AgdaInductiveConstructor{\AgdaUnderscore{}∙⊗\AgdaUnderscore{}}}\<%
\\
\>[0]\AgdaKeyword{infixr}\AgdaSpace{}%
\AgdaNumber{50}\AgdaSpace{}%
\AgdaOperator{\AgdaInductiveConstructor{\AgdaUnderscore{}∙⊚\AgdaUnderscore{}}}\<%
\\
\>[0]\AgdaKeyword{infix}\AgdaSpace{}%
\AgdaNumber{100}\AgdaSpace{}%
\AgdaOperator{\AgdaFunction{!∙\AgdaUnderscore{}}}\<%
\end{code}
\newcommand{\PIPFUdef}{%
\begin{code}%
\>[0]\AgdaFunction{Singleton}\AgdaSpace{}%
\AgdaSymbol{:}\AgdaSpace{}%
\AgdaSymbol{(}\AgdaBound{A}\AgdaSpace{}%
\AgdaSymbol{:}\AgdaSpace{}%
\AgdaPrimitiveType{Set}\AgdaSymbol{)}\AgdaSpace{}%
\AgdaSymbol{→}\AgdaSpace{}%
\AgdaSymbol{(}\AgdaBound{v}\AgdaSpace{}%
\AgdaSymbol{:}\AgdaSpace{}%
\AgdaBound{A}\AgdaSymbol{)}\AgdaSpace{}%
\AgdaSymbol{→}\AgdaSpace{}%
\AgdaPrimitiveType{Set}\<%
\\
\>[0]\AgdaFunction{Singleton}\AgdaSpace{}%
\AgdaBound{A}\AgdaSpace{}%
\AgdaBound{v}\AgdaSpace{}%
\AgdaSymbol{=}\AgdaSpace{}%
\AgdaFunction{∃}\AgdaSpace{}%
\AgdaSymbol{(λ}\AgdaSpace{}%
\AgdaBound{●}\AgdaSpace{}%
\AgdaSymbol{→}\AgdaSpace{}%
\AgdaBound{v}\AgdaSpace{}%
\AgdaOperator{\AgdaDatatype{≡}}\AgdaSpace{}%
\AgdaBound{●}\AgdaSymbol{)}\<%
\\
\\[\AgdaEmptyExtraSkip]%
\>[0]\AgdaFunction{Recip}\AgdaSpace{}%
\AgdaSymbol{:}\AgdaSpace{}%
\AgdaSymbol{(}\AgdaBound{A}\AgdaSpace{}%
\AgdaSymbol{:}\AgdaSpace{}%
\AgdaPrimitiveType{Set}\AgdaSymbol{)}\AgdaSpace{}%
\AgdaSymbol{→}\AgdaSpace{}%
\AgdaSymbol{(}\AgdaBound{v}\AgdaSpace{}%
\AgdaSymbol{:}\AgdaSpace{}%
\AgdaBound{A}\AgdaSymbol{)}\AgdaSpace{}%
\AgdaSymbol{→}\AgdaSpace{}%
\AgdaPrimitiveType{Set}\<%
\\
\>[0]\AgdaFunction{Recip}\AgdaSpace{}%
\AgdaBound{A}\AgdaSpace{}%
\AgdaBound{v}\AgdaSpace{}%
\AgdaSymbol{=}\AgdaSpace{}%
\AgdaFunction{Singleton}\AgdaSpace{}%
\AgdaBound{A}\AgdaSpace{}%
\AgdaBound{v}\AgdaSpace{}%
\AgdaSymbol{→}\AgdaSpace{}%
\AgdaRecord{⊤}\<%
\\
\\[\AgdaEmptyExtraSkip]%
\>[0]\AgdaKeyword{data}\AgdaSpace{}%
\AgdaDatatype{∙𝕌}\AgdaSpace{}%
\AgdaSymbol{:}\AgdaSpace{}%
\AgdaPrimitiveType{Set}\AgdaSpace{}%
\AgdaKeyword{where}\<%
\\
\>[0][@{}l@{\AgdaIndent{0}}]%
\>[2]\AgdaOperator{\AgdaInductiveConstructor{\AgdaUnderscore{}\#\AgdaUnderscore{}}}%
\>[10]\AgdaSymbol{:}\AgdaSpace{}%
\AgdaSymbol{(}\AgdaBound{t}\AgdaSpace{}%
\AgdaSymbol{:}\AgdaSpace{}%
\AgdaDatatype{𝕌}\AgdaSymbol{)}\AgdaSpace{}%
\AgdaSymbol{→}\AgdaSpace{}%
\AgdaSymbol{(}\AgdaBound{v}\AgdaSpace{}%
\AgdaSymbol{:}\AgdaSpace{}%
\AgdaOperator{\AgdaFunction{⟦}}\AgdaSpace{}%
\AgdaBound{t}\AgdaSpace{}%
\AgdaOperator{\AgdaFunction{⟧}}\AgdaSymbol{)}\AgdaSpace{}%
\AgdaSymbol{→}\AgdaSpace{}%
\AgdaDatatype{∙𝕌}\<%
\\
\>[2]\AgdaOperator{\AgdaInductiveConstructor{\AgdaUnderscore{}∙×ᵤ\AgdaUnderscore{}}}%
\>[10]\AgdaSymbol{:}\AgdaSpace{}%
\AgdaDatatype{∙𝕌}\AgdaSpace{}%
\AgdaSymbol{→}\AgdaSpace{}%
\AgdaDatatype{∙𝕌}\AgdaSpace{}%
\AgdaSymbol{→}\AgdaSpace{}%
\AgdaDatatype{∙𝕌}\<%
\\
\>[2]\AgdaOperator{\AgdaInductiveConstructor{❰\AgdaUnderscore{}❱}}%
\>[10]\AgdaSymbol{:}\AgdaSpace{}%
\AgdaDatatype{∙𝕌}\AgdaSpace{}%
\AgdaSymbol{→}\AgdaSpace{}%
\AgdaDatatype{∙𝕌}\<%
\\
\>[2]\AgdaInductiveConstructor{∙𝟙/}%
\>[10]\AgdaSymbol{:}\AgdaSpace{}%
\AgdaDatatype{∙𝕌}\AgdaSpace{}%
\AgdaSymbol{→}\AgdaSpace{}%
\AgdaDatatype{∙𝕌}\<%
\\
\\[\AgdaEmptyExtraSkip]%
\>[0]\AgdaOperator{\AgdaFunction{∙⟦\AgdaUnderscore{}⟧}}\AgdaSpace{}%
\AgdaSymbol{:}\AgdaSpace{}%
\AgdaDatatype{∙𝕌}\AgdaSpace{}%
\AgdaSymbol{→}\AgdaSpace{}%
\AgdaFunction{Σ[}\AgdaSpace{}%
\AgdaBound{A}\AgdaSpace{}%
\AgdaFunction{∈}\AgdaSpace{}%
\AgdaPrimitiveType{Set}\AgdaSpace{}%
\AgdaFunction{]}\AgdaSpace{}%
\AgdaBound{A}\<%
\\
\>[0]\AgdaOperator{\AgdaFunction{∙⟦}}\AgdaSpace{}%
\AgdaBound{t}\AgdaSpace{}%
\AgdaOperator{\AgdaInductiveConstructor{\#}}\AgdaSpace{}%
\AgdaBound{v}\AgdaSpace{}%
\AgdaOperator{\AgdaFunction{⟧}}%
\>[17]\AgdaSymbol{=}\AgdaSpace{}%
\AgdaOperator{\AgdaFunction{⟦}}\AgdaSpace{}%
\AgdaBound{t}\AgdaSpace{}%
\AgdaOperator{\AgdaFunction{⟧}}\AgdaSpace{}%
\AgdaOperator{\AgdaInductiveConstructor{,}}\AgdaSpace{}%
\AgdaBound{v}\<%
\\
\>[0]\AgdaOperator{\AgdaFunction{∙⟦}}\AgdaSpace{}%
\AgdaBound{T₁}\AgdaSpace{}%
\AgdaOperator{\AgdaInductiveConstructor{∙×ᵤ}}\AgdaSpace{}%
\AgdaBound{T₂}\AgdaSpace{}%
\AgdaOperator{\AgdaFunction{⟧}}%
\>[17]\AgdaSymbol{=}%
\>[118I]\AgdaKeyword{let}%
\>[119I]\AgdaSymbol{(}\AgdaBound{t₁}\AgdaSpace{}%
\AgdaOperator{\AgdaInductiveConstructor{,}}\AgdaSpace{}%
\AgdaBound{v₁}\AgdaSymbol{)}\AgdaSpace{}%
\AgdaSymbol{=}\AgdaSpace{}%
\AgdaOperator{\AgdaFunction{∙⟦}}\AgdaSpace{}%
\AgdaBound{T₁}\AgdaSpace{}%
\AgdaOperator{\AgdaFunction{⟧}}\<%
\\
\>[.][@{}l@{}]\<[119I]%
\>[23]\AgdaSymbol{(}\AgdaBound{t₂}\AgdaSpace{}%
\AgdaOperator{\AgdaInductiveConstructor{,}}\AgdaSpace{}%
\AgdaBound{v₂}\AgdaSymbol{)}\AgdaSpace{}%
\AgdaSymbol{=}\AgdaSpace{}%
\AgdaOperator{\AgdaFunction{∙⟦}}\AgdaSpace{}%
\AgdaBound{T₂}\AgdaSpace{}%
\AgdaOperator{\AgdaFunction{⟧}}\<%
\\
\>[.][@{}l@{}]\<[118I]%
\>[19]\AgdaKeyword{in}%
\>[23]\AgdaSymbol{(}\AgdaBound{t₁}\AgdaSpace{}%
\AgdaOperator{\AgdaFunction{×}}\AgdaSpace{}%
\AgdaBound{t₂}\AgdaSymbol{)}\AgdaSpace{}%
\AgdaOperator{\AgdaInductiveConstructor{,}}\AgdaSpace{}%
\AgdaSymbol{(}\AgdaBound{v₁}\AgdaSpace{}%
\AgdaOperator{\AgdaInductiveConstructor{,}}\AgdaSpace{}%
\AgdaBound{v₂}\AgdaSymbol{)}\<%
\\
\>[0]\AgdaOperator{\AgdaFunction{∙⟦}}\AgdaSpace{}%
\AgdaOperator{\AgdaInductiveConstructor{❰}}\AgdaSpace{}%
\AgdaBound{T}\AgdaSpace{}%
\AgdaOperator{\AgdaInductiveConstructor{❱}}\AgdaSpace{}%
\AgdaOperator{\AgdaFunction{⟧}}%
\>[17]\AgdaSymbol{=}\AgdaSpace{}%
\AgdaKeyword{let}\AgdaSpace{}%
\AgdaSymbol{(}\AgdaBound{t}\AgdaSpace{}%
\AgdaOperator{\AgdaInductiveConstructor{,}}\AgdaSpace{}%
\AgdaBound{v}\AgdaSymbol{)}\AgdaSpace{}%
\AgdaSymbol{=}\AgdaSpace{}%
\AgdaOperator{\AgdaFunction{∙⟦}}\AgdaSpace{}%
\AgdaBound{T}\AgdaSpace{}%
\AgdaOperator{\AgdaFunction{⟧}}\AgdaSpace{}%
\AgdaKeyword{in}\AgdaSpace{}%
\AgdaFunction{Singleton}\AgdaSpace{}%
\AgdaBound{t}\AgdaSpace{}%
\AgdaBound{v}\AgdaSpace{}%
\AgdaOperator{\AgdaInductiveConstructor{,}}\AgdaSpace{}%
\AgdaSymbol{(}\AgdaBound{v}\AgdaSpace{}%
\AgdaOperator{\AgdaInductiveConstructor{,}}\AgdaSpace{}%
\AgdaInductiveConstructor{refl}\AgdaSymbol{)}\<%
\\
\>[0]\AgdaOperator{\AgdaFunction{∙⟦}}\AgdaSpace{}%
\AgdaInductiveConstructor{∙𝟙/}\AgdaSpace{}%
\AgdaBound{T}\AgdaSpace{}%
\AgdaOperator{\AgdaFunction{⟧}}%
\>[17]\AgdaSymbol{=}\AgdaSpace{}%
\AgdaKeyword{let}\AgdaSpace{}%
\AgdaSymbol{(}\AgdaBound{t}\AgdaSpace{}%
\AgdaOperator{\AgdaInductiveConstructor{,}}\AgdaSpace{}%
\AgdaBound{v}\AgdaSymbol{)}\AgdaSpace{}%
\AgdaSymbol{=}\AgdaSpace{}%
\AgdaOperator{\AgdaFunction{∙⟦}}\AgdaSpace{}%
\AgdaBound{T}\AgdaSpace{}%
\AgdaOperator{\AgdaFunction{⟧}}\AgdaSpace{}%
\AgdaKeyword{in}\AgdaSpace{}%
\AgdaFunction{Recip}\AgdaSpace{}%
\AgdaBound{t}\AgdaSpace{}%
\AgdaBound{v}\AgdaSpace{}%
\AgdaOperator{\AgdaInductiveConstructor{,}}\AgdaSpace{}%
\AgdaSymbol{λ}\AgdaSpace{}%
\AgdaBound{\AgdaUnderscore{}}\AgdaSpace{}%
\AgdaSymbol{→}\AgdaSpace{}%
\AgdaInductiveConstructor{tt}\<%
\end{code}}
\begin{code}[hide]%
\>[0]\AgdaFunction{∙𝟙}\AgdaSpace{}%
\AgdaSymbol{:}\AgdaSpace{}%
\AgdaDatatype{∙𝕌}\<%
\\
\>[0]\AgdaFunction{∙𝟙}\AgdaSpace{}%
\AgdaSymbol{=}\AgdaSpace{}%
\AgdaInductiveConstructor{𝟙}\AgdaSpace{}%
\AgdaOperator{\AgdaInductiveConstructor{\#}}\AgdaSpace{}%
\AgdaInductiveConstructor{tt}\<%
\end{code}
\newcommand{\PIPFCombDef}{%
\begin{code}%
\>[0]\AgdaKeyword{data}\AgdaSpace{}%
\AgdaOperator{\AgdaDatatype{\AgdaUnderscore{}⧟\AgdaUnderscore{}}}\AgdaSpace{}%
\AgdaSymbol{:}\AgdaSpace{}%
\AgdaDatatype{∙𝕌}\AgdaSpace{}%
\AgdaSymbol{→}\AgdaSpace{}%
\AgdaDatatype{∙𝕌}\AgdaSpace{}%
\AgdaSymbol{→}\AgdaSpace{}%
\AgdaPrimitiveType{Set}\AgdaSpace{}%
\AgdaKeyword{where}\<%
\\
\>[0][@{}l@{\AgdaIndent{0}}]%
\>[2]\AgdaComment{-- lifting from plain Π}\<%
\\
\>[2]\AgdaInductiveConstructor{∙c}%
\>[10]\AgdaSymbol{:}\AgdaSpace{}%
\AgdaSymbol{\{}\AgdaBound{t₁}\AgdaSpace{}%
\AgdaBound{t₂}\AgdaSpace{}%
\AgdaSymbol{:}\AgdaSpace{}%
\AgdaDatatype{𝕌}\AgdaSymbol{\}}\AgdaSpace{}%
\AgdaSymbol{\{}\AgdaBound{v}\AgdaSpace{}%
\AgdaSymbol{:}\AgdaSpace{}%
\AgdaOperator{\AgdaFunction{⟦}}\AgdaSpace{}%
\AgdaBound{t₁}\AgdaSpace{}%
\AgdaOperator{\AgdaFunction{⟧}}\AgdaSymbol{\}}\AgdaSpace{}%
\AgdaSymbol{→}\AgdaSpace{}%
\AgdaSymbol{(}\AgdaBound{c}\AgdaSpace{}%
\AgdaSymbol{:}\AgdaSpace{}%
\AgdaBound{t₁}\AgdaSpace{}%
\AgdaOperator{\AgdaDatatype{⟷}}\AgdaSpace{}%
\AgdaBound{t₂}\AgdaSymbol{)}\AgdaSpace{}%
\AgdaSymbol{→}\AgdaSpace{}%
\AgdaBound{t₁}\AgdaSpace{}%
\AgdaOperator{\AgdaInductiveConstructor{\#}}\AgdaSpace{}%
\AgdaBound{v}\AgdaSpace{}%
\AgdaOperator{\AgdaDatatype{⧟}}\AgdaSpace{}%
\AgdaBound{t₂}\AgdaSpace{}%
\AgdaOperator{\AgdaInductiveConstructor{\#}}\AgdaSpace{}%
\AgdaSymbol{(}\AgdaFunction{eval}\AgdaSpace{}%
\AgdaBound{c}\AgdaSpace{}%
\AgdaBound{v}\AgdaSymbol{)}\<%
\\
\>[2]\AgdaInductiveConstructor{∙times\#}%
\>[217I]\AgdaSymbol{:}\AgdaSpace{}%
\AgdaSymbol{\{}\AgdaBound{t₁}\AgdaSpace{}%
\AgdaBound{t₂}\AgdaSpace{}%
\AgdaSymbol{:}\AgdaSpace{}%
\AgdaDatatype{𝕌}\AgdaSymbol{\}}\AgdaSpace{}%
\AgdaSymbol{\{}\AgdaBound{v₁}\AgdaSpace{}%
\AgdaSymbol{:}\AgdaSpace{}%
\AgdaOperator{\AgdaFunction{⟦}}\AgdaSpace{}%
\AgdaBound{t₁}\AgdaSpace{}%
\AgdaOperator{\AgdaFunction{⟧}}\AgdaSymbol{\}}\AgdaSpace{}%
\AgdaSymbol{\{}\AgdaBound{v₂}\AgdaSpace{}%
\AgdaSymbol{:}\AgdaSpace{}%
\AgdaOperator{\AgdaFunction{⟦}}\AgdaSpace{}%
\AgdaBound{t₂}\AgdaSpace{}%
\AgdaOperator{\AgdaFunction{⟧}}\AgdaSymbol{\}}\<%
\\
\>[.][@{}l@{}]\<[217I]%
\>[10]\AgdaSymbol{→}\AgdaSpace{}%
\AgdaSymbol{((}\AgdaBound{t₁}\AgdaSpace{}%
\AgdaOperator{\AgdaInductiveConstructor{×ᵤ}}\AgdaSpace{}%
\AgdaBound{t₂}\AgdaSymbol{)}\AgdaSpace{}%
\AgdaOperator{\AgdaInductiveConstructor{\#}}\AgdaSpace{}%
\AgdaSymbol{(}\AgdaBound{v₁}\AgdaSpace{}%
\AgdaOperator{\AgdaInductiveConstructor{,}}\AgdaSpace{}%
\AgdaBound{v₂}\AgdaSymbol{))}\AgdaSpace{}%
\AgdaOperator{\AgdaDatatype{⧟}}\AgdaSpace{}%
\AgdaSymbol{((}\AgdaBound{t₁}\AgdaSpace{}%
\AgdaOperator{\AgdaInductiveConstructor{\#}}\AgdaSpace{}%
\AgdaBound{v₁}\AgdaSymbol{)}\AgdaSpace{}%
\AgdaOperator{\AgdaInductiveConstructor{∙×ᵤ}}\AgdaSpace{}%
\AgdaSymbol{(}\AgdaBound{t₂}\AgdaSpace{}%
\AgdaOperator{\AgdaInductiveConstructor{\#}}\AgdaSpace{}%
\AgdaBound{v₂}\AgdaSymbol{))}\<%
\\
\>[2]\AgdaInductiveConstructor{∙\#times}%
\>[247I]\AgdaSymbol{:}\AgdaSpace{}%
\AgdaSymbol{\{}\AgdaBound{t₁}\AgdaSpace{}%
\AgdaBound{t₂}\AgdaSpace{}%
\AgdaSymbol{:}\AgdaSpace{}%
\AgdaDatatype{𝕌}\AgdaSymbol{\}}\AgdaSpace{}%
\AgdaSymbol{\{}\AgdaBound{v₁}\AgdaSpace{}%
\AgdaSymbol{:}\AgdaSpace{}%
\AgdaOperator{\AgdaFunction{⟦}}\AgdaSpace{}%
\AgdaBound{t₁}\AgdaSpace{}%
\AgdaOperator{\AgdaFunction{⟧}}\AgdaSymbol{\}}\AgdaSpace{}%
\AgdaSymbol{\{}\AgdaBound{v₂}\AgdaSpace{}%
\AgdaSymbol{:}\AgdaSpace{}%
\AgdaOperator{\AgdaFunction{⟦}}\AgdaSpace{}%
\AgdaBound{t₂}\AgdaSpace{}%
\AgdaOperator{\AgdaFunction{⟧}}\AgdaSymbol{\}}\<%
\\
\>[.][@{}l@{}]\<[247I]%
\>[10]\AgdaSymbol{→}\AgdaSpace{}%
\AgdaSymbol{((}\AgdaBound{t₁}\AgdaSpace{}%
\AgdaOperator{\AgdaInductiveConstructor{\#}}\AgdaSpace{}%
\AgdaBound{v₁}\AgdaSymbol{)}\AgdaSpace{}%
\AgdaOperator{\AgdaInductiveConstructor{∙×ᵤ}}\AgdaSpace{}%
\AgdaSymbol{(}\AgdaBound{t₂}\AgdaSpace{}%
\AgdaOperator{\AgdaInductiveConstructor{\#}}\AgdaSpace{}%
\AgdaBound{v₂}\AgdaSymbol{))}\AgdaSpace{}%
\AgdaOperator{\AgdaDatatype{⧟}}\AgdaSpace{}%
\AgdaSymbol{((}\AgdaBound{t₁}\AgdaSpace{}%
\AgdaOperator{\AgdaInductiveConstructor{×ᵤ}}\AgdaSpace{}%
\AgdaBound{t₂}\AgdaSymbol{)}\AgdaSpace{}%
\AgdaOperator{\AgdaInductiveConstructor{\#}}\AgdaSpace{}%
\AgdaSymbol{(}\AgdaBound{v₁}\AgdaSpace{}%
\AgdaOperator{\AgdaInductiveConstructor{,}}\AgdaSpace{}%
\AgdaBound{v₂}\AgdaSymbol{))}\<%
\\
\>[2]\AgdaComment{-- multiplicative structure (omitted)}\<%
\\
\>[2]\AgdaComment{-- monad / comonad}\<%
\\
\>[2]\AgdaInductiveConstructor{return}%
\>[11]\AgdaSymbol{:}\AgdaSpace{}%
\AgdaSymbol{\{}\AgdaBound{T}\AgdaSpace{}%
\AgdaSymbol{:}\AgdaSpace{}%
\AgdaDatatype{∙𝕌}\AgdaSymbol{\}}\AgdaSpace{}%
\AgdaSymbol{→}\AgdaSpace{}%
\AgdaBound{T}\AgdaSpace{}%
\AgdaOperator{\AgdaDatatype{⧟}}\AgdaSpace{}%
\AgdaOperator{\AgdaInductiveConstructor{❰}}\AgdaSpace{}%
\AgdaBound{T}\AgdaSpace{}%
\AgdaOperator{\AgdaInductiveConstructor{❱}}\<%
\\
\>[2]\AgdaInductiveConstructor{extract}%
\>[11]\AgdaSymbol{:}\AgdaSpace{}%
\AgdaSymbol{\{}\AgdaBound{T}\AgdaSpace{}%
\AgdaSymbol{:}\AgdaSpace{}%
\AgdaDatatype{∙𝕌}\AgdaSymbol{\}}\AgdaSpace{}%
\AgdaSymbol{→}\AgdaSpace{}%
\AgdaOperator{\AgdaInductiveConstructor{❰}}\AgdaSpace{}%
\AgdaBound{T}\AgdaSpace{}%
\AgdaOperator{\AgdaInductiveConstructor{❱}}\AgdaSpace{}%
\AgdaOperator{\AgdaDatatype{⧟}}\AgdaSpace{}%
\AgdaBound{T}\<%
\\
\>[2]\AgdaComment{-- eta/epsilon}\<%
\\
\>[2]\AgdaInductiveConstructor{η}%
\>[5]\AgdaSymbol{:}%
\>[8]\AgdaSymbol{(}\AgdaBound{T}\AgdaSpace{}%
\AgdaSymbol{:}\AgdaSpace{}%
\AgdaDatatype{∙𝕌}\AgdaSymbol{)}\AgdaSpace{}%
\AgdaSymbol{→}\AgdaSpace{}%
\AgdaFunction{∙𝟙}\AgdaSpace{}%
\AgdaOperator{\AgdaDatatype{⧟}}\AgdaSpace{}%
\AgdaOperator{\AgdaInductiveConstructor{❰}}\AgdaSpace{}%
\AgdaBound{T}\AgdaSpace{}%
\AgdaOperator{\AgdaInductiveConstructor{❱}}\AgdaSpace{}%
\AgdaOperator{\AgdaInductiveConstructor{∙×ᵤ}}\AgdaSpace{}%
\AgdaInductiveConstructor{∙𝟙/}\AgdaSpace{}%
\AgdaBound{T}\<%
\\
\>[2]\AgdaInductiveConstructor{ε}%
\>[5]\AgdaSymbol{:}%
\>[8]\AgdaSymbol{(}\AgdaBound{T}\AgdaSpace{}%
\AgdaSymbol{:}\AgdaSpace{}%
\AgdaDatatype{∙𝕌}\AgdaSymbol{)}\AgdaSpace{}%
\AgdaSymbol{→}\AgdaSpace{}%
\AgdaOperator{\AgdaInductiveConstructor{❰}}\AgdaSpace{}%
\AgdaBound{T}\AgdaSpace{}%
\AgdaOperator{\AgdaInductiveConstructor{❱}}\AgdaSpace{}%
\AgdaOperator{\AgdaInductiveConstructor{∙×ᵤ}}\AgdaSpace{}%
\AgdaInductiveConstructor{∙𝟙/}\AgdaSpace{}%
\AgdaBound{T}\AgdaSpace{}%
\AgdaOperator{\AgdaDatatype{⧟}}\AgdaSpace{}%
\AgdaFunction{∙𝟙}\<%
\end{code}
\begin{code}[hide]%
\>[2]\AgdaInductiveConstructor{∙id⟷}%
\>[12]\AgdaSymbol{:}%
\>[15]\AgdaSymbol{\{}\AgdaBound{T}\AgdaSpace{}%
\AgdaSymbol{:}\AgdaSpace{}%
\AgdaDatatype{∙𝕌}\AgdaSymbol{\}}\AgdaSpace{}%
\AgdaSymbol{→}\AgdaSpace{}%
\AgdaBound{T}\AgdaSpace{}%
\AgdaOperator{\AgdaDatatype{⧟}}\AgdaSpace{}%
\AgdaBound{T}\<%
\\
\>[2]\AgdaOperator{\AgdaInductiveConstructor{\AgdaUnderscore{}∙⊚\AgdaUnderscore{}}}%
\>[12]\AgdaSymbol{:}%
\>[15]\AgdaSymbol{\{}\AgdaBound{T₁}\AgdaSpace{}%
\AgdaBound{T₂}\AgdaSpace{}%
\AgdaBound{T₃}\AgdaSpace{}%
\AgdaSymbol{:}\AgdaSpace{}%
\AgdaDatatype{∙𝕌}\AgdaSymbol{\}}\AgdaSpace{}%
\AgdaSymbol{→}\AgdaSpace{}%
\AgdaSymbol{(}\AgdaBound{T₁}\AgdaSpace{}%
\AgdaOperator{\AgdaDatatype{⧟}}\AgdaSpace{}%
\AgdaBound{T₂}\AgdaSymbol{)}\AgdaSpace{}%
\AgdaSymbol{→}\AgdaSpace{}%
\AgdaSymbol{(}\AgdaBound{T₂}\AgdaSpace{}%
\AgdaOperator{\AgdaDatatype{⧟}}\AgdaSpace{}%
\AgdaBound{T₃}\AgdaSymbol{)}\AgdaSpace{}%
\AgdaSymbol{→}\AgdaSpace{}%
\AgdaSymbol{(}\AgdaBound{T₁}\AgdaSpace{}%
\AgdaOperator{\AgdaDatatype{⧟}}\AgdaSpace{}%
\AgdaBound{T₃}\AgdaSymbol{)}\<%
\\
\>[2]\AgdaComment{-- multiplicative structure}\<%
\\
\>[2]\AgdaInductiveConstructor{∙unite⋆l}%
\>[12]\AgdaSymbol{:}%
\>[15]\AgdaSymbol{\{}\AgdaBound{T}\AgdaSpace{}%
\AgdaSymbol{:}\AgdaSpace{}%
\AgdaDatatype{∙𝕌}\AgdaSymbol{\}}\AgdaSpace{}%
\AgdaSymbol{→}\AgdaSpace{}%
\AgdaFunction{∙𝟙}\AgdaSpace{}%
\AgdaOperator{\AgdaInductiveConstructor{∙×ᵤ}}\AgdaSpace{}%
\AgdaBound{T}\AgdaSpace{}%
\AgdaOperator{\AgdaDatatype{⧟}}\AgdaSpace{}%
\AgdaBound{T}\<%
\\
\>[2]\AgdaInductiveConstructor{∙uniti⋆l}%
\>[12]\AgdaSymbol{:}%
\>[15]\AgdaSymbol{\{}\AgdaBound{T}\AgdaSpace{}%
\AgdaSymbol{:}\AgdaSpace{}%
\AgdaDatatype{∙𝕌}\AgdaSymbol{\}}\AgdaSpace{}%
\AgdaSymbol{→}\AgdaSpace{}%
\AgdaBound{T}\AgdaSpace{}%
\AgdaOperator{\AgdaDatatype{⧟}}\AgdaSpace{}%
\AgdaFunction{∙𝟙}\AgdaSpace{}%
\AgdaOperator{\AgdaInductiveConstructor{∙×ᵤ}}\AgdaSpace{}%
\AgdaBound{T}\<%
\\
\>[2]\AgdaInductiveConstructor{∙unite⋆r}%
\>[12]\AgdaSymbol{:}%
\>[15]\AgdaSymbol{\{}\AgdaBound{T}\AgdaSpace{}%
\AgdaSymbol{:}\AgdaSpace{}%
\AgdaDatatype{∙𝕌}\AgdaSymbol{\}}\AgdaSpace{}%
\AgdaSymbol{→}\AgdaSpace{}%
\AgdaBound{T}\AgdaSpace{}%
\AgdaOperator{\AgdaInductiveConstructor{∙×ᵤ}}\AgdaSpace{}%
\AgdaFunction{∙𝟙}\AgdaSpace{}%
\AgdaOperator{\AgdaDatatype{⧟}}\AgdaSpace{}%
\AgdaBound{T}\<%
\\
\>[2]\AgdaInductiveConstructor{∙uniti⋆r}%
\>[12]\AgdaSymbol{:}%
\>[15]\AgdaSymbol{\{}\AgdaBound{T}\AgdaSpace{}%
\AgdaSymbol{:}\AgdaSpace{}%
\AgdaDatatype{∙𝕌}\AgdaSymbol{\}}\AgdaSpace{}%
\AgdaSymbol{→}\AgdaSpace{}%
\AgdaBound{T}\AgdaSpace{}%
\AgdaOperator{\AgdaDatatype{⧟}}\AgdaSpace{}%
\AgdaBound{T}\AgdaSpace{}%
\AgdaOperator{\AgdaInductiveConstructor{∙×ᵤ}}\AgdaSpace{}%
\AgdaFunction{∙𝟙}\<%
\\
\>[2]\AgdaInductiveConstructor{∙swap⋆}%
\>[12]\AgdaSymbol{:}%
\>[15]\AgdaSymbol{\{}\AgdaBound{T₁}\AgdaSpace{}%
\AgdaBound{T₂}\AgdaSpace{}%
\AgdaSymbol{:}\AgdaSpace{}%
\AgdaDatatype{∙𝕌}\AgdaSymbol{\}}\AgdaSpace{}%
\AgdaSymbol{→}\AgdaSpace{}%
\AgdaBound{T₁}\AgdaSpace{}%
\AgdaOperator{\AgdaInductiveConstructor{∙×ᵤ}}\AgdaSpace{}%
\AgdaBound{T₂}\AgdaSpace{}%
\AgdaOperator{\AgdaDatatype{⧟}}\AgdaSpace{}%
\AgdaBound{T₂}\AgdaSpace{}%
\AgdaOperator{\AgdaInductiveConstructor{∙×ᵤ}}\AgdaSpace{}%
\AgdaBound{T₁}\<%
\\
\>[2]\AgdaInductiveConstructor{∙assocl⋆}%
\>[12]\AgdaSymbol{:}%
\>[15]\AgdaSymbol{\{}\AgdaBound{T₁}\AgdaSpace{}%
\AgdaBound{T₂}\AgdaSpace{}%
\AgdaBound{T₃}\AgdaSpace{}%
\AgdaSymbol{:}\AgdaSpace{}%
\AgdaDatatype{∙𝕌}\AgdaSymbol{\}}\AgdaSpace{}%
\AgdaSymbol{→}\<%
\\
\>[15]\AgdaBound{T₁}\AgdaSpace{}%
\AgdaOperator{\AgdaInductiveConstructor{∙×ᵤ}}\AgdaSpace{}%
\AgdaSymbol{(}\AgdaBound{T₂}\AgdaSpace{}%
\AgdaOperator{\AgdaInductiveConstructor{∙×ᵤ}}\AgdaSpace{}%
\AgdaBound{T₃}\AgdaSymbol{)}\AgdaSpace{}%
\AgdaOperator{\AgdaDatatype{⧟}}\AgdaSpace{}%
\AgdaSymbol{(}\AgdaBound{T₁}\AgdaSpace{}%
\AgdaOperator{\AgdaInductiveConstructor{∙×ᵤ}}\AgdaSpace{}%
\AgdaBound{T₂}\AgdaSymbol{)}\AgdaSpace{}%
\AgdaOperator{\AgdaInductiveConstructor{∙×ᵤ}}\AgdaSpace{}%
\AgdaBound{T₃}\<%
\\
\>[2]\AgdaInductiveConstructor{∙assocr⋆}%
\>[12]\AgdaSymbol{:}%
\>[15]\AgdaSymbol{\{}\AgdaBound{T₁}\AgdaSpace{}%
\AgdaBound{T₂}\AgdaSpace{}%
\AgdaBound{T₃}\AgdaSpace{}%
\AgdaSymbol{:}\AgdaSpace{}%
\AgdaDatatype{∙𝕌}\AgdaSymbol{\}}\AgdaSpace{}%
\AgdaSymbol{→}\<%
\\
\>[15]\AgdaSymbol{(}\AgdaBound{T₁}\AgdaSpace{}%
\AgdaOperator{\AgdaInductiveConstructor{∙×ᵤ}}\AgdaSpace{}%
\AgdaBound{T₂}\AgdaSymbol{)}\AgdaSpace{}%
\AgdaOperator{\AgdaInductiveConstructor{∙×ᵤ}}\AgdaSpace{}%
\AgdaBound{T₃}\AgdaSpace{}%
\AgdaOperator{\AgdaDatatype{⧟}}\AgdaSpace{}%
\AgdaBound{T₁}\AgdaSpace{}%
\AgdaOperator{\AgdaInductiveConstructor{∙×ᵤ}}\AgdaSpace{}%
\AgdaSymbol{(}\AgdaBound{T₂}\AgdaSpace{}%
\AgdaOperator{\AgdaInductiveConstructor{∙×ᵤ}}\AgdaSpace{}%
\AgdaBound{T₃}\AgdaSymbol{)}\<%
\\
\>[2]\AgdaOperator{\AgdaInductiveConstructor{\AgdaUnderscore{}∙⊗\AgdaUnderscore{}}}%
\>[12]\AgdaSymbol{:}%
\>[15]\AgdaSymbol{\{}\AgdaBound{T₁}\AgdaSpace{}%
\AgdaBound{T₂}\AgdaSpace{}%
\AgdaBound{T₃}\AgdaSpace{}%
\AgdaBound{T₄}\AgdaSpace{}%
\AgdaSymbol{:}\AgdaSpace{}%
\AgdaDatatype{∙𝕌}\AgdaSymbol{\}}\AgdaSpace{}%
\AgdaSymbol{→}\AgdaSpace{}%
\AgdaSymbol{(}\AgdaBound{T₁}\AgdaSpace{}%
\AgdaOperator{\AgdaDatatype{⧟}}\AgdaSpace{}%
\AgdaBound{T₃}\AgdaSymbol{)}\AgdaSpace{}%
\AgdaSymbol{→}\AgdaSpace{}%
\AgdaSymbol{(}\AgdaBound{T₂}\AgdaSpace{}%
\AgdaOperator{\AgdaDatatype{⧟}}\AgdaSpace{}%
\AgdaBound{T₄}\AgdaSymbol{)}\AgdaSpace{}%
\AgdaSymbol{→}\<%
\\
\>[15]\AgdaSymbol{(}\AgdaBound{T₁}\AgdaSpace{}%
\AgdaOperator{\AgdaInductiveConstructor{∙×ᵤ}}\AgdaSpace{}%
\AgdaBound{T₂}\AgdaSpace{}%
\AgdaOperator{\AgdaDatatype{⧟}}\AgdaSpace{}%
\AgdaBound{T₃}\AgdaSpace{}%
\AgdaOperator{\AgdaInductiveConstructor{∙×ᵤ}}\AgdaSpace{}%
\AgdaBound{T₄}\AgdaSymbol{)}\<%
\end{code}}
\newcommand{\PIPFCombderive}{%
\begin{code}%
\>[0]\AgdaFunction{∙Singᵤ}\AgdaSpace{}%
\AgdaSymbol{:}\AgdaSpace{}%
\AgdaSymbol{\{}\AgdaBound{T₁}\AgdaSpace{}%
\AgdaBound{T₂}\AgdaSpace{}%
\AgdaSymbol{:}\AgdaSpace{}%
\AgdaDatatype{∙𝕌}\AgdaSymbol{\}}\AgdaSpace{}%
\AgdaSymbol{→}\AgdaSpace{}%
\AgdaSymbol{(}\AgdaBound{T₁}\AgdaSpace{}%
\AgdaOperator{\AgdaDatatype{⧟}}\AgdaSpace{}%
\AgdaBound{T₂}\AgdaSymbol{)}\AgdaSpace{}%
\AgdaSymbol{→}\AgdaSpace{}%
\AgdaOperator{\AgdaInductiveConstructor{❰}}\AgdaSpace{}%
\AgdaBound{T₁}\AgdaSpace{}%
\AgdaOperator{\AgdaInductiveConstructor{❱}}\AgdaSpace{}%
\AgdaOperator{\AgdaDatatype{⧟}}\AgdaSpace{}%
\AgdaOperator{\AgdaInductiveConstructor{❰}}\AgdaSpace{}%
\AgdaBound{T₂}\AgdaSpace{}%
\AgdaOperator{\AgdaInductiveConstructor{❱}}\<%
\\
\>[0]\AgdaFunction{∙Singᵤ}\AgdaSpace{}%
\AgdaSymbol{\{}\AgdaBound{T₁}\AgdaSymbol{\}}\AgdaSpace{}%
\AgdaSymbol{\{}\AgdaBound{T₂}\AgdaSymbol{\}}\AgdaSpace{}%
\AgdaBound{c}\AgdaSpace{}%
\AgdaSymbol{=}\AgdaSpace{}%
\AgdaInductiveConstructor{extract}\AgdaSpace{}%
\AgdaOperator{\AgdaInductiveConstructor{∙⊚}}\AgdaSpace{}%
\AgdaBound{c}\AgdaSpace{}%
\AgdaOperator{\AgdaInductiveConstructor{∙⊚}}\AgdaSpace{}%
\AgdaInductiveConstructor{return}\<%
\\
\\[\AgdaEmptyExtraSkip]%
\>[0]\AgdaFunction{tensor}\AgdaSpace{}%
\AgdaSymbol{:}\AgdaSpace{}%
\AgdaSymbol{\{}\AgdaBound{T₁}\AgdaSpace{}%
\AgdaBound{T₂}\AgdaSpace{}%
\AgdaSymbol{:}\AgdaSpace{}%
\AgdaDatatype{∙𝕌}\AgdaSymbol{\}}\AgdaSpace{}%
\AgdaSymbol{→}\AgdaSpace{}%
\AgdaOperator{\AgdaInductiveConstructor{❰}}\AgdaSpace{}%
\AgdaBound{T₁}\AgdaSpace{}%
\AgdaOperator{\AgdaInductiveConstructor{❱}}\AgdaSpace{}%
\AgdaOperator{\AgdaInductiveConstructor{∙×ᵤ}}\AgdaSpace{}%
\AgdaOperator{\AgdaInductiveConstructor{❰}}\AgdaSpace{}%
\AgdaBound{T₂}\AgdaSpace{}%
\AgdaOperator{\AgdaInductiveConstructor{❱}}\AgdaSpace{}%
\AgdaOperator{\AgdaDatatype{⧟}}\AgdaSpace{}%
\AgdaOperator{\AgdaInductiveConstructor{❰}}\AgdaSpace{}%
\AgdaBound{T₁}\AgdaSpace{}%
\AgdaOperator{\AgdaInductiveConstructor{∙×ᵤ}}\AgdaSpace{}%
\AgdaBound{T₂}\AgdaSpace{}%
\AgdaOperator{\AgdaInductiveConstructor{❱}}\<%
\\
\>[0]\AgdaFunction{tensor}\AgdaSpace{}%
\AgdaSymbol{\{}\AgdaBound{T₁}\AgdaSymbol{\}}\AgdaSpace{}%
\AgdaSymbol{\{}\AgdaBound{T₂}\AgdaSymbol{\}}\AgdaSpace{}%
\AgdaSymbol{=}\AgdaSpace{}%
\AgdaSymbol{(}\AgdaInductiveConstructor{extract}\AgdaSpace{}%
\AgdaOperator{\AgdaInductiveConstructor{∙⊗}}\AgdaSpace{}%
\AgdaInductiveConstructor{extract}\AgdaSymbol{)}\AgdaSpace{}%
\AgdaOperator{\AgdaInductiveConstructor{∙⊚}}\AgdaSpace{}%
\AgdaInductiveConstructor{return}\<%
\\
\\[\AgdaEmptyExtraSkip]%
\>[0]\AgdaFunction{cotensor}\AgdaSpace{}%
\AgdaSymbol{:}\AgdaSpace{}%
\AgdaSymbol{\{}\AgdaBound{T₁}\AgdaSpace{}%
\AgdaBound{T₂}\AgdaSpace{}%
\AgdaSymbol{:}\AgdaSpace{}%
\AgdaDatatype{∙𝕌}\AgdaSymbol{\}}\AgdaSpace{}%
\AgdaSymbol{→}\AgdaSpace{}%
\AgdaOperator{\AgdaInductiveConstructor{❰}}\AgdaSpace{}%
\AgdaBound{T₁}\AgdaSpace{}%
\AgdaOperator{\AgdaInductiveConstructor{∙×ᵤ}}\AgdaSpace{}%
\AgdaBound{T₂}\AgdaSpace{}%
\AgdaOperator{\AgdaInductiveConstructor{❱}}\AgdaSpace{}%
\AgdaOperator{\AgdaDatatype{⧟}}\AgdaSpace{}%
\AgdaOperator{\AgdaInductiveConstructor{❰}}\AgdaSpace{}%
\AgdaBound{T₁}\AgdaSpace{}%
\AgdaOperator{\AgdaInductiveConstructor{❱}}\AgdaSpace{}%
\AgdaOperator{\AgdaInductiveConstructor{∙×ᵤ}}\AgdaSpace{}%
\AgdaOperator{\AgdaInductiveConstructor{❰}}\AgdaSpace{}%
\AgdaBound{T₂}\AgdaSpace{}%
\AgdaOperator{\AgdaInductiveConstructor{❱}}\<%
\\
\>[0]\AgdaFunction{cotensor}\AgdaSpace{}%
\AgdaSymbol{\{}\AgdaBound{T₁}\AgdaSymbol{\}}\AgdaSpace{}%
\AgdaSymbol{\{}\AgdaBound{T₂}\AgdaSymbol{\}}\AgdaSpace{}%
\AgdaSymbol{=}\AgdaSpace{}%
\AgdaInductiveConstructor{extract}\AgdaSpace{}%
\AgdaOperator{\AgdaInductiveConstructor{∙⊚}}\AgdaSpace{}%
\AgdaSymbol{(}\AgdaInductiveConstructor{return}\AgdaSpace{}%
\AgdaOperator{\AgdaInductiveConstructor{∙⊗}}\AgdaSpace{}%
\AgdaInductiveConstructor{return}\AgdaSymbol{)}\<%
\\
\\[\AgdaEmptyExtraSkip]%
\>[0]\AgdaFunction{join}\AgdaSpace{}%
\AgdaSymbol{:}\AgdaSpace{}%
\AgdaSymbol{\{}\AgdaBound{T₁}\AgdaSpace{}%
\AgdaSymbol{:}\AgdaSpace{}%
\AgdaDatatype{∙𝕌}\AgdaSymbol{\}}\AgdaSpace{}%
\AgdaSymbol{→}%
\>[20]\AgdaOperator{\AgdaInductiveConstructor{❰}}\AgdaSpace{}%
\AgdaOperator{\AgdaInductiveConstructor{❰}}\AgdaSpace{}%
\AgdaBound{T₁}\AgdaSpace{}%
\AgdaOperator{\AgdaInductiveConstructor{❱}}\AgdaSpace{}%
\AgdaOperator{\AgdaInductiveConstructor{❱}}\AgdaSpace{}%
\AgdaOperator{\AgdaDatatype{⧟}}\AgdaSpace{}%
\AgdaOperator{\AgdaInductiveConstructor{❰}}\AgdaSpace{}%
\AgdaBound{T₁}\AgdaSpace{}%
\AgdaOperator{\AgdaInductiveConstructor{❱}}\<%
\\
\>[0]\AgdaFunction{join}\AgdaSpace{}%
\AgdaSymbol{\{}\AgdaBound{T₁}\AgdaSymbol{\}}\AgdaSpace{}%
\AgdaSymbol{=}\AgdaSpace{}%
\AgdaInductiveConstructor{extract}\<%
\\
\\[\AgdaEmptyExtraSkip]%
\>[0]\AgdaFunction{duplicate}\AgdaSpace{}%
\AgdaSymbol{:}\AgdaSpace{}%
\AgdaSymbol{\{}\AgdaBound{T₁}\AgdaSpace{}%
\AgdaSymbol{:}\AgdaSpace{}%
\AgdaDatatype{∙𝕌}\AgdaSymbol{\}}\AgdaSpace{}%
\AgdaSymbol{→}\AgdaSpace{}%
\AgdaOperator{\AgdaInductiveConstructor{❰}}\AgdaSpace{}%
\AgdaBound{T₁}\AgdaSpace{}%
\AgdaOperator{\AgdaInductiveConstructor{❱}}\AgdaSpace{}%
\AgdaOperator{\AgdaDatatype{⧟}}\AgdaSpace{}%
\AgdaOperator{\AgdaInductiveConstructor{❰}}\AgdaSpace{}%
\AgdaOperator{\AgdaInductiveConstructor{❰}}\AgdaSpace{}%
\AgdaBound{T₁}\AgdaSpace{}%
\AgdaOperator{\AgdaInductiveConstructor{❱}}\AgdaSpace{}%
\AgdaOperator{\AgdaInductiveConstructor{❱}}\<%
\\
\>[0]\AgdaFunction{duplicate}\AgdaSpace{}%
\AgdaSymbol{\{}\AgdaBound{T₁}\AgdaSymbol{\}}\AgdaSpace{}%
\AgdaSymbol{=}\AgdaSpace{}%
\AgdaInductiveConstructor{return}\<%
\end{code}}

\newcommand{\PIPFrev}{%
\begin{code}%
\>[0]\AgdaOperator{\AgdaFunction{!∙\AgdaUnderscore{}}}\AgdaSpace{}%
\AgdaSymbol{:}\AgdaSpace{}%
\AgdaSymbol{\{}\AgdaBound{A}\AgdaSpace{}%
\AgdaBound{B}\AgdaSpace{}%
\AgdaSymbol{:}\AgdaSpace{}%
\AgdaDatatype{∙𝕌}\AgdaSymbol{\}}\AgdaSpace{}%
\AgdaSymbol{→}\AgdaSpace{}%
\AgdaBound{A}\AgdaSpace{}%
\AgdaOperator{\AgdaDatatype{⧟}}\AgdaSpace{}%
\AgdaBound{B}\AgdaSpace{}%
\AgdaSymbol{→}\AgdaSpace{}%
\AgdaBound{B}\AgdaSpace{}%
\AgdaOperator{\AgdaDatatype{⧟}}\AgdaSpace{}%
\AgdaBound{A}\<%
\end{code}
\begin{code}[hide]%
\>[0]\AgdaOperator{\AgdaFunction{!∙}}\AgdaSpace{}%
\AgdaSymbol{(}\AgdaInductiveConstructor{∙c}\AgdaSpace{}%
\AgdaSymbol{\{}\AgdaBound{t₁}\AgdaSymbol{\}}\AgdaSpace{}%
\AgdaSymbol{\{}\AgdaBound{t₂}\AgdaSymbol{\}}\AgdaSpace{}%
\AgdaSymbol{\{}\AgdaBound{v}\AgdaSymbol{\}}\AgdaSpace{}%
\AgdaBound{c}\AgdaSymbol{)}\AgdaSpace{}%
\AgdaSymbol{=}\AgdaSpace{}%
\AgdaFunction{subst}%
\>[661I]\AgdaSymbol{(λ}\AgdaSpace{}%
\AgdaBound{x}\AgdaSpace{}%
\AgdaSymbol{→}\AgdaSpace{}%
\AgdaBound{t₂}\AgdaSpace{}%
\AgdaOperator{\AgdaInductiveConstructor{\#}}\AgdaSpace{}%
\AgdaFunction{eval}\AgdaSpace{}%
\AgdaBound{c}\AgdaSpace{}%
\AgdaBound{v}\AgdaSpace{}%
\AgdaOperator{\AgdaDatatype{⧟}}\AgdaSpace{}%
\AgdaBound{t₁}\AgdaSpace{}%
\AgdaOperator{\AgdaInductiveConstructor{\#}}\AgdaSpace{}%
\AgdaBound{x}\AgdaSymbol{)}\<%
\\
\>[.][@{}l@{}]\<[661I]%
\>[32]\AgdaSymbol{(}\AgdaFunction{ΠisRev}\AgdaSpace{}%
\AgdaBound{c}\AgdaSpace{}%
\AgdaBound{v}\AgdaSymbol{)}\AgdaSpace{}%
\AgdaSymbol{(}\AgdaInductiveConstructor{∙c}\AgdaSpace{}%
\AgdaSymbol{(}\AgdaOperator{\AgdaFunction{!}}\AgdaSpace{}%
\AgdaBound{c}\AgdaSymbol{))}\<%
\\
\>[0]\AgdaOperator{\AgdaFunction{!∙}}\AgdaSpace{}%
\AgdaInductiveConstructor{∙times\#}\AgdaSpace{}%
\AgdaSymbol{=}\AgdaSpace{}%
\AgdaInductiveConstructor{∙\#times}\<%
\\
\>[0]\AgdaOperator{\AgdaFunction{!∙}}\AgdaSpace{}%
\AgdaInductiveConstructor{∙\#times}\AgdaSpace{}%
\AgdaSymbol{=}\AgdaSpace{}%
\AgdaInductiveConstructor{∙times\#}\<%
\\
\>[0]\AgdaOperator{\AgdaFunction{!∙}}\AgdaSpace{}%
\AgdaInductiveConstructor{∙id⟷}\AgdaSpace{}%
\AgdaSymbol{=}\AgdaSpace{}%
\AgdaInductiveConstructor{∙id⟷}\<%
\\
\>[0]\AgdaOperator{\AgdaFunction{!∙}}\AgdaSpace{}%
\AgdaSymbol{(}\AgdaBound{c₁}\AgdaSpace{}%
\AgdaOperator{\AgdaInductiveConstructor{∙⊚}}\AgdaSpace{}%
\AgdaBound{c₂}\AgdaSymbol{)}\AgdaSpace{}%
\AgdaSymbol{=}\AgdaSpace{}%
\AgdaSymbol{(}\AgdaOperator{\AgdaFunction{!∙}}\AgdaSpace{}%
\AgdaBound{c₂}\AgdaSymbol{)}\AgdaSpace{}%
\AgdaOperator{\AgdaInductiveConstructor{∙⊚}}\AgdaSpace{}%
\AgdaSymbol{(}\AgdaOperator{\AgdaFunction{!∙}}\AgdaSpace{}%
\AgdaBound{c₁}\AgdaSymbol{)}\<%
\\
\>[0]\AgdaOperator{\AgdaFunction{!∙}}\AgdaSpace{}%
\AgdaInductiveConstructor{∙unite⋆l}\AgdaSpace{}%
\AgdaSymbol{=}\AgdaSpace{}%
\AgdaInductiveConstructor{∙uniti⋆l}\<%
\\
\>[0]\AgdaOperator{\AgdaFunction{!∙}}\AgdaSpace{}%
\AgdaInductiveConstructor{∙uniti⋆l}\AgdaSpace{}%
\AgdaSymbol{=}\AgdaSpace{}%
\AgdaInductiveConstructor{∙unite⋆l}\<%
\\
\>[0]\AgdaOperator{\AgdaFunction{!∙}}\AgdaSpace{}%
\AgdaInductiveConstructor{∙unite⋆r}\AgdaSpace{}%
\AgdaSymbol{=}\AgdaSpace{}%
\AgdaInductiveConstructor{∙uniti⋆r}\<%
\\
\>[0]\AgdaOperator{\AgdaFunction{!∙}}\AgdaSpace{}%
\AgdaInductiveConstructor{∙uniti⋆r}\AgdaSpace{}%
\AgdaSymbol{=}\AgdaSpace{}%
\AgdaInductiveConstructor{∙unite⋆r}\<%
\\
\>[0]\AgdaOperator{\AgdaFunction{!∙}}\AgdaSpace{}%
\AgdaInductiveConstructor{∙swap⋆}\AgdaSpace{}%
\AgdaSymbol{=}\AgdaSpace{}%
\AgdaInductiveConstructor{∙swap⋆}\<%
\\
\>[0]\AgdaOperator{\AgdaFunction{!∙}}\AgdaSpace{}%
\AgdaInductiveConstructor{∙assocl⋆}\AgdaSpace{}%
\AgdaSymbol{=}\AgdaSpace{}%
\AgdaInductiveConstructor{∙assocr⋆}\<%
\\
\>[0]\AgdaOperator{\AgdaFunction{!∙}}\AgdaSpace{}%
\AgdaInductiveConstructor{∙assocr⋆}\AgdaSpace{}%
\AgdaSymbol{=}\AgdaSpace{}%
\AgdaInductiveConstructor{∙assocl⋆}\<%
\\
\>[0]\AgdaOperator{\AgdaFunction{!∙}}\AgdaSpace{}%
\AgdaSymbol{(}\AgdaBound{c₁}\AgdaSpace{}%
\AgdaOperator{\AgdaInductiveConstructor{∙⊗}}\AgdaSpace{}%
\AgdaBound{c₂}\AgdaSymbol{)}\AgdaSpace{}%
\AgdaSymbol{=}\AgdaSpace{}%
\AgdaSymbol{(}\AgdaOperator{\AgdaFunction{!∙}}\AgdaSpace{}%
\AgdaBound{c₁}\AgdaSymbol{)}\AgdaSpace{}%
\AgdaOperator{\AgdaInductiveConstructor{∙⊗}}\AgdaSpace{}%
\AgdaSymbol{(}\AgdaOperator{\AgdaFunction{!∙}}\AgdaSpace{}%
\AgdaBound{c₂}\AgdaSymbol{)}\<%
\\
\>[0]\AgdaOperator{\AgdaFunction{!∙}}\AgdaSpace{}%
\AgdaInductiveConstructor{return}\AgdaSpace{}%
\AgdaSymbol{=}\AgdaSpace{}%
\AgdaInductiveConstructor{extract}\<%
\\
\>[0]\AgdaOperator{\AgdaFunction{!∙}}\AgdaSpace{}%
\AgdaInductiveConstructor{extract}\AgdaSpace{}%
\AgdaSymbol{=}\AgdaSpace{}%
\AgdaInductiveConstructor{return}\<%
\\
\>[0]\AgdaOperator{\AgdaFunction{!∙}}\AgdaSpace{}%
\AgdaInductiveConstructor{η}\AgdaSpace{}%
\AgdaBound{T}\AgdaSpace{}%
\AgdaSymbol{=}\AgdaSpace{}%
\AgdaInductiveConstructor{ε}\AgdaSpace{}%
\AgdaBound{T}\<%
\\
\>[0]\AgdaOperator{\AgdaFunction{!∙}}\AgdaSpace{}%
\AgdaInductiveConstructor{ε}\AgdaSpace{}%
\AgdaBound{T}\AgdaSpace{}%
\AgdaSymbol{=}\AgdaSpace{}%
\AgdaInductiveConstructor{η}\AgdaSpace{}%
\AgdaBound{T}\<%
\end{code}}
\newcommand{\PIPFeval}{
\begin{code}%
\>[0]\AgdaFunction{∙eval}\AgdaSpace{}%
\AgdaSymbol{:}%
\>[9]\AgdaSymbol{\{}\AgdaBound{T₁}\AgdaSpace{}%
\AgdaBound{T₂}\AgdaSpace{}%
\AgdaSymbol{:}\AgdaSpace{}%
\AgdaDatatype{∙𝕌}\AgdaSymbol{\}}\AgdaSpace{}%
\AgdaSymbol{→}\AgdaSpace{}%
\AgdaSymbol{(}\AgdaBound{C}\AgdaSpace{}%
\AgdaSymbol{:}\AgdaSpace{}%
\AgdaBound{T₁}\AgdaSpace{}%
\AgdaOperator{\AgdaDatatype{⧟}}\AgdaSpace{}%
\AgdaBound{T₂}\AgdaSymbol{)}\AgdaSpace{}%
\AgdaSymbol{→}\<%
\\
\>[9]\AgdaKeyword{let}\AgdaSpace{}%
\AgdaSymbol{(}\AgdaBound{t₁}\AgdaSpace{}%
\AgdaOperator{\AgdaInductiveConstructor{,}}\AgdaSpace{}%
\AgdaBound{v₁}\AgdaSymbol{)}\AgdaSpace{}%
\AgdaSymbol{=}\AgdaSpace{}%
\AgdaOperator{\AgdaFunction{∙⟦}}\AgdaSpace{}%
\AgdaBound{T₁}\AgdaSpace{}%
\AgdaOperator{\AgdaFunction{⟧}}\AgdaSymbol{;}\AgdaSpace{}%
\AgdaSymbol{(}\AgdaBound{t₂}\AgdaSpace{}%
\AgdaOperator{\AgdaInductiveConstructor{,}}\AgdaSpace{}%
\AgdaBound{v₂}\AgdaSymbol{)}\AgdaSpace{}%
\AgdaSymbol{=}\AgdaSpace{}%
\AgdaOperator{\AgdaFunction{∙⟦}}\AgdaSpace{}%
\AgdaBound{T₂}\AgdaSpace{}%
\AgdaOperator{\AgdaFunction{⟧}}\<%
\\
\>[9]\AgdaKeyword{in}%
\>[13]\AgdaRecord{Σ}\AgdaSpace{}%
\AgdaSymbol{(}\AgdaBound{t₁}\AgdaSpace{}%
\AgdaSymbol{→}\AgdaSpace{}%
\AgdaBound{t₂}\AgdaSymbol{)}\AgdaSpace{}%
\AgdaSymbol{(λ}\AgdaSpace{}%
\AgdaBound{f}\AgdaSpace{}%
\AgdaSymbol{→}\AgdaSpace{}%
\AgdaBound{f}\AgdaSpace{}%
\AgdaBound{v₁}\AgdaSpace{}%
\AgdaOperator{\AgdaDatatype{≡}}\AgdaSpace{}%
\AgdaBound{v₂}\AgdaSymbol{)}\<%
\end{code}
\begin{code}[hide]%
\>[0]\AgdaFunction{∙eval}\AgdaSpace{}%
\AgdaInductiveConstructor{∙id⟷}\AgdaSpace{}%
\AgdaSymbol{=}\AgdaSpace{}%
\AgdaSymbol{(λ}\AgdaSpace{}%
\AgdaBound{x}\AgdaSpace{}%
\AgdaSymbol{→}\AgdaSpace{}%
\AgdaBound{x}\AgdaSymbol{)}\AgdaSpace{}%
\AgdaOperator{\AgdaInductiveConstructor{,}}\AgdaSpace{}%
\AgdaInductiveConstructor{refl}\<%
\\
\>[0]\AgdaFunction{∙eval}\AgdaSpace{}%
\AgdaSymbol{(}\AgdaInductiveConstructor{∙c}\AgdaSpace{}%
\AgdaBound{c}\AgdaSymbol{)}\AgdaSpace{}%
\AgdaSymbol{=}\AgdaSpace{}%
\AgdaFunction{eval}\AgdaSpace{}%
\AgdaBound{c}\AgdaSpace{}%
\AgdaOperator{\AgdaInductiveConstructor{,}}\AgdaSpace{}%
\AgdaInductiveConstructor{refl}\<%
\\
\>[0]\AgdaFunction{∙eval}\AgdaSpace{}%
\AgdaSymbol{(}\AgdaBound{C₁}\AgdaSpace{}%
\AgdaOperator{\AgdaInductiveConstructor{∙⊚}}\AgdaSpace{}%
\AgdaBound{C₂}\AgdaSymbol{)}\AgdaSpace{}%
\AgdaKeyword{with}\AgdaSpace{}%
\AgdaFunction{∙eval}\AgdaSpace{}%
\AgdaBound{C₁}\AgdaSpace{}%
\AgdaSymbol{|}\AgdaSpace{}%
\AgdaFunction{∙eval}\AgdaSpace{}%
\AgdaBound{C₂}\<%
\\
\>[0]\AgdaSymbol{...}\AgdaSpace{}%
\AgdaSymbol{|}\AgdaSpace{}%
\AgdaSymbol{(}\AgdaBound{f}\AgdaSpace{}%
\AgdaOperator{\AgdaInductiveConstructor{,}}\AgdaSpace{}%
\AgdaBound{p}\AgdaSymbol{)}\AgdaSpace{}%
\AgdaSymbol{|}\AgdaSpace{}%
\AgdaSymbol{(}\AgdaBound{g}\AgdaSpace{}%
\AgdaOperator{\AgdaInductiveConstructor{,}}\AgdaSpace{}%
\AgdaBound{q}\AgdaSymbol{)}\AgdaSpace{}%
\AgdaSymbol{=}\AgdaSpace{}%
\AgdaBound{g}\AgdaSpace{}%
\AgdaOperator{\AgdaFunction{∘}}\AgdaSpace{}%
\AgdaBound{f}\AgdaSpace{}%
\AgdaOperator{\AgdaInductiveConstructor{,}}\AgdaSpace{}%
\AgdaFunction{trans}\AgdaSpace{}%
\AgdaSymbol{(}\AgdaFunction{cong}\AgdaSpace{}%
\AgdaBound{g}\AgdaSpace{}%
\AgdaBound{p}\AgdaSymbol{)}\AgdaSpace{}%
\AgdaBound{q}\<%
\\
\>[0]\AgdaFunction{∙eval}\AgdaSpace{}%
\AgdaInductiveConstructor{∙unite⋆l}\AgdaSpace{}%
\AgdaSymbol{=}\AgdaSpace{}%
\AgdaSymbol{(λ}\AgdaSpace{}%
\AgdaSymbol{\{(}\AgdaInductiveConstructor{tt}\AgdaSpace{}%
\AgdaOperator{\AgdaInductiveConstructor{,}}\AgdaSpace{}%
\AgdaBound{x}\AgdaSymbol{)}\AgdaSpace{}%
\AgdaSymbol{→}\AgdaSpace{}%
\AgdaBound{x}\AgdaSymbol{\})}\AgdaSpace{}%
\AgdaOperator{\AgdaInductiveConstructor{,}}\AgdaSpace{}%
\AgdaInductiveConstructor{refl}\<%
\\
\>[0]\AgdaFunction{∙eval}\AgdaSpace{}%
\AgdaInductiveConstructor{∙uniti⋆l}\AgdaSpace{}%
\AgdaSymbol{=}\AgdaSpace{}%
\AgdaSymbol{(λ}\AgdaSpace{}%
\AgdaSymbol{\{}\AgdaBound{x}\AgdaSpace{}%
\AgdaSymbol{→}\AgdaSpace{}%
\AgdaSymbol{(}\AgdaInductiveConstructor{tt}\AgdaSpace{}%
\AgdaOperator{\AgdaInductiveConstructor{,}}\AgdaSpace{}%
\AgdaBound{x}\AgdaSymbol{)\})}\AgdaSpace{}%
\AgdaOperator{\AgdaInductiveConstructor{,}}\AgdaSpace{}%
\AgdaInductiveConstructor{refl}\<%
\\
\>[0]\AgdaFunction{∙eval}\AgdaSpace{}%
\AgdaInductiveConstructor{∙unite⋆r}\AgdaSpace{}%
\AgdaSymbol{=}\AgdaSpace{}%
\AgdaSymbol{(λ}\AgdaSpace{}%
\AgdaSymbol{\{(}\AgdaBound{x}\AgdaSpace{}%
\AgdaOperator{\AgdaInductiveConstructor{,}}\AgdaSpace{}%
\AgdaInductiveConstructor{tt}\AgdaSymbol{)}\AgdaSpace{}%
\AgdaSymbol{→}\AgdaSpace{}%
\AgdaBound{x}\AgdaSymbol{\})}\AgdaSpace{}%
\AgdaOperator{\AgdaInductiveConstructor{,}}\AgdaSpace{}%
\AgdaInductiveConstructor{refl}\<%
\\
\>[0]\AgdaFunction{∙eval}\AgdaSpace{}%
\AgdaInductiveConstructor{∙uniti⋆r}\AgdaSpace{}%
\AgdaSymbol{=}\AgdaSpace{}%
\AgdaSymbol{(λ}\AgdaSpace{}%
\AgdaSymbol{\{}\AgdaBound{x}\AgdaSpace{}%
\AgdaSymbol{→}\AgdaSpace{}%
\AgdaSymbol{(}\AgdaBound{x}\AgdaSpace{}%
\AgdaOperator{\AgdaInductiveConstructor{,}}\AgdaSpace{}%
\AgdaInductiveConstructor{tt}\AgdaSymbol{)\})}\AgdaSpace{}%
\AgdaOperator{\AgdaInductiveConstructor{,}}\AgdaSpace{}%
\AgdaInductiveConstructor{refl}\<%
\\
\>[0]\AgdaFunction{∙eval}\AgdaSpace{}%
\AgdaInductiveConstructor{∙swap⋆}\AgdaSpace{}%
\AgdaSymbol{=}\AgdaSpace{}%
\AgdaSymbol{(λ}\AgdaSpace{}%
\AgdaSymbol{\{(}\AgdaBound{x}\AgdaSpace{}%
\AgdaOperator{\AgdaInductiveConstructor{,}}\AgdaSpace{}%
\AgdaBound{y}\AgdaSymbol{)}\AgdaSpace{}%
\AgdaSymbol{→}\AgdaSpace{}%
\AgdaBound{y}\AgdaSpace{}%
\AgdaOperator{\AgdaInductiveConstructor{,}}\AgdaSpace{}%
\AgdaBound{x}\AgdaSymbol{\})}\AgdaSpace{}%
\AgdaOperator{\AgdaInductiveConstructor{,}}\AgdaSpace{}%
\AgdaInductiveConstructor{refl}\<%
\\
\>[0]\AgdaFunction{∙eval}\AgdaSpace{}%
\AgdaInductiveConstructor{∙assocl⋆}\AgdaSpace{}%
\AgdaSymbol{=}\AgdaSpace{}%
\AgdaSymbol{(λ}\AgdaSpace{}%
\AgdaSymbol{\{(}\AgdaBound{x}\AgdaSpace{}%
\AgdaOperator{\AgdaInductiveConstructor{,}}\AgdaSpace{}%
\AgdaSymbol{(}\AgdaBound{y}\AgdaSpace{}%
\AgdaOperator{\AgdaInductiveConstructor{,}}\AgdaSpace{}%
\AgdaBound{z}\AgdaSymbol{))}\AgdaSpace{}%
\AgdaSymbol{→}\AgdaSpace{}%
\AgdaSymbol{((}\AgdaBound{x}\AgdaSpace{}%
\AgdaOperator{\AgdaInductiveConstructor{,}}\AgdaSpace{}%
\AgdaBound{y}\AgdaSymbol{)}\AgdaSpace{}%
\AgdaOperator{\AgdaInductiveConstructor{,}}\AgdaSpace{}%
\AgdaBound{z}\AgdaSymbol{)\})}\AgdaSpace{}%
\AgdaOperator{\AgdaInductiveConstructor{,}}\AgdaSpace{}%
\AgdaInductiveConstructor{refl}\<%
\\
\>[0]\AgdaFunction{∙eval}\AgdaSpace{}%
\AgdaInductiveConstructor{∙assocr⋆}\AgdaSpace{}%
\AgdaSymbol{=}\AgdaSpace{}%
\AgdaSymbol{(λ}\AgdaSpace{}%
\AgdaSymbol{\{((}\AgdaBound{x}\AgdaSpace{}%
\AgdaOperator{\AgdaInductiveConstructor{,}}\AgdaSpace{}%
\AgdaBound{y}\AgdaSymbol{)}\AgdaSpace{}%
\AgdaOperator{\AgdaInductiveConstructor{,}}\AgdaSpace{}%
\AgdaBound{z}\AgdaSymbol{)}\AgdaSpace{}%
\AgdaSymbol{→}\AgdaSpace{}%
\AgdaSymbol{(}\AgdaBound{x}\AgdaSpace{}%
\AgdaOperator{\AgdaInductiveConstructor{,}}\AgdaSpace{}%
\AgdaSymbol{(}\AgdaBound{y}\AgdaSpace{}%
\AgdaOperator{\AgdaInductiveConstructor{,}}\AgdaSpace{}%
\AgdaBound{z}\AgdaSymbol{))\})}\AgdaSpace{}%
\AgdaOperator{\AgdaInductiveConstructor{,}}\AgdaSpace{}%
\AgdaInductiveConstructor{refl}\<%
\\
\>[0]\AgdaFunction{∙eval}\AgdaSpace{}%
\AgdaSymbol{(}\AgdaBound{C₀}\AgdaSpace{}%
\AgdaOperator{\AgdaInductiveConstructor{∙⊗}}\AgdaSpace{}%
\AgdaBound{C₁}\AgdaSymbol{)}\AgdaSpace{}%
\AgdaKeyword{with}\AgdaSpace{}%
\AgdaFunction{∙eval}\AgdaSpace{}%
\AgdaBound{C₀}\AgdaSpace{}%
\AgdaSymbol{|}\AgdaSpace{}%
\AgdaFunction{∙eval}\AgdaSpace{}%
\AgdaBound{C₁}\<%
\\
\>[0]\AgdaSymbol{...}\AgdaSpace{}%
\AgdaSymbol{|}\AgdaSpace{}%
\AgdaSymbol{(}\AgdaBound{f}\AgdaSpace{}%
\AgdaOperator{\AgdaInductiveConstructor{,}}\AgdaSpace{}%
\AgdaBound{p}\AgdaSymbol{)}\AgdaSpace{}%
\AgdaSymbol{|}\AgdaSpace{}%
\AgdaSymbol{(}\AgdaBound{g}\AgdaSpace{}%
\AgdaOperator{\AgdaInductiveConstructor{,}}\AgdaSpace{}%
\AgdaBound{q}\AgdaSymbol{)}\AgdaSpace{}%
\AgdaSymbol{=}\AgdaSpace{}%
\AgdaSymbol{(λ}\AgdaSpace{}%
\AgdaSymbol{\{(}\AgdaBound{t₁}\AgdaSpace{}%
\AgdaOperator{\AgdaInductiveConstructor{,}}\AgdaSpace{}%
\AgdaBound{t₂}\AgdaSymbol{)}\AgdaSpace{}%
\AgdaSymbol{→}\AgdaSpace{}%
\AgdaBound{f}\AgdaSpace{}%
\AgdaBound{t₁}\AgdaSpace{}%
\AgdaOperator{\AgdaInductiveConstructor{,}}\AgdaSpace{}%
\AgdaBound{g}\AgdaSpace{}%
\AgdaBound{t₂}\AgdaSymbol{\})}\AgdaSpace{}%
\AgdaOperator{\AgdaInductiveConstructor{,}}\AgdaSpace{}%
\AgdaFunction{cong₂}\AgdaSpace{}%
\AgdaOperator{\AgdaInductiveConstructor{\AgdaUnderscore{},\AgdaUnderscore{}}}\AgdaSpace{}%
\AgdaBound{p}\AgdaSpace{}%
\AgdaBound{q}\<%
\\
\>[0]\AgdaFunction{∙eval}\AgdaSpace{}%
\AgdaInductiveConstructor{∙times\#}\AgdaSpace{}%
\AgdaSymbol{=}\AgdaSpace{}%
\AgdaSymbol{(λ}\AgdaSpace{}%
\AgdaBound{x}\AgdaSpace{}%
\AgdaSymbol{→}\AgdaSpace{}%
\AgdaBound{x}\AgdaSymbol{)}\AgdaSpace{}%
\AgdaOperator{\AgdaInductiveConstructor{,}}\AgdaSpace{}%
\AgdaInductiveConstructor{refl}\<%
\\
\>[0]\AgdaFunction{∙eval}\AgdaSpace{}%
\AgdaInductiveConstructor{∙\#times}\AgdaSpace{}%
\AgdaSymbol{=}\AgdaSpace{}%
\AgdaSymbol{(λ}\AgdaSpace{}%
\AgdaBound{x}\AgdaSpace{}%
\AgdaSymbol{→}\AgdaSpace{}%
\AgdaBound{x}\AgdaSymbol{)}\AgdaSpace{}%
\AgdaOperator{\AgdaInductiveConstructor{,}}\AgdaSpace{}%
\AgdaInductiveConstructor{refl}\<%
\\
\>[0]\AgdaFunction{∙eval}\AgdaSpace{}%
\AgdaSymbol{(}\AgdaInductiveConstructor{return}\AgdaSpace{}%
\AgdaSymbol{\{}\AgdaBound{T}\AgdaSymbol{\})}\AgdaSpace{}%
\AgdaSymbol{=}\AgdaSpace{}%
\AgdaSymbol{(λ}\AgdaSpace{}%
\AgdaBound{\AgdaUnderscore{}}\AgdaSpace{}%
\AgdaSymbol{→}\AgdaSpace{}%
\AgdaField{proj₂}\AgdaSpace{}%
\AgdaOperator{\AgdaFunction{∙⟦}}\AgdaSpace{}%
\AgdaBound{T}\AgdaSpace{}%
\AgdaOperator{\AgdaFunction{⟧}}\AgdaSpace{}%
\AgdaOperator{\AgdaInductiveConstructor{,}}\AgdaSpace{}%
\AgdaInductiveConstructor{refl}\AgdaSymbol{)}\AgdaSpace{}%
\AgdaOperator{\AgdaInductiveConstructor{,}}\AgdaSpace{}%
\AgdaInductiveConstructor{refl}\<%
\\
\>[0]\AgdaFunction{∙eval}\AgdaSpace{}%
\AgdaInductiveConstructor{extract}\AgdaSpace{}%
\AgdaSymbol{=}\AgdaSpace{}%
\AgdaSymbol{(λ}\AgdaSpace{}%
\AgdaSymbol{\{(}\AgdaBound{w}\AgdaSpace{}%
\AgdaOperator{\AgdaInductiveConstructor{,}}\AgdaSpace{}%
\AgdaInductiveConstructor{refl}\AgdaSymbol{)}\AgdaSpace{}%
\AgdaSymbol{→}\AgdaSpace{}%
\AgdaBound{w}\AgdaSymbol{\})}\AgdaSpace{}%
\AgdaOperator{\AgdaInductiveConstructor{,}}\AgdaSpace{}%
\AgdaInductiveConstructor{refl}\<%
\\
\>[0]\AgdaFunction{∙eval}\AgdaSpace{}%
\AgdaSymbol{(}\AgdaInductiveConstructor{η}\AgdaSpace{}%
\AgdaBound{T}\AgdaSymbol{)}\AgdaSpace{}%
\AgdaSymbol{=}\AgdaSpace{}%
\AgdaSymbol{(λ}\AgdaSpace{}%
\AgdaBound{tt}\AgdaSpace{}%
\AgdaSymbol{→}\AgdaSpace{}%
\AgdaSymbol{(}\AgdaField{proj₂}\AgdaSpace{}%
\AgdaOperator{\AgdaFunction{∙⟦}}\AgdaSpace{}%
\AgdaBound{T}\AgdaSpace{}%
\AgdaOperator{\AgdaFunction{⟧}}\AgdaSpace{}%
\AgdaOperator{\AgdaInductiveConstructor{,}}\AgdaSpace{}%
\AgdaInductiveConstructor{refl}\AgdaSymbol{)}\AgdaSpace{}%
\AgdaOperator{\AgdaInductiveConstructor{,}}\AgdaSpace{}%
\AgdaSymbol{λ}\AgdaSpace{}%
\AgdaBound{\AgdaUnderscore{}}\AgdaSpace{}%
\AgdaSymbol{→}\AgdaSpace{}%
\AgdaBound{tt}\AgdaSymbol{)}\AgdaSpace{}%
\AgdaOperator{\AgdaInductiveConstructor{,}}\AgdaSpace{}%
\AgdaInductiveConstructor{refl}\<%
\\
\>[0]\AgdaFunction{∙eval}\AgdaSpace{}%
\AgdaSymbol{(}\AgdaInductiveConstructor{ε}\AgdaSpace{}%
\AgdaBound{T}\AgdaSymbol{)}\AgdaSpace{}%
\AgdaSymbol{=}\AgdaSpace{}%
\AgdaSymbol{(λ}\AgdaSpace{}%
\AgdaSymbol{\{}\AgdaSpace{}%
\AgdaSymbol{((\AgdaUnderscore{}}\AgdaSpace{}%
\AgdaOperator{\AgdaInductiveConstructor{,}}\AgdaSpace{}%
\AgdaInductiveConstructor{refl}\AgdaSymbol{)}\AgdaSpace{}%
\AgdaOperator{\AgdaInductiveConstructor{,}}\AgdaSpace{}%
\AgdaBound{f}\AgdaSymbol{)}\AgdaSpace{}%
\AgdaSymbol{→}\AgdaSpace{}%
\AgdaBound{f}\AgdaSpace{}%
\AgdaSymbol{(}\AgdaField{proj₂}\AgdaSpace{}%
\AgdaOperator{\AgdaFunction{∙⟦}}\AgdaSpace{}%
\AgdaBound{T}\AgdaSpace{}%
\AgdaOperator{\AgdaFunction{⟧}}\AgdaSpace{}%
\AgdaOperator{\AgdaInductiveConstructor{,}}\AgdaSpace{}%
\AgdaInductiveConstructor{refl}\AgdaSymbol{)\})}\AgdaSpace{}%
\AgdaOperator{\AgdaInductiveConstructor{,}}\AgdaSpace{}%
\AgdaInductiveConstructor{refl}\<%
\end{code}}
\newcommand{\PIPFrevrev}{%
\begin{code}%
\>[0]\AgdaFunction{revrev}\AgdaSpace{}%
\AgdaSymbol{:}\AgdaSpace{}%
\AgdaSymbol{\{}\AgdaBound{A}\AgdaSpace{}%
\AgdaSymbol{:}\AgdaSpace{}%
\AgdaDatatype{∙𝕌}\AgdaSymbol{\}}\AgdaSpace{}%
\AgdaSymbol{→}\AgdaSpace{}%
\AgdaInductiveConstructor{∙𝟙/}\AgdaSpace{}%
\AgdaSymbol{(}\AgdaInductiveConstructor{∙𝟙/}\AgdaSpace{}%
\AgdaBound{A}\AgdaSymbol{)}\AgdaSpace{}%
\AgdaOperator{\AgdaDatatype{⧟}}\AgdaSpace{}%
\AgdaOperator{\AgdaInductiveConstructor{❰}}\AgdaSpace{}%
\AgdaBound{A}\AgdaSpace{}%
\AgdaOperator{\AgdaInductiveConstructor{❱}}\<%
\\
\>[0]\AgdaFunction{revrev}\AgdaSpace{}%
\AgdaSymbol{\{}\AgdaBound{A}\AgdaSymbol{\}}\AgdaSpace{}%
\AgdaSymbol{=}%
\>[1029I]\AgdaInductiveConstructor{∙uniti⋆l}\AgdaSpace{}%
\AgdaOperator{\AgdaInductiveConstructor{∙⊚}}\<%
\\
\>[.][@{}l@{}]\<[1029I]%
\>[13]\AgdaSymbol{(}\AgdaInductiveConstructor{η}\AgdaSpace{}%
\AgdaBound{A}\AgdaSpace{}%
\AgdaOperator{\AgdaInductiveConstructor{∙⊗}}\AgdaSpace{}%
\AgdaInductiveConstructor{∙id⟷}\AgdaSymbol{)}\AgdaSpace{}%
\AgdaOperator{\AgdaInductiveConstructor{∙⊚}}\<%
\\
\>[13]\AgdaSymbol{((}\AgdaInductiveConstructor{∙id⟷}\AgdaSpace{}%
\AgdaOperator{\AgdaInductiveConstructor{∙⊗}}\AgdaSpace{}%
\AgdaInductiveConstructor{return}\AgdaSymbol{)}\AgdaSpace{}%
\AgdaOperator{\AgdaInductiveConstructor{∙⊗}}\AgdaSpace{}%
\AgdaInductiveConstructor{∙id⟷}\AgdaSymbol{)}\AgdaSpace{}%
\AgdaOperator{\AgdaInductiveConstructor{∙⊚}}\<%
\\
\>[13]\AgdaInductiveConstructor{∙assocr⋆}\AgdaSpace{}%
\AgdaOperator{\AgdaInductiveConstructor{∙⊚}}\<%
\\
\>[13]\AgdaInductiveConstructor{∙id⟷}\AgdaSpace{}%
\AgdaOperator{\AgdaInductiveConstructor{∙⊗}}\AgdaSpace{}%
\AgdaInductiveConstructor{ε}\AgdaSpace{}%
\AgdaSymbol{(}\AgdaInductiveConstructor{∙𝟙/}\AgdaSpace{}%
\AgdaBound{A}\AgdaSymbol{)}\AgdaSpace{}%
\AgdaOperator{\AgdaInductiveConstructor{∙⊚}}\<%
\\
\>[13]\AgdaInductiveConstructor{∙unite⋆r}\<%
\end{code}}
\newcommand{\PIPFExample}{%
\begin{code}%
\>[0]\AgdaFunction{Ex}\AgdaSpace{}%
\AgdaSymbol{:}\AgdaSpace{}%
\AgdaRecord{Σ}\AgdaSpace{}%
\AgdaSymbol{((}\AgdaBound{x}\AgdaSpace{}%
\AgdaSymbol{:}\AgdaSpace{}%
\AgdaOperator{\AgdaFunction{⟦}}\AgdaSpace{}%
\AgdaFunction{𝔹}\AgdaSpace{}%
\AgdaOperator{\AgdaFunction{⟧}}\AgdaSymbol{)}\AgdaSpace{}%
\AgdaSymbol{→}\AgdaSpace{}%
\AgdaOperator{\AgdaFunction{⟦}}\AgdaSpace{}%
\AgdaFunction{𝔹}\AgdaSpace{}%
\AgdaOperator{\AgdaFunction{⟧}}\AgdaSymbol{)}\AgdaSpace{}%
\AgdaSymbol{(λ}\AgdaSpace{}%
\AgdaBound{f}\AgdaSpace{}%
\AgdaSymbol{→}\AgdaSpace{}%
\AgdaBound{f}\AgdaSpace{}%
\AgdaInductiveConstructor{𝔽}\AgdaSpace{}%
\AgdaOperator{\AgdaDatatype{≡}}\AgdaSpace{}%
\AgdaInductiveConstructor{𝕋}\AgdaSymbol{)}\<%
\\
\>[0]\AgdaFunction{Ex}\AgdaSpace{}%
\AgdaSymbol{=}\AgdaSpace{}%
\AgdaFunction{∙eval}\AgdaSpace{}%
\AgdaSymbol{(}\AgdaInductiveConstructor{∙c}\AgdaSpace{}%
\AgdaFunction{NOT}\AgdaSymbol{)}\<%
\end{code}}

\begin{code}[hide]%
\>[0]\AgdaSymbol{\{-\#}\AgdaSpace{}%
\AgdaKeyword{OPTIONS}\AgdaSpace{}%
\AgdaPragma{--safe}\AgdaSpace{}%
\AgdaSymbol{\#-\}}\<%
\\
\>[0]\AgdaKeyword{module}\AgdaSpace{}%
\AgdaModule{Extraction}\AgdaSpace{}%
\AgdaKeyword{where}\<%
\\
\>[0]\AgdaKeyword{open}\AgdaSpace{}%
\AgdaKeyword{import}\AgdaSpace{}%
\AgdaModule{Data.Empty}\<%
\\
\>[0]\AgdaKeyword{open}\AgdaSpace{}%
\AgdaKeyword{import}\AgdaSpace{}%
\AgdaModule{Data.Unit}\<%
\\
\>[0]\AgdaKeyword{open}\AgdaSpace{}%
\AgdaKeyword{import}\AgdaSpace{}%
\AgdaModule{Data.Product}\<%
\\
\>[0]\AgdaKeyword{open}\AgdaSpace{}%
\AgdaKeyword{import}\AgdaSpace{}%
\AgdaModule{Data.Sum}\<%
\\
\>[0]\AgdaKeyword{open}\AgdaSpace{}%
\AgdaKeyword{import}\AgdaSpace{}%
\AgdaModule{Data.Maybe}\<%
\\
\>[0]\AgdaKeyword{open}\AgdaSpace{}%
\AgdaKeyword{import}\AgdaSpace{}%
\AgdaModule{Relation.Binary.PropositionalEquality}\<%
\\
\>[0]\AgdaKeyword{open}\AgdaSpace{}%
\AgdaKeyword{import}\AgdaSpace{}%
\AgdaModule{Relation.Binary.HeterogeneousEquality}\AgdaSpace{}%
\AgdaKeyword{hiding}\AgdaSpace{}%
\AgdaSymbol{(}\AgdaFunction{subst}\AgdaSymbol{)}\<%
\\
\>[0]\AgdaKeyword{open}\AgdaSpace{}%
\AgdaKeyword{import}\AgdaSpace{}%
\AgdaModule{Relation.Nullary}\<%
\\
\>[0]\AgdaKeyword{open}\AgdaSpace{}%
\AgdaKeyword{import}\AgdaSpace{}%
\AgdaModule{Pi}\<%
\\
\>[0]\AgdaKeyword{open}\AgdaSpace{}%
\AgdaKeyword{import}\AgdaSpace{}%
\AgdaModule{Pi.Pointed}\<%
\\
\>[0]\AgdaKeyword{open}\AgdaSpace{}%
\AgdaKeyword{import}\AgdaSpace{}%
\AgdaModule{Pi.D}\AgdaSpace{}%
\AgdaKeyword{renaming}\AgdaSpace{}%
\AgdaSymbol{(}\AgdaDatatype{𝕌}\AgdaSpace{}%
\AgdaSymbol{to}\AgdaSpace{}%
\AgdaDatatype{𝕌D}\AgdaSymbol{;}\AgdaSpace{}%
\AgdaOperator{\AgdaFunction{⟦\AgdaUnderscore{}⟧}}\AgdaSpace{}%
\AgdaSymbol{to}\AgdaSpace{}%
\AgdaOperator{\AgdaFunction{⟦\AgdaUnderscore{}⟧D}}\AgdaSymbol{;}\AgdaSpace{}%
\AgdaOperator{\AgdaDatatype{\AgdaUnderscore{}⟷\AgdaUnderscore{}}}\AgdaSpace{}%
\AgdaSymbol{to}\AgdaSpace{}%
\AgdaOperator{\AgdaDatatype{\AgdaUnderscore{}⟷D\AgdaUnderscore{}}}\AgdaSymbol{)}\<%
\end{code}

\newcommand{\EXTsig}{%
\begin{code}%
\>[0]\AgdaFunction{Ext𝕌}%
\>[6]\AgdaSymbol{:}\AgdaSpace{}%
\AgdaDatatype{∙𝕌}\AgdaSpace{}%
\AgdaSymbol{→}\AgdaSpace{}%
\AgdaFunction{Σ[}\AgdaSpace{}%
\AgdaBound{t}\AgdaSpace{}%
\AgdaFunction{∈}\AgdaSpace{}%
\AgdaDatatype{𝕌D}\AgdaSpace{}%
\AgdaFunction{]}\AgdaSpace{}%
\AgdaOperator{\AgdaFunction{⟦}}\AgdaSpace{}%
\AgdaBound{t}\AgdaSpace{}%
\AgdaOperator{\AgdaFunction{⟧D}}\<%
\\
\\[\AgdaEmptyExtraSkip]%
\>[0]\AgdaFunction{Ext⧟}\AgdaSpace{}%
\AgdaSymbol{:}\AgdaSpace{}%
\AgdaSymbol{∀}\AgdaSpace{}%
\AgdaSymbol{\{}\AgdaBound{t₁}\AgdaSpace{}%
\AgdaBound{t₂}\AgdaSymbol{\}}\AgdaSpace{}%
\AgdaSymbol{→}\AgdaSpace{}%
\AgdaBound{t₁}\AgdaSpace{}%
\AgdaOperator{\AgdaDatatype{⧟}}\AgdaSpace{}%
\AgdaBound{t₂}\AgdaSpace{}%
\AgdaSymbol{→}%
\>[753I]\AgdaKeyword{let}%
\>[754I]\AgdaSymbol{(}\AgdaBound{s₁}\AgdaSpace{}%
\AgdaOperator{\AgdaInductiveConstructor{,}}\AgdaSpace{}%
\AgdaBound{w₁}\AgdaSymbol{)}\AgdaSpace{}%
\AgdaSymbol{=}\AgdaSpace{}%
\AgdaFunction{Ext𝕌}\AgdaSpace{}%
\AgdaBound{t₁}\<%
\\
\>[.][@{}l@{}]\<[754I]%
\>[33]\AgdaSymbol{(}\AgdaBound{s₂}\AgdaSpace{}%
\AgdaOperator{\AgdaInductiveConstructor{,}}\AgdaSpace{}%
\AgdaBound{w₂}\AgdaSymbol{)}\AgdaSpace{}%
\AgdaSymbol{=}\AgdaSpace{}%
\AgdaFunction{Ext𝕌}\AgdaSpace{}%
\AgdaBound{t₂}\<%
\\
\>[.][@{}l@{}]\<[753I]%
\>[29]\AgdaKeyword{in}\AgdaSpace{}%
\AgdaBound{s₁}\AgdaSpace{}%
\AgdaOperator{\AgdaDatatype{⟷D}}\AgdaSpace{}%
\AgdaBound{s₂}\<%
\\
\\[\AgdaEmptyExtraSkip]%
\>[0]\AgdaFunction{Ext≡}%
\>[6]\AgdaSymbol{:}\AgdaSpace{}%
\AgdaSymbol{∀}\AgdaSpace{}%
\AgdaSymbol{\{}\AgdaBound{t₁}\AgdaSpace{}%
\AgdaBound{t₂}\AgdaSymbol{\}}\AgdaSpace{}%
\AgdaSymbol{→}\AgdaSpace{}%
\AgdaSymbol{(}\AgdaBound{c}\AgdaSpace{}%
\AgdaSymbol{:}\AgdaSpace{}%
\AgdaBound{t₁}\AgdaSpace{}%
\AgdaOperator{\AgdaDatatype{⧟}}\AgdaSpace{}%
\AgdaBound{t₂}\AgdaSymbol{)}\<%
\\
\>[6]\AgdaSymbol{→}\AgdaSpace{}%
\AgdaFunction{interp}\AgdaSpace{}%
\AgdaSymbol{(}\AgdaFunction{Ext⧟}\AgdaSpace{}%
\AgdaBound{c}\AgdaSymbol{)}\AgdaSpace{}%
\AgdaSymbol{(}\AgdaField{proj₂}\AgdaSpace{}%
\AgdaSymbol{(}\AgdaFunction{Ext𝕌}\AgdaSpace{}%
\AgdaBound{t₁}\AgdaSymbol{))}\AgdaSpace{}%
\AgdaOperator{\AgdaDatatype{≡}}\AgdaSpace{}%
\AgdaInductiveConstructor{just}\AgdaSpace{}%
\AgdaSymbol{(}\AgdaField{proj₂}\AgdaSpace{}%
\AgdaSymbol{(}\AgdaFunction{Ext𝕌}\AgdaSpace{}%
\AgdaBound{t₂}\AgdaSymbol{))}\<%
\end{code}}

\newcommand{\EXTu}{%
\begin{code}%
\>[0]\AgdaFunction{Ext𝕌}\AgdaSpace{}%
\AgdaSymbol{(}\AgdaInductiveConstructor{∙𝟙/}\AgdaSpace{}%
\AgdaBound{T}\AgdaSymbol{)}\AgdaSpace{}%
\AgdaSymbol{=}%
\>[831I]\AgdaKeyword{let}\AgdaSpace{}%
\AgdaSymbol{(}\AgdaBound{t}\AgdaSpace{}%
\AgdaOperator{\AgdaInductiveConstructor{,}}\AgdaSpace{}%
\AgdaBound{v}\AgdaSymbol{)}\AgdaSpace{}%
\AgdaSymbol{=}\AgdaSpace{}%
\AgdaFunction{Ext𝕌}\AgdaSpace{}%
\AgdaBound{T}\<%
\\
\>[.][@{}l@{}]\<[831I]%
\>[15]\AgdaKeyword{in}\AgdaSpace{}%
\AgdaOperator{\AgdaInductiveConstructor{𝟙/}}\AgdaSpace{}%
\AgdaBound{v}\AgdaSpace{}%
\AgdaOperator{\AgdaInductiveConstructor{,}}\AgdaSpace{}%
\AgdaInductiveConstructor{↻}\<%
\end{code}}

\newcommand{\EXTcomb}{%
\begin{code}%
\>[0]\AgdaFunction{Ext⧟}\AgdaSpace{}%
\AgdaSymbol{(}\AgdaInductiveConstructor{η}\AgdaSpace{}%
\AgdaBound{T}\AgdaSymbol{)}\AgdaSpace{}%
\AgdaSymbol{=}\AgdaSpace{}%
\AgdaInductiveConstructor{η}\AgdaSpace{}%
\AgdaSymbol{(}\AgdaField{proj₂}\AgdaSpace{}%
\AgdaSymbol{(}\AgdaFunction{Ext𝕌}\AgdaSpace{}%
\AgdaBound{T}\AgdaSymbol{))}\<%
\\
\>[0]\AgdaFunction{Ext⧟}\AgdaSpace{}%
\AgdaSymbol{(}\AgdaInductiveConstructor{ε}\AgdaSpace{}%
\AgdaBound{T}\AgdaSymbol{)}\AgdaSpace{}%
\AgdaSymbol{=}\AgdaSpace{}%
\AgdaInductiveConstructor{ε}\AgdaSpace{}%
\AgdaSymbol{(}\AgdaField{proj₂}\AgdaSpace{}%
\AgdaSymbol{(}\AgdaFunction{Ext𝕌}\AgdaSpace{}%
\AgdaBound{T}\AgdaSymbol{))}\<%
\end{code}}

\newcommand{\EXTeq}{%
\begin{code}%
\>[0]\AgdaFunction{Ext≡}\AgdaSpace{}%
\AgdaSymbol{(}\AgdaInductiveConstructor{ε}\AgdaSpace{}%
\AgdaBound{T}\AgdaSymbol{)}\AgdaSpace{}%
\AgdaKeyword{with}\AgdaSpace{}%
\AgdaSymbol{(}\AgdaField{proj₂}\AgdaSpace{}%
\AgdaSymbol{(}\AgdaFunction{Ext𝕌}\AgdaSpace{}%
\AgdaBound{T}\AgdaSymbol{)}\AgdaSpace{}%
\AgdaOperator{\AgdaFunction{≟ᵤ}}\AgdaSpace{}%
\AgdaField{proj₂}\AgdaSpace{}%
\AgdaSymbol{(}\AgdaFunction{Ext𝕌}\AgdaSpace{}%
\AgdaBound{T}\AgdaSymbol{))}\<%
\\
\>[0]\AgdaSymbol{...}\AgdaSpace{}%
\AgdaSymbol{|}\AgdaSpace{}%
\AgdaInductiveConstructor{yes}\AgdaSpace{}%
\AgdaBound{p}\AgdaSpace{}%
\AgdaSymbol{=}\AgdaSpace{}%
\AgdaInductiveConstructor{refl}\<%
\\
\>[0]\AgdaSymbol{...}\AgdaSpace{}%
\AgdaSymbol{|}\AgdaSpace{}%
\AgdaInductiveConstructor{no}\AgdaSpace{}%
\AgdaBound{¬p}\AgdaSpace{}%
\AgdaSymbol{=}\AgdaSpace{}%
\AgdaFunction{⊥-elim}\AgdaSpace{}%
\AgdaSymbol{(}\AgdaBound{¬p}\AgdaSpace{}%
\AgdaInductiveConstructor{refl}\AgdaSymbol{)}\<%
\end{code}}

\begin{code}[hide]%
\>[0]\AgdaSymbol{\{-\#}\AgdaSpace{}%
\AgdaKeyword{OPTIONS}\AgdaSpace{}%
\AgdaPragma{--safe}\AgdaSpace{}%
\AgdaSymbol{\#-\}}\<%
\\
\>[0]\AgdaKeyword{module}\AgdaSpace{}%
\AgdaModule{Examples}\AgdaSpace{}%
\AgdaKeyword{where}\<%
\\
\>[0]\AgdaKeyword{open}\AgdaSpace{}%
\AgdaKeyword{import}\AgdaSpace{}%
\AgdaModule{Data.Unit}\<%
\\
\>[0]\AgdaKeyword{open}\AgdaSpace{}%
\AgdaKeyword{import}\AgdaSpace{}%
\AgdaModule{Data.Product}\<%
\\
\>[0]\AgdaKeyword{open}\AgdaSpace{}%
\AgdaKeyword{import}\AgdaSpace{}%
\AgdaModule{Data.Fin}\<%
\\
\>[0]\AgdaKeyword{open}\AgdaSpace{}%
\AgdaKeyword{import}\AgdaSpace{}%
\AgdaModule{Data.Vec}\<%
\\
\>[0]\AgdaKeyword{open}\AgdaSpace{}%
\AgdaKeyword{import}\AgdaSpace{}%
\AgdaModule{Pi.Pointed}\<%
\\
\>[0]\AgdaKeyword{import}\AgdaSpace{}%
\AgdaModule{Pi}\<%
\\
\>[0]\AgdaKeyword{open}\AgdaSpace{}%
\AgdaKeyword{import}\AgdaSpace{}%
\AgdaModule{Pi.D}\<%
\\
\>[0]\AgdaKeyword{open}\AgdaSpace{}%
\AgdaKeyword{import}\AgdaSpace{}%
\AgdaModule{Pi.DMemSem}\<%
\\
\>[0]\AgdaKeyword{open}\AgdaSpace{}%
\AgdaKeyword{import}\AgdaSpace{}%
\AgdaModule{Extraction}\<%
\\
\\[\AgdaEmptyExtraSkip]%
\>[0]\AgdaFunction{revrev-ext}\AgdaSpace{}%
\AgdaSymbol{:}\AgdaSpace{}%
\AgdaOperator{\AgdaInductiveConstructor{𝟙/\AgdaUnderscore{}}}\AgdaSpace{}%
\AgdaSymbol{\{}\AgdaOperator{\AgdaInductiveConstructor{𝟙/}}\AgdaSpace{}%
\AgdaInductiveConstructor{𝕋}\AgdaSymbol{\}}\AgdaSpace{}%
\AgdaInductiveConstructor{↻}\AgdaSpace{}%
\AgdaOperator{\AgdaDatatype{⟷}}\AgdaSpace{}%
\AgdaFunction{𝔹}\<%
\\
\>[0]\AgdaFunction{revrev-ext}\AgdaSpace{}%
\AgdaSymbol{=}\AgdaSpace{}%
\AgdaFunction{Ext⧟}\AgdaSpace{}%
\AgdaSymbol{(}\AgdaFunction{Pi.Pointed.revrev}\AgdaSpace{}%
\AgdaSymbol{\{}\AgdaFunction{Pi.𝔹}\AgdaSpace{}%
\AgdaOperator{\AgdaInductiveConstructor{\#}}\AgdaSpace{}%
\AgdaInductiveConstructor{𝕋}\AgdaSymbol{\})}\<%
\\
\\[\AgdaEmptyExtraSkip]%
\>[0]\AgdaFunction{trace₃}\AgdaSpace{}%
\AgdaSymbol{:}\AgdaSpace{}%
\AgdaDatatype{Vec}\AgdaSpace{}%
\AgdaDatatype{State'}\AgdaSpace{}%
\AgdaNumber{15}\<%
\\
\>[0]\AgdaFunction{trace₃}\AgdaSpace{}%
\AgdaSymbol{=}\AgdaSpace{}%
\AgdaFunction{run}\AgdaSpace{}%
\AgdaSymbol{\AgdaUnderscore{}}\AgdaSpace{}%
\AgdaOperator{\AgdaInductiveConstructor{⟪}}\AgdaSpace{}%
\AgdaFunction{revrev-ext}\AgdaSpace{}%
\AgdaOperator{\AgdaInductiveConstructor{∥}}\AgdaSpace{}%
\AgdaFunction{Enum}\AgdaSpace{}%
\AgdaSymbol{(}\AgdaOperator{\AgdaInductiveConstructor{𝟙/\AgdaUnderscore{}}}\AgdaSpace{}%
\AgdaSymbol{\{}\AgdaOperator{\AgdaInductiveConstructor{𝟙/\AgdaUnderscore{}}}\AgdaSpace{}%
\AgdaSymbol{\{}\AgdaFunction{𝔹}\AgdaSymbol{\}}\AgdaSpace{}%
\AgdaInductiveConstructor{𝕋}\AgdaSymbol{\}}\AgdaSpace{}%
\AgdaInductiveConstructor{↻}\AgdaSymbol{)}\AgdaSpace{}%
\AgdaOperator{\AgdaInductiveConstructor{[}}\AgdaSpace{}%
\AgdaInductiveConstructor{zero}\AgdaSpace{}%
\AgdaOperator{\AgdaInductiveConstructor{]⟫}}\<%
\end{code}

\begin{code}[hide]%
\>[0]\AgdaSymbol{\{-\#}\AgdaSpace{}%
\AgdaKeyword{OPTIONS}\AgdaSpace{}%
\AgdaPragma{--rewriting}\AgdaSpace{}%
\AgdaSymbol{\#-\}}\<%
\\
\>[0]\AgdaKeyword{module}\AgdaSpace{}%
\AgdaModule{Toffoli}\AgdaSpace{}%
\AgdaKeyword{where}\<%
\\
\>[0]\AgdaKeyword{open}\AgdaSpace{}%
\AgdaKeyword{import}\AgdaSpace{}%
\AgdaModule{Relation.Binary.PropositionalEquality}\<%
\\
\>[0]\AgdaKeyword{open}\AgdaSpace{}%
\AgdaKeyword{import}\AgdaSpace{}%
\AgdaModule{Data.Product}\<%
\\
\>[0]\AgdaKeyword{open}\AgdaSpace{}%
\AgdaKeyword{import}\AgdaSpace{}%
\AgdaModule{Data.Maybe}\<%
\\
\>[0]\AgdaKeyword{open}\AgdaSpace{}%
\AgdaKeyword{import}\AgdaSpace{}%
\AgdaModule{Pi}\AgdaSpace{}%
\AgdaKeyword{hiding}\AgdaSpace{}%
\AgdaSymbol{(}\AgdaOperator{\AgdaDatatype{\AgdaUnderscore{}⟷\AgdaUnderscore{}}}\AgdaSymbol{)}\<%
\\
\>[0]\AgdaKeyword{open}\AgdaSpace{}%
\AgdaKeyword{import}\AgdaSpace{}%
\AgdaModule{Pi.Pointed}\<%
\\
\>[0]\AgdaKeyword{open}\AgdaSpace{}%
\AgdaKeyword{import}\AgdaSpace{}%
\AgdaModule{Pi.D}\AgdaSpace{}%
\AgdaKeyword{hiding}\AgdaSpace{}%
\AgdaSymbol{(}\AgdaOperator{\AgdaFunction{⟦\AgdaUnderscore{}⟧}}\AgdaSpace{}%
\AgdaSymbol{;}\AgdaSpace{}%
\AgdaFunction{𝔹}\AgdaSpace{}%
\AgdaSymbol{;}\AgdaSpace{}%
\AgdaFunction{𝔹³}\AgdaSymbol{)}\<%
\\
\>[0]\AgdaKeyword{open}\AgdaSpace{}%
\AgdaKeyword{import}\AgdaSpace{}%
\AgdaModule{Agda.Builtin.Equality.Rewrite}\<%
\\
\>[0]\AgdaKeyword{open}\AgdaSpace{}%
\AgdaKeyword{import}\AgdaSpace{}%
\AgdaModule{Extraction}\<%
\\
\\[\AgdaEmptyExtraSkip]%
\>[0]\AgdaKeyword{infixr}\AgdaSpace{}%
\AgdaNumber{20}\AgdaSpace{}%
\AgdaOperator{\AgdaFunction{\AgdaUnderscore{}\&\AgdaUnderscore{}}}\<%
\\
\>[0]\AgdaKeyword{infixr}\AgdaSpace{}%
\AgdaNumber{20}\AgdaSpace{}%
\AgdaOperator{\AgdaFunction{\AgdaUnderscore{}\textasciicircum{}\AgdaUnderscore{}}}\<%
\\
\\[\AgdaEmptyExtraSkip]%
\>[0]\AgdaOperator{\AgdaFunction{\AgdaUnderscore{}\&\AgdaUnderscore{}}}\AgdaSpace{}%
\AgdaSymbol{:}\AgdaSpace{}%
\AgdaOperator{\AgdaFunction{⟦}}\AgdaSpace{}%
\AgdaFunction{𝔹}\AgdaSpace{}%
\AgdaOperator{\AgdaFunction{⟧}}\AgdaSpace{}%
\AgdaSymbol{→}\AgdaSpace{}%
\AgdaOperator{\AgdaFunction{⟦}}\AgdaSpace{}%
\AgdaFunction{𝔹}\AgdaSpace{}%
\AgdaOperator{\AgdaFunction{⟧}}\AgdaSpace{}%
\AgdaSymbol{→}\AgdaSpace{}%
\AgdaOperator{\AgdaFunction{⟦}}\AgdaSpace{}%
\AgdaFunction{𝔹}\AgdaSpace{}%
\AgdaOperator{\AgdaFunction{⟧}}\<%
\\
\>[0]\AgdaInductiveConstructor{𝔽}\AgdaSpace{}%
\AgdaOperator{\AgdaFunction{\&}}\AgdaSpace{}%
\AgdaInductiveConstructor{𝔽}\AgdaSpace{}%
\AgdaSymbol{=}\AgdaSpace{}%
\AgdaInductiveConstructor{𝔽}\<%
\\
\>[0]\AgdaInductiveConstructor{𝔽}\AgdaSpace{}%
\AgdaOperator{\AgdaFunction{\&}}\AgdaSpace{}%
\AgdaInductiveConstructor{𝕋}\AgdaSpace{}%
\AgdaSymbol{=}\AgdaSpace{}%
\AgdaInductiveConstructor{𝔽}\<%
\\
\>[0]\AgdaInductiveConstructor{𝕋}\AgdaSpace{}%
\AgdaOperator{\AgdaFunction{\&}}\AgdaSpace{}%
\AgdaInductiveConstructor{𝔽}\AgdaSpace{}%
\AgdaSymbol{=}\AgdaSpace{}%
\AgdaInductiveConstructor{𝔽}\<%
\\
\>[0]\AgdaInductiveConstructor{𝕋}\AgdaSpace{}%
\AgdaOperator{\AgdaFunction{\&}}\AgdaSpace{}%
\AgdaInductiveConstructor{𝕋}\AgdaSpace{}%
\AgdaSymbol{=}\AgdaSpace{}%
\AgdaInductiveConstructor{𝕋}\<%
\\
\\[\AgdaEmptyExtraSkip]%
\>[0]\AgdaOperator{\AgdaFunction{\AgdaUnderscore{}\textasciicircum{}\AgdaUnderscore{}}}\AgdaSpace{}%
\AgdaSymbol{:}\AgdaSpace{}%
\AgdaOperator{\AgdaFunction{⟦}}\AgdaSpace{}%
\AgdaFunction{𝔹}\AgdaSpace{}%
\AgdaOperator{\AgdaFunction{⟧}}\AgdaSpace{}%
\AgdaSymbol{→}\AgdaSpace{}%
\AgdaOperator{\AgdaFunction{⟦}}\AgdaSpace{}%
\AgdaFunction{𝔹}\AgdaSpace{}%
\AgdaOperator{\AgdaFunction{⟧}}\AgdaSpace{}%
\AgdaSymbol{→}\AgdaSpace{}%
\AgdaOperator{\AgdaFunction{⟦}}\AgdaSpace{}%
\AgdaFunction{𝔹}\AgdaSpace{}%
\AgdaOperator{\AgdaFunction{⟧}}\<%
\\
\>[0]\AgdaInductiveConstructor{𝔽}\AgdaSpace{}%
\AgdaOperator{\AgdaFunction{\textasciicircum{}}}\AgdaSpace{}%
\AgdaInductiveConstructor{𝔽}\AgdaSpace{}%
\AgdaSymbol{=}\AgdaSpace{}%
\AgdaInductiveConstructor{𝔽}\<%
\\
\>[0]\AgdaInductiveConstructor{𝔽}\AgdaSpace{}%
\AgdaOperator{\AgdaFunction{\textasciicircum{}}}\AgdaSpace{}%
\AgdaInductiveConstructor{𝕋}\AgdaSpace{}%
\AgdaSymbol{=}\AgdaSpace{}%
\AgdaInductiveConstructor{𝕋}\<%
\\
\>[0]\AgdaInductiveConstructor{𝕋}\AgdaSpace{}%
\AgdaOperator{\AgdaFunction{\textasciicircum{}}}\AgdaSpace{}%
\AgdaInductiveConstructor{𝔽}\AgdaSpace{}%
\AgdaSymbol{=}\AgdaSpace{}%
\AgdaInductiveConstructor{𝕋}\<%
\\
\>[0]\AgdaInductiveConstructor{𝕋}\AgdaSpace{}%
\AgdaOperator{\AgdaFunction{\textasciicircum{}}}\AgdaSpace{}%
\AgdaInductiveConstructor{𝕋}\AgdaSpace{}%
\AgdaSymbol{=}\AgdaSpace{}%
\AgdaInductiveConstructor{𝔽}\<%
\\
\\[\AgdaEmptyExtraSkip]%
\>[0]\AgdaFunction{lid\textasciicircum{}}\AgdaSpace{}%
\AgdaSymbol{:}\AgdaSpace{}%
\AgdaSymbol{∀}\AgdaSpace{}%
\AgdaBound{b}\AgdaSpace{}%
\AgdaSymbol{→}\AgdaSpace{}%
\AgdaInductiveConstructor{𝔽}\AgdaSpace{}%
\AgdaOperator{\AgdaFunction{\textasciicircum{}}}\AgdaSpace{}%
\AgdaBound{b}\AgdaSpace{}%
\AgdaOperator{\AgdaDatatype{≡}}\AgdaSpace{}%
\AgdaBound{b}\<%
\\
\>[0]\AgdaFunction{lid\textasciicircum{}}\AgdaSpace{}%
\AgdaInductiveConstructor{𝔽}\AgdaSpace{}%
\AgdaSymbol{=}\AgdaSpace{}%
\AgdaInductiveConstructor{refl}\<%
\\
\>[0]\AgdaFunction{lid\textasciicircum{}}\AgdaSpace{}%
\AgdaInductiveConstructor{𝕋}\AgdaSpace{}%
\AgdaSymbol{=}\AgdaSpace{}%
\AgdaInductiveConstructor{refl}\<%
\\
\>[0]\AgdaSymbol{\{-\#}\AgdaSpace{}%
\AgdaKeyword{REWRITE}\AgdaSpace{}%
\AgdaFunction{lid\textasciicircum{}}\AgdaSpace{}%
\AgdaSymbol{\#-\}}\<%
\\
\\[\AgdaEmptyExtraSkip]%
\>[0]\AgdaFunction{rid\textasciicircum{}}\AgdaSpace{}%
\AgdaSymbol{:}\AgdaSpace{}%
\AgdaSymbol{∀}\AgdaSpace{}%
\AgdaBound{b}\AgdaSpace{}%
\AgdaSymbol{→}\AgdaSpace{}%
\AgdaBound{b}\AgdaSpace{}%
\AgdaOperator{\AgdaFunction{\textasciicircum{}}}\AgdaSpace{}%
\AgdaInductiveConstructor{𝔽}\AgdaSpace{}%
\AgdaOperator{\AgdaDatatype{≡}}\AgdaSpace{}%
\AgdaBound{b}\<%
\\
\>[0]\AgdaFunction{rid\textasciicircum{}}\AgdaSpace{}%
\AgdaInductiveConstructor{𝔽}\AgdaSpace{}%
\AgdaSymbol{=}\AgdaSpace{}%
\AgdaInductiveConstructor{refl}\<%
\\
\>[0]\AgdaFunction{rid\textasciicircum{}}\AgdaSpace{}%
\AgdaInductiveConstructor{𝕋}\AgdaSpace{}%
\AgdaSymbol{=}\AgdaSpace{}%
\AgdaInductiveConstructor{refl}\<%
\\
\>[0]\AgdaSymbol{\{-\#}\AgdaSpace{}%
\AgdaKeyword{REWRITE}\AgdaSpace{}%
\AgdaFunction{rid\textasciicircum{}}\AgdaSpace{}%
\AgdaSymbol{\#-\}}\<%
\\
\\[\AgdaEmptyExtraSkip]%
\>[0]\AgdaFunction{inv\textasciicircum{}}\AgdaSpace{}%
\AgdaSymbol{:}\AgdaSpace{}%
\AgdaSymbol{∀}\AgdaSpace{}%
\AgdaBound{b}\AgdaSpace{}%
\AgdaSymbol{→}\AgdaSpace{}%
\AgdaBound{b}\AgdaSpace{}%
\AgdaOperator{\AgdaFunction{\textasciicircum{}}}\AgdaSpace{}%
\AgdaBound{b}\AgdaSpace{}%
\AgdaOperator{\AgdaDatatype{≡}}\AgdaSpace{}%
\AgdaInductiveConstructor{𝔽}\<%
\\
\>[0]\AgdaFunction{inv\textasciicircum{}}\AgdaSpace{}%
\AgdaInductiveConstructor{𝔽}\AgdaSpace{}%
\AgdaSymbol{=}\AgdaSpace{}%
\AgdaInductiveConstructor{refl}\<%
\\
\>[0]\AgdaFunction{inv\textasciicircum{}}\AgdaSpace{}%
\AgdaInductiveConstructor{𝕋}\AgdaSpace{}%
\AgdaSymbol{=}\AgdaSpace{}%
\AgdaInductiveConstructor{refl}\<%
\\
\>[0]\AgdaSymbol{\{-\#}\AgdaSpace{}%
\AgdaKeyword{REWRITE}\AgdaSpace{}%
\AgdaFunction{inv\textasciicircum{}}\AgdaSpace{}%
\AgdaSymbol{\#-\}}\<%
\\
\>[0]\<%
\end{code}
\newcommand{\PToffoli}{%
\begin{code}%
\>[0]\AgdaFunction{∙times\#³}%
\>[143I]\AgdaSymbol{:}\AgdaSpace{}%
\AgdaSymbol{∀}\AgdaSpace{}%
\AgdaSymbol{\{}\AgdaBound{t₁}\AgdaSpace{}%
\AgdaBound{t₂}\AgdaSpace{}%
\AgdaBound{t₃}\AgdaSpace{}%
\AgdaBound{v₁}\AgdaSpace{}%
\AgdaBound{v₂}\AgdaSpace{}%
\AgdaBound{v₃}\AgdaSymbol{\}}\<%
\\
\>[.][@{}l@{}]\<[143I]%
\>[9]\AgdaSymbol{→}\AgdaSpace{}%
\AgdaSymbol{((}\AgdaBound{t₁}\AgdaSpace{}%
\AgdaOperator{\AgdaInductiveConstructor{×ᵤ}}\AgdaSpace{}%
\AgdaSymbol{(}\AgdaBound{t₂}\AgdaSpace{}%
\AgdaOperator{\AgdaInductiveConstructor{×ᵤ}}\AgdaSpace{}%
\AgdaBound{t₃}\AgdaSymbol{))}\AgdaSpace{}%
\AgdaOperator{\AgdaInductiveConstructor{\#}}\AgdaSpace{}%
\AgdaSymbol{(}\AgdaBound{v₁}\AgdaSpace{}%
\AgdaOperator{\AgdaInductiveConstructor{,}}\AgdaSpace{}%
\AgdaBound{v₂}\AgdaSpace{}%
\AgdaOperator{\AgdaInductiveConstructor{,}}\AgdaSpace{}%
\AgdaBound{v₃}\AgdaSymbol{))}\<%
\\
\>[9]\AgdaOperator{\AgdaDatatype{⧟}}\AgdaSpace{}%
\AgdaSymbol{((}\AgdaBound{t₁}\AgdaSpace{}%
\AgdaOperator{\AgdaInductiveConstructor{\#}}\AgdaSpace{}%
\AgdaBound{v₁}\AgdaSymbol{)}\AgdaSpace{}%
\AgdaOperator{\AgdaInductiveConstructor{∙×ᵤ}}\AgdaSpace{}%
\AgdaSymbol{(}\AgdaBound{t₂}\AgdaSpace{}%
\AgdaOperator{\AgdaInductiveConstructor{\#}}\AgdaSpace{}%
\AgdaBound{v₂}\AgdaSymbol{)}\AgdaSpace{}%
\AgdaOperator{\AgdaInductiveConstructor{∙×ᵤ}}\AgdaSpace{}%
\AgdaSymbol{(}\AgdaBound{t₃}\AgdaSpace{}%
\AgdaOperator{\AgdaInductiveConstructor{\#}}\AgdaSpace{}%
\AgdaBound{v₃}\AgdaSymbol{))}\<%
\\
\>[0]\AgdaFunction{∙times\#³}\AgdaSpace{}%
\AgdaSymbol{=}\AgdaSpace{}%
\AgdaInductiveConstructor{∙times\#}\AgdaSpace{}%
\AgdaOperator{\AgdaInductiveConstructor{∙⊚}}\AgdaSpace{}%
\AgdaInductiveConstructor{∙id⟷}\AgdaSpace{}%
\AgdaOperator{\AgdaInductiveConstructor{∙⊗}}\AgdaSpace{}%
\AgdaInductiveConstructor{∙times\#}\<%
\\
\\[\AgdaEmptyExtraSkip]%
\>[0]\AgdaFunction{∙\#times³}%
\>[179I]\AgdaSymbol{:}\AgdaSpace{}%
\AgdaSymbol{∀}\AgdaSpace{}%
\AgdaSymbol{\{}\AgdaBound{t₁}\AgdaSpace{}%
\AgdaBound{t₂}\AgdaSpace{}%
\AgdaBound{t₃}\AgdaSpace{}%
\AgdaBound{v₁}\AgdaSpace{}%
\AgdaBound{v₂}\AgdaSpace{}%
\AgdaBound{v₃}\AgdaSymbol{\}}\<%
\\
\>[.][@{}l@{}]\<[179I]%
\>[9]\AgdaSymbol{→}\AgdaSpace{}%
\AgdaSymbol{((}\AgdaBound{t₁}\AgdaSpace{}%
\AgdaOperator{\AgdaInductiveConstructor{\#}}\AgdaSpace{}%
\AgdaBound{v₁}\AgdaSymbol{)}\AgdaSpace{}%
\AgdaOperator{\AgdaInductiveConstructor{∙×ᵤ}}\AgdaSpace{}%
\AgdaSymbol{(}\AgdaBound{t₂}\AgdaSpace{}%
\AgdaOperator{\AgdaInductiveConstructor{\#}}\AgdaSpace{}%
\AgdaBound{v₂}\AgdaSymbol{)}\AgdaSpace{}%
\AgdaOperator{\AgdaInductiveConstructor{∙×ᵤ}}\AgdaSpace{}%
\AgdaSymbol{(}\AgdaBound{t₃}\AgdaSpace{}%
\AgdaOperator{\AgdaInductiveConstructor{\#}}\AgdaSpace{}%
\AgdaBound{v₃}\AgdaSymbol{))}\<%
\\
\>[9]\AgdaOperator{\AgdaDatatype{⧟}}\AgdaSpace{}%
\AgdaSymbol{((}\AgdaBound{t₁}\AgdaSpace{}%
\AgdaOperator{\AgdaInductiveConstructor{×ᵤ}}\AgdaSpace{}%
\AgdaSymbol{(}\AgdaBound{t₂}\AgdaSpace{}%
\AgdaOperator{\AgdaInductiveConstructor{×ᵤ}}\AgdaSpace{}%
\AgdaBound{t₃}\AgdaSymbol{))}\AgdaSpace{}%
\AgdaOperator{\AgdaInductiveConstructor{\#}}\AgdaSpace{}%
\AgdaSymbol{(}\AgdaBound{v₁}\AgdaSpace{}%
\AgdaOperator{\AgdaInductiveConstructor{,}}\AgdaSpace{}%
\AgdaBound{v₂}\AgdaSpace{}%
\AgdaOperator{\AgdaInductiveConstructor{,}}\AgdaSpace{}%
\AgdaBound{v₃}\AgdaSymbol{))}\<%
\\
\>[0]\AgdaFunction{∙\#times³}\AgdaSpace{}%
\AgdaSymbol{=}\AgdaSpace{}%
\AgdaInductiveConstructor{∙id⟷}\AgdaSpace{}%
\AgdaOperator{\AgdaInductiveConstructor{∙⊗}}\AgdaSpace{}%
\AgdaInductiveConstructor{∙\#times}\AgdaSpace{}%
\AgdaOperator{\AgdaInductiveConstructor{∙⊚}}\AgdaSpace{}%
\AgdaInductiveConstructor{∙\#times}\<%
\\
\\[\AgdaEmptyExtraSkip]%
\>[0]\AgdaFunction{∙TOFFOLI}\AgdaSpace{}%
\AgdaSymbol{:}\AgdaSpace{}%
\AgdaSymbol{∀}\AgdaSpace{}%
\AgdaSymbol{\{}\AgdaBound{a}\AgdaSpace{}%
\AgdaBound{b}\AgdaSpace{}%
\AgdaBound{c}\AgdaSymbol{\}}%
\>[220I]\AgdaSymbol{→}%
\>[25]\AgdaSymbol{(}\AgdaFunction{𝔹}\AgdaSpace{}%
\AgdaOperator{\AgdaInductiveConstructor{\#}}\AgdaSpace{}%
\AgdaBound{a}\AgdaSpace{}%
\AgdaOperator{\AgdaInductiveConstructor{∙×ᵤ}}\AgdaSpace{}%
\AgdaFunction{𝔹}\AgdaSpace{}%
\AgdaOperator{\AgdaInductiveConstructor{\#}}\AgdaSpace{}%
\AgdaBound{b}\AgdaSpace{}%
\AgdaOperator{\AgdaInductiveConstructor{∙×ᵤ}}\AgdaSpace{}%
\AgdaFunction{𝔹}\AgdaSpace{}%
\AgdaOperator{\AgdaInductiveConstructor{\#}}\AgdaSpace{}%
\AgdaBound{c}\AgdaSymbol{)}\<%
\\
\>[.][@{}l@{}]\<[220I]%
\>[21]\AgdaOperator{\AgdaDatatype{⧟}}\AgdaSpace{}%
\AgdaSymbol{(}\AgdaFunction{𝔹}\AgdaSpace{}%
\AgdaOperator{\AgdaInductiveConstructor{\#}}\AgdaSpace{}%
\AgdaBound{a}\AgdaSpace{}%
\AgdaOperator{\AgdaInductiveConstructor{∙×ᵤ}}\AgdaSpace{}%
\AgdaFunction{𝔹}\AgdaSpace{}%
\AgdaOperator{\AgdaInductiveConstructor{\#}}\AgdaSpace{}%
\AgdaBound{b}\AgdaSpace{}%
\AgdaOperator{\AgdaInductiveConstructor{∙×ᵤ}}\AgdaSpace{}%
\AgdaFunction{𝔹}\AgdaSpace{}%
\AgdaOperator{\AgdaInductiveConstructor{\#}}\AgdaSpace{}%
\AgdaSymbol{((}\AgdaBound{a}\AgdaSpace{}%
\AgdaOperator{\AgdaFunction{\&}}\AgdaSpace{}%
\AgdaBound{b}\AgdaSymbol{)}\AgdaSpace{}%
\AgdaOperator{\AgdaFunction{\textasciicircum{}}}\AgdaSpace{}%
\AgdaBound{c}\AgdaSymbol{))}\<%
\\
\>[0]\AgdaFunction{∙TOFFOLI}\AgdaSpace{}%
\AgdaSymbol{=}\AgdaSpace{}%
\AgdaFunction{∙\#times³}\AgdaSpace{}%
\AgdaOperator{\AgdaInductiveConstructor{∙⊚}}\AgdaSpace{}%
\AgdaFunction{TOFFOLI'}\AgdaSpace{}%
\AgdaOperator{\AgdaInductiveConstructor{∙⊚}}\AgdaSpace{}%
\AgdaFunction{∙times\#³}\<%
\\
\>[0][@{}l@{\AgdaIndent{0}}]%
\>[2]\AgdaKeyword{where}\<%
\\
\>[2][@{}l@{\AgdaIndent{0}}]%
\>[4]\AgdaFunction{TOFFOLI'}\AgdaSpace{}%
\AgdaSymbol{:}\AgdaSpace{}%
\AgdaSymbol{∀}\AgdaSpace{}%
\AgdaSymbol{\{}\AgdaBound{a}\AgdaSpace{}%
\AgdaBound{b}\AgdaSpace{}%
\AgdaBound{c}\AgdaSymbol{\}}\AgdaSpace{}%
\AgdaSymbol{→}\AgdaSpace{}%
\AgdaSymbol{(}\AgdaFunction{𝔹³}\AgdaSpace{}%
\AgdaOperator{\AgdaInductiveConstructor{\#}}\AgdaSpace{}%
\AgdaSymbol{(}\AgdaBound{a}\AgdaSpace{}%
\AgdaOperator{\AgdaInductiveConstructor{,}}\AgdaSpace{}%
\AgdaBound{b}\AgdaSpace{}%
\AgdaOperator{\AgdaInductiveConstructor{,}}\AgdaSpace{}%
\AgdaBound{c}\AgdaSymbol{))}\AgdaSpace{}%
\AgdaOperator{\AgdaDatatype{⧟}}\AgdaSpace{}%
\AgdaSymbol{(}\AgdaFunction{𝔹³}\AgdaSpace{}%
\AgdaOperator{\AgdaInductiveConstructor{\#}}\AgdaSpace{}%
\AgdaSymbol{(}\AgdaBound{a}\AgdaSpace{}%
\AgdaOperator{\AgdaInductiveConstructor{,}}\AgdaSpace{}%
\AgdaBound{b}\AgdaSpace{}%
\AgdaOperator{\AgdaInductiveConstructor{,}}\AgdaSpace{}%
\AgdaSymbol{((}\AgdaBound{a}\AgdaSpace{}%
\AgdaOperator{\AgdaFunction{\&}}\AgdaSpace{}%
\AgdaBound{b}\AgdaSymbol{)}\AgdaSpace{}%
\AgdaOperator{\AgdaFunction{\textasciicircum{}}}\AgdaSpace{}%
\AgdaBound{c}\AgdaSymbol{)))}\<%
\\
\>[4]\AgdaFunction{TOFFOLI'}\AgdaSpace{}%
\AgdaSymbol{\{}\AgdaInductiveConstructor{𝔽}\AgdaSymbol{\}}\AgdaSpace{}%
\AgdaSymbol{\{}\AgdaInductiveConstructor{𝔽}\AgdaSymbol{\}}\AgdaSpace{}%
\AgdaSymbol{\{}\AgdaBound{c}\AgdaSymbol{\}}\AgdaSpace{}%
\AgdaSymbol{=}\AgdaSpace{}%
\AgdaInductiveConstructor{∙c}\AgdaSpace{}%
\AgdaFunction{TOFFOLI}\<%
\\
\>[4]\AgdaFunction{TOFFOLI'}\AgdaSpace{}%
\AgdaSymbol{\{}\AgdaInductiveConstructor{𝔽}\AgdaSymbol{\}}\AgdaSpace{}%
\AgdaSymbol{\{}\AgdaInductiveConstructor{𝕋}\AgdaSymbol{\}}\AgdaSpace{}%
\AgdaSymbol{\{}\AgdaBound{c}\AgdaSymbol{\}}\AgdaSpace{}%
\AgdaSymbol{=}\AgdaSpace{}%
\AgdaInductiveConstructor{∙c}\AgdaSpace{}%
\AgdaFunction{TOFFOLI}\<%
\\
\>[4]\AgdaFunction{TOFFOLI'}\AgdaSpace{}%
\AgdaSymbol{\{}\AgdaInductiveConstructor{𝕋}\AgdaSymbol{\}}\AgdaSpace{}%
\AgdaSymbol{\{}\AgdaInductiveConstructor{𝔽}\AgdaSymbol{\}}\AgdaSpace{}%
\AgdaSymbol{\{}\AgdaBound{c}\AgdaSymbol{\}}\AgdaSpace{}%
\AgdaSymbol{=}\AgdaSpace{}%
\AgdaInductiveConstructor{∙c}\AgdaSpace{}%
\AgdaFunction{TOFFOLI}\<%
\\
\>[4]\AgdaFunction{TOFFOLI'}\AgdaSpace{}%
\AgdaSymbol{\{}\AgdaInductiveConstructor{𝕋}\AgdaSymbol{\}}\AgdaSpace{}%
\AgdaSymbol{\{}\AgdaInductiveConstructor{𝕋}\AgdaSymbol{\}}\AgdaSpace{}%
\AgdaSymbol{\{}\AgdaInductiveConstructor{𝔽}\AgdaSymbol{\}}\AgdaSpace{}%
\AgdaSymbol{=}\AgdaSpace{}%
\AgdaInductiveConstructor{∙c}\AgdaSpace{}%
\AgdaFunction{TOFFOLI}\<%
\\
\>[4]\AgdaFunction{TOFFOLI'}\AgdaSpace{}%
\AgdaSymbol{\{}\AgdaInductiveConstructor{𝕋}\AgdaSymbol{\}}\AgdaSpace{}%
\AgdaSymbol{\{}\AgdaInductiveConstructor{𝕋}\AgdaSymbol{\}}\AgdaSpace{}%
\AgdaSymbol{\{}\AgdaInductiveConstructor{𝕋}\AgdaSymbol{\}}\AgdaSpace{}%
\AgdaSymbol{=}\AgdaSpace{}%
\AgdaInductiveConstructor{∙c}\AgdaSpace{}%
\AgdaFunction{TOFFOLI}\<%
\end{code}}
\newcommand{\PToffoliExtEq}{%
\begin{code}%
\>[0]\AgdaFunction{ExtEq}\AgdaSpace{}%
\AgdaSymbol{:}\AgdaSpace{}%
\AgdaFunction{Ext⧟}\AgdaSpace{}%
\AgdaSymbol{(}\AgdaFunction{∙TOFFOLI}\AgdaSpace{}%
\AgdaSymbol{\{}\AgdaInductiveConstructor{𝔽}\AgdaSymbol{\}}\AgdaSpace{}%
\AgdaSymbol{\{}\AgdaInductiveConstructor{𝔽}\AgdaSymbol{\}}\AgdaSpace{}%
\AgdaSymbol{\{}\AgdaInductiveConstructor{𝔽}\AgdaSymbol{\})}\AgdaSpace{}%
\AgdaOperator{\AgdaDatatype{≡}}\AgdaSpace{}%
\AgdaFunction{Ext⧟}\AgdaSpace{}%
\AgdaSymbol{(}\AgdaFunction{∙TOFFOLI}\AgdaSpace{}%
\AgdaSymbol{\{}\AgdaInductiveConstructor{𝕋}\AgdaSymbol{\}}\AgdaSpace{}%
\AgdaSymbol{\{}\AgdaInductiveConstructor{𝕋}\AgdaSymbol{\}}\AgdaSpace{}%
\AgdaSymbol{\{}\AgdaInductiveConstructor{𝕋}\AgdaSymbol{\})}\<%
\\
\>[0]\AgdaFunction{ExtEq}\AgdaSpace{}%
\AgdaSymbol{=}\AgdaSpace{}%
\AgdaInductiveConstructor{refl}\<%
\end{code}}
\newcommand{\PToffolifour}{%
\begin{code}%
\>[0]\AgdaFunction{∙TOFFOLI₄}\AgdaSpace{}%
\AgdaSymbol{:}%
\>[13]\AgdaSymbol{∀}\AgdaSpace{}%
\AgdaSymbol{\{}\AgdaBound{a}\AgdaSpace{}%
\AgdaBound{b}\AgdaSpace{}%
\AgdaBound{c}\AgdaSpace{}%
\AgdaBound{d}\AgdaSymbol{\}}\AgdaSpace{}%
\AgdaSymbol{→}\<%
\\
\>[13]\AgdaSymbol{(}\AgdaFunction{𝔹}\AgdaSpace{}%
\AgdaOperator{\AgdaInductiveConstructor{\#}}\AgdaSpace{}%
\AgdaBound{a}\AgdaSpace{}%
\AgdaOperator{\AgdaInductiveConstructor{∙×ᵤ}}\AgdaSpace{}%
\AgdaFunction{𝔹}\AgdaSpace{}%
\AgdaOperator{\AgdaInductiveConstructor{\#}}\AgdaSpace{}%
\AgdaBound{b}\AgdaSpace{}%
\AgdaOperator{\AgdaInductiveConstructor{∙×ᵤ}}\AgdaSpace{}%
\AgdaFunction{𝔹}\AgdaSpace{}%
\AgdaOperator{\AgdaInductiveConstructor{\#}}\AgdaSpace{}%
\AgdaBound{c}\AgdaSpace{}%
\AgdaOperator{\AgdaInductiveConstructor{∙×ᵤ}}\AgdaSpace{}%
\AgdaFunction{𝔹}\AgdaSpace{}%
\AgdaOperator{\AgdaInductiveConstructor{\#}}\AgdaSpace{}%
\AgdaBound{d}\AgdaSymbol{)}\AgdaSpace{}%
\AgdaOperator{\AgdaDatatype{⧟}}\<%
\\
\>[13]\AgdaSymbol{(}\AgdaFunction{𝔹}\AgdaSpace{}%
\AgdaOperator{\AgdaInductiveConstructor{\#}}\AgdaSpace{}%
\AgdaBound{a}\AgdaSpace{}%
\AgdaOperator{\AgdaInductiveConstructor{∙×ᵤ}}\AgdaSpace{}%
\AgdaFunction{𝔹}\AgdaSpace{}%
\AgdaOperator{\AgdaInductiveConstructor{\#}}\AgdaSpace{}%
\AgdaBound{b}\AgdaSpace{}%
\AgdaOperator{\AgdaInductiveConstructor{∙×ᵤ}}\AgdaSpace{}%
\AgdaFunction{𝔹}\AgdaSpace{}%
\AgdaOperator{\AgdaInductiveConstructor{\#}}\AgdaSpace{}%
\AgdaBound{c}\AgdaSpace{}%
\AgdaOperator{\AgdaInductiveConstructor{∙×ᵤ}}\AgdaSpace{}%
\AgdaFunction{𝔹}\AgdaSpace{}%
\AgdaOperator{\AgdaInductiveConstructor{\#}}\AgdaSpace{}%
\AgdaSymbol{(((}\AgdaBound{a}\AgdaSpace{}%
\AgdaOperator{\AgdaFunction{\&}}\AgdaSpace{}%
\AgdaBound{b}\AgdaSymbol{)}\AgdaSpace{}%
\AgdaOperator{\AgdaFunction{\&}}\AgdaSpace{}%
\AgdaBound{c}\AgdaSymbol{)}\AgdaSpace{}%
\AgdaOperator{\AgdaFunction{\textasciicircum{}}}\AgdaSpace{}%
\AgdaBound{d}\AgdaSymbol{))}\<%
\\
\>[0]\AgdaFunction{∙TOFFOLI₄}\AgdaSpace{}%
\AgdaSymbol{=}\<%
\\
\>[0][@{}l@{\AgdaIndent{0}}]%
\>[2]\AgdaInductiveConstructor{∙assocl⋆}\AgdaSpace{}%
\AgdaOperator{\AgdaInductiveConstructor{∙⊚}}\<%
\\
\>[2]\AgdaSymbol{((}\AgdaInductiveConstructor{∙uniti⋆r}\AgdaSpace{}%
\AgdaOperator{\AgdaInductiveConstructor{∙⊚}}\AgdaSpace{}%
\AgdaSymbol{(}\AgdaInductiveConstructor{∙id⟷}\AgdaSpace{}%
\AgdaOperator{\AgdaInductiveConstructor{∙⊗}}\AgdaSpace{}%
\AgdaSymbol{(}\AgdaInductiveConstructor{η}\AgdaSpace{}%
\AgdaSymbol{(}\AgdaFunction{𝔹}\AgdaSpace{}%
\AgdaOperator{\AgdaInductiveConstructor{\#}}\AgdaSpace{}%
\AgdaInductiveConstructor{𝔽}\AgdaSymbol{)}\AgdaSpace{}%
\AgdaOperator{\AgdaInductiveConstructor{∙⊚}}\AgdaSpace{}%
\AgdaSymbol{(}\AgdaInductiveConstructor{extract}\AgdaSpace{}%
\AgdaOperator{\AgdaInductiveConstructor{∙⊗}}\AgdaSpace{}%
\AgdaInductiveConstructor{∙id⟷}\AgdaSymbol{))))}\AgdaSpace{}%
\AgdaOperator{\AgdaInductiveConstructor{∙⊗}}\AgdaSpace{}%
\AgdaInductiveConstructor{∙id⟷}\AgdaSymbol{)}\<%
\\
\>[2]\AgdaOperator{\AgdaInductiveConstructor{∙⊚}}\<%
\\
\>[2]\AgdaSymbol{((}\AgdaInductiveConstructor{∙assocl⋆}\AgdaSpace{}%
\AgdaOperator{\AgdaInductiveConstructor{∙⊚}}\AgdaSpace{}%
\AgdaSymbol{((}\AgdaInductiveConstructor{∙assocr⋆}\AgdaSpace{}%
\AgdaOperator{\AgdaInductiveConstructor{∙⊚}}\AgdaSpace{}%
\AgdaFunction{∙TOFFOLI}\AgdaSymbol{)}\AgdaSpace{}%
\AgdaOperator{\AgdaInductiveConstructor{∙⊗}}\AgdaSpace{}%
\AgdaInductiveConstructor{∙id⟷}\AgdaSymbol{)}\AgdaSpace{}%
\AgdaOperator{\AgdaInductiveConstructor{∙⊚}}\AgdaSpace{}%
\AgdaFunction{shuffle}\AgdaSymbol{)}\AgdaSpace{}%
\AgdaOperator{\AgdaInductiveConstructor{∙⊗}}\AgdaSpace{}%
\AgdaInductiveConstructor{∙id⟷}\AgdaSymbol{)}\<%
\\
\>[2]\AgdaOperator{\AgdaInductiveConstructor{∙⊚}}\<%
\\
\>[2]\AgdaInductiveConstructor{∙assocr⋆}\AgdaSpace{}%
\AgdaOperator{\AgdaInductiveConstructor{∙⊚}}\AgdaSpace{}%
\AgdaSymbol{(}\AgdaInductiveConstructor{∙id⟷}\AgdaSpace{}%
\AgdaOperator{\AgdaInductiveConstructor{∙⊗}}\AgdaSpace{}%
\AgdaFunction{∙TOFFOLI}\AgdaSymbol{)}\AgdaSpace{}%
\AgdaOperator{\AgdaInductiveConstructor{∙⊚}}\AgdaSpace{}%
\AgdaInductiveConstructor{∙assocl⋆}\<%
\\
\>[2]\AgdaOperator{\AgdaInductiveConstructor{∙⊚}}\<%
\\
\>[2]\AgdaSymbol{((}\AgdaFunction{shuffle}\AgdaSpace{}%
\AgdaOperator{\AgdaInductiveConstructor{∙⊚}}\AgdaSpace{}%
\AgdaSymbol{((}\AgdaFunction{∙TOFFOLI}\AgdaSpace{}%
\AgdaOperator{\AgdaInductiveConstructor{∙⊚}}\AgdaSpace{}%
\AgdaInductiveConstructor{∙assocl⋆}\AgdaSymbol{)}\AgdaSpace{}%
\AgdaOperator{\AgdaInductiveConstructor{∙⊗}}\AgdaSpace{}%
\AgdaInductiveConstructor{∙id⟷}\AgdaSymbol{)}\AgdaSpace{}%
\AgdaOperator{\AgdaInductiveConstructor{∙⊚}}\AgdaSpace{}%
\AgdaInductiveConstructor{∙assocr⋆}\AgdaSymbol{)}\AgdaSpace{}%
\AgdaOperator{\AgdaInductiveConstructor{∙⊗}}\AgdaSpace{}%
\AgdaInductiveConstructor{∙id⟷}\AgdaSymbol{)}\<%
\\
\>[2]\AgdaOperator{\AgdaInductiveConstructor{∙⊚}}\<%
\\
\>[2]\AgdaSymbol{(((}\AgdaInductiveConstructor{∙id⟷}\AgdaSpace{}%
\AgdaOperator{\AgdaInductiveConstructor{∙⊗}}\AgdaSpace{}%
\AgdaSymbol{((}\AgdaInductiveConstructor{return}\AgdaSpace{}%
\AgdaOperator{\AgdaInductiveConstructor{∙⊗}}\AgdaSpace{}%
\AgdaInductiveConstructor{∙id⟷}\AgdaSymbol{)}\AgdaSpace{}%
\AgdaOperator{\AgdaInductiveConstructor{∙⊚}}\AgdaSpace{}%
\AgdaInductiveConstructor{ε}\AgdaSpace{}%
\AgdaSymbol{(}\AgdaFunction{𝔹}\AgdaSpace{}%
\AgdaOperator{\AgdaInductiveConstructor{\#}}\AgdaSpace{}%
\AgdaInductiveConstructor{𝔽}\AgdaSymbol{)))}\AgdaSpace{}%
\AgdaOperator{\AgdaInductiveConstructor{∙⊚}}\AgdaSpace{}%
\AgdaInductiveConstructor{∙unite⋆r}\AgdaSymbol{)}\AgdaSpace{}%
\AgdaOperator{\AgdaInductiveConstructor{∙⊗}}\AgdaSpace{}%
\AgdaInductiveConstructor{∙id⟷}\AgdaSymbol{)}\<%
\\
\>[2]\AgdaOperator{\AgdaInductiveConstructor{∙⊚}}\<%
\\
\>[2]\AgdaInductiveConstructor{∙assocr⋆}\<%
\\
\\[\AgdaEmptyExtraSkip]%
\>[2]\AgdaKeyword{where}\<%
\\
\>[2][@{}l@{\AgdaIndent{0}}]%
\>[4]\AgdaFunction{shuffle}\AgdaSpace{}%
\AgdaSymbol{:}%
\>[15]\AgdaSymbol{∀}\AgdaSpace{}%
\AgdaSymbol{\{}\AgdaBound{A}\AgdaSpace{}%
\AgdaBound{B}\AgdaSpace{}%
\AgdaBound{C}\AgdaSpace{}%
\AgdaBound{D}\AgdaSymbol{\}}\AgdaSpace{}%
\AgdaSymbol{→}\<%
\\
\>[15]\AgdaSymbol{(}\AgdaBound{A}\AgdaSpace{}%
\AgdaOperator{\AgdaInductiveConstructor{∙×ᵤ}}\AgdaSpace{}%
\AgdaBound{B}\AgdaSpace{}%
\AgdaOperator{\AgdaInductiveConstructor{∙×ᵤ}}\AgdaSpace{}%
\AgdaBound{C}\AgdaSymbol{)}\AgdaSpace{}%
\AgdaOperator{\AgdaInductiveConstructor{∙×ᵤ}}\AgdaSpace{}%
\AgdaBound{D}\AgdaSpace{}%
\AgdaOperator{\AgdaDatatype{⧟}}\AgdaSpace{}%
\AgdaSymbol{(}\AgdaBound{A}\AgdaSpace{}%
\AgdaOperator{\AgdaInductiveConstructor{∙×ᵤ}}\AgdaSpace{}%
\AgdaBound{B}\AgdaSpace{}%
\AgdaOperator{\AgdaInductiveConstructor{∙×ᵤ}}\AgdaSpace{}%
\AgdaBound{D}\AgdaSymbol{)}\AgdaSpace{}%
\AgdaOperator{\AgdaInductiveConstructor{∙×ᵤ}}\AgdaSpace{}%
\AgdaBound{C}\<%
\\
\>[4]\AgdaFunction{shuffle}\AgdaSpace{}%
\AgdaSymbol{=}%
\>[15]\AgdaInductiveConstructor{∙assocr⋆}\AgdaSpace{}%
\AgdaOperator{\AgdaInductiveConstructor{∙⊚}}\AgdaSpace{}%
\AgdaSymbol{(}\AgdaInductiveConstructor{∙id⟷}\AgdaSpace{}%
\AgdaOperator{\AgdaInductiveConstructor{∙⊗}}\AgdaSpace{}%
\AgdaSymbol{(}\AgdaInductiveConstructor{∙assocr⋆}\AgdaSpace{}%
\AgdaOperator{\AgdaInductiveConstructor{∙⊚}}\AgdaSpace{}%
\AgdaSymbol{(}\AgdaInductiveConstructor{∙id⟷}\AgdaSpace{}%
\AgdaOperator{\AgdaInductiveConstructor{∙⊗}}\AgdaSpace{}%
\AgdaInductiveConstructor{∙swap⋆}\AgdaSymbol{)))}\AgdaSpace{}%
\AgdaOperator{\AgdaInductiveConstructor{∙⊚}}\<%
\\
\>[15]\AgdaSymbol{(}\AgdaInductiveConstructor{∙id⟷}\AgdaSpace{}%
\AgdaOperator{\AgdaInductiveConstructor{∙⊗}}\AgdaSpace{}%
\AgdaInductiveConstructor{∙assocl⋆}\AgdaSymbol{)}\AgdaSpace{}%
\AgdaOperator{\AgdaInductiveConstructor{∙⊚}}\AgdaSpace{}%
\AgdaInductiveConstructor{∙assocl⋆}\<%
\end{code}}
\newcommand{\PToffoliExt}{%
\begin{code}%
\>[0]\AgdaFunction{TOFFOLI₄}\AgdaSpace{}%
\AgdaSymbol{:}\AgdaSpace{}%
\AgdaFunction{𝔹⁴}\AgdaSpace{}%
\AgdaOperator{\AgdaDatatype{⟷}}\AgdaSpace{}%
\AgdaFunction{𝔹⁴}\<%
\\
\>[0]\AgdaFunction{TOFFOLI₄}\AgdaSpace{}%
\AgdaSymbol{=}\AgdaSpace{}%
\AgdaFunction{Ext⧟}\AgdaSpace{}%
\AgdaSymbol{(}\AgdaFunction{∙TOFFOLI₄}\AgdaSpace{}%
\AgdaSymbol{\{}\AgdaInductiveConstructor{𝔽}\AgdaSymbol{\}}\AgdaSpace{}%
\AgdaSymbol{\{}\AgdaInductiveConstructor{𝔽}\AgdaSymbol{\}}\AgdaSpace{}%
\AgdaSymbol{\{}\AgdaInductiveConstructor{𝔽}\AgdaSymbol{\}}\AgdaSpace{}%
\AgdaSymbol{\{}\AgdaInductiveConstructor{𝔽}\AgdaSymbol{\})}\<%
\end{code}}
\newcommand{\PToffoliExtTest}{%
\begin{code}%
\>[0]\AgdaFunction{TOFFOLI₄Test₁}\AgdaSpace{}%
\AgdaSymbol{:}\AgdaSpace{}%
\AgdaFunction{interp}\AgdaSpace{}%
\AgdaFunction{TOFFOLI₄}\AgdaSpace{}%
\AgdaSymbol{(}\AgdaInductiveConstructor{𝔽}\AgdaSpace{}%
\AgdaOperator{\AgdaInductiveConstructor{,}}\AgdaSpace{}%
\AgdaInductiveConstructor{𝔽}\AgdaSpace{}%
\AgdaOperator{\AgdaInductiveConstructor{,}}\AgdaSpace{}%
\AgdaInductiveConstructor{𝔽}\AgdaSpace{}%
\AgdaOperator{\AgdaInductiveConstructor{,}}\AgdaSpace{}%
\AgdaInductiveConstructor{𝔽}\AgdaSymbol{)}\AgdaSpace{}%
\AgdaOperator{\AgdaDatatype{≡}}\AgdaSpace{}%
\AgdaInductiveConstructor{just}\AgdaSpace{}%
\AgdaSymbol{(}\AgdaInductiveConstructor{𝔽}\AgdaSpace{}%
\AgdaOperator{\AgdaInductiveConstructor{,}}\AgdaSpace{}%
\AgdaInductiveConstructor{𝔽}\AgdaSpace{}%
\AgdaOperator{\AgdaInductiveConstructor{,}}\AgdaSpace{}%
\AgdaInductiveConstructor{𝔽}\AgdaSpace{}%
\AgdaOperator{\AgdaInductiveConstructor{,}}\AgdaSpace{}%
\AgdaInductiveConstructor{𝔽}\AgdaSymbol{)}\<%
\\
\>[0]\AgdaFunction{TOFFOLI₄Test₁}\AgdaSpace{}%
\AgdaSymbol{=}\AgdaSpace{}%
\AgdaInductiveConstructor{refl}\<%
\\
\\[\AgdaEmptyExtraSkip]%
\>[0]\AgdaFunction{TOFFOLI₄Test₂}\AgdaSpace{}%
\AgdaSymbol{:}\AgdaSpace{}%
\AgdaFunction{interp}\AgdaSpace{}%
\AgdaFunction{TOFFOLI₄}\AgdaSpace{}%
\AgdaSymbol{(}\AgdaInductiveConstructor{𝕋}\AgdaSpace{}%
\AgdaOperator{\AgdaInductiveConstructor{,}}\AgdaSpace{}%
\AgdaInductiveConstructor{𝕋}\AgdaSpace{}%
\AgdaOperator{\AgdaInductiveConstructor{,}}\AgdaSpace{}%
\AgdaInductiveConstructor{𝕋}\AgdaSpace{}%
\AgdaOperator{\AgdaInductiveConstructor{,}}\AgdaSpace{}%
\AgdaInductiveConstructor{𝔽}\AgdaSymbol{)}\AgdaSpace{}%
\AgdaOperator{\AgdaDatatype{≡}}\AgdaSpace{}%
\AgdaInductiveConstructor{just}\AgdaSpace{}%
\AgdaSymbol{(}\AgdaInductiveConstructor{𝕋}\AgdaSpace{}%
\AgdaOperator{\AgdaInductiveConstructor{,}}\AgdaSpace{}%
\AgdaInductiveConstructor{𝕋}\AgdaSpace{}%
\AgdaOperator{\AgdaInductiveConstructor{,}}\AgdaSpace{}%
\AgdaInductiveConstructor{𝕋}\AgdaSpace{}%
\AgdaOperator{\AgdaInductiveConstructor{,}}\AgdaSpace{}%
\AgdaInductiveConstructor{𝕋}\AgdaSymbol{)}\<%
\\
\>[0]\AgdaFunction{TOFFOLI₄Test₂}\AgdaSpace{}%
\AgdaSymbol{=}\AgdaSpace{}%
\AgdaInductiveConstructor{refl}\<%
\end{code}}

\begin{document}
\title{Fractional Types}
\subtitle{Expressive and Safe Space Management for Ancilla Bits}
%
%
\author{
  Chao-Hong Chen\inst{1} \and
  Vikraman Choudhury\inst{1} \and
  Jacques Carette\inst{2} \and
  Amr Sabry\inst{1}
  }
\authorrunning{C. Chen et al.}
%
\institute{
  Indiana University, Bloomington IN, USA\\
  \email{\{chen464,vikraman,sabry\}@indiana.edu}
  \and
  McMaster University, Ontario, Canada\\
  \email{carette@mcmaster.ca}
}
\maketitle              
\begin{abstract}
In reversible computing, the management of space is
subject to two broad classes of constraints. First, as with
general-purpose computation, every allocation must be paired
with a matching de-allocation. Second, space can only be safely
de-allocated if its contents are restored to their initial value
from allocation time. Generally speaking, the state of the art
provides limited partial solutions that address the first
constraint by imposing a stack discipline and by leaving the second
constraint to programmers' assertions.

We propose a novel approach based on the idea of \emph{fractional
types}. As a simple intuitive example, allocation of a new boolean
value initialized to \AgdaInductiveConstructor{false} also creates a value
$\oneover{\AgdaInductiveConstructor{false}}$ that can be thought of as a
garbage collection (GC) process specialized to reclaim, and only reclaim,
storage containing the value $\AgdaInductiveConstructor{false}$. This GC
process is a first-class entity that can be manipulated, decomposed
into smaller processes and combined with other GC processes.

We formalize this idea in the context of a reversible language founded
on type isomorphisms, prove its fundamental correctness properties, and
illustrate its expressiveness using a wide variety of examples. The
development is backed by a fully-formalized Agda
implementation~\footnote{\url{https://github.com/DreamLinuxer/FracAncilla}}.

\end{abstract}
\section{Introduction}

We solve the ancilla problem in reversible computation
using a novel concept: \emph{fractional types}. In
the next section, we introduce the problem of ancilla management, motivate its
importance, and explain the limitations of current ``stack-based''
solutions with runtime checks.

Although the concept of fractional types could potentially be integrated with
general-purpose languages, its natural technical definition exploits
symmetries present in the categorical model of type isomorphisms.  To
that end, we first review in Sec.~\ref{sec:pi} our previous
work~\cite{rc2011,DBLP:conf/esop/CaretteS16,rc2012,James:2012:IE:2103656.2103667}
on reversible programming language built using type isomorphisms. In
Sec.~\ref{sec:dyn}, we introduce a simple version of fractional types
that allows allocation and de-allocation of ancillae in patterns
beyond the model in scope but, like existing stack-based solutions, still
requires a runtime check to verify the safety of de-allocation. In
Sec.~\ref{sec:dep} we show how to remove this runtime check, by
lifting programs to a richer type system with pointed types,
expressing the proofs of safety in that setting, and then, from the proofs,
extracting programs with guaranteed safe de-allocations and no runtime
checks.
The last section concludes with a summary of our results.






\section{Ancilla Bits: Review and a Type-Based Approach}
\label{sec:examples}

Restricting a reversible circuit to use no ancilla bits is like restricting
a Turing machine to use no memory other than the $n$ bits used to represent
the input~\cite{aaronson_et_al:LIPIcs:2017:8173}. As such a restriction
disallows countless computations for trivial reasons, reversible
models of computation have,
since their inception, included management for scratch storage in the form of ancilla
bits~\cite{Toffoli:1980} with the fundamental restriction that such bits must be returned
to their initial states before being safely reused or de-allocated.

\subsection{Review}

In programming languages, the common way to handle ancilla bits is to use a stack model in which the
lifetime of the ancilla bits coincides with their \emph{static scope}
and augment this discipline with a runtime check that ensures that the
ancilla bit has the correct value. For example:
\begin{itemize}
\item Quipper~\cite{Green:2013:QSQ:2491956.2462177} uses a scoped way to manage ancilla bits via:
\begin{verbatim}
    with_ancilla :: (Qubit -> Circ a) -> Circ a
\end{verbatim}
  The operator takes a block of gates parameterized by an ancilla value,
  allocates a new ancilla value of type \verb|Qubit| initialized to
  $\ket{0}$, and runs the given block of gates. At the end of its
  execution, the block is expected to return the ancilla value to the state
  $\ket{0}$ at which point it is de-allocated. The expectation that
  the ancilla value is in the state $\ket{0}$ is enforced via a runtime check.
\item Ricercar~\cite{10.1007/978-3-319-20860-2_13} also uses a scoped
  way to manage ancillae. The expression $\alpha x . A$ allocates an ancilla wire
  $x$ for the gate $A$ requiring that $x$ is set to $0$ after the
  evaluation of $A$,
  as the following rule of the operational semantics shows:
\[
\Rule{}
{\sigma \vdash x \to b \qquad
  \sigma[x \mapsto 0] \vdash A \to \sigma' \qquad
  \sigma' \vdash x \to 0}
{\sigma \vdash \alpha x.A \to \sigma'[x \mapsto b]}
{}
\]
where $\sigma$ is the global memory mapping each variable to its
value and $\rightarrow$ represents evaluation.
\item Janus~\cite{yokoyama:2007:rpl:1244381.1244404} is a reversible
  imperative programming language that is not based on the circuit
  model but as Rose~\cite{rose2015arrow} explains, its treatment is
  essentially similar to above:
  \begin{quote}
    All variables in original Janus are global, but in the University
    of Copenhagen interpreter you can allocate local variables with
    the \textsf{local} statement. The inverse of the \textsf{local}
    statement is the \textsf{delocal} statement, which performs
    deallocation. When inverted, the deallocation becomes the
    allocation and vice versa. In order to invert deallocation, the
    value of the variable at deallocation time must be known, so the
    syntax is \verb|delocal <variable> = <value>|. Again the onus is
    on the programmer to ensure that the equality actually holds.
\end{quote}
\end{itemize}

\subsection{A Type-Based Approach}

This scoped model is a pragmatic choice which is however
limited. To understand its limitations more vividly, consider
the following analogy: allocating an ancilla bit by creating a new wire in
the circuit is like borrowing some money from a global external
entity (the memory manager); the computation has access to a new resource
temporarily. De-allocating the ancilla bit is like returning the borrowed
money to the global entity; the computation no longer has access to
that resource. It would however be unreasonably restrictive
to insist that the person (function) borrowing the money must be
the same person (function) returning it.
Indeed, as far as reversible computation is concerned, the
only important invariant is that information is conserved, i.e., that
money is conserved. The identities of bits are not
observable as they are all interchangeable in the same way that
particular bills with different serial numbers are interchangeable in
financial transactions. Thus the only invariant is that the net flow
of money between the computation and the global entity is zero. This
observation allows us to go even further than just switching the
identities of borrowers. It is even possible for one person to borrow
\$10, and have three different persons collectively collaborate to
pay back the debt with one person paying \$5, another \$2, and a third
\$3, nor the opposite situation of gradually borrowing \$10 and returning
it all at once.

Computationally, this extra generality is not a gratuitous concern:
since scope is a \emph{static property} of programs, it does
not allow the flexibility of heap allocation in which the lifetime of
resources is dynamically determined. Furthermore, limiting ancilla
bits to static scope does not help in solving the fundamental problem of
ensuring that their value is properly restored to their initial value
before de-allocation.

We demonstrate that both problems can be solved with a typing discipline.
The main idea is simple: we
introduce a type representing ``processes specialized to
garbage-collect specific values.'' The infrastructure of reversible
computing will ensure that the information inherent in this
process will never be duplicated or erased, enforcing that proper
safe de-allocation must happen in a complete program. Furthermore, since
reversible computation focuses on conservation of \emph{information}
rather than syntactic entities, this approach will permit fascinating
mechanisms in which allocations and de-allocations can be sliced and
diced, decomposed and recomposed, run forwards and backwards, in
arbitrary ways as long as the net balance is 0.


\begin{figure*}[t]
\[\begin{array}{rrcll}
\idc :& \tau & \iso & \tau &: \idc \\
\\
\identlp :&  0 + \tau & \iso & \tau &: \identrp \\
\swapp :&  \tau_1 + \tau_2 & \iso & \tau_2 + \tau_1 &: \swapp \\
\assoclp :&  \tau_1 + (\tau_2 + \tau_3) & \iso & (\tau_1 + \tau_2) + \tau_3 &: \assocrp \\ [1.5ex]
\identlt :&  1 \times \tau & \iso & \tau &: \identrt \\
\swapt :&  \tau_1 \times \tau_2 & \iso & \tau_2 \times \tau_1 &: \swapt \\
\assoclt :&  \tau_1 \times (\tau_2 \times \tau_3) & \iso & (\tau_1 \times\tau_2) \times \tau_3 &: \assocrt \\ [1.5ex]
\distz :&~ 0 \times \tau & \iso & 0 ~ &: \factorzl \\
\dist :&~ (\tau_1 + \tau_2) \times \tau_3 & \iso & (\tau_1 \times \tau_3) + (\tau_2 \times \tau_3)~ &: \factor
\end{array}\]
{\scalebox{0.85}{%
\Rule{}
{\jdg{}{}{c_1 : \tau_1 \iso \tau_2} \quad \vdash c_2 : \tau_2 \iso \tau_3}
{\jdg{}{}{c_1 \fatsemi c_2 : \tau_1 \iso \tau_3}}
{}
\Rule{}
{\jdg{}{}{c_1 : \tau_1 \iso \tau_2} \quad \vdash c_2 : \tau_3 \iso \tau_4}
{\jdg{}{}{c_1 \oplus c_2 : \tau_1 + \tau_3 \iso \tau_2 + \tau_4}}
{}
\Rule{}
{\jdg{}{}{c_1 : \tau_1 \iso \tau_2} \quad \vdash c_2 : \tau_3 \iso \tau_4}
{\jdg{}{}{c_1 \otimes c_2 : \tau_1 \times \tau_3 \iso \tau_2 \times \tau_4}}
{}
}}




\caption{$\Pi$-terms and combinators.}
\label{pi-terms}
\end{figure*}

\section{Preliminaries: $\Pi$}
\label{sec:pi}
\label{sub:core}

The syntax of the language $\Pi$ consists of several sorts:
\[\begin{array}{lrcl}
\textit{Value types} & \tau &::=& 0 \alt 1 \alt \tau+\tau \alt \tau\times\tau \\
\textit{Values}      & v &::=& \textsf{tt} \alt \inlv{v} \alt \inrv{v} \alt (v,v) \\
\textit{Program types} &&& \tau \leftrightarrow \tau \\
\textit{Programs} & c &::=& (\textrm{See Fig.~\ref{pi-terms}})
\end{array}\]

Focusing on finite types, the building blocks of the type theory are:
the empty type ($0$), the unit type ($1$), the sum type ($+$), and the
product ($\times$) type. One may view each type $\tau$ as a collection
of physical wires that can transmit $|\tau|$ distinct values where
$|\tau|$ is a natural number that indicates the size of a type, computed
as: $| 0 | = 0$; $| 1 | = 1$;
$| \tau_1 + \tau_2 | = | \tau_1 | + | \tau_2 |$; and
$| \tau_1 \times \tau_2 | = | \tau_1 | * | \tau_2 |$.
Thus the type $\mathbb{B} = 1 + 1$ corresponds to a wire that can
transmit one of two values, i.e., bits, with the convention that
$\inlv{\textsf{tt}}$ represents \AgdaFunction{𝔽} 
and $\inrv{\textsf{tt}}$ represents
\AgdaFunction{𝕋}. 
The type $\mathbb{B} \times \mathbb{B} \times \mathbb{B}$ corresponds to a
collection of wires that can transmit three bits. From that
perspective, a type isomorphism between types $\tau_1$ and $\tau_2$ (such that
$|\tau_1|=|\tau_2|=n$) models a \emph{reversible} combinational circuit that
\emph{permutes} the $n$ different values. These type isomorphisms are
collected in Fig.~\ref{pi-terms}. It is known that these type
isomorphisms are sound and complete for all permutations on finite
types~\cite{Fiore:2004,fiore-remarks} and hence that they are
\emph{complete} for expressing combinational
circuits~\cite{fredkin1982conservative,James:2012:IE:2103656.2103667,Toffoli:1980}. Algebraically,
these types and combinators form a \emph{commutative semiring} (up to
type isomorphism). Logically they form a superstructural logic
capturing space-time tradeoffs~\cite{superstructural}. Categorically,
they form a \emph{distributive bimonoidal category}~\cite{laplaza72}.

\label{sub:toff}

Below, we show code, in our Agda formalization, that defines types corresponding to bits (booleans),
two-bits, and three-bits. We then define an operator that
builds a controlled version of a given combinator $c$. This controlled
version takes an additional ``control'' bit and only applies $c$ if
the control bit is true. The code then iterates the control operation
several times starting from boolean negation.

\medskip
\Bexamples{}



Although austere, this combinator-based language
has the advantage of being more amenable to formal analysis for at
least two reasons: (i) it is conceptually simple and small, and (ii)
it has direct and evident connections to type theory and category
theory. Indeed our solution for managing ancillae is inspired by the
construction of \emph{compact closed
  categories}~\cite{kelly1972many}. These categories extend the
monoidal
categories~\cite{benabou1963algebre,benabou1964algebre,maclane1963natural}
which are used to model many resource-aware (e.g., based on linear
types) programming
languages~\cite{benton1994mixed,krishnaswami2015integrating}
(including $\Pi$) with a new type constructor that creates duals or
inverses to existing types. This dual will be our fractional type.

\section{First-Class Garbage Collectors}
\label{sec:dyn}

The main idea is to extend the $\Pi$ terms with two combinators $\eta$
and $\epsilon$ witnessing the isomorphism $A * \oneover{A} = 1$. The
names and types of these operations are inspired by compact closed
categories which are extensions of the monoidal categories that model
$\Pi$. Intuitively, $\eta$~allows one, from ``no information,'' to
create a pair of a value of type $A$ and a value of type
$\oneover{A}$. We interpret the latter value as a GC process
specialized to collect the created value. Dually, $\epsilon$ applies
the GC process to the appropriate value annihilating
both.\footnote{Another interesting interpretation is that these
  operations correspond to creation and annihilation of entangled
  particle/antiparticle pairs in quantum
  physics~\cite{Panangaden2011}.}

To make this idea work, several technical issues need to be
dealt with. Most
notably, we must exclude the empty type from this creation and
annihilation process. Otherwise, we would be able to prove that:
\[\begin{array}{rcl@{\qquad}l}
1 &=& 0 \times 1/0 & \textrm{by~} \eta \\
&=& 0 & \textrm{by~} \absorbr
\end{array}\]
The second important issues is to ensure that the GC process is
specialized to collect a particular value. We therefore exploit ideas
from dependent type theory to treat individual values as singleton
types. More precisely, we extend
the syntax of core $\Pi$ in Sec.~\ref{sub:core} as follows:

\[\begin{array}{lrcl}
\textit{Value types} & \tau &::=& \cdots \alt \oneover{v} \\
\textit{Values}      & v &::=& \cdots \alt \gcv \\
\textit{Program types} &&& \tau \leftrightarrow \tau \\
\textit{Programs} & c &::=& \cdots \alt
   \eta_{v:t} : 1 \leftrightarrow (\tau \times \oneover{v})
   ~~\alt \epsilon_{v:t} : (\tau \times \oneover{v}) \leftrightarrow 1
\end{array}\]

\noindent
For now, the core $\Pi$ language is simply extended with a new type
$\oneover{v}$ which represents a GC process specialized to collect the
value $v$. Since all relevant information is present in the type, at
runtime, this GC process is represented using a trivial value denoted
by $\gcv$. The combinators $\eta$ and $\epsilon$ are parameterized by
the value $v$ (and its type $t$) which serves two purposes. First it
guarantees that the combinators operate on non-empty types, and second
it fixes the type of the GC process. At this point, however, although
the language guarantees that the GC process can only collect a
particular value, the type system does not track the value created by
$\eta$, nor does it predict the value that reaches
$\epsilon$. In other words, it is possible to write programs in which
$\epsilon$ expects one value but is instead applied to another value. In
this section, we will deal with such situations by including a runtime
check in the formal semantics, and show how to remove it,
via a safety proof, in the next
section.

Our Agda formalization clarifies our semantics, with the new type as:
\medskip

\PIFDUdef{}

\noindent The new combinators are defined as follows:
\medskip

\PIFDCombdef{}

\noindent  The most relevant excerpt of the formal semantics is given below:

\medskip
\PIFDinterp{}

The interpreter either returns a proper value $(\Acon{just}~\ldots)$
or throws an exception \Acon{nothing}. The
semantics of the core $\Pi$ combinators performs the appropriate
isomorphism and returns a proper value. At $\eta$, the $v$ that
parameterizes the combinator is used to create a new value $v$ and a
GC process specialized to collect it. By the time evaluation
reaches $\epsilon$, the value created by $\eta$ may have undergone
arbitrary transformations and is not guaranteed to be the value
expected by the GC process. A runtime check is performed: if the value
is the expected one, it is annihilated together with the GC process;
otherwise an exception is thrown which is demonstrated in the following
example which returns normally if given \AgdaFunction{𝔽} and otherwise
throws an exception:
\medskip
\PIFDExample{}

For future reference, we will call this language $\Pi/D$ for the
fractional extension of $\Pi$ with a dynamic check. We illustrate the
expressiveness of the language with two small examples. The Agda code
for the examples is written in a style that reveals the intermediate
steps for expository purposes.

The first circuit has one input and one output. Immediately after
receiving the input, the circuit generates an ancilla wire and its
corresponding GC process (first two steps in the Agda definition). The
original input and the ancilla wire interact using two
\AgdaFunction{CNOT} gates, after which the ancilla wire is redirected
to the output (next three steps in the Agda code). Finally the
original input is GC'ed (last two steps in the Agda code). The entire
circuit is extensionally equivalent to the identity function but it
does highlight an important functionality beyond scoped ancilla
management: the allocated ancilla bit is redirected to the output and
a completely different bit (with the proper default value) is
collected instead.

\begin{minipage}[t]{.4\columnwidth}
\vspace{45pt}%
\PIFDEtaEpsilonExampleone{}
\end{minipage}
\begin{minipage}[t]{.6\columnwidth}
  \vspace{0pt}%
  \begin{tikzpicture}[yscale=0.3,xscale=0.7,every node/.style={scale=0.6}]
	\begin{pgfonlayer}{nodelayer}
		\node [style=none,label=$\mathbb{B}$] (0) at (-7, 25) {};
		\node [style=none] (1) at (-3, 25) {};
		\node [style=none] (2) at (-1, 25) {};
		\node [style=none] (3) at (1, 25) {};
		\node [style=none,label=left:$\mathbb{B}~~$] (4) at (-5, 23) {};
		\node [style=none] (5) at (-4, 23) {};
		\node [style=none] (6) at (-5, 25) {};
		\node [style=none] (7) at (-4, 25) {};
		\node [style=none] (8) at (-3, 23) {};
		\node [style=none] (9) at (-1, 23) {};
		\node [style=none] (10) at (-6, 22) {};
		\node [style=none,label=below:$\oneover{\mathbb{B}}$] (11) at (-5, 21) {};
		\node [style=none] (13) at (0, 22) {};
		\node [style=none] (14) at (-1, 21) {};
		\node [style=control,scale=0.5] (15) at (-5, 25) {};
		\node [style=control,scale=0.5] (16) at (-4, 23) {};
		\node [style=none] (18) at (-5, 23) {$\oplus$};
		\node [style=none] (19) at (-4, 25) {$\oplus$};
	\end{pgfonlayer}
	\begin{pgfonlayer}{edgelayer}
		\draw (0.center) to (1.center);
		\draw (1.center) to (9.center);
		\draw (4.center) to (8.center);
		\draw (8.center) to (2.center);
		\draw (2.center) to (3.center);
		\draw (11.center) to (14.center);
		\draw (6.center) to (4.center);
		\draw (7.center) to (5.center);
		\draw [bend right=45, looseness=2.00] (4.center) to (11.center);
		\draw [bend left=45, looseness=1.75] (9.center) to (14.center);
	\end{pgfonlayer}
\end{tikzpicture}
\end{minipage}

The second example illustrates the manipulation of GC processes. A
process for collecting a pair of values can be decomposed into two
processes each collecting one of the values (and vice-versa):

\medskip
\PIFDrevprod{}

\section{Dependently-Typed Garbage Collectors}
\label{sec:dep}

By lifting the scoping restriction, the development in the previous
sections is already more general than the state of the art in ancilla
management.  It still shares the same limitation of needing a
runtime check to ensure ancillae values are properly restored to their
allocation
value~\cite{10.1007/978-3-319-20860-2_13,Green:2013:QSQ:2491956.2462177}.
We now address this limitation using a combination of
pointed types, singleton types, monads, and comonads.

\subsection{Lifting Evaluation to the Type System}

Before giving all the (rather involved) technical details, we
highlight the main idea using the toy language below:

\medskip
\Texample{}

The toy language has two types (natural numbers and booleans) and two
functions \AgdaInductiveConstructor{square} and
\AgdaInductiveConstructor{isZero} and their compositions. Say we
wanted to prove that \AgdaInductiveConstructor{compose isZero square}
always returns \AgdaInductiveConstructor{false} when applied to a
non-zero natural number. We can certainly do this proof in Agda (i.e.,
in the meta-language of our formalization) but we would like to do the
proof within the toy language itself. The most important reason is
that it can then be used within the language to optimize programs (or,
for the case of $\Pi/D$, to remove a runtime check).

The strategy we adopt is to create a lifted version of the toy
language with \emph{pointed types}~\cite{hottbook}, i.e., types paired with
a value of the type. In the lifted language, the evaluation
function has an interesting type: it keeps track of the result of
evaluation within the type:

\medskip
\Texamplecont{}

This allows various properties of
\AgdaInductiveConstructor{compose isZero square} to be derived within
the extended type system. For example:

\medskip
\Texampletest{}

The first two tests show that the type system can track exact concrete
values. More interestingly, \AgdaFunction{test3} shows a
property that holds for all natural numbers~$n$; its proof
uses ``symbolic'' evaluation within the type system. In
more detail, from the definition of \Afun{eval}, we see that
\Afun{eval}~\Acon{square}~(\Acon{suc}~\Avar{n}) produces
(\Acon{suc}~\Avar{n}) * (\Acon{suc}~\Avar{n}); by definition of
multiplication, this is an expression with a leading \Acon{suc} constructor
which is enough to determine that evaluating \Acon{isZero} on it
yields \Acon{false}. This form of partial evaluation is
quite expressive, and sufficient to allow to keep track of ancilla values
throughout complex programs.

\subsection{Pointed and Singleton Types: $\Pi/\bullet$}

We now use the above idea to create a version of the $\Pi$
language, which we call $\Pi/\bullet$,
in which all types are pointed, i.e.,
for each type $t$ some value $v$ of type $t$ is ``in focus''
$t$\Acon{\#}$v$. As the goal of the language is to keep track of
fractional types, it is sufficient to inherit the multiplicative
structure of $\Pi$. We also need a special kind of pointed type that
includes just one value, a singleton type. The singleton types will
allow the type system to track the flow of one particular value (the
ancilla value), which is exactly what is needed to prove the safety of
deallocation. We present the relevant definitions from our
formalization and explain each:

\medskip
\PIPFUdef{}

Given a set $A$ with an element $v$, the singleton set containing $v$ is
the subset of $A$ whose elements are equal to $v$. In Agda's type
theory, this is encoded using the \AgdaFunction{Singleton} type. For a
given type $A$, and a value $v$ of type $A$, the type
\AgdaFunction{Singleton}~$A$~$v$ is inhabited by a choice of point
$\bullet$ in $A$, along with a proof that $v$ is equal to $\bullet$. In
other words, it is possible to refer to a singleton value $v$ using
several distinct syntactic expressions that all evaluate to $v$. Put
differently, any claim that a value belongs to the singleton type must
come with a proof that this value is equal to $v$. The reciprocal type
\AgdaFunction{Recip}~$A$~$v$ consumes exactly this singleton value. The
universe of pointed types \AgdaDatatype{∙𝕌} contains plain $\Pi$ types
together with a selection of a value in focus; products of pointed
types; singleton types; and reciprocal types. Note that the actual value
in focus for reciprocals, i.e., the runtime value of a GC process, is a
function that disregards its argument returning the constant value of
the unit type. As we show, this is safe, as the type system prevents the
GC process being applied to anything but the particular singleton value
in question.

The combinators in the lifted language $\Pi/\bullet$ consist of all
the combinators in the core $\Pi$ language together with their
multiplicative structure. The types for \Acon{$\eta$} and
\Acon{$\epsilon$} are now specialized to guarantee safety of
de-allocation as follows. When applying \Acon{$\eta$} at a pointed
type, the current witness value is put in focus in a singleton type
and a GC process for that particular singleton type is created. To
apply this process using \Acon{$\epsilon$} the very same singleton
value must be the current one.

\medskip
\PIPFCombDef{}

The mediation between general pointed types and singleton types is
done via \Acon{return} and \Acon{extract}, which form a dual
monad/comonad pair, from which many structural properties can be
derived: specifically a pair of singleton types is a singleton of the
pair of underlying types, and a singleton of a singleton is the same
singleton.

\begin{proposition}
  $\llparenthesis \cdot \rrparenthesis$ is both an idempotent strong
  monad and an idempotent costrong comonad over pointed types.
\end{proposition}
\begin{proof}
  The main insight needed is to define the functor \Afun{∙Singᵤ}, the
  \Afun{tensor}/\Afun{cotensor}, and the \Afun{join}/\Afun{cojoin}
  (\Afun{duplicate}):
  \medskip
  \PIPFCombderive{}
\end{proof}

Like for the toy language, evaluation is reflected in the type
system, and in this case we have the additional property that
evaluation is reversible:

\medskip
\PIPFeval{}
\PIPFrev{}

The type of evaluation now states that given a combinator mapping
pointed type $T_1$ to pointed type $T_2$ where $T_i$ consists of an
underlying type $t_i$ and value $v_i$, evaluation succeeds if applying
the combinator to $v_1$ produces $v_2$. In other words, the result of
evaluation is completely determined by the type system:
\medskip
\PIPFExample{}

To summarize, if a combinator expects a singleton
type, then it would only typecheck in the lifted language, if
it is given the unique value it expects.
A particularly intriguing instance of that
situation is the following program:

\begin{minipage}[t]{.6\linewidth}
  \PIPFrevrev{}
\end{minipage}%
\begin{minipage}[t]{.4\linewidth}
  \vspace{0pt}%
  \begin{tikzpicture}[scale=0.6,every node/.style={scale=0.6}]
	\begin{pgfonlayer}{nodelayer}
		\node [style=none, label=1/(1/A)] (0) at (-5, 2) {};
		\node [style=none] (2) at (-4, 3.5) {};
		\node [style=none, label=$\llparenthesis A\rrparenthesis$] (3) at (-3, 4) {};
		\node [style=none, label=below:1/A] (4) at (-3, 3) {};
		\node [style=none] (5) at (-1, 4) {};
		\node [style=none] (7) at (-1, 2) {};
		\node [style=none, label=$\llparenthesis 1/A \rrparenthesis$] (8) at (-1, 3) {};
		\node [style=none] (9) at (0, 2.5) {};
		\node [style=none] (10) at (1, 4) {};
	\end{pgfonlayer}
	\begin{pgfonlayer}{edgelayer}
		\draw (0.center) to (7.center);
		\draw [bend right=75, looseness=3.00] (3.center) to (4.center);
		\draw [bend left=75, looseness=3.50] (8.center) to (7.center);
		\draw (3.center) to (10.center);
		\draw (4.center) to (8.center);
	\end{pgfonlayer}
\end{tikzpicture}
\end{minipage}%

The program takes a value of type
\Acon{∙𝟙/}~(\Acon{∙𝟙/}~\Avar{A}). This would be a GC process
specialized to collect another GC process! By collecting this process,
the corresponding singleton value is ``rematerialized.'' At runtime,
there would be no information other than the functions that ignore
their argument but the type system provides enough guarantees to
ensure that this process is well-defined and safe.

\subsection{Extraction of Safe Programs}
\label{sec:ext}
By lifting programs and their evaluation to the type level, we can
naturally leverage the typechecking process to verify properties of
interest, including the safe de-allocation of ancillae. One ``could''
just forget about $\Pi/D$ and instead use $\Pi/\bullet$ as \emph{the}
programming language for ancilla management. Indeed the dual nature of
proofs and programs is more and more exploited in languages like the
one used to formalize this paper (Agda).

However, it is also often the case than constructive proofs are
further processed to extract native efficient programs that eschew
the overhead of maintaining information needed just for proof
invariants. In our case, the question is whether we can extract from
a $\Pi/\bullet$ program, a program in $\Pi/D$ that uses a simpler type
system, a simpler runtime representation, and yet is guaranteed to be
safe and hence can run without the runtime checks associated with
de-allocation sites. In this section, we show that this indeed the
case.


We demonstrate this by constructing an extraction map from the syntax of
$\Pi/\bullet$ to $\Pi/D$. This is fully implemented in the underlying
Agda formalization, but we present the most significant
highlights. There are three important functions whose signatures are
given below:

\medskip
\EXTsig{}

\noindent The function \Afun{Ext𝕌} maps a $\Pi/\bullet$ type to a
$\Pi/D$ type and a value in the type. The function \Afun{Ext⧟} maps a
$\Pi/\bullet$ combinator to a $\Pi/D$ combinator, whose types are fixed
by \Afun{Ext𝕌}. And finally, the function \Afun{Ext≡} asserts that the
extracted code cannot throw an exception (it must return a
\Acon{just}~value).

Each of these functions has one or two enlightening cases which we
explain below. In $\Pi/D$ the fractional type expresses that it expects
a particular value but lacks any mechanisms to enforce this
requirement. Thus we have no choice when mapping a fractional type from
$\Pi/\bullet$ to $\Pi/D$ but to use the \Acon{𝟙/}~\Avar{v}~type with the
trivial value:

\medskip
\EXTu{}


When mapping $\Pi/\bullet$ combinators to $\Pi/D$ combinators, the main
interesting cases are for \Acon{η} and \Acon{ε}. In each of those, we
use the values from the pointed type as choices for the ancilla value,
and the expectation for the GC process respectively:

\medskip
\EXTcomb{}



Finally we can prove the correctness of extraction. The punchline is in
the following case:

\medskip
\EXTeq{}

\noindent Here, the singleton type in $\Pi/\bullet$ guarantees that the
runtime check cannot fail!

\subsection{Example}

This new language not only allows us to \emph{verify} circuits but
also allows us to merge verification with programming. To clarify this
idea, we show how to implement a 4-bit Toffoli gate using proper
ancilla management while at the same time proving its correctness.

We start with verification of the Toffoli gate implementation we have in Sec.~\ref{sec:pi}
in $\Pi/\bullet$ using pattern matching:
\medskip
\PToffoli{}

\noindent
Since we use the same implementation in all the cases so it does not matter which value we use to
instantiate extraction:

\medskip
\PToffoliExtEq{}

\noindent Using this as building block we can use Toffoli's construction~\cite{Toffoli:1980} to construct 4-bit
Toffoli gate using an additional ancilla bit:
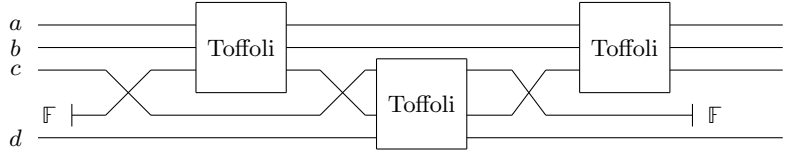
\begin{center}
\begin{tikzpicture}[scale=0.6, font={\small}]
	\begin{pgfonlayer}{nodelayer}
		\node [style=none] (0) at (-28.5, 25.5) {};
		\node [style=none] (1) at (-25, 25.5) {};
		\node [style=none] (2) at (-25, 26) {};
		\node [style=none] (3) at (-25, 24) {};
		\node [style=none] (4) at (-23, 24) {};
		\node [style=none] (5) at (-23, 26) {};
		\node [style=none] (6) at (-23, 25.5) {};
		\node [style=none] (7) at (-16.5, 25.5) {};
		\node [style=none] (8) at (-16.5, 26) {};
		\node [style=none] (9) at (-14.5, 26) {};
		\node [style=none] (10) at (-14.5, 25.5) {};
		\node [style=none] (11) at (-16.5, 24) {};
		\node [style=none] (12) at (-14.5, 24) {};
		\node [style=none] (18) at (-29, 25.5) {$a$};
		\node [style=none] (19) at (-28.5, 25) {};
		\node [style=none] (20) at (-28.5, 24.5) {};
		\node [style=none] (21) at (-25, 24.5) {};
		\node [style=none] (22) at (-25, 25) {};
		\node [style=none] (23) at (-27, 24.5) {};
		\node [style=none] (24) at (-29, 25) {$b$};
		\node [style=none] (26) at (-29, 24.5) {$c$};
		\node [style=none] (27) at (-27.75, 23.5) {};
		\node [style=none] (28) at (-27, 23.5) {};
		\node [style=none] (29) at (-27.75, 23.75) {};
		\node [style=none] (30) at (-27.75, 23.25) {};
		\node [style=none] (31) at (-28.25, 23.5) {$\mathbb{F}$};
		\node [style=none] (32) at (-26, 23.5) {};
		\node [style=none] (33) at (-26, 24.5) {};
		\node [style=none] (34) at (-24, 25) {Toffoli};
		\node [style=none] (36) at (-29, 23) {$d$};
		\node [style=none] (37) at (-28.5, 23) {};
		\node [style=none] (38) at (-25, 23) {};
		\node [style=none] (39) at (-23, 25) {};
		\node [style=none] (40) at (-23, 24.5) {};
		\node [style=none] (41) at (-16.5, 25) {};
		\node [style=none] (45) at (-22.25, 24.5) {};
		\node [style=none] (46) at (-22.25, 23.5) {};
		\node [style=none] (47) at (-21.25, 23.5) {};
		\node [style=none] (48) at (-21.25, 24.5) {};
		\node [style=none] (49) at (-21, 24.75) {};
		\node [style=none] (50) at (-21, 22.75) {};
		\node [style=none] (51) at (-19, 22.75) {};
		\node [style=none] (52) at (-19, 24.5) {};
		\node [style=none] (53) at (-19, 23.5) {};
		\node [style=none] (54) at (-19, 24.75) {};
		\node [style=none] (55) at (-21, 23) {};
		\node [style=none] (56) at (-18, 23.5) {};
		\node [style=none] (57) at (-18, 22.75) {};
		\node [style=none] (58) at (-16.5, 24.5) {};
		\node [style=none] (59) at (-18, 24.5) {};
		\node [style=none] (60) at (-19, 23) {};
		\node [style=none] (61) at (-12, 23) {};
		\node [style=none] (62) at (-17.25, 24.5) {};
		\node [style=none] (63) at (-17.25, 23.5) {};
		\node [style=none] (64) at (-14, 23.5) {};
		\node [style=none] (65) at (-14, 23.75) {};
		\node [style=none] (66) at (-14, 23.25) {};
		\node [style=none] (67) at (-13.75, 23.5) {};
		\node [style=none] (68) at (-13.5, 23.5) {$\mathbb{F}$};
		\node [style=none] (69) at (-12, 25.5) {};
		\node [style=none] (70) at (-12, 25) {};
		\node [style=none] (71) at (-14.5, 25) {};
		\node [style=none] (73) at (-21, 24.5) {};
		\node [style=none] (74) at (-21, 23.5) {};
		\node [style=none] (128) at (-14.5, 24.5) {};
		\node [style=none] (129) at (-12, 24.5) {};
		\node [style=none] (130) at (-15.5, 25) {Toffoli};
		\node [style=none] (131) at (-20, 23.75) {Toffoli};
	\end{pgfonlayer}
	\begin{pgfonlayer}{edgelayer}
		\draw (0.center) to (1.center);
		\draw (2.center) to (1.center);
		\draw (1.center) to (3.center);
		\draw (3.center) to (4.center);
		\draw (4.center) to (5.center);
		\draw (5.center) to (2.center);
		\draw (6.center) to (7.center);
		\draw (8.center) to (11.center);
		\draw (8.center) to (9.center);
		\draw (9.center) to (12.center);
		\draw (11.center) to (12.center);
		\draw (19.center) to (22.center);
		\draw (20.center) to (23.center);
		\draw (29.center) to (30.center);
		\draw (27.center) to (28.center);
		\draw (23.center) to (32.center);
		\draw (28.center) to (33.center);
		\draw (33.center) to (21.center);
		\draw (37.center) to (38.center);
		\draw (39.center) to (41.center);
		\draw (40.center) to (45.center);
		\draw (32.center) to (46.center);
		\draw (45.center) to (47.center);
		\draw (46.center) to (48.center);
		\draw (49.center) to (50.center);
		\draw (38.center) to (55.center);
		\draw (50.center) to (51.center);
		\draw (49.center) to (54.center);
		\draw (54.center) to (51.center);
		\draw (52.center) to (59.center);
		\draw (53.center) to (56.center);
		\draw (60.center) to (61.center);
		\draw (59.center) to (63.center);
		\draw (56.center) to (62.center);
		\draw (62.center) to (58.center);
		\draw (65.center) to (66.center);
		\draw (63.center) to (64.center);
		\draw (10.center) to (69.center);
		\draw (71.center) to (70.center);
		\draw (48.center) to (73.center);
		\draw (47.center) to (74.center);
		\draw (128.center) to (129.center);
	\end{pgfonlayer}
\end{tikzpicture}
\end{center}

The code is written in a conventional $\Pi/D$ style except for the
pervasive lifting to pointed types:

\medskip
\PToffolifour{}

\noindent
With this construction however, we can verify that the circuit
satisfies the specification of 4-bit Toffoli gate and the ancilla bit
is correctly garbage collected without pattern matching.
And using the extraction mechanism, we obtain a
fully verified 4-bit Toffoli gate in $\Pi/D$:

\medskip
\PToffoliExt{}

\noindent Note that, as the type has shown our implementation is independent of
any input so it does not matter which value we use to instantiate the
extraction:

\medskip
\PToffoliExtTest{}

\section{Conclusion}

We have introduced, in the context of reversible languages,
the concept of fractional types as descriptions of specialized GC
processes. Although the basic idea is rather simple and intuitive, the
technical details needed to reason about individual values are
somewhat intricate. The use of fractional types, however, enables a
complete elegant type-based solution to the management of ancilla values in
reversible programming languages.

\bibliographystyle{splncs04}
\bibliography{cites}

\end{document}